\documentclass{amsart}
\usepackage[a4paper,bottom=2.5cm]{geometry}

\makeatletter
\def\@seccntformat#1{%
  \protect\textup{\protect\@secnumfont
    \ifnum\pdfstrcmp{subsection}{#1}=0 \bfseries\fi
    \ifnum\pdfstrcmp{subsubsection}{#1}=0 \itshape\fi
    \csname the#1\endcsname
    \protect\@secnumpunct
  }%
}
\def\part{\@startsection{part}{0}%
  \z@{\linespacing\@plus\linespacing}{.5\linespacing}%
  {\normalfont\bfseries\raggedright}}
\def\section{\@startsection{section}{1}%
  \z@{.8\linespacing\@plus\linespacing}{.5\linespacing}%
  {\normalfont\scshape\centering}}
\def\subsection{\vspace{0.8\linespacing}\@startsection{subsection}{2}%
  \z@{-0.01\linespacing}{0.3\linespacing}%
  {\normalfont\bfseries\S\ }}
\def\subsubsection{\@startsection{subsubsection}{3}%
  \z@{.8\linespacing\@plus.7\linespacing}{-.5em}%
  {\normalfont\itshape}}
\def\paragraph{\@startsection{paragraph}{4}%
  \z@{.8\linespacing}{-\fontdimen2\font}%
  {\normalfont\itshape}}
\makeatother

\makeatletter
\newcommand{\leqnomode}{\tagsleft@true\let\veqno\@@leqno}
\newcommand{\reqnomode}{\tagsleft@false\let\veqno\@@eqno}
\makeatother

\usepackage[T1,T2A]{fontenc}
\usepackage[utf8]{inputenc}
\usepackage{microtype}
\usepackage{newunicodechar}
\usepackage[greek,russian,english]{babel}
\usepackage[babel]{csquotes}

\usepackage{etoolbox}
\usepackage{enumitem}

\usepackage[dvipsnames]{xcolor}
\usepackage{tikz}
\usetikzlibrary{3d,%
  calc,%
  math,%
  fadings,%
  arrows.meta,%
  patterns,patterns.meta,%
  decorations.pathmorphing}
\usepackage{caption}
\newcommand\subcaption[1]{\captionsetup{skip=1pt,font=footnotesize}\caption*{#1}}

\usepackage[pdftex,hypertexnames=false,
colorlinks,%
citecolor=black!70!blue,%
linkcolor=black!70!blue,%
urlcolor={black!70!blue}]{hyperref}
\usepackage[nameinlink,nosort]{cleveref}
\usepackage{crossreftools}

\usepackage[backend=biber,%
  style=alphabetic,%
  autolang=other]{biblatex}
\DeclareFieldFormat{hal}{%
  \mkbibacro{HAL}\addcolon\space
  \ifhyperref{\href{https://hal.science/#1}{\nolinkurl{#1}}}{\nolinkurl{#1}}
}
\DeclareFieldAlias{eprint:hal}{hal}
\DeclareFieldAlias{eprint:HAL}{eprint:hal}
\AtBeginBibliography{\footnotesize}
\usepackage{stmaryrd}
\SetSymbolFont{stmry}{bold}{U}{stmry}{m}{n}
\makeatletter

\renewcommand*{\@cdots}{%
  \mathinner{{\cdotp}{\cdotp}{\cdotp}}}
\makeatother
\usepackage{amsmath}
\usepackage{amssymb}
\usepackage[old]{old-arrows} 
\usepackage{bm}
\usepackage{keytheorems}
\keytheoremset{predefined={parent=section}}
\providekeytheorem{claim}[style=plain]
\providekeytheorem{claim*}[style=plain,numbered=false]
\providekeytheorem{remark*}[name=Remark,style=remark,numbered=false]
\providekeytheorem{comment*}[name=Comment,style=remark,numbered=false]
\NewDocumentEnvironment{details}{O{Details}}%
{\begingroup\color{black!75}\fontsize{9pt}{10.5pt}\selectfont%
  \begin{proof}[#1]}%
  {\end{proof}\endgroup}
\NewDocumentEnvironment{numerics}{O{Numerics}}%
{\begingroup\color{black!75}\fontsize{9pt}{10.5pt}\selectfont%
  \begin{proof}[#1]}%
  {\end{proof}\endgroup}

\usepackage{tcolorbox}
\tcbuselibrary{skins}
\tcbuselibrary{breakable}
\tcbuselibrary{theorems}
\tcolorboxenvironment{details}{%
  blanker,breakable,left=1em,
  before skip=10pt,after skip=10pt,
  borderline west={2pt}{0pt}{gray},
  use color stack=true}
\tcolorboxenvironment{numerics}{%
  blanker,breakable,left=1em,
  before skip=10pt,after skip=10pt,
  borderline west={2pt}{0pt}{gray},
  use color stack=true}
\newtcolorbox{tbox}{%
  enhanced,colframe=black!80,boxrule=0.4pt,colback=white,boxsep=0.8\linespacing,left=0pt,right=0pt,top=0pt,bottom=0pt,before={\noindent}}

\colorlet{TilingGrid}{black!15!white}
\definecolor{RainbowA}{HTML}{ec111a}
\definecolor{RainbowB}{HTML}{fb6330}
\definecolor{RainbowC}{HTML}{eed42f}
\definecolor{RainbowD}{HTML}{3ea908}
\definecolor{RainbowE}{HTML}{009dd6}
\definecolor{RainbowF}{HTML}{7849b8}
\definecolor{RainbowG}{HTML}{da70d6}

\colorlet{LevelI}{RainbowC!50!RainbowB}
\colorlet{LevelI+1}{RainbowD}
\colorlet{LevelK}{RainbowE}
\colorlet{LevelINext}{RainbowF}
\colorlet{LevelKLast}{RainbowG!50!RainbowA}
\colorlet{LevelKNext}{black}

\NewDocumentCommand{\TODO}{O{}}{\textcolor{red}{\bfseries [TODO\ifstrempty{#1}{}{ #1}]}}
\NewDocumentCommand{\N}{}{\mathbb{N}}
\NewDocumentCommand{\Z}{}{\mathbb{Z}}
\NewDocumentCommand{\Q}{}{\mathbb{Q}}
\NewDocumentCommand{\R}{}{\mathbb{R}}
\NewDocumentCommand{\bigO}{}{\mathcal{O}}
\NewDocumentCommand{\pcolon}{}{\mspace{3mu}{\scriptstyle\subseteq}\mspace{-2mu}\colon}
\NewDocumentCommand{\mto}{}{\rightrightarrows}
\makeatletter
\newcommand{\mapsmapsto}{\mathpalette\@mapsmapsto\relax}
\newcommand*{\@mapsmapsto}[2]{%
   \dimen@\fontdimen8
       \ifx#1\displaystyle\textfont\else
       \ifx#1\textstyle\textfont\else
       \ifx#1\scriptstyle\scriptfont\else
       \scriptscriptfont\fi\fi\fi 3
   \mathrel{%
      \vcenter{%
         \vbox{%
            \baselineskip\z@skip
            \lineskip\z@
            \ialign{##\cr
              \noalign{\kern 3\dimen@}\cr
              $#1\mapstochar\varrightarrow$\cr
              \noalign{\kern\dimen@}%
            $#1\mapstochar\varrightarrow$\cr}%
         }%
      }%
   }%
}
\makeatother

\usepackage{accents}
\newcommand{\ubar}[1]{\underaccent{\bar}{#1}}

\RenewDocumentCommand{\vec}{m}{\bm{#1}}
\NewDocumentCommand{\pxvec}{m}{\vec{#1}'}
\NewDocumentCommand{\basis}{m}{\vec{e}^{#1}}
\NewDocumentCommand{\interval}{mm}{\{#1,\ldots,#2\}}
\NewDocumentCommand{\ibox}{m}{\llbracket #1 \rrbracket}
\NewDocumentCommand{\restr}{mm}{#1|_{#2}}
\NewDocumentCommand{\norm}{m}{\lVert #1 \rVert}
\DeclareMathOperator{\len}{len}
\DeclareMathOperator{\polylog}{polylog}
\makeatletter
\newcommand{\bigcomp}{%
  \DOTSB
  \mathop{\vphantom{\sum}\mathpalette\bigcomp@\relax}%
  \slimits@
}
\newcommand{\bigcomp@}[2]{%
  \begingroup\m@th
  \sbox\z@{$#1\sum$}%
  \setlength{\unitlength}{0.9\dimexpr\ht\z@+\dp\z@}%
  \vcenter{\hbox{%
    \begin{picture}(1,1)
    \bigcomp@linethickness{#1}
    \put(0.5,0.5){\circle{1}}
    \end{picture}%
  }}%
  \endgroup
}
\newcommand{\bigcomp@linethickness}[1]{%
  \linethickness{%
      \ifx#1\displaystyle 2\fontdimen8\textfont\else
      \ifx#1\textstyle 1.65\fontdimen8\textfont\else
      \ifx#1\scriptstyle 1.65\fontdimen8\scriptfont\else
      1.65\fontdimen8\scriptscriptfont\fi\fi\fi 3
  }%
}
\makeatother
\NewDocumentCommand{\alphabet}{O{A}}{\mathcal{#1}}
\NewDocumentCommand{\blank}{}{\mathtt{B}}
\NewDocumentCommand{\dom}{O{}}{\mathrm{dom}_{#1}}
\NewDocumentCommand{\subpattern}{}{\sqsubseteq}
\RenewDocumentCommand{\square}{O{}}{\begin{tikzpicture}[x=1.4ex,y=1.4ex,baseline=.15ex]
    \ifstrempty{#1}
    {\draw[black!30] (0,0) rectangle (1,1);}
    {\draw[black!30, fill=#1] (0,0) rectangle(1,1);}
  \end{tikzpicture}}
\NewDocumentCommand{\wang}{m}{{#1}}
\NewDocumentCommand{\Swang}{O{}O{T}}{#2_{#1}}
\NewDocumentCommand{\Slabel}{O{}O{\alphabet[L]}}{#2_{#1}}
\NewDocumentCommand{\Scolor}{O{}O{\alphabet[C]}}{#2_{#1}}

\NewDocumentCommand{\lift}{m}{{#1}^{\uparrow}}
\NewDocumentCommand{\decsymbol}{}{
  \begin{tikzpicture}[baseline=0ex,x=0.5em,y=0.5em]
      \begin{scope}[rotate=45]
        \fill[rounded corners=0.5pt, even odd rule] (0,1) -- ++(0.3,-0.3) -- ++(0,-0.7) -- ++(-0.6,0) -- ++(0,0.7) -- cycle (0,0.7) circle (0.08);
      \end{scope}
    \end{tikzpicture}\mspace{-1mu}}
  \NewDocumentCommand{\dec}{}{\smash{f_{\mspace{-1mu}\decsymbol}}}
\NewDocumentCommand{\wire}{}{\mspace{-2mu}\begin{tikzpicture}[x=0.5em,y=0.5em]
    \draw[thick] (0,0.6) -- ++(1,0);
    \draw[thick] (0,0) -- ++(0.2,0) -- ++(0.25,0.25) -- ++(0.55,0);
  \end{tikzpicture}\mspace{-2mu}}
\NewDocumentCommand{\route}{}{\mspace{-2mu}\begin{tikzpicture}[x=0.5em,y=0.5em]
    \draw[thick] (0,0) -- ++(0.2,0) -- ++(0.5,0.6) -- ++(0.25,0);
    \draw[thick] (0,0.6) -- ++(0.2,0) -- ++(0.5,-0.6) -- ++(0.25,0);
  \end{tikzpicture}\mspace{-2mu}}
\NewDocumentCommand{\Srect}{}{\mathfrak{R}}
\NewDocumentCommand{\rect}{m}{\mathfrak{#1}}
\NewDocumentCommand{\pxrect}{m}{\tilde{\mathfrak{#1}}}
\NewDocumentCommand{\rborder}{m}{\partial\mspace{1mu}#1}
\NewDocumentCommand{\rvol}{m}{\lVert #1\rVert}
\NewDocumentCommand{\rtime}{}{\lambda}
\NewDocumentCommand{\rspace}{}{\mu}
\NewDocumentCommand{\facet}{mm}{f\ifstrempty{#1}{}{_{#1}}\ifstrempty{#2}{}{^{\scriptscriptstyle #2}}}
\NewDocumentCommand{\sfacet}{}{\longS}
\NewDocumentCommand{\rnorm}{m}{[#1]}
\DeclareTextSymbol{\textlongs}{TS1}{"73}
\DeclareTextSymbolDefault{\textlongs}{TS1}
\NewDocumentCommand{\longS}{}{\textit{\fontfamily{lmr}\textlongs\normalfont}\mspace{3mu}}
\NewDocumentCommand{\Sgrid}{}{\mathfrak{G}}
\DeclareDocumentCommand{\grid}{m}{\mathfrak{#1}}
\NewDocumentCommand{\pxgrid}{m}{\tilde{\mathfrak{#1}}}
\NewDocumentCommand{\subst}{O{\tau}}{#1}
\NewDocumentCommand{\gsubst}{O{\tau}}{\subst[#1]}
\NewDocumentCommand{\Ssubst}{}{\mathcal{S}}
\NewDocumentCommand{\ndill}{O{\tau}}{#1}
\NewDocumentCommand{\Sndill}{}{\mathcal{D}}
\NewDocumentCommand{\substep}{O{\tau}O{}}{\ifstrempty{#2}
  {\xrightarrow{#1}}
  {\xrightarrow[{\raisebox{0.5ex}[0.5ex]{$\scriptstyle #2$}}]{#1}}}
\NewDocumentCommand{\rsubstep}{O{\tau}O{}}{\ifstrempty{#2}
  {\xleftarrow{#1}}
  {\xleftarrow[{\raisebox{0.5ex}[0.5ex]{$\scriptstyle #2$}}]{#1}}}
\NewDocumentCommand{\limspace}{mO{}O{}}{\smash{\overleftarrow{#1}\ifstrempty{#2}{}{_{\!#2}}\ifstrempty{#3}{}{^{\raisebox{0.7ex}{$\scriptscriptstyle #3$}}}}}
\NewDocumentCommand{\odnone}{}{{\lozenge}}
\NewDocumentCommand{\K}{O{}O{}}{K%
  \ifstrempty{#2}
  {}
  {_{\mathrm{#2}}}}%
\NewDocumentCommand{\pxspace}{mO{}}{\smash{\tilde{\mu}_{#1}\ifstrempty{#2}{}{^{\mspace{2.5mu}#2}}}}
\NewDocumentCommand{\nghb}{}{\mathcal{N}}
\NewDocumentCommand{\field}{m}{{\boldmath\textbf{#1}}}
\NewDocumentCommand{\mwang}{mO{}}{\smash{\wang{\tilde{#1}\ifstrempty{#2}{}{^{(#2)}}}}}
\NewDocumentCommand{\compzone}{O{}}{\tilde{\mathfrak{c}}_{#1}}

\NewDocumentCommand{\ram}{mO{}O{}}{\mathtt{#1}\ifstrempty{#2}{}{\mspace{-1mu}\textnormal{\texttt{[}}#2\textnormal{\texttt{]}}}\ifstrempty{#3}{}{\mspace{-1mu}\textnormal{\texttt{[}}#3\textnormal{\texttt{]}}}}
\NewDocumentCommand{\funram}{m}{\varphi_{#1}}
\NewDocumentCommand{\halt}{O{}}{{\downarrow\ifstrempty{#1}{}{_{[#1]}}}}

\NewDocumentCommand{\code}{m}{\langle #1 \rangle}

\newcommand{\card}[1]{\left|#1\right|}
\newcommand{\sett}[3][]{\left\{#2\ifstrempty{#1}{}{\in #1}\colon #3\right\}}

\newcommand\X{X}
\newcommand\freq{p}

\title[Pimp my fixpoint: sofic realization of substitution-based shift spaces]{Pimp my fixpoint: sofic realization of multidimensional substitution-based shift spaces}

\author{Antonin Callard}
\address{ENS de Lyon, UCBL, CNRS, LIP UMR 5668, Lyon, France}
\curraddr{Université de Montpellier, CNRS, LIRMM UMR 5506, Montpellier, France}
\email{contact@acallard.net}

\author{Pierre Guillon}
\address{CNRS, Aix-Marseille Université, I2M UMR 7373, Marseille, France}
\email{pguillon@math.cnrs.fr}

\author{Léo Paviet Salomon}
\address{ENS de Lyon, UCBL, CNRS, LIP UMR 5668, Lyon, France}
\email{mail@lpaviets.org}

\author{Pascal Vanier}
\address{Université Caen Normandie, ENSICAEN, CNRS, Normandie Univ, GREYC UMR 6072, Caen, France}
\email{pascal.vanier@unicaen.fr}

\subjclass[2020]{Primary 37B10; Secondary 68Q09}

\thanks{No AI tool was used in the research and creation of this article.}

\begin{document}

\begin{abstract}
  In symbolic dynamics, the fixed-point construction from~\cite{Durand-Romashchenko-Shen_2012:fixed-point-tilesets-and-applications} defines shift spaces of finite type whose configurations embed infinite hierarchies of tilings. This article provides a ``black box'' abstraction of this method phrased in terms of substitutions and $S$-adic limit spaces operating over sequences of increasingly large alphabets. By quantifying the amount of information computed by the substitutions at each level, and using a suitable parallel model of computation, we provide a simple positive criterion of multidimensional soficity that generalizes classical examples from the literature.
\end{abstract}

\maketitle

%
\section{Introduction}
Shift spaces are sets of colorings (or ``configurations'') of a discrete space (usually $\Z^d$) with a finite alphabet of colors $\alphabet$ that abide to some family of forbidden patterns. Most commonly, this family is composed of finitely many patterns, and thus defines a shift of finite type (SFT). Perhaps surprisingly, these shifts may already exhibit complex behaviors: it is undecidable whether they are empty~\cite{Berger_1964:phd:undecidability-domino-problem}, they may have only non-computable configurations~\cite{Hanf_1974:nonrecursive-tilings-of-the-plane-1,Myers_1974:nonrecursive-tilings-of-the-plane-2}, they may contain only maximally complex configurations in the sense of Kolmogorov complexity~\cite{Durand-Levin-Shen_2008:complex-tilings}\dots

On the other side of the complexity spectrum are the effective shifts, whose forbidden constraints can be enumerated algorithmically. Classically, the most considered intermediate class of shifts is the class of \emph{sofic shifts}, which are defined as the cellwise recolorings of the shifts of finite type (their \emph{covers}). While they still admit a finite description, they form a superclass of SFTs that exhibit strictly more complex behaviors: for instance, the shift on the alphabet $\{\square,\square[black]\}$ whose tilings contain at most a single $\square[black]$ cell is sofic but not of finite type. In fact, sofic shifts can manifest computationally complex properties: as illustrated by a landmark result of M.~Hochman~\cite{Hochman_2009:dynamics-recursive-properties-of-multidimensional-symbolic-systems}, arbitrary effective shifts can be realized as subsystems of higher-dimensional sofic shifts.

\medskip
While sofic shifts form a strict subset of effective shifts, and enjoy a simple algebraic characterization on $\Z$~\cite[Chapter 3.2]{Lind-Marcus_1995:introduction-symbolic-dynamics}, they do not admit such a straightforward separation criterion in dimension $d \geq 2$. Indeed, the only arguments known to the authors proving that an effective shift is not sofic are based on information-theoretic tools quantifying the amount of information inside their patterns.

More precisely, since shifts of finite type are defined by finite families of forbidden patterns, the validity of patterns in an SFT is a local property. In a shift of finite type, the compatibility of a pattern of domain $\interval{0}{n}^d$ with its exterior is thus constrained by the size of their sufficiently thick border, which is only composed of $\bigO(n^{d-1})$ cells. With only a finite alphabet, this border cannot embed enough information to fully describe the interior of large enough patterns. As cellwise projections of SFTs, this $\bigO(n^{d-1})$ information bound still applies to sofic shifts.

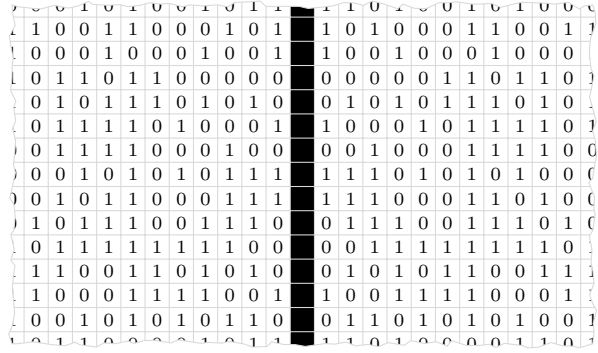
\begin{figure}[t]
  \begin{tikzpicture}[scale=0.32, decoration={random steps,segment length=0.4em,amplitude=0.15em}]
    \pgfmathsetseed{4}
    \draw [TilingGrid,decorate,path picture={
      \fill[black] (0,0) rectangle ++(1,15);
      \foreach \i in {1,...,12} {
        \foreach \j in {0,...,14} {
          \pgfmathrandominteger{\c}{0}{1}
          \node[black,font=\scriptsize] at (\i.5,\j.5) {$\c$};
          \node[black,font=\scriptsize] at ($(-\i,\j) + (0.5,0.5)$) {$\c$};
        }
      }
      \draw (-12,0) grid (13,15);
    }] (-11.5,0.5) rectangle (12.5,14.5);
  \end{tikzpicture}
  \caption{A typical tiling of the mirror shift.}
\end{figure}

The textbook example of this obstruction is the $2$-dimensional mirror shift: on the alphabet~$\{0,1,\square[black]\}$, only a single $\square[black]$ symbol can appear on each line, all $\square[black]$ symbols must constitute an entire column, and the two induced $\{0,1\}$-colored half-planes must be mirrors of one another. The mirror shift defines $2^{n^2}$ binary $\{0,1\}$-colored patterns of size $n \times n$, which must be mirrored across the $\square[black]$-column. However large the alphabet of a cover for the mirror shift, these borders are too small to transmit the description of their interior to the other side with only SFT constraints.

The quest to find the exact separation between sofic and non-sofic effective shifts has found many highly complex examples of sofic multidimensional shifts: substitutive
shifts~\cite{Mozes_1989:tiling_substitution_systems_dynamical_systems_generated_by_them};
S-adic effective shifts~\cite{Aubrun-Sablik_2014:multid-effective-s-adic-subshifts-are-sofic}; seas of squares~\cite{Westrick_2017:seas-of-squares}; all effective shifts of low density~\cite{Destombes_2021:phd:algorithmic-complexity-soficness-subshifts-multidimensional}\dots

All these examples can actually be recovered using ad-hoc modifications of the so-called \emph{fixed-point construction} on tilings (\textit{c.f.}~\cite{Durand-Romashchenko-Shen_2012:fixed-point-tilesets-and-applications} and \Cref{sof:sec}). Based on a fixed-point theorem from computability theory, this construction defines shifts of finite type whose configurations simulate an infinite hierarchy of layers. Inside each layer, a simulation groups the tiles into finite blocks, called \emph{macro-tiles}, in which the computations of a Turing machine emulate a tile of the next level.

Furthermore, this construction is entirely algorithmic: in order to change the tiling simulated at some level, one only needs to change its ``program'' drawn in the macro-tiles of the previous layer. This versatility certainly explains the variety of shifts which have been proved sofic by this construction. For example, if a substitution is a map $\tau \colon \alphabet \to \alphabet^{D}$ that transforms individual cells into finite rectangular patterns, then the resulting \emph{substitutive shift} (in which every tiling admits preimages for the sequence of substitutions $\tau,\, \tau^2,\, \tau^3\dots$) is actually sofic~\cite{Mozes_1989:tiling_substitution_systems_dynamical_systems_generated_by_them}. For every computable (or even $\Pi^{0}_{1}$-computable) set of integers $S \subseteq \N$,~\cite{Westrick_2017:seas-of-squares} proves the soficity of the seas of squares of edge lengths $S$ (where a sea of square is a tiling drawing independent $\square[black]$-squares over a $\square$ background). For every $\alpha < 1$,~\cite{Destombes_2021:phd:algorithmic-complexity-soficness-subshifts-multidimensional} proves the soficity of every effective shift $X \subseteq \{\square,\square[black]\}^{\Z^2}$ in which every $n \times n$ pattern contains at most $\bigO(n^{\alpha})$ $\square[black]$-colored cells.

However, the fixed-point construction has proven difficult to formalize into a single theorem, as each new application has required (sometimes substantial) modifications of the original construction to fit its specificities. Nevertheless, all the aforementioned examples seem to share a common key ingredient: namely, there appear to be sublinear bounds on the amount of ``useful information'' contained inside the patterns of their tilings.

\begin{figure}[ht]
  \begin{tikzpicture}[RainbowF!85,scale=0.14]
    \begin{scope}[shift={(0,15.5)}]
      \fill (0,0) rectangle ++(1,1);
      \draw[TilingGrid] (0,0) grid (1,1);
    \end{scope}
    \begin{scope}[shift={(6,15)}]
      \fill (1,1) rectangle ++(1,1);
      \fill (0,0) rectangle ++(2,1);
      \draw[TilingGrid] (0,0) grid (2,2);
    \end{scope}
    \begin{scope}[shift={(13,14)}]
      \fill (0,0) rectangle ++(4,1);
      \fill (3,1) rectangle ++(1,3);
      \fill (1,1) rectangle ++(1,1);
      \fill (2,2) rectangle ++(1,1);
      \draw[TilingGrid] (0,0) grid (4,4);
    \end{scope}
    \begin{scope}[shift={(22,12)}]
      \fill (0,0) rectangle ++(8,1);
      \fill (7,1) rectangle ++(1,7);
      \fill (3,1) rectangle ++(1,3);
      \fill (4,4) rectangle ++(3,1);
      \fill (1,1) rectangle ++(1,1);
      \fill (2,2) rectangle ++(1,1);
      \fill (5,1) rectangle ++(1,1);
      \fill (6,2) rectangle ++(1,1);
      \fill (5,5) rectangle ++(1,1);
      \fill (6,6) rectangle ++(1,1);
      \draw[TilingGrid] (0,0) grid (8,8);
    \end{scope}
    \begin{scope}[shift={(35,8)}]
      \fill (0,0) rectangle ++(16,1);
      \fill (15,1) rectangle ++(1,15);
      \fill (7,1) rectangle ++(1,7);
      \fill (8,8) rectangle ++(7,1);
      \begin{scope}
        \fill (3,1) rectangle ++(1,3);
        \fill (4,4) rectangle ++(3,1);
        \fill (1,1) rectangle ++(1,1);
        \fill (2,2) rectangle ++(1,1);
        \fill (5,1) rectangle ++(1,1);
        \fill (6,2) rectangle ++(1,1);
        \fill (5,5) rectangle ++(1,1);
        \fill (6,6) rectangle ++(1,1);
      \end{scope}
      \begin{scope}[shift={(8,0)}]
        \fill (3,1) rectangle ++(1,3);
        \fill (4,4) rectangle ++(3,1);
        \fill (1,1) rectangle ++(1,1);
        \fill (2,2) rectangle ++(1,1);
        \fill (5,1) rectangle ++(1,1);
        \fill (6,2) rectangle ++(1,1);
        \fill (5,5) rectangle ++(1,1);
        \fill (6,6) rectangle ++(1,1);
      \end{scope}
      \begin{scope}[shift={(8,8)}]
        \fill (3,1) rectangle ++(1,3);
        \fill (4,4) rectangle ++(3,1);
        \fill (1,1) rectangle ++(1,1);
        \fill (2,2) rectangle ++(1,1);
        \fill (5,1) rectangle ++(1,1);
        \fill (6,2) rectangle ++(1,1);
        \fill (5,5) rectangle ++(1,1);
        \fill (6,6) rectangle ++(1,1);
      \end{scope}
      \draw[TilingGrid] (0,0) grid (16,16);
    \end{scope}
    \begin{scope}[shift={(56,1)}]
      \fill (0,0) rectangle ++(32,1);
      \fill (31,1) rectangle ++(1,31);
      \fill (15,1) rectangle ++(1,15);
      \fill (16,16) rectangle ++(15,1);
      \fill (7,1) rectangle ++(1,7);
      \fill (8,8) rectangle ++(7,1);
      \fill (23,1) rectangle ++(1,7);
      \fill (24,8) rectangle ++(7,1);
      \fill (23,17) rectangle ++(1,7);
      \fill (24,24) rectangle ++(7,1);
      \begin{scope}[shift={(0,0)}]
        \fill (3,1) rectangle ++(1,3);
        \fill (4,4) rectangle ++(3,1);
        \fill (1,1) rectangle ++(1,1);
        \fill (2,2) rectangle ++(1,1);
        \fill (5,1) rectangle ++(1,1);
        \fill (6,2) rectangle ++(1,1);
        \fill (5,5) rectangle ++(1,1);
        \fill (6,6) rectangle ++(1,1);
      \end{scope}
      \begin{scope}[shift={(8,0)}]
        \fill (3,1) rectangle ++(1,3);
        \fill (4,4) rectangle ++(3,1);
        \fill (1,1) rectangle ++(1,1);
        \fill (2,2) rectangle ++(1,1);
        \fill (5,1) rectangle ++(1,1);
        \fill (6,2) rectangle ++(1,1);
        \fill (5,5) rectangle ++(1,1);
        \fill (6,6) rectangle ++(1,1);
      \end{scope}
      \begin{scope}[shift={(16,0)}]
        \fill (3,1) rectangle ++(1,3);
        \fill (4,4) rectangle ++(3,1);
        \fill (1,1) rectangle ++(1,1);
        \fill (2,2) rectangle ++(1,1);
        \fill (5,1) rectangle ++(1,1);
        \fill (6,2) rectangle ++(1,1);
        \fill (5,5) rectangle ++(1,1);
        \fill (6,6) rectangle ++(1,1);
      \end{scope}
      \begin{scope}[shift={(24,0)}]
        \fill (3,1) rectangle ++(1,3);
        \fill (4,4) rectangle ++(3,1);
        \fill (1,1) rectangle ++(1,1);
        \fill (2,2) rectangle ++(1,1);
        \fill (5,1) rectangle ++(1,1);
        \fill (6,2) rectangle ++(1,1);
        \fill (5,5) rectangle ++(1,1);
        \fill (6,6) rectangle ++(1,1);
      \end{scope}
      \begin{scope}[shift={(8,8)}]
        \fill (3,1) rectangle ++(1,3);
        \fill (4,4) rectangle ++(3,1);
        \fill (1,1) rectangle ++(1,1);
        \fill (2,2) rectangle ++(1,1);
        \fill (5,1) rectangle ++(1,1);
        \fill (6,2) rectangle ++(1,1);
        \fill (5,5) rectangle ++(1,1);
        \fill (6,6) rectangle ++(1,1);
      \end{scope}
      \begin{scope}[shift={(24,8)}]
        \fill (3,1) rectangle ++(1,3);
        \fill (4,4) rectangle ++(3,1);
        \fill (1,1) rectangle ++(1,1);
        \fill (2,2) rectangle ++(1,1);
        \fill (5,1) rectangle ++(1,1);
        \fill (6,2) rectangle ++(1,1);
        \fill (5,5) rectangle ++(1,1);
        \fill (6,6) rectangle ++(1,1);
      \end{scope}
      \begin{scope}[shift={(16,16)}]
        \fill (3,1) rectangle ++(1,3);
        \fill (4,4) rectangle ++(3,1);
        \fill (1,1) rectangle ++(1,1);
        \fill (2,2) rectangle ++(1,1);
        \fill (5,1) rectangle ++(1,1);
        \fill (6,2) rectangle ++(1,1);
        \fill (5,5) rectangle ++(1,1);
        \fill (6,6) rectangle ++(1,1);
      \end{scope}
      \begin{scope}[shift={(24,16)}]
        \fill (3,1) rectangle ++(1,3);
        \fill (4,4) rectangle ++(3,1);
        \fill (1,1) rectangle ++(1,1);
        \fill (2,2) rectangle ++(1,1);
        \fill (5,1) rectangle ++(1,1);
        \fill (6,2) rectangle ++(1,1);
        \fill (5,5) rectangle ++(1,1);
        \fill (6,6) rectangle ++(1,1);
      \end{scope}
      \begin{scope}[shift={(24,24)}]
        \fill (3,1) rectangle ++(1,3);
        \fill (4,4) rectangle ++(3,1);
        \fill (1,1) rectangle ++(1,1);
        \fill (2,2) rectangle ++(1,1);
        \fill (5,1) rectangle ++(1,1);
        \fill (6,2) rectangle ++(1,1);
        \fill (5,5) rectangle ++(1,1);
        \fill (6,6) rectangle ++(1,1);
      \end{scope}
      \draw[TilingGrid] (0,0) grid (32,32);
    \end{scope}
    \begin{scope}[black]
      \draw (3.5,16) node {$\mapsto$};
      \draw (10.5,16) node {$\mapsto$};
      \draw (19.5,16) node {$\mapsto$};
      \draw (32.5,16) node {$\mapsto$};
      \draw (53.5,16) node {$\mapsto$};
    \end{scope}
  \end{tikzpicture}
  \caption{A few iterations of the substitution
    $\tau \colon \mspace{1mu}
    \begin{tikzpicture}[RainbowF!85,x=1.6ex,y=1.6ex,baseline=0.3ex]
      \protect\fill (0,0) rectangle ++(1,1);
      \protect\draw[TilingGrid] (0,0) rectangle (1,1);
    \end{tikzpicture}
    \mspace{1mu}
    \mapsto
    \mspace{1mu}
    \begin{tikzpicture}[RainbowF!85,x=1.6ex,y=1.6ex,baseline=1.1ex]
      \protect\fill (1,1) rectangle ++(1,1);
      \protect\fill (0,0) rectangle ++(2,1);
      \protect\draw[TilingGrid,step=1] (0,0) grid (2,2);
    \end{tikzpicture}\mspace{2mu},\;
    \begin{tikzpicture}[x=1.6ex,y=1.6ex,baseline=0.3ex]
      \protect\draw[TilingGrid] (0,0) rectangle (1,1);
    \end{tikzpicture}\mspace{1mu}
    \mapsto
    \mspace{1mu}
    \begin{tikzpicture}[x=1.6ex,y=1.6ex,baseline=1.1ex]
      \protect\draw[TilingGrid,step=1] (0,0) grid (2,2);
  \end{tikzpicture}\mspace{1mu}$
  realizing the Sierpiński triangle.}
\end{figure}

\pagebreak
In this article, we achieve a generalization of the fixed-point theorem that covers all these previous applications while abstracting the underlying construction under the formalism of \emph{substitutions} and \emph{$S$-adic shifts} (where a shift space is \emph{$S$-adic} if its tilings admit preimages by an infinite sequence of substitutions $\tau_0,\, {\tau_0 \circ \tau_1},\, {\tau_0 \circ \tau_1 \circ \tau_2}\dots$).
Informally, we prove under some computability conditions that a sequence of substitutions operating on alphabets of sublinear sizes (with respect to pattern domains) yields an \emph{$S$-adic shift space} that is actually sofic.

We actually consider a model of non-deterministic substitutions (\textit{a.k.a.} \emph{random substitutions}, see~\cite{Rust-Spindeler_2018:dynamical-systems-arising-from-random-substitutions}), which may non-deterministically substitute a single symbol into several finite rectangular patterns.
As did the simultaneity of the constructions \cite{Aubrun-Sablik_2013:simulation-of-effective-subshifts-by-2d-sfts,Durand-Romashchenko-Shen_2012:fixed-point-tilesets-and-applications}, our result emphasizes a duality between the concepts of simulations and substitutions: the notion of self-simulation (from computer science) is intrinsically constructive, and focuses on understanding multi-scale building blocks in tilings; while the substitution/$S$-adic point of view from mathematics can be seen as studying abstract hierarchical objects with specified properties at each level.

More precisely, our main results include:
\begin{itemize}
\item A ``black box'' lemma for the fixed-point construction~(\Cref{sof:lem:fixpoint}) phrased in terms of a computable substitution operating on an infinite alphabet of Wang tiles and generating infinite hierarchies of valid tilings.
\item Two applications of this lemma on limit shift spaces generated by infinite sequences of \emph{ndill maps}: \Cref{sof:thm:square-d-adic-sofic} operating on square substitutions of possibly non-monotonous sizes, and \Cref{sof:thm:bounded-d-adic-sofic} operating on rectangular substitutions of bounded sizes.
\end{itemize}

Ndill maps~\cite{Salo-Törmä_2015:block-maps-between-primitive-uniform-pisot-substitutions} generalize both block maps and classical substitutions by substituting symbols differently depending on the colors of their neighbors. As with substitutions, a sequence of dill maps $(\tau_\ell)_{\ell \in \N}$ defines an \emph{$N$-adic} limit space. Informally, $N$-adic systems allow to consider $S$-adic configurations that satisfy some intermediate validity conditions, \textit{e.g.}~limit configurations
\begin{center}
  \begin{tikzpicture}[scale=1.2]
    \foreach \i in {0,...,3} {
      \node (n\i) at (\i,0) {$x^{(\i)}$};
      \draw (\i,0) node[anchor=west,rotate=-45,shift={(0.15,-0.1)}] {$\scriptstyle \in X_{\i}$};
    }
    \node[anchor=east] at (n0.west) {$x\vphantom{x^{(0)}} = \!\!$};
    \foreach \i in {0,...,2} {
      \tikzmath{int \j; \j=\i+1;}
      \draw[<-] (n\i) edge node[above,shift={(0.04,0)}] {$\scriptstyle \subst_{\i}$} (n\j);
    }
    \node[anchor=west] at (n3.east) {$\!\!\vphantom{x^{(0)}}\dots$};
  \end{tikzpicture}
  \vspace*{-.6em}
\end{center}
in which all intermediate configurations $x^{(\ell)}$ are valid in some shifts of finite type $X_{\ell}$.
$N$-adic limit spaces were, to our knowledge, first introduced in~\cite{Zinoviadis_2015:hierarchy-expansiveness-2d-subshifts-finite-type} as \emph{substitution schemes}.

\medskip
While \Cref{sof:thm:square-d-adic-sofic,sof:thm:bounded-d-adic-sofic} are weaker than \Cref{sof:lem:fixpoint}, they actually are both easier to manipulate and general enough to recover all the classical applications of the fixed-point construction (\textit{e.g.}~\cite{Durand-Romashchenko-Shen_2012:fixed-point-tilesets-and-applications,Aubrun-Sablik_2014:multid-effective-s-adic-subshifts-are-sofic,Westrick_2017:seas-of-squares,Destombes_2021:phd:algorithmic-complexity-soficness-subshifts-multidimensional}). Indeed, the structure of macro-tiles in the fixed-point construction naturally induces a sequence of substitutions $\smash{\tau_{\ell} \colon \alphabet_{\ell+1} \mapsto \alphabet_{\ell}^{\interval{0}{N_{\ell}}^d}}$, each simulating a tile from the $\ell+1$\textsuperscript{th} level by a block of size $\smash{\interval{0}{N_{\ell}}^d}$ on the previous layer.

\Cref{sof:lem:fixpoint} offers technical improvements over existing applications in the literature :
\begin{itemize}
\item Classically, macro-tiles of domain $\interval{0}{N_{\ell}}^{d}$ embed the computations emulating the next layer of the simulation. Thus, the size $N_{\ell}$ must usually grow fast enough to embed increasingly more computations, \textit{e.g.}~$\smash{N_{\ell} = 2^{2^{\ell}}}$. In \Cref{sof:lem:fixpoint}, we actually allow a large variety of possibly non-monotonous sequences $(N_{\ell})_{\ell \in \N}$, \textit{e.g.}~$N_{\ell} = 2$ for infinitely many~$\ell$'s. To bypass the lack of computational space in the $\ell$\textsuperscript{th} layer, we thus embed the computations for the layer of level $\ell+1$ into some deeper layer $\ell' \leq \ell$, which requires to solve some technical challenges in trans-layer communication.
\item Classically, the fixed-point construction is defined with square macro-tiles (\textit{i.e.}~square substitutions). Similarly to~\cite{Zinoviadis_2015:hierarchy-expansiveness-2d-subshifts-finite-type,Zinoviadis_2016:phd:expansiveness-2d-sfts}, we relax this condition by allowing rectangular macro-tiles of possibly unbounded eccentricity;
\item Classically, computability conditions on the next layer of the fixed-point construction are given in terms of $N_{\ell}$, \textit{i.e.}~the size of macro-tiles on the current level of index $\ell$. Since \Cref{sof:lem:fixpoint} actually allows size $N_{\ell}=2$ arbitrarily often, we must give more flexible conditions: we formulate them in terms of $\prod_{i=0}^{\ell} N_i$, \textit{i.e.}~in terms of the size of a tile $t_{\ell+1}$ after $\ell+1$ steps of substitutions $\tau_0 \circ \dots \circ \tau_{\ell}(t_{\ell+1})$.
\item Classically, the fixed-point construction is defined in dimension $d=2$. Generalizations to higher dimensions are often claimed possible, although it is not obvious how the classical computational embedding (space-time diagrams of Turing machines) should be adapted.\\\pagebreak
  In this article, we instead implement a highly parallel model of computation called \emph{processor arrays} (see~\Cref{calc:sec:processors}). By operating on multidimensional grids of processors, it naturally processes inputs of size $\bigO(N^{d-1})$ in time $\bigO(N)$. To the best of our knowledge, this is the first use of this computation model in multidimensional symbolic dynamics.
\end{itemize}

For convenience, we state the computational conditions of our results in the $\log$-RAM model of computation, as the associated time complexity closely resembles those of ``real-world'' programming languages (and is thus easier to compute than with Turing machines).

\subsection*{Outline of the paper}

\begin{itemize}
\item \Cref{prel:sec} provides classical background in computation models, symbolic dynamics (shift spaces, sofic shifts) and substitutions ($S$-adic limit spaces, and their generalization as ndill maps);
\item \Cref{sof:sec} motivates and states our main ``fixed-point''  \Cref{sof:lem:fixpoint}. It also relates \Cref{sof:lem:fixpoint} to other dynamical notions from the literature about shift spaces in \Cref{sof:thm:square-d-adic-sofic,sof:thm:bounded-d-adic-sofic}: namely, $S$-adic and $N$-adic limit spaces;
\item \Cref{app:sec} illustrates our results by considering their applications to shift dynamics (frequencies and growth-type invariants, universality\dots) and recovers or generalizes some well-known examples of multidimensional soficity~\cite{Durand-Romashchenko-Shen_2012:fixed-point-tilesets-and-applications,Westrick_2017:seas-of-squares,Destombes_2021:phd:algorithmic-complexity-soficness-subshifts-multidimensional};
\item \Cref{calc:sec,pos:sec} develop the required tools for the proof of \Cref{sof:lem:fixpoint}: some results from computability theory, and some notions of \emph{positioning} and \emph{wiring} (\textit{c.f.}~\crtlnameref{pos:lem:routing-lemma});
\item \Cref{fix:sec} contains the proof of our main ``fixed-point'' \Cref{sof:lem:fixpoint}. As we aimed at both providing a high-level overview of the construction and the associated technical details, this proof spans most of the article.
\item Finally, \Cref{persp:sec} discusses some future work on/using this construction.
\end{itemize}

%
\section{Preliminaries}
\label{prel:sec}

\subsection{Standard notions and notations}

\emph{While skipping this section may be quite tempting, we strongly recommend spending the time to familiarize oneself about the non-standard notions and notations on rectangles and grids: smallest facets, longest edges, grid products, etc\dots}

\medskip
In dimension $d \in \N$, we will often denote $\vec{i} = (i_1,\ldots,i_d) \in \Z^d$ an element of $\Z^d$, and $\basis{1},\ldots,\basis{d}$ the standard basis. For $n_1,\ldots,n_d \in \N$, we denote by $\ibox{n_1,\ldots,n_d}$ the product of intervals $\interval{0}{n_1-1} \times \cdots \times \interval{0}{n_d-1}$.

\paragraph*{Rectangles}
In dimension $d \in \N$, we denote by $\Srect_0 = \{ \ibox{n_1,\ldots,n_d} : (n_1,\ldots,n_d) \in \N^d \}$ the set of rectangles with bottom-left corner $\vec{0}$; and $\Srect = \bigcup_{\vec{i} \in \Z^d} \vec{i} + \Srect_0$ the set of $d$-dimensional rectangles. We will often denote individual rectangles $\rect{r} \in \Srect$; and for any rectangle $\rect{r} \in \Srect$, we define $\rnorm{\rect{r}}$ the \emph{normalization} of $\rect{r}$, \textit{i.e.}~the unique translation of $\rect{r}$ that belongs to $\Srect_0$. For any position $\vec{i} = (i_1,\dots,i_d)\in \Z^d$ and rectangle $\rect{r} = \ibox{n_1,\dots,n_d}$, we denote by $\vec{i} \bmod \rect{r}$ the position $(i_1 \bmod n_1,\dots,i_d \bmod n_d) \in \rect{r}$. For a given rectangle $\rect{r} \in \Srect_0$, a \emph{cut} of $\rect{r} = \ibox{n_1,\ldots,n_d}$ is a rectangle $\rect{c} = \ibox{n_1',\dots,n_d'}$ such that $n_1' \leq n_1,\ldots,n_d' \leq n_d$, \textit{i.e.}~a rectangle $\rect{c} \in \Srect_0$ such that $\rect{c} \subseteq \rect{r}$.

For a given rectangle $\rect{r} = \ibox{n_1,\ldots,n_d}$, we will denote $\rvol{\rect{r}} = \prod_{k=1}^{d} n_k$ the \emph{volume} of $\rect{r}$, and $\rtime(\rect{r})$ and $\rspace(\rect{r})$ the \emph{time} and \emph{space} of $\rect{r}$: these respectively refer to the length $\rtime(\rect{r}) = \max \{ n_k : 1 \leq k \leq d \}$ of a longest edge and to the area $\rspace(\rect{r}) = \frac{1}{\rtime(\rect{r})} \rvol{\rect{r}}$ of a smallest facet in $\rect{r}$ (this terminology is motivated by the embeddings of space-time diagrams in \Cref{fix:sec:computation-zone}).

Finally, we denote by $\facet{k}{-}(\rect{r}) = \{ \vec{i} \in \rect{r} : \vec{i}_k = 0 \}$ and $\facet{k}{+}(\rect{r}) = \{ \vec{i} \in \rect{r} : \vec{i}_k = n_k-1\}$ the $(d-1)$-dimensional facets of minimal and maximal index in $\rect{r}$ along the direction $\basis{k}$ for $1 \leq k \leq d$. In this paper, we will often refer to \emph{smallest facets}, \textit{i.e.}~facets of smallest area $\rspace(\rect{r})$. Since there are several of them, we arbitrarily pick a convention and denote by $\sfacet(\rect{r})$ the facet $\facet{h}{-}(\rect{r})$ of smallest index $1 \leq h \leq d$ that has area $\rspace(\rect{r})$.

\begin{figure}[ht]
  \begin{tikzpicture}[scale=0.75,z={(0.45cm,0.5cm)}]
    \begin{scope}[canvas is xy plane at z=2]
      \draw (0,0) grid (4,4);
      \foreach \i in {0,...,3} {
        \foreach \j in {0,...,3} {
          \tikzmath{int \c; if Mod(\j,2) == 0 then {\c = 4*\j+\i;} else {\c = 4*(\j+1)-\i-1;};}
          \draw (\i.5,\j.5) node {\c};
        }
      }
    \end{scope}
    \begin{scope}[shift={(8,0)}]
      \foreach \i in {0,...,3} {
        \foreach \j in {0,...,3} {
          \foreach \k in {0,...,3} {
            \tikzmath{int \c;
              if Mod(\k,2) == 0 then
              {if Mod(\j,2) == 0 then
                {\c=16*\k+4*\j+\i;}
                else
                {\c=16*\k+4*(\j+1)-\i-1;};}
              else
              {if Mod(\j,2) == 1 then
                {\c=16*(\k+1)-4*(\j+1)+\i;}
                else
                {\c=16*(\k+1)-4*\j-\i-1;};};}
            \ifnum\k=0{\draw[canvas is xy plane at z=0] (\i.5,\j.5) node {\c};}\fi
            \ifnum\j=3{\draw[canvas is xz plane at y=4, transform shape] (\i.5,\k.5) node[font=\Large] {\c};}\fi
            \ifnum\i=3{\draw[canvas is yz plane at x=4, transform shape] (\j.5,\k.5) node[yscale=-1,rotate=270,font=\Large] {\c};}\fi
          }
        }
      }
      \foreach \i in {1,...,3} {
        \draw (0,\i,0) -- ++(4,0,0);
        \draw (0,4,\i) -- ++(4,0,0);
        \draw (\i,0,0) -- ++(0,4,0);
        \draw (\i,4,0) -- ++(0,0,4);
        \draw (4,\i,0) -- ++(0,0,4);
        \draw (4,0,\i) -- ++(0,4,0);
      }
      \draw (0,0,0) -- ++(4,0,0) -- ++(0,0,4) -- ++(0,4,0) -- ++(-4,0,0) -- ++(0,0,-4) -- cycle (0,4,0) -- ++(4,0,0) -- ++(0,0,4) (4,0,0) -- ++(0,4,0);

    \end{scope}
  \end{tikzpicture}
  \caption{Arrays on $\ibox{4}^d$ for $d=2$ and $d=3$ sorted in boustrophedon ordering.}
\end{figure}
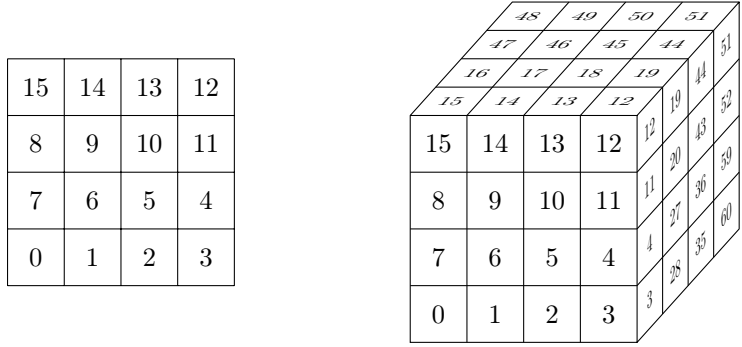

\paragraph*{Boustrophedon order} The \emph{boustrophedon order}\footnote{From the Greek ``\textgreek{βουστροφηδόν}'' (lit.~``in the way an ox turns [while plowing]''), it is also called the \emph{snake} or the \emph{shear order}; thus illustrating a mathematician's love for bucolic metaphors.} on $\ibox{n_1,\dots,n_d}$ is defined inductively: an array ${a \in \N^{\ibox{n_1}}}$ is sorted in boustrophedon order if $a_{i} \leq a_{i+1}$ for all ${i \in \interval{0}{n_1-2}}$; and a $d$\nobreakdash-dimensional array $a \in \N^{\ibox{n_1,\dots,n_d}}$ is sorted in boustrophedon order if
\begin{enumerate}[label=\roman*.]
\item Each subarray $\restr{a}{\{i\} \times \ibox{n_2,\dots,n_d}}$ is sorted in increasing (resp.~decreasing) boustrophedon order if  $i \in \interval{0}{n_1-1}$ is even (resp.~odd);
\item Adjacent subarrays are themselves sorted, \textit{i.e.}~$a_{i,\vec{j}} \leq a_{i+1,\vec{j'}}$ for every ${i \in \interval{0}{n_1-1}}$ and $\vec{j'},\vec{j'} \in \ibox{n_2,\dots,n_d}$.
\end{enumerate}

We are especially interested in the boustrophedon order since it preserves locality: for any rectangle $\rect{r} = \ibox{n_1,\dots,n_d}$ and for $n < \prod_{k=1}^d n_k$, denote by $\vec{i}$ and $\vec{j}$ the positions of index $n$ and $n+1$ in the boustrophedon ordering of $\rect{r}$; then $\vec{i}$ and $\vec{j}$ denote two adjacent positions in~$\rect{r}$.

\paragraph*{Grids} We call \emph{grid} a partition $\grid{g} = (\rect{r}_{\vec{i}})_{\vec{i} \in \Z^d}$ of $\Z^d$ into facet-adjacent disjoint rectangles (\textit{i.e.}~for all $\vec{i} \in \Z^d$ and $1 \leq k \leq d$, the facets $\facet{k}{+}(\rect{r}_{\vec{i}})$ and $\facet{k}{-}(\rect{r}_{\vec{i} + \basis{k}})$ match). For $\grid{g} = (\rect{r}_{\vec{i}})_{\vec{i} \in \Z^d}$ a grid and $\vec{j} \in \Z^d$ a position, there exists a unique $\vec{i} \in \Z^d$ such that $\vec{j} \in \rect{r}_{\vec{i}}$; and we denote by $\vec{j} \bmod \grid{g} \in \N^d$ the translation of $\vec{j} \in \rect{r}_{\vec{i}}$ in the normalization $\rnorm{\rect{r}_{\vec{i}}} \in \Srect_0$.

Generalizing the previous notions of time and space, the \emph{time} of $\grid{g}$ will denote the supremum $\rtime(\grid{g}) = \sup_{\vec{i} \in \Z^d} \rtime(\rect{r}_{\vec{i}})$ and the \emph{space} of $\grid{g}$ will denote the infimum $\rspace(\grid{g}) = \inf_{\vec{i} \in \Z^d} \rspace(\rect{r}_{\vec{i}})$.

We also define the following \emph{composition}: for a rectangle $\rect{r}' \in \Srect$ and a family of rectangles $(\rect{r}_{\vec{i}})_{\vec{i} \in \rect{r}'}$ indexed by the positions $\vec{i} \in \rect{r}'$, we define $\rect{r}' \circ (\rect{r}_{\vec{i}})_{\vec{i} \in \rect{r}'}$ as the rectangle $\pxrect{r} = \bigcup_{\vec{i} \in \rect{r}'} \rect{r}_{\vec{i}}$. Generalizing from individual rectangles to complete grids, we define $\grid{g}' \circ \grid{g} = (\rect{r}'_i \circ (\rect{r}_{\vec{j}})_{\vec{j} \in \rect{r}'_i})_{\vec{i} \in \Z^d}$ for $\grid{g}' = (\rect{r}'_{\vec{i}})_{\vec{i} \in \Z^d}$ and $\grid{g} = (\rect{r}_{\vec{j}})_{\vec{j} \in \Z^d}$; in other words, $\grid{g}' \circ \grid{g}$ is the grid of individual rectangles $\smash{\bigcup_{\vec{j} \in \rect{r}'_{\vec{i}}} \rect{r}_{\vec{j}}}$ for $\vec{i} \in \Z^d$. When composing multiple grids $(\grid{g}_i)_{0 \leq i < n}$, we denote by $\smash{\bigcomp_{i=0}^{n-1} \grid{g}_{i} = \grid{g}_{n-1} \circ \cdots \circ \grid{g}_{0}}$.

Notice that $\pxgrid{g} = \grid{g}' \circ \grid{g}$ is a \emph{subgrid} of $\grid{g}$ (we also say that $\grid{g}$ is \emph{nested} inside $\pxgrid{g}$), in the sense that $\pxgrid{g}$ is a coarser partition of $\Z^d$ than $\grid{g}$. This elementary observation will become especially relevant when we will later consider \emph{sequences of grids} $(\grid{g}_{i})_{0 \leq i < n}$: the products $\pxgrid{g}_i = \bigcomp_{j=0}^{i-1} \grid{g}_j$ define a sequence of \emph{nested} grids, each $\pxgrid{g}_{i+1}$ being a subgrid of $\pxgrid{g}_{i}$.

\paragraph*{(Multi)functions} Formally, a \emph{multifunction} (or \emph{non-deterministic function}) $f \colon A \mto B$ is a binary relation ${\mathcal{R}_{f} \subseteq A \times B}$. For $x \in A$, we denote $f(x) = \sett[B]y{x \mathbin{\mathcal{R}_{f}} y}$ the set of possible images of $x$ (which may be empty), and $f^{-1}\colon B\mto A$ the symmetric relation.

\enlargethispage{\baselineskip}
\paragraph*{Computability} While computability theory is not a prerequisite for the main result of this paper, some examples in \Cref{app:sec} will assume some basic familiarity with the computability of sets.

A set $S \subseteq \{0,1\}^{\ast}$ is \emph{computable} if there exists a total computable function $f \colon \{0,1\}^{\ast} \to \{\top,\bot\}$ such that $f(u) = \top$ if and only if $u \in S$; $S$ is \emph{computably enumerable} if there exists a partial computable function $f \colon \{0,1\}^{\ast} \to \{\top,\bot\}$ of domain $S$, \textit{i.e.}~such that $f(u)$ is defined if and only if $u \in S$.
Equivalently, there exists a computable total function $f \colon \N \to \{0,1\}^{\ast}$ which \emph{enumerates} $S$, in the sense that $\{ \funram{e}(n) : n \in \N\} = U$.
$S$ is \emph{computably co-enumerable} if its complement is computably enumerable.
\emph{Using standard binary encodings, these notions implicitly generalize to sets of integers, tuples, etc\dots}

For a finite alphabet $\alphabet$, a set $S \subseteq \alphabet^{\N}$ is \emph{effectively closed} (or~$\Pi^{0}_{1}$) if there exists a computably enumerable set of forbidden prefixes $\mathcal{F} \subseteq \alphabet^{\ast}$ such that ${S = \{ x \in \alphabet^{\N} : \forall n \in \N, \restr{x}{\{0,\dots,n\}} \notin \mathcal{F} \}}$. Equivalently, it is the set of infinite inputs upon which some given (Turing, RAM, etc\dots) machine never halts (in which case, said machine is said to \emph{recognize} $S$).

\subsection{The RAM model}

Theoretic computations are usually abstracted from the specifics of a computation model, since most models (\textit{e.g.}~Turing machines, variations of register machines\dots) end up defining the same class of polynomial-time computable functions. Unfortunately, this article is not interested in time complexity as a theory, but rather as the precise number of operations performed by an algorithm, which depends heavily on the choice of a computation model.

In this section, we thus recall the notion of \emph{computable function} in the context of \emph{$\log$-RAM machines}: as the latter are very close to real-world programming languages, they define a somewhat intuitive measure of time complexity.

\begin{definition}[RAM machine]
  A \emph{Random Access Machine} (RAM) is an abstract computing machine operating with a fixed \emph{instruction set} $S$. It is controlled by a \emph{program} $(V,P)$, where $V$ is a finite set of symbolic names indexing the \emph{variables} of the program, and $P \in S^{\ast}$ is a finite sequence of instructions. A configuration is then a tuple $(\ram{PC},(\ram{var}_i)_{i \in V},\ram{M},\ram{I},\ram{O})$, where:
  \begin{itemize}
  \item $\ram{PC} \in \N$ is a \emph{program counter} indexing the current instruction in the sequence~$P$;
  \item $(\ram{var}_i)_{i \in V}$ is a collection of words $\ram{var}_i \in \{0,1\}^{\ast}$ assigning values to the variables of $V$;
  \item $\ram{M} = (\ram{M}[i])_{i \in \N}$ is an infinite \emph{memory}, in which individual memory cells each hold a binary word $\ram{M}[i] \in \{0,1\}^{\ast}$, and are indexed by the address $i \in \N$;
  \item $\ram{I} \in \{0,1\}^{\ast\ast}$ and $\ram{O} \in \{0,1\}^{\ast\ast}$ are the \emph{input} and the \emph{output arrays}, whose elements respectively correspond to the words given as input to, and printed by, the machine.
  \end{itemize}
  A \emph{run} is then a sequence of configurations corresponding to an execution of the machine: at each time step, the machine applies the instruction given by the program counter $\ram{PC}$ (thus updating its memory, variables, and typically incrementing $\ram{PC}$ by one). The instruction set $S$ of non-deterministic RAM machines is usually composed of:
  \begin{itemize}[topsep=1pt]
  \item \emph{Arithmetic operations} (such as add, subtract, copy, multiply, constant\dots) operating directly on the \emph{variables} of the program by interpreting the associated binary strings as integers: typically, such instructions include $\ram{var}_i \leftarrow \ram{var}_j \mathbin{\ram{op}} \ram{var}_k$ (resp.~${\ram{var}_i \leftarrow \ram{op}\mspace{4mu} \ram{var}_j}$);
  \item A \emph{non-deterministic} instruction $\ram{var}_i \leftarrow \ram{NDET}(\ram{var}_j)$ that non-deterministically sets a variable $\ram{var}_i$ of the program to any integer in the interval $\interval{0}{\ram{var}_j}$;
  \item \emph{Control flow instructions} (such as $\ram{HALT}$, $\ram{ERROR}$, conditional $\ram{GOTO}$\dots) whose instructions directly operate on the program counter $\ram{PC}$ to allow for conditional execution branching;
  \item \emph{Memory input/output}, which consists in transferring data between the variables of~$V$ and either the memory ($\ram{LOAD} : \ram{var}_i \leftarrow \ram{M}[\ram{var}_j]$ and $\ram{STORE} : \ram{M}[\ram{var}_i] \leftarrow \ram{var}_j$), or the input/output arrays ($\ram{READ} : \ram{var}_i \leftarrow \ram{I}[\ram{var}_j][\ram{var}_k]$ and $\ram{WRITE} : \ram{O}[\ram{var}_i][\ram{var}_j] \leftarrow \ram{var}_k$).
  \end{itemize}
\end{definition}

\noindent For the rest of this article, we fix a finite instruction set $S$ for our RAM machines. Without listing it exhaustively (as it would not yield much insight into our results), we informally assume that its instructions allow to perform  classical arithmetic operations (addition, multiplication, division\dots) and conditional branching in constant time.

\medskip
As with classical programming languages, a \emph{program} will be represented as (and identified with) a \emph{code} $e \in \{0,1\}^{\ast}$.
Before a computation begins, we assume that the program counter is set to zero, that the variables, memory and output arrays are all empty, and that an input array $\ram{I} \in \{0,1\}^{\ast\ast}$ has been set by the user. By considering the associations of inputs and outputs obtained by running the program $e \in \{0,1\}^*$, we obtain a multifunction:

\begin{definition}[Computability]
  Let $\funram{e} \pcolon \{0,1\}^{\ast\ast} \mto \{0,1\}^{\ast\ast}$ be the partial multifunction computed by the RAM program $e \in \{0,1\}^{\ast}$, \textit{i.e.}~the association of arrays $\ram{I} \in \{0,1\}^{\ast\ast}$ and $\ram{O} \in \{0,1\}^{\ast\ast}$ such that $e$ admits a run on the input $\ram{I}$ that halts and outputs $\ram{O}$.
\end{definition}

For $t \in \N$, we denote by $\funram{e}(\ram{I})\halt[t]$ the set of all output arrays $\ram{O}$ for which there exists a run of the program $e$ on $\ram{I}$ of at most $t$ steps halting and outputting $\ram{O}$. We say that $f \pcolon \{0,1\}^{\ast\ast} \mto \{0,1\}^{\ast\ast}$ is \emph{computable} (in time $t(s)$) if there exists a program $e \in \{0,1\}^{\ast}$ such that $f = \funram{e}$ (resp.~for all $\ram{I} \in \{0,1\}^{\ast\ast}$ with \emph{bit length} $\smash{s = \sum_{i=0}^{\len(\ram{I})-1} |\ram{I}[i]|}$, $f(\ram{I}) = \funram{e}(\ram{I})\halt[t(s)]$).

\begin{remark}[name={Data types},label={calc:rem:data-types}]
  In the same way we identified binary words with integers through their binary expansions, computations can actually be performed on arbitrary countable sets $S$ through the choice of an \emph{encoding} $\eta \colon S \to \{0,1\}^{\ast}$. Since most reasonable encodings of classical data types are computable from one another, this article will implicitly consider that RAM algorithms can manipulate integers of $\N$ and $\Z$, booleans $\{\top,\bot\}$, tuples $(\{0,1\}^{\ast})^{m}$ (for arbitrary $m \in \N$), etc\dots
\end{remark}

\subsubsection*{Word-RAM}

RAM machines perform arithmetic operations on integers of arbitrary sizes in unit time. While this simplifies time complexities, it also has the unintended consequence of allowing to solve PSPACE-complete problems in polynomial time~\cite{Hartmanis-Simon_1974:power-multiplication-ram}. We will thus bound the size of words in variables and memory cells to avoid this issue~\cite{Slissenko_1978:models-computation-address-organization-storage_RU,Angluin-Valiant_1979:fast-probabilistic-algorithms-hamiltonian-and-rac}.

\begin{definition}[word-RAM]
  A RAM program $e \in \{0,1\}^{\ast}$ has \emph{word length $w(s)$} if, for any input $\ram{I} \in \{0,1\}^{\ast\ast}$ of bit length $s = \sum_{i=0}^{\len(\ram{I})-1} |\ram{I}[i]|$ and for any execution on $\ram{I}$ halting in an accepting state, the word length of all memory cells and variables\footnote{In particular, we assume that $w(s) \geq \log|e|$ to include the program counter. However, since the $\ram{READ}$ instruction does not load a whole input word $\ram{I}[i]$ but operates one input bit at a time, this $w(s)$ bound does not apply to input words $\ram{I}[i]$.} is bounded at any time by~$w(s)$.\\
  In particular, if $w(s) = \bigO(\log s)$, $e$ is said to be a \emph{$\log$-RAM program}.
\end{definition}

\medskip
The $\log$-RAM model is actually very close to real-world \emph{machine code instructions} (that operate in unit time on integers of bounded lengths). By classical compilation processes (simulating a call stack, local variables and scopes, etc\dots), any C-like\footnote{Replace C by your favorite Turing-complete programming language: Common Lisp, Haskell, OCaml, Python\dots}\textsuperscript{,}\footnote{Since non-determinism is not a canonical feature of the C standard, we assume that some $\ram{NDET}$-like instruction has been added. Similarly, we ignore any 64-bit lengths limitations.} program of polynomial time complexity~$t(s)$ and word size $w(s)$ can be compiled into a RAM program with time $\bigO(t(s))$ and word size $w(s) + \bigO(\log t(s))$: in particular, if $t(s) = s^{\bigO(1)}$ and $w(s) = \bigO(\log s)$, then $e$ is a $\log$-RAM program. \emph{We thus invite the reader to consider the time-complexity and word-length requirements in \Cref{sof:lem:fixpoint} and \Cref{app:sec} as they would in their favorite programming language.}

%
\subsection{Symbolic dynamics}\label{prel:sec:symbolic-dynamics}

\subsubsection*{Shift spaces}

For a finite \emph{alphabet} $\mathcal{A}$, and a \emph{dimension} $d \in \N$, a \emph{pattern} is a coloring $w \colon D \to \alphabet$ of a \emph{domain} $D \subseteq \Z^{d}$ (denoted $\dom(w)=D$) with symbols from $\alphabet$; and we denote by $w_{\vec{p}} \in \alphabet$ the symbol of $\alphabet$ appearing at position $\vec{p} \in \dom(w)$. In case the domain $\dom(w)$ is finite, $w$ is called a \emph{finite} pattern; and in case the domain $\dom(w)$ equals the whole space $\Z^{d}$, $w$ will often be called a \emph{configuration}. For $w$ a finite pattern over the alphabet $\alphabet$ and for $a \in \alphabet$, we denote $|w|_{a} = |\{\vec{i} \in \dom(w) : w_{\vec{i}} = a\}|$ the number of symbols $a$ in $w$.

Abusing notations, we denote by $\alphabet^{\Srect}$ the set of $d$-dimensional rectangular finite patterns, \textit{i.e.}~$\alphabet^{\Srect} = \bigcup_{\rect{r} \in \Srect} \alphabet^{\rect{r}}$ (and $\alphabet^{\Srect_0} = \bigcup_{\rect{r} \in \Srect_0} \alphabet^{\rect{r}}$); and by $\alphabet^{\ast d}$ (resp.~$\alphabet^{\circledast d}$) the set of arbitrarily-shaped $d$-dimensional finite (resp.~possibly infinite) patterns. The set of all $d$-dimensional configurations over $\alphabet$ is denoted $\smash{\alphabet^{\Z^{d}}}$.

A pattern $u$ is said to \emph{appear} in a pattern $v$ (or, conversely, that $v$ \emph{contains} the pattern~$u$), denoted $u \subpattern v$, if there exists a position $\vec{p_0} \in \Z^{d}$ such that $(\vec{p_0} + \dom(u)) \subseteq \dom(v)$ and $\forall \vec{i} \in \dom(u),\, u_{\vec{p}} = v_{\vec{p_0} + \vec{p}}$. For $w \in \alphabet^{\circledast d}$ a pattern and $S \subseteq \dom(w)$, we denote by $\restr{w}{S} \in \alphabet^{S}$ the coloring it induces on $S$. For two patterns $u, v \in \alphabet^{\circledast d}$, we denote $u = v$ if the two patterns are equal and $\dom(u) = \dom(v)$; and $u \equiv v$ if $u$ and $v$ are equal up to translation (\textit{i.e.}~$u \subpattern v$ and $v \subpattern u$). In particular, if $w \in \alphabet^{\Srect}$ is a rectangular pattern, we denote $\rnorm{w}$ the unique pattern in $\alphabet^{\Srect_0}$ such that $w \equiv \rnorm{w}$.

\medskip
Equipping $\alphabet$ with a discrete topology, and $\smash{\alphabet^{\Z^d}}$ with a product topology, the set of configurations $\smash{\alphabet^{\Z^d}}$ is a Cantor space; and its subsystems are:
\begin{definition}[Shift space]
  A \emph{shift space} (or \emph{shift} for short) is a closed and translation-invariant subset of $\alphabet^{\Z^{d}}$. Equivalently, a subset $\smash{X \subseteq \alphabet^{\Z^{d}}}$ is a shift space if and only if there exists a (possibly infinite) family of forbidden finite patterns $\mathcal{F}$ such that $X = X_{\mathcal{F}}$, where
  \[ X_{\mathcal{F}} = \{ x \in \alphabet^{\Z^{d}} : \forall w \in \mathcal{F},\, w \not\subpattern x\}. \]
\end{definition}
A family of forbidden finite patterns $\mathcal{F}$ such that $X = X_{\mathcal{F}}$ may be considered as a \emph{presentation} of the shift space $X$. As usual with presentations, two different families of forbidden finite patterns may define the same shift space.

From this definition, the group $\Z^d$ acts on the configurations of any shift space $X \subseteq \alphabet^{\Z^d}$ by translation: more precisely, $\Z^d$ defines a left group action $\Z^d \curvearrowright X$ given by $(\vec{t} \cdot x)_{\vec{p}} = x_{\vec{p} - \vec{t}}$ (for any translation $\vec{t} \in \Z^d$, position $\vec{p} \in \Z^d$ and configuration $x \in X$). This action, called the \emph{shift}, thus defines a dynamics on the shift space $X$.

\subsubsection*{Factor maps}
Morphisms preserve the structure of shift spaces as dynamical systems:
\begin{definition}[Morphism]
  For $X \subseteq \alphabet^{\Z^{d}}$ and $Y \subseteq \alphabet[B]^{\Z^{d}}$ two shift spaces, a map $\varphi \colon X \to Y$ is a (shift) \emph{morphism}\footnote{Symbolic dynamicists may also use the word \emph{morphism} for morphisms of the free monoid, which yield one-dimensional substitutions. Since our work involves both objects, the reader should be aware that we do not use the latter terminology.} if it is continuous and translation-equivariant.

  Equivalently, $\varphi \colon X \to Y$ is a morphism if it is a (sliding) block map (or some kind of alphabet-changing cellular automaton), \textit{i.e.}~if there exists a finite neighborhood $N \subseteq \Z^{d}$ and a local rule $f \colon \alphabet^{N} \to \alphabet[B]$ such that $\varphi(x)_{\vec{p}} = f(\restr{x}{\vec{p}+N})$ for every $x \in X$ and every $\vec{p} \in \Z^{d}$.
\end{definition}

\begin{definition}[Factor and conjugacy]
  For $X \subseteq \alphabet^{\Z^d}$ and $Y \subseteq \alphabet[B]^{\Z^d}$, a surjective morphism $\varphi \colon X \to Y$ is called a \emph{factor map}; in which case, $Y$ is said to be a \emph{factor} of $X$ (or, conversely, that $X$ is a \emph{cover} of $Y$).

  If $\varphi$ is also bijective, then $\varphi^{-1}$ can also be shown to be a factor map \cite[Curtis-Hedlund-Lyndon theorem]{Hedlund_1969:endomorphisms-automorphisms-of-shift-dynamical-systems}, and $\varphi$ is called a \emph{conjugacy}; in which case, $X$ and $Y$ are said to be \emph{conjugate}.
\end{definition}

\subsubsection*{Classes of shift spaces}
Shift spaces are traditionally classified depending how complex are their presentations:
\begin{itemize}
\item A shift space $X \subseteq \alphabet^{\Z^{d}}$ has \emph{finite type} (Shift of Finite Type, abbreviated SFT) if there exists a finite family $\mathcal{F}$ of forbidden finite patterns such that $X = X_{\mathcal{F}}$;
\item A shift space $Y \subseteq \alphabet^{\Z^{d}}$ is \emph{sofic} if it is a factor of a shift of finite type $X$, called an \emph{SFT cover} of $Y$;
\item A shift space $X \subseteq \alphabet^{\Z^{d}}$ is \emph{effective} if there exists a computably enumerable family $\mathcal{F}$ of forbidden finite patterns such that $X = X_{\mathcal{F}}$.
\end{itemize}

These classes are all closed under conjugacy, and strictly included in one another:
\[ [ \mspace{1mu}\text{SFTs}\mspace{1mu} ] \subsetneq [ \mspace{1mu}\text{Sofic shifts}\mspace{1mu} ] \subsetneq [ \mspace{1mu}\text{Effective shifts}\mspace{1mu} ].\]
Since shifts of finite type are conjugate to their higher-block code (\textit{c.f.}~\cite[Chapter~1]{Lind-Marcus_1995:introduction-symbolic-dynamics}), a shift $\smash{Y \subseteq \alphabet[B]^{\Z^d}}$ is sofic if and only if it is the symbol-by-symbol projections $\pi \colon \alphabet \to \alphabet[B]$ of some (nearest-neighbor) shift of finite type $X \subseteq \alphabet^{\Z^d}$.

\subsubsection*{Wang tiles}
For a set $\alphabet[C]$ of colors and a set $\alphabet[L]$ of labels, a \emph{(decorated) Wang tile} is a tuple ${\wang{t} = (c_1,\dots,c_{2d},c_{\decsymbol}) \in \alphabet[C]^{2d} \times \alphabet[L]}$. We identify Wang tiles with colored unit cubes by arbitrarily indexing its facets: more precisely, for any $1 \leq k \leq d$, we denote $\facet{k}{-}(\wang{t}) = c_{k}$ and $\facet{k}{+}(\wang{t}) = c_{d+k}$ the colors of the facets $\facet{k}{\pm}$; and $\dec(\wang{t}) = c_{\decsymbol}$ the \emph{decoration} of the tile $\wang{t}$. A \emph{Wang tileset} is then a set of tiles $\Swang \subseteq \alphabet[C]^{2d} \times \alphabet[L]$.

Any Wang tileset $\Swang$ defines a set of valid configurations in which any two adjacent tiles bear the same color on their common facet: $X_{\Swang} = \{x \in \Swang^{\Z^d} : \forall \vec{i} \in \Z^d, \forall 1 \leq k \leq d\, \facet{k}{+}(x_{\vec{i}}) = \facet{k}{-}(x_{\vec{i} + \basis{k}}) \}$. If $\Swang$ is \emph{finite}, then $X_{\Swang}$ is compact and thus defines a \emph{nearest-neighbor} shift of finite type.
Conversely, any SFT is essentially equivalent to some set of tilings by a finite Wang tileset.

\smallskip
In this article, we consider \emph{decorated} Wang tiles for two reasons. First, they allow some notational and writing convenience in the proof of our main \Cref{sof:lem:fixpoint}. Second, the configurations of a sofic shift can actually be obtained by projecting the tilings on some decorated Wang tileset
$\Swang$ to their decorations.

\subsection{Substitutions, ndill maps and limit spaces}
\label{prel:sec:subst}

This subsection provides more specific background for our main results in \Cref{sof:sec}: namely, substitutions, ndill maps (\textit{i.e.}~substitutions ``with context'') and the resulting limit shift spaces.

\subsubsection*{Substitutions and ndill maps}

Recall that for $\alphabet$ an alphabet, $\alphabet^{\Srect_0}$ is the set of normalized $d$\nobreakdash-dimensional rectangular patterns over the alphabet $\alphabet$. Rectangular substitutions are then defined as follows:
\begin{definition}[Substitution]
  Given two alphabets $\alphabet$ and $\alphabet[B]$, a $d$-dimensional \emph{(rectangular) substitution} is a multi-valued map $\subst \colon \alphabet \mto \alphabet[B]^{\Srect_0}\!$.\\
  A substitution $\subst$ extends from individual symbols of $\alphabet$ to whole configurations as follows: for $\smash{x \in \alphabet^{\Z^d}}$, we define $\gsubst(x)$ as the set of configurations $\smash{y \in \alphabet[B]^{\Z^d}}$ for which there exists a grid $\grid{g} = (\rect{r}_{\vec{i}})_{\vec{i} \in \Z^d}$ such that $\rnorm{\restr{y}{\rect{r}_{\vec{i}}}} \in \subst(x_{\vec{i}})$: for such a grid $\grid{g}$ and configuration $y$, we denote $\smash{x \substep[\subst][\grid{g}] y}$.
\end{definition}

In the literature, non-deterministic substitutions are also known as \emph{random} substitutions. \cite{Rust-Spindeler_2018:dynamical-systems-arising-from-random-substitutions} began a first systematic study of random substitutions over $\Z$, mostly focusing on the dynamical and ergodic properties of the associated shifts;
and in~\cite{Gohlke-Rust-Spindeler_2019:shift-finite-type-random-substitutions}, the authors prove that one-dimensional random substitutions are general enough to generate (up to conjugacy) all the transitive shifts of finite type.
In higher dimensions, specific random substitutions have been studied (\textit{e.g.}~\cite{Godrèche-Luck_1989:quasiperiodicity-random-tilings-plane} for Penrose tilings).
Outside of symbolic dynamics, L-systems~\cite{Rozenberg-Salomaa_1976:mathematical-theory-of-L-systems} appear to form a closely related framework, but their limit spaces are not generally studied for their dynamical properties.

Usually, substitutions are extended from individual letters to whole configurations by assuming or providing some gluing/compatibility conditions on the images of adjacent letters~(see \textit{e.g.}~\cite{Aubrun-Sablik_2014:multid-effective-s-adic-subshifts-are-sofic}): in our definition, we circumvent these issues altogether by only considering in $\ndill(x)$ the configurations that explicitly admit a desubstitution along a grid.

In this article, we additionally consider substitutions that are ``context-sensitive'', \textit{i.e.}~whose images depend on the symbols held by neighboring cells.
In their deterministic version, these objects were introduced in~\cite{Salo-Törmä_2015:block-maps-between-primitive-uniform-pisot-substitutions} under the name of \emph{dill maps}\footnote{According to~\cite{Salo-Törmä_2015:block-maps-between-primitive-uniform-pisot-substitutions}, ``dill'' comes from Deterministic Interactive Lindenmayer systems on Long words.}:

\begin{definition}[Ndill map]
  An \emph{ndill map}\footnote{Since the ``d'' in \emph{dill map} means ``deterministic'', we looked for alternative names when defining their non-deterministic counterparts. We decided to avoid \emph{nill maps} due to the possible confusions with nilsystems, and eventually settled upon \emph{ndill map}, whose pronunciation is left to the reader.} is the composition of a shift morphism and a substitution. More precisely, given two alphabets $\alphabet$ and $\alphabet[B]$ and a finite neighborhood $V \subseteq \Z^d$, a $d$-dimensional \emph{ndill map} is a non-deterministic map $\ndill \colon \alphabet^V \mto \alphabet[B]^{\Srect_{0}}\!$ that extends to whole configurations as follows: for $x \in \alphabet^{\Z^d}$, we define $\ndill(x)$ as the set of configurations $y \in \alphabet[B]^{\Z^d}$ for which there exists a grid $\grid{g} = (\rect{r}_{\vec{i}})_{\vec{i} \in \Z^d}$ such that $\rnorm{\restr{y}{\rect{r}_{\vec{i}}}} \in \ndill(\restr{x}{\vec{i} + V})$.
  We still write $\smash{x \substep[\subst][\grid{g}] y}$.
  If $V \subseteq \interval{-r}{r}^d$, we say that $\ndill$ has \emph{radius} $r \in \N$.
\end{definition}

In particular, substitutions form exactly the ndill maps of radius $0$; and morphisms are deterministic ndill maps whose images consist of a single cell. For readability purposes, we will sometimes denote $\smash{x \substep[\ndill] y}$ if $y \in \ndill(x)$. For more information about (one-dimensional deterministic) ndill maps, we refer to~\cite{Salo-Törmä_2015:block-maps-between-primitive-uniform-pisot-substitutions,Radhane_2023:phd:symbolic-dynamics-systems-defined-via-edit-distances}.

\medskip
Notice that, with our definitions, substitutions and ndill maps are non-uniform (\textit{i.e.}~all image patterns do not necessarily have the same shape), and not necessarily total (that is, $\ndill(x)$ might be empty for some configurations~$x$: this feature is actually very useful in the examples from \Cref{app:sec}). As multifunctions, they are also non-deterministic.

Finally, we say that an ndill map $\ndill \colon \alphabet^{V} \to \alphabet[B]^{\Srect_0}$ is \emph{expanding} if for every pattern $u \in \alphabet^{V}$ and every image $w \in \ndill(u)$, the domain $\dom(w) = \ibox{n_1,\dots,n_d} \in \Srect_0$ satisfies $n_k \geq 2$ for every $1 \leq k \leq d$.
{In the rest of this paper, we essentially only consider expanding ndill maps.}

\enlargethispage{\baselineskip}
\subsubsection*{$S$-adic and $N$-adic spaces}

A substitution $\subst \colon \alphabet \mto \alphabet^{\Srect}$ naturally defines a \emph{limit} shift space $\smash{\limspace{X}[\subst] \subseteq \alphabet^{\Z^d}}$ by considering the limit configurations obtained after infinitely many substitution steps. While substitutive words and shift spaces have been extensively studied in the one-dimensional setting~\cite{PytheasFogg_2002:substitutions-dynamics-arithmetics-combinatorics,Akiyama-Arnoux_2020:substitution-tiling-dynamics-self-inducing-structures}, they also frame one of the first major results in multidimensional symbolic dynamics: if $\alphabet$ is a finite alphabet and $\subst$ has finite image, then $\limspace{X}[\subst]$ is actually a sofic shift~\cite{Mozes_1989:tiling_substitution_systems_dynamical_systems_generated_by_them} in dimension $d \geq 2$.

$S$-adic systems extend the notion of substitutive shifts by allowing a more flexible structure of substitutions: instead of considering a single substitution $\subst$, an $S$-adic configuration is obtained as a limit of the form
\[ x = x^{(0)} \rsubstep[\subst_0] x^{(1)} \rsubstep[\subst_1] x^{(2)} \rsubstep[\subst_2] x^{(3)} \dots \]
for $(\subst_{\ell})_{\ell \in \N}$ an infinite sequence of substitutions. One can naturally impose various restrictions on such sequences: operating on a fixed alphabet, ranging over a finite set of substitutions, etc\dots\
They enjoy abundant literature in the one-dimensional setting~\cite{Berthé-Delecroix_2014:beyond-substitutive-dynamics-systems-s-adic-expansions} (though the most usual definition slightly differs from ours, see \Cref{sof:sec:comments}), and a natural generalization of~\cite{Mozes_1989:tiling_substitution_systems_dynamical_systems_generated_by_them}'s theorem to multidimensional $S$-adic systems has been proved in~\cite{Aubrun-Sablik_2014:multid-effective-s-adic-subshifts-are-sofic}.

Like sequences of substitutions can define an $S$-adic configuration, ndill maps can be used to define limit \emph{$N$-adic} configurations. More precisely, for $(\ndill_{\ell})_{\ell \in \N}$ a sequence of ndill maps where $\ndill_{\ell} \colon \alphabet_{\ell+1}^{V_{\ell}} \mto \alphabet_{\ell}^{\Srect_0}$,
we say that a configuration $\smash{x \in \alphabet_{0}^{\mspace{2mu}\Z^d}}$ is \emph{$N$-adic} if there exists a sequence of configurations $\smash{(x^{(\ell)})_{\ell \in \N}}$ such that $\smash{x^{(\ell)} \in \alphabet_{\ell}^{\mspace{2mu}\Z^d}}$ and
\[ x = x^{(0)} \rsubstep[\ndill_0] x^{(1)} \rsubstep[\ndill_1] x^{(2)} \rsubstep[\ndill_2] x^{(3)} \dots \]
The sequence $(\ndill_{\ell})_{\ell \in \N}$ is called a \emph{directive sequence} for $x$.

\begin{definition}[name={$N$-adic limit space}, label={sof:def:limit-space}]
  Let $(\alphabet_{\ell})_{\ell \in \N}$ be a sequence of finite alphabets, and let $(\Sndill_{\ell})_{\ell \in \N}$ be a sequence of finite sets $\Sndill_{\ell} \subseteq \big(\alphabet_{\ell+1}^{V_{\ell}} \mto \alphabet_{\ell}^{\Srect_0} \big)$ of ndill maps. A shift space $\smash{X \subseteq \alphabet_0^{\mspace{2mu}\Z^d}}$ and a set $\Sndill \subseteq \prod_{\ell \in \N} \Sndill_{\ell}$ of \emph{directive sequences} together define a \emph{$N$-adic limit shift space}
  \[ \limspace{X}[\Sndill] = \big\{ x^{(0)} \in X : \exists (\ndill_{\ell})_{\ell \in \N} \in \Sndill,\, \exists (x^{(\ell)})_{\ell \geq 1},\; x^{(0)} \rsubstep[\ndill_0] x^{(1)} \rsubstep[\ndill_1] x^{(2)} \rsubstep[\ndill_2] x^{(3)} \dots \big\}. \]
  If all ndill maps have radius $0$ (they are substitutions), we talk about $S$-adic limit shift space.
  When the initial shift $X$ is unspecified, the limit space $\limspace{X}[\Sndill]$ will implicitly be defined for $X = \smash{\alphabet_0^{\mspace{2mu}\Z^d}}$.
\end{definition}


\enlargethispage{\baselineskip}
\paragraph*{Alternative limit spaces}
As is usual with substitutions, one can actually define two possibly different limit spaces: for $\Sndill$ a set of directive sequences, one can consider
\begin{align*}
  \limspace{X}[\Sndill] & = \big\{ x \in X : \exists (\ndill_{\ell})_{\ell \in \N} \in \Sndill,\, \forall n \in \N, \exists (x^{(\ell)})_{0 \leq \ell \leq n+1},\; x = x^{(0)} \rsubstep[\ndill_0] \dots \rsubstep[\ndill_n] x^{(n+1)} \big\}, \\
  \intertext{\textit{i.e.}~the limit shift space (as defined above) whose configurations can be infinitely de-substituted along a directive sequence of $\Sndill$; or the shift space in which we only consider subpatterns of iterated substitions of ``single symbols''\footnotemark\ from the alphabets $\alphabet_{\ell}$'s:}
  \limspace{X}[\Sndill][\subpattern] & = \big\{ x \in X : \exists (\ndill_{\ell})_{\ell \in \N} \in \Sndill,\, \forall w \subpattern x, \exists n \in \N, \exists (x^{(\ell)})_{0 \leq \ell \leq n+1},\; x^{(0)} \rsubstep[\ndill_0][\grid{g}_{0}] \dots \rsubstep[\ndill_n][\grid{g}_{n}] x^{(n+1)} \\[-9pt]
  & \hspace{6.5cm} \text{ and } w \subpattern \restr{x^{(0)}}{\pxrect{r}_{\vec{0}}} \text{ } \big(\text{where } (\pxrect{r}_{\vec{i}})_{\vec{i} \in \Z^d} = \bigcomp_{\ell=0}^{n} \grid{g}_{\ell} \big) \big\}.
\end{align*}
\footnotetext{In the case of substitutions instead of ndill maps, the second condition can be rewritten: ${\forall w \subpattern x},\allowbreak {\exists n \in \N},\allowbreak {\exists a \in \alphabet_n},\allowbreak w \subpattern \subst_0 \circ \dots \circ \subst_n(a)$.}
In this article, we only prove the soficity of $\limspace{X}[\Sndill]$ for some sets of directive sequences $\Sndill$. The soficity of the associated $\limspace{X}[\Sndill][\subpattern]$ is usually harder to prove, and might depend on technical properties of the ndill maps: we claim that, in later \Cref{sof:thm:square-d-adic-sofic,sof:thm:bounded-d-adic-sofic}, additionally respecting the \emph{Property~A} from~\cite{Mozes_1989:tiling_substitution_systems_dynamical_systems_generated_by_them} (see also~\cite[Section~3.1]{Aubrun-Sablik_2014:multid-effective-s-adic-subshifts-are-sofic}) is sufficient to prove the soficity of the limit space $\limspace{X}[\Sndill][\subpattern]$.

\paragraph*{Ndill maps and SFTs} Our main motivation for considering ndill maps in this article comes from their ability to define additional validity conditions in the $S$-adic structure of a configuration. More precisely, let $(X_{\ell})_{\ell \in \N}$ be an arbitrary sequence of shifts of finite type $X_{\ell} \subseteq \alphabet_{\ell}^{\mspace{2mu}\Z^d}\!$. For any sequence $(\subst_{\ell})_{\ell \in \N}$ of substitutions $\subst_{\ell} \colon \alphabet_{\ell+1} \mto \alphabet_{\ell}^{\Srect_0}$, one can consider the limit configurations $x \in X_0$ such that there exists $(x^{(\ell)})_{\ell \in \N}$ satisfying:
\begin{center}
  \begin{tikzpicture}[scale=1.2]
    \foreach \i in {0,...,3} {
      \node (n\i) at (\i,0) {$x^{(\i)}$};
      \draw (\i,0) node[anchor=west,rotate=-45,shift={(0.15,-0.1)}] {$\scriptstyle \in X_{\i}$};
    }
    \node[anchor=east] at (n0.west) {$x\vphantom{x^{(0)}} = \!\!$};
    \foreach \i in {0,...,2} {
      \tikzmath{int \j; \j=\i+1;}
      \draw[<-] (n\i) edge node[above,shift={(0.04,0)}] {$\scriptstyle \subst_{\i}$} (n\j);
    }
    \node[anchor=west] at (n3.east) {$\!\!\vphantom{x^{(0)}}\dots$};
  \end{tikzpicture}
\end{center}
\textit{i.e.}~such that each $x^{(\ell)}$ is furthermore a valid configuration in the SFT $X_{\ell}$. In other words, the corresponding limit configurations are obtained by an alternation of infinitely many substitution steps and validity conditions of finite type.

We claim that such limit configurations can actually be obtained by a corresponding sequence of ndill maps. Indeed, denoting $V_{\ell}$ a neighborhood of the SFT $X_{\ell+1}$, \textit{i.e.}~a finite subset of~$\Z^d$ large enough to contain a family of forbidden patterns, define the ndill maps $\ndill_{\ell}' \colon \alphabet_{\ell+1}^{V_{\ell}} \mto \alphabet_{\ell}^{\Srect_0}$:
\[ \ndill_{\ell}'(u) = \begin{cases}
  \subst_{\ell}(u_{\vec{0}}) & \text{if $u$ is locally valid in $X_{\ell+1}$} \\
  \emptyset & \text{otherwise.}\end{cases}
\]
As each ndill map $\tau_{\ell}'$ checks the validity of its source configuration in $X_{\ell+1}$, the limit configurations in $X_0$ obtained by the ndill maps $(\ndill_{\ell}')_{\ell \in \N}$ exactly correspond to the limit configurations by the substitutions $(\subst_{\ell})_{\ell \in \N}$ in which all intermediate configurations are valid in their respective SFTs. In particular, ndill maps of radius $1$ allow to consider substitutions of valid Wang tilings.

\subsubsection*{Parallelism and computability}
\label{prel:sec:parallel-dill-maps}

As a way to parallelize the computations of an ndill map~$\ndill$, we will distribute the computations $w \in \ndill(u)$ over the individual cells of the image rectangles $\rect{r} = \dom(w)$. To ensure that the concatenation of these independent subcomputations results in a valid image of $\ndill$, we define a notion of parallel ndill map:

\begin{definition}[Parallel ndill map]
  An ndill map $\ndill \colon \alphabet^V \mto \alphabet[B]^{\Srect_{0}}$ is \emph{parallel} if there exists some \emph{local function} $\varphi \colon \Srect_{0} \times \Z^d \times \alphabet^V \mto \alphabet[B]$ that satisfies:
  \[ \forall u \in \alphabet^V, \; \ndill(u) = \!\!\bigcup_{\rect{r} \in \Srect_{0}} \{ w \in \alphabet[B]^{\rect{r}} : \forall \vec{i} \in \rect{r},\, w_{\vec{i}} \in \varphi(\rect{r},\vec{i},u) \}. \]
\end{definition}

\begin{remark}[label={prel:rem:parallelism}] While all deterministic computable ndill maps are parallel, notice that for any rectangle $\rect{r} \in \Srect_0$ and input $u \in \alphabet^V$, the patterns $\{ w \in \alphabet[B]^{\rect{r}} : {w \in \ndill(u)} \} = \prod_{\vec{i} \in \rect{r}} \varphi(\rect{r},\vec{i},u)$ are a product of independent pixels ranges. As a consequence, even the simple substitution
  \[ \subst \colon \square[white] \mapsto \begin{tikzpicture}[baseline=0.85ex,x=1.4ex,y=1.4ex]
      \fill[black!85] (0,1) rectangle (2,2);
      \draw[TilingGrid,step=1] (0,0) grid (2,2);
    \end{tikzpicture}
    \qquad
    \subst \colon \square[white] \mapsto \begin{tikzpicture}[baseline=0.85ex,x=1.4ex,y=1.4ex]
      \fill[black!85] (0,0) rectangle (2,1);
      \draw[TilingGrid,step=1] (0,0) grid (2,2);
    \end{tikzpicture}
    \qquad
    \subst \colon \square[black!85] \mapsto \begin{tikzpicture}[baseline=0.85ex,x=1.4ex,y=1.4ex]
      \fill[black!85] (0,0) rectangle (1,2);
      \draw[TilingGrid,step=1] (0,0) grid (2,2);
    \end{tikzpicture}
    \qquad
    \subst \colon \square[black!85] \mapsto \begin{tikzpicture}[baseline=0.85ex,x=1.4ex,y=1.4ex]
      \fill[black!85] (1,0) rectangle (2,2);
      \draw[TilingGrid,step=1] (0,0) grid (2,2);
    \end{tikzpicture}
  \]
  is not parallel. Nevertheless, it is always possible to change the ``work alphabet'' of the ndill maps to realize a non-parallel $N$-adic limit space using an alternative well-chosen $N$-adic representation. More precisely, for $(\ndill_{\ell})_{\ell \in \N}$ a sequence of arbitrary ndill maps $\ndill_{\ell} \colon \alphabet_{\ell+1}^{V_{\ell}} \mto \alphabet_{\ell}^{\Srect_0}$, define:
  \begin{itemize}
  \item The alphabets $\alphabet[B]_{\ell}$ as the sets of symbols $(a,v) \in \alphabet_{\ell} \times \alphabet_{\ell-1}^{\Srect_0}$; and $\alphabet[B]_{0} = \alphabet_{0}$;
  \item The ndill maps $\ndill_{\ell}' \colon \alphabet[B]_{\ell+1}^{V_{\ell}} \to \alphabet[B]_{\ell}^{\Srect_0}$ by the following local maps $\varphi_{\ell}$: for any pattern  $u' \in \alphabet[B]_{\ell+1}^{V_{\ell}}$, if we denote $(a,v) = u'_{\vec{0}} \in \alphabet[B]_{\ell+1}$ the central tile in $u'$ and $\smash{u \in \alphabet_{\ell+1}^{V_{\ell}}}$ the projection of $u'$ to the alphabet $\alphabet_{\ell+1}$; then for any $\rect{r} \in \Srect_0$, $\vec{i} \in \rect{r}$ and any symbol $(b,w) \in \alphabet[B]_{\ell}$, we have $(b,w) \in \varphi_{\ell}(\rect{r},\vec{i},u')$ if and only if $v \in \ndill_{\ell}(u)$, $\rect{r} = \dom(v)$ and $b = v_{\vec{i}}$.
  \end{itemize}
  Then the ndill maps $\ndill_{\ell}' \colon \alphabet[B]_{\ell+1}^{V_{\ell}} \mto \alphabet[B]_{\ell}^{\Srect_0}$ are parallel and realize the same set of limit configurations. In other words, generic ndill maps yield limit spaces that can also be realized by \emph{parallel} ndill maps, at the cost of operating on larger alternative alphabets; and where the size blowup between $\alphabet_{\ell}$ and $\alphabet[B]_{\ell}$ depends on the size of the images of the ndill maps $\ndill_{\ell-1}$.
\end{remark}

In order to prove the soficity of a $N$-adic shift space, we will consider additional computational restrictions on our directive sequences of ndill maps. More precisely:
\begin{definition}
  A parallel ndill map $\ndill \colon \alphabet^V \mto \alphabet[B]^{\Srect_0}$ is said to be \emph{computable} if there exists a RAM program $e \in \{0,1\}^{\ast}$ computing its local map, \textit{i.e.}~such that:
  \[ \forall u \in \alphabet^V, \; \ndill(u) = \!\!\bigcup_{\rect{r} \in \Srect_{0}} \{ w \in \alphabet[B]^{\rect{r}} : \forall \vec{i} \in \rect{r},\, w_{\vec{i}} \in \funram{e}(\rect{r},\vec{i},u) \}. \]

  \vspace*{-0.5em}
  \noindent Furthermore, it is \emph{computable in time} $t \in \N$ if $\funram{e}(\rect{r},\vec{i},u)\halt[t] = \funram{e}(\rect{r},\vec{i},u)$.
\end{definition}

As is usual in computability theory, any statement operating on multiple computable objects requires some \emph{uniformity conditions}, \textit{i.e.}~that a single program $e \in \{0,1\}^{\ast}$ can enumerate them all. To this end, we say that:
\begin{itemize}
\item A sequence of parallel ndill maps $(\ndill_{\ell})_{\ell \in \N}$ is \emph{computable in time $t(\ell)$} if there exists a program $e \in \{0,1\}^{\ast}$ such that, for all $\ell \in \N$, the function $(\rect{r},\vec{i},u) \mapsto \funram{e}(\ell,\rect{r},\vec{i},u)\halt[t(\ell)]$  is a local map of $\ndill_{\ell}$;
\item A sequence $(\Sndill_{\ell})_{\ell \in \N}$ of sets $\Sndill_{\ell} = (\ndill_{\ell,i})_{i \in I_{\ell}}$ of parallel ndill maps is \emph{computable in time $t(\ell)$} if there exists a program $e \in \{0,1\}^{\ast}$ such that, for all $\ell \in \N$ and $i \in I_{\ell}$, the function $(\rect{r},\vec{i},u) \mapsto \funram{e}(\ell,i,\rect{r},\vec{i},u)\halt[t(\ell)]$ is a local map of $\ndill_{\ell,i}$.
\end{itemize}

\subsection{Recognizability and simulation}

We conclude these definitions by relating the dual notions of \emph{simulation} (from computer science), upon which the fixed-point construction is based; and the notion of \emph{recognizable substitutions} (from mathematics), which allows to recover the source configurations from their substitutions.

More precisely, the \emph{recognizability} of a substitution is a strong injectivity property, which amounts to saying that preimages and the associated grids can always be recognized in substituted configurations. Since the seminal theorem by Mossé that primitive $\Z$-substitutions are recognizable over their limit space \cite{Mossé_1992:puissance-de-mots-et-reconnaissabilité-points-fixes-substitutions}, the notion has appeared under various names and turned out to be a convenient property to prove dynamical properties on each level of an $S$-adic structure:

\begin{definition}[name={Recognizability},label={fix:def:rec}]
  An ndill map (or a substitution) $\tau:\alphabet^{V}\mto\alphabet[B]^{\Srect_0}$ is \emph{recognizable} over $\smash{Y\subset\alphabet[B]^{\Z^d}}$ if for every $y \in Y$, there exist at most one grid $\grid g$ and one $\smash{x\in\alphabet^{\Z^d}}$ such that $\smash{x \substep[\subst][\grid{g}] y}$.
  A directive sequence $(\tau_\ell)_{\ell\in\N}$ is \emph{recognizable} if every level $\ell\in\N$ can be recognized, successively: $\tau_\ell$~is recognizable over the set of configurations $\smash{x^{(\ell)} \in \alphabet_\ell^{\mspace{2mu}\Z^d}}$ which appear in some sequence
  \[ x = x^{(0)} \rsubstep[\ndill_0] x^{(1)} \rsubstep[\ndill_1] x^{(2)} \rsubstep[\ndill_2] x^{(3)} \dots\]
  for some configurations $\smash{x^{(k)} \in \alphabet_k^{\mspace{2mu}\Z^d}}$.
\end{definition}
\noindent One way to enforce the recognizability of a substitution consists in explicitly drawing/hard-coding the grid and the preimage symbols in the substituted configurations. This method is actually known under another name: the notion of simulation, in a sense close to what was developed in \cite{Delorme-Mazoyer-Ollinger-Theyssier_2011:bulking-2-classifications-cellular-automata,Durand-Romashchenko-Shen_2012:fixed-point-tilesets-and-applications}. If we denote $\Sgrid$ the set of all grids on $\Z^d$:

\begin{definition}[name={Simulation},label={fix:def:tiling-simulation}]
  A (block) \emph{simulation} of a shift $\smash{Y \subseteq \alphabet[B]^{\Z^d}}$ by a shift $\smash{X \subseteq \alphabet^{\Z^d}}$ is a partial map $\phi \colon \alphabet^{\Srect_0} \to \alphabet[B]$ such that the resulting partial mapping $\Phi \colon \alphabet^{\Z^d}\mspace{-6mu} \times \Sgrid \to \alphabet[B]^{\Z^d}$ defined by $\Phi(x,(\rect{r}_{\vec i})_{\vec i\in\Z^d}) = (\phi(\rnorm{\restr{x}{\rect{r}_{\vec i}}})_{\vec i\in\Z^d}\!)$ satisfies:
\begin{itemize}
\item For every configuration $x \in X$, there exists a unique grid $\grid{g}$ such that $(x,\grid g)\in\dom(\Phi)$;
\item The image set of $X$ by $\Phi$ is $Y$.
\end{itemize}
We say that $X$ \emph{simulates} $Y$ (through simulation $\phi$).

In the particular context of Wang tilesets, a (block) \emph{simulation} of a tileset $\Swang[][T]$ by another tileset~$\Swang[][S]$ is a simulation between the corresponding shifts of finite type that preserves the local validity of the tilings: given blocks $w,w' \in \Swang[][S]^{\Srect_0}$ compatible along the direction $\basis{k}$ (\textit{i.e.}~satisfying ${\facet{k}{+}(w) = \facet{k}{-}(w')}$), the tiles $\sigma(w)$ and $\sigma(w')$ must also be compatible along $\basis{k}$.

We call \emph{macro-tiles} the rectangular $\alphabet$-patterns that belong to the domain $\dom(\sigma)$.
\end{definition}

In particular, the inverse $\sigma^{-1}$ of a simulation $\sigma$ is a substitution that is recognizable over the full shift. Note however that this substitution may not be deterministic: the injectivity requirements in the definition focus only on the grids, and not the simulated symbols themselves.

\enlargethispage{-\baselineskip}
\section{Fixed-point lemma and sofic realization}
\label{sof:sec}

We now develop the main results of this article. In the context of self-simulation, the fixed-point construction was introduced in a series of articles~\cite{Durand-Romashchenko-Shen_2008:fixed-point-aperiodic-tilings}--\cite{Durand-Romashchenko-Shen_2012:fixed-point-tilesets-and-applications}; themselves based on earlier work on self-simulating cellular automata~\cite{Gács_1986:reliable-computation-cellular-automata,Gács_2001:reliable-cellular-automata-self-organization}. For a more complete history of the construction, we refer to the extensive survey~\cite{Törma_2021:fixed-point-constructions-tilings-cellular-automata}.

This construction creates shifts of finite type that simulate an infinite hierarchy of configurations.
It provides a flexible framework that can be adapted to a large variety of contexts (\textit{e.g.}~projections and subdynamics of SFTs~\cite{Durand-Romashchenko-Shen_2010:effective-closed-subshifts-1D-implemented-2D}, characterization of expansive directions of SFTs~\cite{Zinoviadis_2016:phd:expansiveness-2d-sfts}, construction of a nontrivial uniquely ergodic cellular automaton~\cite{Törmä_2015:uniquely-ergodic-cellular-automaton}, soficity of complex shifts~\cite{Westrick_2017:seas-of-squares}, minimal shifts~\cite{Durand-Romashchenko_2021:expressiveness-quasiperiodic-minimal-sfts}\dots).
In this section, we provide a technical abstraction of the construction under the formalism of non-deterministic substitutions on infinite tilesets, along with several rephrasings in terms of $N$-adic limit spaces.

\pagebreak
\subsection{Main fixed-point lemma}

Recall that for $\rect{r} \in \Srect_0$ a rectangle, we denote $\rtime(\rect{r})$ and $\rspace(\rect{r})$ the length of its longest edges and the area of its smallest facets, respectively referred to as the \emph{time} and \emph{space} of $\rect{r}$. These quantities extend to an infinite grid $\grid{g}$ by defining $\rtime(\grid{g}) = \sup_{\rect{r} \in \grid{g}} \rtime(\rect{r})$ and $\rspace(\grid{g}) = \inf_{\rect{r} \in \grid{g}} \rspace(\rect{r})$. Finally, we denote by $\Swang[\ast]$ the infinite set of $d$-dimensional Wang tiles with binary string colors $\alphabet[C] = \{0,1\}^{\ast}$.

Our main lemma provides a positive criterion for a limit space of Wang tilings to be sofic:
\begin{lemma}[name={Fixed-point lemma},label={sof:lem:fixpoint}]
  Let $\code{\tau} \in \{0,1\}^{\ast}$ be a $\log$-RAM program defining a parallel expanding substitution ${\tau \colon \Swang[\ast] \mto \Swang[\ast]^{\Srect_0}}$ operating on the tileset $\Swang[\ast]$. Let also $\alpha < 1$ and $\delta < \frac{1}{2}$ be two rational numbers, $\K \in \N$ a bound, $\alphabet \subseteq \{0,1\}^{\ast}$ a finite alphabet and $X \subseteq \smash{\alphabet^{\Z^d}}$ a sofic shift. If $\limspace{X}[\code{\tau}]$ denotes the set of configurations $x \in X$ such that:
  \begin{enumerate}[itemsep=2pt]
  \item there exists a sequence of grids $(\grid{g}_{\ell})_{\ell \in \N}$ such that, denoting $\pxspace{\ell} = \prod_{i=0}^{\ell-1} \rspace(\grid{g}_{i})$, we have:
    \begin{enumerate}[label=(\roman*),itemsep=2pt,topsep=1pt]
    \item $2 \leq \rtime(\grid{g}_{\ell}) \leq 2^{\K \cdot \pxspace{\ell}[\delta]}$ {\small (upper bound on the size of the rectangles)};
    \item $\rspace(\grid{g}_{\ell})^{\K \cdot \log \pxspace{\ell+1}} \geq \rtime(\grid{g}_{\ell})$ {\small (upper bound on the eccentricity of the rectangles)};
    \end{enumerate}
  \item there exists a sequence of \underline{valid} Wang tilings $(x^{(\ell)})_{\ell \in \N}$ in $(\Swang[\ast])^{\Z^d}$ such that each tiling $x^{(\ell+1)}$ satisfies $\smash{x^{(\ell+1)} \substep[\tau][\grid{g}_{\ell}] x^{(\ell)}}$ and:
    \begin{enumerate}[label=(\roman*),itemsep=2pt,topsep=2pt]
    \item for every $\ell \in \N$, the colors in $x^{(\ell)}$ have length at most $\norm{x^{(\ell)}} = \K[t] \cdot \pxspace{\ell-1}[\alpha]$;
    \item for every $\ell \in \N$, the substitution step $\smash{x^{(\ell+1)} \substep[\tau][\grid{g}_{\ell}] x^{(\ell)}}$ is computed in time $\K[t] \cdot \pxspace{\ell}[\alpha]$;
    \item $\dec(x^{(0)}) = x$, \textit{i.e.}~the decorations of $x^{(0)}$ project onto $x$.
    \end{enumerate}
  \end{enumerate}
  \noindent Then $\limspace{X}[\code{\tau}]$ is a sofic shift space.
\end{lemma}

\medskip
The rest of \Cref{sof:sec} discusses some properties of \Cref{sof:lem:fixpoint} and develops its implications in terms of $S$-adic and $N$-adic spaces (more concrete examples will be considered in in \Cref{app:sec}). Due to its length, the proof of \Cref{sof:lem:fixpoint} is postponed to \Cref{fix:sec}.

\enlargethispage{\baselineskip}
\subsection{Additional features}

The statement of our main~\Cref{sof:lem:fixpoint} emphasizes the pleasant (existential) soficity result, in order to make it short and understandable. In this section, we sketch some additional features that can be enforced in conjunction with \Cref{sof:lem:fixpoint}, and that can be easily seen when understanding the main construction.

\subsubsection{Effectivity}\label{sof:par:comment-constructivism}
Even though the notion of shift soficity is intrinsically existential (given a shift, does it admit a cover of finite type?), \Cref{sof:lem:fixpoint} is actually constructive: when given the program ${\code{\tau} \in \{0,1\}^{\ast}}$ and the constants $\alpha,\gamma < 1$, the proof in \Cref{fix:sec} effectively builds a cover of finite type $\mathcal{X}$ for the limit space $\limspace{X}[\code{\tau}]$.

In particular, \Cref{sof:lem:fixpoint} can be used to recover many undecidability results from the literature on shifts of finite type: for example, the \emph{Domino problem} (is a given SFT empty?) can be proved undecidable by defining substitutions that incrementally compute the halting problem in any given number of steps (see~\cite[Section~5]{Durand-Romashchenko-Shen_2012:fixed-point-tilesets-and-applications}).
One can also note that, given a configuration from $\limspace{X}[\code{\tau}]$, its set of preimages in the cover is computable: in other words, the construction does not introduce any artificial uncomputability in the cover $\mathcal{X}$.

In the following, all corollaries and applications will have to be understood with this constructive point of view: even though they will only claim soficity for the sake of conciseness, \Cref{sof:lem:fixpoint} will actually provide, in each case, an explicit cover of finite type.

\subsubsection{$S$-adic structure of the cover $\mathcal{X}$}\label{sof:par:comment-structure}
Denote $\mathcal{X}$ the cover of finite type for $\limspace{X}[\code{\tau}]$ built in \Cref{fix:sec}, and $\pi \colon \mathcal{X} \to \limspace{X}[\code{\tau}]$ the associated factor map.

Then $\mathcal{X}$ itself can be written as an $S$-adic limit space $\limspace{X}[\code{\tilde\tau}]$.
Furthermore, each macro-tile of each level $\ell\in\N$ in $\limspace{X}[\code{\tau}]$ is actually the projection of a macro-tile with the same support in $\mathcal{X}$ (see the notion of \emph{simulation} from \Cref{fix:def:tiling-simulation}).
Hence there exists a sequence $(\pi^{(\ell)})_{\ell\in\N}$ of projection maps such that, if $x \in \limspace{X}[\code{\tau}]$ comes from an $S$-adic structure
\[ x=x^{(0)} \rsubstep[\tau][\grid{g}_0] x^{(1)} \rsubstep[\tau][\grid{g}_1] x^{(2)} \rsubstep[\tau][\grid{g}_2] x^{(3)} \dots, \]
then any preimage $\chi\in\mathcal{X}$ of $x$ by $\pi$ comes from an $S$-adic structure
\[ \chi=\chi^{(0)} \rsubstep[\tilde\tau][\grid{g}_0] \chi^{(1)} \rsubstep[\tilde\tau][\grid{g}_1] \chi^{(2)} \rsubstep[\tilde\tau][\grid{g}_2] \chi^{(3)} \dots \]
with the same grids, and $\pi^{(\ell)}(\chi^{(\ell)})=x^{(\ell)}$. Furthermore, this $S$-adic decomposition is \emph{recognizable}: hence, each configuration $\chi^{(\ell)}$ can be recovered (in a continuous way) from $\chi$, along with the grids $(\grid{g}_i)_{i < \ell}$ and the configuration $x^{(\ell)}$.

\subsubsection{Entropy}\label{sof:par:comment-entropy}
The \emph{pattern complexity} $\smash{N_X(n) = \card{\{w \in \alphabet^{\ibox{n}^d} : \exists x \in X,\, w \subpattern x \}}}$ counts the number of patterns of domain $\ibox{n}^d$ appearing in the configurations of a shift $\smash{X \subseteq \alphabet^{\Z^d}}$. Its asymptotic growth rate defines a conjugacy invariant known as the (topological) \emph{entropy}:
\[ h(X) = \lim_{n \to +\infty} \frac{\log N_X(n)}{n^d}. \]
Informally, $h(X) = \gamma$ if the pattern complexity roughly amounts to $2^{\mspace{1mu}\gamma \cdot n^d}$. The existence of this limit is standard and follows, for example, from the (multivariate) subadditive lemma~\cite{Capobianco_2008:cellular-automata-generalization-fekete-lemma}.

Our construction can be made more precise as to not introduce any artificial entropy in the finite-type cover $\mathcal{X}$.
More precisely, a pattern in $\mathcal{X}$ is essentially given by the corresponding pattern in $\limspace{X}[\code{\tau}]$ and in the different levels in the $S$-adic hierarchy, as well as the possible computations associated with these substitution steps.
Hence, if one additionally assumes in \Cref{sof:lem:fixpoint} that $\tau$ is recognizable (over its limit set, or even simply injective over finite patterns) and that the program $\code{\tau}$ is \emph{unambiguous} (that is: every valid $\wang{t}' \in \funram{\code{\tau}}(\rect{r},\vec{i},\wang{t})$ is accepted by a single run), then the entropy of $\mathcal X$ is equal to that of $\limspace{X}[\code{\tau}]$.
One could even prove a polynomial relation between the two complexity functions, but we do not enter into these details.

The fact that our construction does not introduce any artificial entropy is closely related to Weiss's question~\cite{Desai_2006:subsystem-entropy-for-Zd-sofic-shifts} about the general existence of finite type covers of equal entropy.

\subsubsection{Ndill maps} Instead of a substitution operating on the infinite tileset $\Swang[\ast]$, it is possible to apply \Cref{sof:lem:fixpoint} with ndill maps $\ndill \colon \alphabet_{\ast}^{\Srect_0} \mto \alphabet_{\ast}^{\Srect_0}$ on the infinite alphabet $\alphabet_{\ast} = \{0,1\}^{\ast}$. Indeed, one can reencode the ndill map $\ndill$ by a substitution on a higher-block code Wang tiling (see~\cite[Chapter~1]{Lind-Marcus_1995:introduction-symbolic-dynamics}) of large enough domain to fit the radius of $\ndill$.

More precisely, \Cref{sof:lem:fixpoint} also holds for parallel expanding and computable ndill maps $\ndill \colon \alphabet_{\ast}^{\Srect_0} \mto \alphabet_{\ast}^{\Srect_0}$ if, in a substitution step $x^{(\ell+1)} \substep[\ndill] x^{(\ell)}$, the radius needed to substitute each cell is bounded by $\bigO(\pxspace{\ell}[\beta])$ for some $\beta > 0$ such that $\alpha + d \cdot \beta < 1$. We apply this principle in \Cref{sof:thm:square-d-adic-sofic,sof:thm:bounded-d-adic-sofic} with polylogarithmic radii.

\subsubsection{Non-parallel ndill maps}\label{sof:par:comment-parallelism}
The notion of parallelism for substitutions or ndill maps (see \Cref{prel:sec:parallel-dill-maps}) allows to distribute the computation of a substitution step geometrically, and result in better numeric bounds in \Cref{sof:lem:fixpoint}. Unfortunately, while deterministic ndill maps are always parallel, non-deterministic ndill maps may fail to satisfy this property.

Following \Cref{prel:rem:parallelism}, \Cref{sof:lem:fixpoint} may also be applied to arbitrary (possibly non-parallel) expanding ndill maps $\ndill \colon \alphabet_{\ast}^{\Srect_0} \mto \alphabet_{\ast}^{\Srect_0}$ if the size of the image patterns is not too large, \textit{i.e.}~if the growth of the grids satisfies $\rtime(\grid{g}_{\ell}) = \bigO(\pxspace{\ell}[\beta'])$ for some $\beta' > 0$ such that $\alpha + d \cdot \beta' < 1$.  This idea will be applied in \Cref{sof:thm:square-d-adic-sofic,sof:thm:bounded-d-adic-sofic} to images of polylogarithmic/constant sizes.

\enlargethispage{\baselineskip}
\subsection{\boldmath Applications of the \texorpdfstring{\protect\crtlnameref{sof:lem:fixpoint}}{fixed-point lemma} to \texorpdfstring{$N$}{N}-adic limit spaces}

To provide some broader context for the \crtlnameref{sof:lem:fixpoint}~(\Cref{sof:lem:fixpoint}), we consider two corollaries: their phrasing in the more familiar setting of $S$-adic/$N$-adic limit spaces should help to understand and contextualize the possibilities of this construction.

For the first corollary, we weaken the possibilities on the grids of \Cref{sof:lem:fixpoint}: instead of allowing substitutions of level $\ell$ to generate a large variety of grid sizes and shapes, we consider a square ndill map $\tau_{\ell}$ of fixed computable size $N_{\ell} \in \N$.
\begin{theorem}[label={sof:thm:square-d-adic-sofic}]
For $d \geq 2$, let $(\subst_{\ell})_{\ell \in \N}$ be a sequence of parallel expanding ndill maps $\smash{\ndill_{\ell} \colon \alphabet_{\ell+1}^{V_{\ell}} \mto
  \alphabet_{\ell}^{\mspace{2mu}\ibox{N_{\ell}}^d}}$ uniformly computed by a $\log$-RAM program $e \in \{0,1\}^{\ast}$.
For $\smash{X \subseteq \alphabet_0^{\mspace{2mu}\Z^d}}$ a sofic shift, and denoting $\Sndill = \{ (\ndill_{\ell})_{\ell \in \N} \}$ and $L_{\ell} = \prod_{i < \ell} N_i$, assume that there exists some $\alpha < d-1$ such that:
  \begin{enumerate}[topsep=2pt,itemsep=1pt]
  \item\label{i:tile} For every $\ell \in \N$, we have $\smash{2 \leq N_{\ell} \leq 2^{\bigO(L_{\ell}^{\mspace{3mu}\delta})}}$ for some $\delta < \frac{d-1}{2}$;
  \item\label{i:symbol} For every $\ell \in \N$, symbols in the alphabet $\alphabet_{\ell+1}$ have bit length $\bigO(L_{\ell}^{\mspace{3mu}\alpha})$;
  \item\label{i:radius} For every $\ell \in \N$, the neighborhood $V_{\ell}$ of the ndill map $\tau_{\ell}$ has radius $\bigO(\polylog L_{\ell})$;
  \item\label{i:comput} For every $\ell \in \N$, the RAM program $\funram{e}(\ell,\ldots)$ computes the ndill map $\ndill_{\ell}$ in time $\bigO(L_{\ell}^{\mspace{3mu}\alpha})$.
  \end{enumerate}

  \smallskip
  \noindent Then the $N$-adic limit space $\limspace{X}[\Sndill]$ is a sofic shift space.
\end{theorem}

For the second corollary, we again weaken the growth of the grids in \Cref{sof:lem:fixpoint}: instead of a fixed sequence of square substitutions, we consider effectively closed sets of directive sequences whose ndill maps are of uniformly bounded sizes.
\begin{theorem}[label={sof:thm:bounded-d-adic-sofic}]
  For $d \geq 2$ and $M \in \N$, let $(\Sndill_{\ell})_{\ell \in \N}$ be computably uniform sets $\Sndill_{\ell} = (\ndill_{\ell,i})_{0 \leq i < M}$ of sublinear expanding ndill maps $\smash{\ndill_{\ell,i} \colon
    \alphabet_{\ell+1}^{V_{\ell}} \mto \alphabet_{\ell}^{\Srect_0}}$ uniformly computed by a $\log$-RAM program $e \in \{0,1\}^{\ast}$ and uniformly bounded by some $M \in \N$, \textit{i.e.}~there exists some $\alpha < 1$ such that:
  \begin{enumerate}[topsep=2pt,itemsep=1pt]
  \item For any $\ell \in \N$ and $i \in \interval{0}{M-1}$, the ndill map $\ndill_{\ell,i}$ has radius at most $M$ and every image pattern $w = \ndill_{\ell,i}(u)$ has bounded domain $\smash{\dom(w) \subseteq \ibox{M}^d}$;
  \item For any $\ell \in \N$, symbols in the alphabet $\alphabet_{\ell}$ have bit length $\bigO(2^{\ell \cdot \alpha})$;
  \item For any $\ell \in \N$ and $i < M$, $\funram{e}(\ell,i,\ldots)$ computes the ndill map $\ndill_{\ell,i}$ in time $\bigO(2^{\ell \cdot \alpha})$.
  \end{enumerate}
  Then for any sofic shift $X \subseteq \alphabet_0^{\mspace{2mu}\Z^d}$ and any effectively closed set $\Sndill \subseteq \prod_{\ell \in \N} \Sndill_{\ell}$ of directive sequences, the $N$-adic limit space $\limspace{X}[\Sndill]$ is also sofic.
\end{theorem}

\noindent Notice that, by the comment in \Cref{sof:par:comment-parallelism}, \Cref{sof:thm:square-d-adic-sofic} can be applied to non-parallel ndill maps with constant or polylogarithmic growth $N_{\ell} = \polylog(L_{\ell})$; while the parallel requirement in \Cref{sof:thm:bounded-d-adic-sofic} has already been dropped altogether.

\begin{proof}[Proof of \Cref{sof:thm:square-d-adic-sofic}]
  Following \Cref{sof:par:comment-parallelism}, we encode configurations on the alphabet~$\alphabet_{\ell}$ by higher-block code tilings of domain $\{ - \bigO(\polylog(L_{\ell})), \dots, \bigO(\polylog(L_{\ell})) \}^d$; and then apply \Cref{sof:lem:fixpoint} by replacing the ndill maps $(\tau_{\ell})_{\ell \in \N}$ by equivalent substitutions.

\noindent More precisely, fix a constant $\K \in \N$ for the $\bigO(\dots)$ of \Cref{sof:thm:square-d-adic-sofic}, \textit{i.e.}~such that for all $\ell \in \N$:
  \begin{itemize}
  \item Images by the ndill map $\tau_{\ell}$ are of size $N_{\ell} \leq 2^{\K \cdot L_{\ell}^{\mspace{2mu}\delta}}$;
  \item Symbols in the alphabet $\alphabet_{\ell+1}$ have bit length $\K \cdot L_{\ell}^{\mspace{2mu}\alpha}$;
  \item The neighborhood $V_{\ell}$ satisfies $V_{\ell} \subseteq \{-\K \cdot (\log L_{\ell})^{\K},\dots, \K \cdot (\log L_{\ell})^{\K} \}^d$;
  \item The RAM program $\funram{e}(\ell,\ldots)$ computes the ndill map $\ndill_{\ell}$ in time $\K \cdot L_{\ell}^{\mspace{2mu}\alpha}$.
  \end{itemize}
  We define a program $\code{\tau}$ satisfying the hypotheses of \Cref{sof:lem:fixpoint} and such that $\limspace{X}[\code{\tau}] = \limspace{X}[\Sndill]$. Up to an implicit binary encoding, we consider colors of the form $(\ell, u) \in \N \times \alphabet_{\ell}^{\Srect_0}$, where:
  \begin{itemize}
  \item $\ell \in \N$ is an integer called the \field{level} of the tuple;
  \item $u \in \alphabet_{\ell}^{\Srect_0}$ is a pattern of domain $\dom(u) \subseteq \{-\K \cdot (\log L_{\ell})^{\K},\dots, \K \cdot (\log L_{\ell})^{\K} \}^d$, called the \field{pattern} of the tuple.
  \end{itemize}
  and denote by $\Scolor$ the associated set of colors. We then define the program $\code{\tau}$ as the following algorithm substituting Wang tiles of $\Scolor^{2d+1}$, and the local map $\funram{\code{\tau}}(\rect{r},\vec{i},\wang{t}) \mapsmapsto \wang{t}'$ given by:
  \begin{itemize}
  \item Denoting $(\ell+1, u)$ the color of $\dec(\wang{t})$, check that $\rect{r} = \ibox{N_{\ell}}^{d}$, that $\vec{i} \in \rect{r}$, and that the pattern $u \in \alphabet_{\ell+1}^{\Srect_0}$ has domain $V_{\ell} = \{-\K \cdot (\log L_{\ell})^{\K},\dots, \K \cdot (\log L_{\ell})^{\K} \}^d$;
  \item (Consistency) Check that all facets $\facet{k}{\pm}(\wang{t})$ also have \field{level} $\ell+1$;
  \item (Higher block code) Check that, if $\facet{k}{\pm}(\wang{t})$ has \field{pattern} $v \in \alphabet_{\ell+1}^{\Srect_0}$, then $v = \rnorm{\restr{u}{V_{\ell} \setminus \facet{k}{\pm}(V_{\ell})}}$,
  \item (Ndill map) Run the computation $\funram{e}(\ell,\rect{r},\vec{i},u)$ and let $a \in \alphabet_{\ell}$ be its output;
  \item (Return) Non-deterministically choose a pattern $v \in \alphabet_{\ell}^{V_{\ell-1}}$ such that $v_{\vec{0}} = a$, and return the tile $\wang{t}' \in \Scolor^{2d+1}$ defined as
    \begin{itemize}
    \item $\dec(\wang{t}') = (\ell,, v)$;
    \item $\facet{k}{\pm}(\wang{t}') = (\ell,\rnorm{\restr{v}{V_{\ell-1} \setminus \facet{k}{\pm}(V_{\ell-1})}})$.
    \end{itemize}
  \end{itemize}
  Then, in a sequence of simulations $\smash{x^{(0)} \rsubstep[\tau] \dots \rsubstep[\tau] x^{(\ell)} \rsubstep[\tau] \dots}$ along a sequence of regular grids of size $(N_{\ell})_{\ell \in \N}$, each substitution step $\smash{x^{(\ell+1)} \substep[\tau] x^{(\ell)}}$ is computed in time $\bigO(L_{\ell}^{\mspace{2mu}\alpha} \cdot \polylog L_{\ell})$. This concludes the proof.
\end{proof}

\begin{proof}[Proof of \Cref{sof:thm:bounded-d-adic-sofic}]
  A similar proof applies: in addition to the \field{level} and \field{pattern} of a color, we also embed in the tiles of \field{level} $\ell \in \N$ a sequence of indices $(i_{k})_{0 \leq k < \ell}$ representing the choice of ndill maps $\ndill_{\ell,i_{k}}$ used when substituting from level $k+1$ to $k$. Said finite sequence is then checked (for some $\bigO(\ell)$ steps of computations) against an enumeration of the forbidden prefixes of the effectively closed set of directive sequence $\Sndill$.
\end{proof}

\pagebreak
\subsection{Comments}\label{sof:sec:comments}

We conclude this section with some general comments about \Cref{sof:lem:fixpoint}.

\paragraph*{Relation to existing literature}
The fixed-point construction depends on a notion of local \emph{simulation} between tilings~\cite[Section~2.1]{Durand-Romashchenko-Shen_2012:fixed-point-tilesets-and-applications} that is actually a form a (de)substitution. While the similarity between simulations and substitutions was already noticed in~\cite{Törma_2021:fixed-point-constructions-tilings-cellular-automata}, our \crtlnameref{sof:lem:fixpoint} is -- to the best of our knowledge -- the first abstraction of the construction that is explicitly phrased in terms of substitutions. Furthermore, \Cref{sof:lem:fixpoint} generalizes existing applications of the construction:
\begin{enumerate}
\item The fixed-point construction traditionally requires the sequence
  $\rspace(\grid{g}_{\ell})_{\ell \in \N}$ to increase rather substantially (for example, $\rspace(\grid{g}_{\ell}) = \smash{\raisebox{-0.4ex}{$2^{2^{2^{\ell}}}$}}$ in~\cite{Destombes_2021:phd:algorithmic-complexity-soficness-subshifts-multidimensional}); and in particular never allows for constant-shaped substitutions of size, say, $\rtime(\grid{g}_{\ell}) = 2$.\footnote{Not including here the first author's PhD thesis~\cite{Callard_2025:phd:soficity-multidimensional-subshifts}, upon which this article is built. Notice that, conversely,~\cite[Theorem~10.11]{Callard_2025:phd:soficity-multidimensional-subshifts} only considers substitutions of size $\rtime(\grid{g}_{\ell}) = 2$.}
\item Computability conditions on the configuration $x^{(\ell)}$ usually depend on~$\rspace(\grid{g}_{\ell-1})$ (the size of the next step of substitution) instead of $\pxspace{\ell}$ (the product of all previous sizes). The latter is more flexible, and made necessary here in order to allow for constant or non-monotonous sequences $(\rspace(\grid{g}_{\ell}))_{\ell \in \N}$.
\item The fixed-point construction is usually applied on $\Z^2$, where its computations embed the space-time diagrams of a (multitape) Turing machine~\cite{Durand-Romashchenko-Shen_2012:fixed-point-tilesets-and-applications,Westrick_2017:seas-of-squares} or cellular automaton~\cite{Zinoviadis_2016:phd:expansiveness-2d-sfts,Destombes_2021:phd:algorithmic-complexity-soficness-subshifts-multidimensional}, \textit{i.e.}~formalisms in which programming --~and, hence, precise considerations about time complexity~-- is quite difficult to achieve. In our lemma, we actually embed the space-time diagram of a \emph{processor array} (\textit{c.f.}~\Cref{calc:sec:processors}), and rely on \Cref{calc:lem:pa-ram-simulation} to state our result in the more intuitive setting of $\log$-RAM machines.
\end{enumerate}

\Cref{sof:lem:fixpoint} is also a result of multidimensional soficity on limit spaces defined by sequences of substitutions. To the best of our knowledge, the only similar statement is a characterization of some $S$-adic shifts in~\cite{Aubrun-Sablik_2014:multid-effective-s-adic-subshifts-are-sofic}. Our result mainly differs in that it applies to sequences of substitutions of possibly \emph{infinite alphabet rank} (\textit{i.e.}~the alphabets at different levels of (de)substitution may be of increasing and unbounded cardinality).

\paragraph*{Numeric bounds}
\Cref{sof:lem:fixpoint} and its corollaries call for some technical comments on the optimality of the rational parameters $\alpha < 1$ and $\delta < \frac{1}{2}$.

The bound $\delta < \frac{1}{2}$ on the sizes of the substitutions appears to be rather arbitrary, and is indeed a residue of the bound $\pxspace{\ell}[2\delta]$ appearing in the~``\crtlnameref{fix:lem:staircase-lemma}''~(\Cref{fix:lem:staircase-lemma}). A fixed-point construction for the more natural condition $\delta < 1$ is definitely possible with our proof, but such upper bounds would not allow grids to contain rectangles of arbitrary small sizes (\textit{e.g.}~$\rspace(\grid{g}_{\ell}) = 2$).

The bound on the color lengths $\norm{x^{(\ell)}} = \bigO(\pxspace{\ell-1}[\alpha])$, which is less natural than the expected $\norm{x^{(\ell)}} = \bigO(\pxspace{\ell}[\alpha])$, is also a proof artifact. It comes up when synchronizing information between consecutive levels of the construction. While results from~\cite{Westrick_2017:seas-of-squares,Destombes_2021:phd:algorithmic-complexity-soficness-subshifts-multidimensional} rely on global assumptions on the density of information in their specific shifts to solve the communication between levels as a graph flow problem, such a solution cannot be applied to a generic statement such as ours\footnote{Nevertheless, we still recover~\cite{Westrick_2017:seas-of-squares,Destombes_2021:phd:algorithmic-complexity-soficness-subshifts-multidimensional} in \Cref{app:sec}: indeed, we argue that such density assumptions actually translate into alternative substitutions \emph{of bounded sizes} computing the same space $\limspace{X}[\code{\tau}]$.}.

\paragraph*{Differences between \Cref{sof:thm:square-d-adic-sofic,sof:thm:bounded-d-adic-sofic} and \Cref{sof:lem:fixpoint}}

Both \Cref{sof:thm:square-d-adic-sofic,sof:thm:bounded-d-adic-sofic} substantially weaken our main \crtlnameref{sof:lem:fixpoint} (\Cref{sof:lem:fixpoint}) in two different ways: the first one reduces the shapes that can appear to uniform squares of fixed computable sizes along a single sequence of ndill maps $(\tau_{\ell})_{\ell \in \N}$; the second to arbitrary rectangles of bounded size along branching successions of ndill maps from $\prod_{\ell \in \N} \Sndill_{\ell}$. In both cases, the main lost feature is the non-uniformity of a level: in \Cref{sof:lem:fixpoint}, the possible grids at level $\ell$ can actually be of any sizes $\rtime(\grid{g}_{\ell})$ ranging in $2 \leq \rtime(\grid{g}_{\ell}) \leq 2^{\bigO(\pxspace{\ell}[\delta])}$, thus defining a wide range of possible computational capacities for the next levels of substitutions. Nevertheless, these two theorems together should offer good insight into the potential applications of the \crtlnameref{sof:lem:fixpoint}, and are general enough to recover all examples from \Cref{app:sec}.

\paragraph*{Soficity}
As mentioned in the introduction, in sofic shift spaces, patterns of domain $\ibox{n}^d$ define an information flow bounded by $\bigO(n^{d-1})$. \Cref{sof:thm:square-d-adic-sofic,sof:thm:bounded-d-adic-sofic}, which prove the soficity of shifts with information flow $\bigO(n^{\alpha})$ for $\alpha < d-1$, can be understood as a partial converse statement.

More precisely, the ndill maps operate on alphabets $\alphabet_{\ell}$ of increasing cardinality. As such, the size of the alphabet $\alphabet_{\ell+1}$ provides a bound on the amount of information that $\ndill_{\ell}$ can process. Thus, a sequence of ndill maps $(\ndill_{\ell})_{\ell \in \N}$ defines a recursive hierarchical scheme that verifies the validity of patterns in the configurations of the limit space. This defines a quantification of the amount of information appearing in the patterns of $\limspace{X}[\Sndill]$, which is somewhat orthogonal to other more classical complexity measures\footnote{For example, a full shift $X = \{0,1\}^{\Z^d}$ contains many patterns of maximal Kolmogorov complexity, but can be defined using a non-deterministic $\tau_0 \colon \square \mapsmapsto \{0,1\}$, and $\tau_{\ell} = \square \mapsto \square$ otherwise; thus, with $\bigO(1)$ bits of information at each level.} (\textit{e.g.}~Kolmogorov complexity).

\begin{remark}[A tautology]
  One should note that \Cref{sof:lem:fixpoint}, due to its non-deterministic ndill maps, realizes all sofic shifts $\smash{Y \subseteq \alphabet^{\Z^d}}$ as limit spaces $\limspace{X}[\code{\tau}]$ of the full shift $\smash{X = \alphabet^{\Z^d}}$. Indeed, assuming $Y$ to be sofic, let $\smash{Z \subseteq \alphabet[B]^{\Z^d}}$ be a cover of finite type for $Y$ and $\pi \colon Z \to Y$ be an associated factor map. We define a sequence of non-deterministic ndill maps $\smash{\ndill_{\ell} \colon \alphabet_{\ell} \to \alphabet_{\ell-1}^{\ibox{2}^d}}$ by:
  \begin{itemize}
  \item At level $\ell \geq 2$, all ndill maps $\ndill_{\ell}$ substitute a blank color to blank patterns;
  \item At level $\ell = 1$, the ndill map $\ndill_{1}$ non-deterministically ``guesses'' a configuration $z \in \alphabet[B]^{\Z^d}$;
  \item At level $\ell = 0$, the ndill map $\ndill_{0}$ checks the local validity of the configuration $z$ in the cover~$Z$, and then maps each cell to its image by the factor map $\pi$.
  \end{itemize}
  Since alphabets $\alphabet_{\ell}$ have cardinality $1$ for $\ell \geq 2$ and that the associated ndill maps can be computed in constant time, this $N$-adic structure satisfies the hypotheses of \Cref{sof:thm:square-d-adic-sofic}. The main difficulty in applying \Cref{sof:lem:fixpoint} thus lies in finding a suitable $N$-adic structure for an effective shift that is \emph{not yet} known to be sofic.
\end{remark}

\paragraph*{Alternative $N$-adic shifts}
In this article, we prove, for some sets of directive sequences $\Sndill$, the soficity of the limit space $\limspace{X}[\Sndill]$ whose configurations can be infinitely desubstituted along a directive sequence of $\Sndill$.
In the literature, what is called \emph{$S$-adic shift} is most often a bit different, built from the subpatterns of iterated substitutions of ``single symbols''\footnotemark\ from the alphabets $\alphabet_{n}\!$'s:
\begin{align*}\limspace{X}[\Sndill][\subpattern] & = \big\{ x \in X : \exists (\ndill_{\ell})_{\ell \geq 0} \in \Sndill,\, \forall w \subpattern x, \exists n \in \N, \exists (x^{(\ell)})_{0 \leq \ell \leq n},\; x^{(0)} \rsubstep[\ndill_0][\grid{g}_{0}] \dots \rsubstep[\ndill_{n-1}][\grid{g}_{n-1}] x^{(n)} \\[-8pt]
                                                 & \hspace{6.5cm} \text{ and } w \subpattern \restr{x^{(0)}}{\pxrect{r}_{\vec{0}}} \text{ } \big(\text{where } (\pxrect{r}_{\vec{i}})_{\vec{i} \in \Z^d} = \bigcomp_{\ell=0}^{n-1} \grid{g}_{\ell} \big) \big\}.
\end{align*}
\footnotetext{In the case of substitutions instead of ndill maps, the second condition can be rewritten: ${\forall w \subpattern x},\allowbreak {\exists n \in \N},\allowbreak {\exists a \in \alphabet_n},\allowbreak w \subpattern \subst_0 \circ \dots \circ \subst_n(a)$.}

In both shifts $\limspace{X}[\Sndill]$ and $\limspace{X}[\Sndill][\subpattern]$, we cannot avoid the existence of degenerate configurations admitting two hierarchies of macro-tiles that are neighbors at every level but do not admit any common ancestor. Such a phenomenon is sometimes called a \emph{fault line}, as any pattern intersecting the boundary of these neighboring hierarchies is never fully contained in a macro-tile (see~\Cref{def:fig:fault-line}). The difference between $\limspace{X}[\Sndill]$ and $\limspace{X}[\Sndill][\subpattern]$ is that the latter only allows degenerate configurations in which the patterns overlapping such fault lines can be locally interpreted as also deriving from a common ancestor macro-tile (in another hierarchical interpretation).

\begin{figure}[ht]
  \begin{tikzpicture}[scale=.067,decoration={random steps,segment length=0.4em,amplitude=0.15em}]
    \pgfmathsetmacro{\xmin}{-91}
    \pgfmathsetmacro{\ymin}{-50}
    \pgfmathsetmacro{\height}{101}
    \pgfmathsetmacro{\width}{182}
    \draw[TilingGrid,decorate,path picture={
      \begin{scope} 
        \draw[thick,opacity=.1] (\xmin,\ymin) grid ++(\width,\height);
        \draw[thick,LevelI,opacity=.1,step=3,shift={(0,-1)}] (\xmin,\ymin) grid ++($(\width,\height) + (0,1)$);
        \draw[thick,LevelI+1,opacity=.1,step=9,shift={(0,-4)}] (\xmin,\ymin) grid  ++($(\width,\height) + (0,4)$);
        \draw[very thick,LevelK,opacity=.2,step=27,shift={(0,-13)}] (\xmin,\ymin) grid ++($(\width,\height) + (0,13)$);
      \end{scope}
      \begin{scope}
        \draw[thick] (-3,-1) grid ($(3,2*1)$);
        \draw[thick,LevelI,step=3,shift={(0,-1)}] (-9,-3) grid ($(9,2*3)$);
        \draw[thick,LevelI+1,step=9,shift={(0,-4)}] (-27,-9) grid ($(27,2*9)$);
        \draw[thick,LevelK,step=27,shift={(0,-13)}] (-81,-27) grid ($(81,2*27)$);
        \draw[thick,LevelINext,step=81,shift={(0,-40)}] (-243,-81) grid ($(243,2*81)$);
      \end{scope}
      \draw[black!10!red,line width=.9mm,dash pattern={on 2.5mm off 1.5mm}] (0,\ymin) -- ++(0,\height);
    }] (\xmin,\ymin) rectangle ++(\width,\height);
  \end{tikzpicture}
  \caption{A degenerate hierarchy of grids resulting in a fault line.}
  \subcaption{In a sequence of square grids of size $3 \times 3$, any pattern overlapping the fault line
    $\mspace{2mu}\begin{tikzpicture}[baseline=0.3ex,scale=0.3]
      \draw[very thick,black!10!red,dashed] (0,0) -- ++(0,1.1);
    \end{tikzpicture}\mspace{2mu}$
    will never be fully contained in a macro-tile. Nevertheless, these grids could form a valid $S$-adic structure: one can just desubstitute both halves of the plane independently.}
  \label{def:fig:fault-line}
\end{figure}
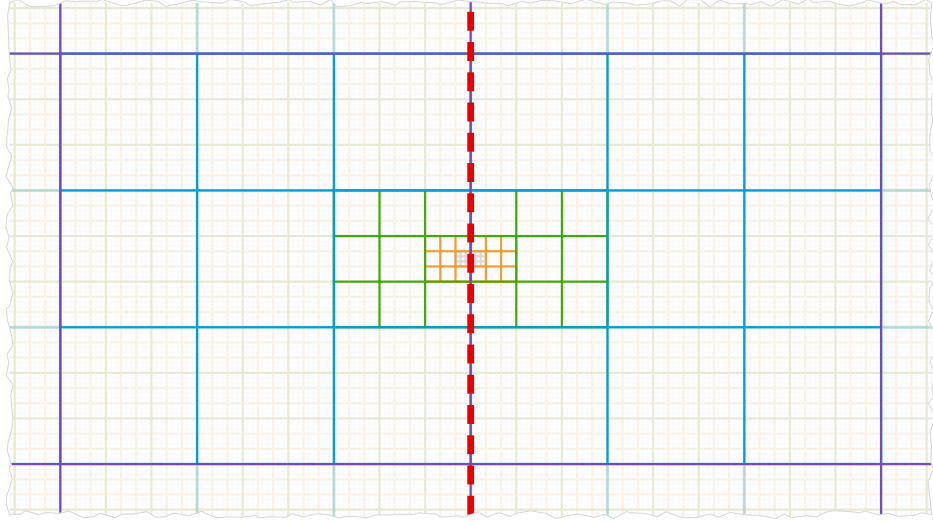

In the case of expanding substitutions, the difference between $\limspace{X}[\Sndill]$ and $\limspace{X}[\Sndill][\subpattern]$ mostly amounts to checking that any pair of $2$ neighboring macro-tiles actually appears in the iterated image of some higher-level symbol. In our framework, this check can be implemented at each level, among the local constraints of the ndill maps. In particular, the $N$-adic shift $\limspace{X}[\Sndill][\subpattern]$ can be realized as a $N$-adic limit space $\limspace{X}[\Sndill']$ with the same grids, same alphabet, and comparable neighborhood (the same, up to increasing it to a generating set of $\Z^d$ if it did not already contain one).

In general, these new local constraints may become too computationally intensive for the fixed-point lemma to still apply. But in many natural cases (in particular, a stationary directive sequence), \Cref{sof:thm:square-d-adic-sofic,sof:thm:bounded-d-adic-sofic} will still hold for $N$-adic shifts $\limspace{X}[\Sndill][\subpattern]$ by simply simulating a constant additional number of substitutions to perform this check.

\enlargethispage{\baselineskip}
%
\section{Applications}\label{app:sec}

We now illustrate \Cref{sof:lem:fixpoint} and its corollaries by proving the soficity of a large variety of dynamical and computational-based multidimensional shifts. Most of these applications are written as two-dimensional shift spaces for notational convenience, but could be immediately extended to any higher dimension.

\subsection{\boldmath\texorpdfstring{$S$}{S}-adic shift spaces}

As already mentioned, our results allow to recover~\cite{Aubrun-Sablik_2014:multid-effective-s-adic-subshifts-are-sofic}, which proves the soficity of multidimensional $S$-adic shifts whose $S$-adic structures involves a finite set of substitutions:
\begin{theorem}[note={\cite[Theorem 5]{Aubrun-Sablik_2014:multid-effective-s-adic-subshifts-are-sofic}},label={app:thm:effective-s-adic}]
  Let $\Ssubst_0$ be a finite set of expanding multidimensional substitutions operating on some finite alphabet $\alphabet$, and let $\Ssubst \subseteq \Ssubst_0^{\mathbb{N}}$ be an effectively closed set of directive sequences. Then the limit space $\limspace{X}[\Ssubst]$ is sofic.
\end{theorem}

\begin{proof}
  This is a special case of \Cref{sof:thm:bounded-d-adic-sofic}.
\end{proof}

As a special case with $\Ssubst_0 = \{\tau\}$, we also recover the seminal result of~\cite{Mozes_1989:tiling_substitution_systems_dynamical_systems_generated_by_them}:
\begin{theorem}[note={\cite{Mozes_1989:tiling_substitution_systems_dynamical_systems_generated_by_them}}, label={prel:thm:mozes-theorem}]
  For $\subst \colon \alphabet \mto \alphabet^{\Srect_0}\!$ a substitution on some finite $\alphabet$, the shift  $\limspace{X}[\gsubst]$ is sofic.
\end{theorem}

Many shift spaces can be defined through this $S$-adic formalism. For example, shifts describing discrete planes~\cite{Arnoux-Berthé-Ito_2002:discrete-planes-z2-actions-jacobi-perron-algorithm-substitutions} are $S$-adic, and the associated sets of directive sequences are computable if the equations of the planes also satisfy computable constraints. Unfortunately, some of the substitutions involved are not expanding; but some assumption on the computability and growth of the partial quotients in the corresponding continuous fraction expansion should help make them instances of our theorems.

\subsection{Embedding computably enumerable information}\label{app:sec:medvedev-equivalence}
We say that a set $S \subseteq \alphabet^\N$ is \emph{Medvedev-reducible} to a set $S' \subseteq \alphabet^\N$ if there exists a computable partial function $\Phi: \alphabet^\N \to \alphabet^\N$ whose domain includes $S'$ and such that $\Phi(S') \subseteq S$. Intuitively, from any element $z \in S'$ (often called an \emph{oracle}), one can compute an associated element $\Phi(z) \in S$. This yields a preorder, whose equivalence classes are called \emph{Medvedev degrees}.
These notions can be extended to other sets endowed with computable structures: for instance, $\smash{\alphabet^{\Z^d}}$ can inherit Medvedev reducibility from $\alphabet^\N$ through any natural enumeration $\N \to \Z^d$.

\begin{theorem}[label={app:thm:medvedev-degrees}]
  Let $S \subseteq \{0,1\}^\N$ be an effectively closed set of binary infinite words.
  Then there exists a sofic shift $\smash{\X\subset\{0,1\}^{\Z^2}}$ in the Medvedev degree of $S$.
\end{theorem}

Additionally considering \Cref{sof:par:comment-constructivism}, we recover the main result of~\cite{Simpson_2014:medvedev-degrees-of-2d-sfts}: the cover of finite type $\mathcal{X}$ associated with \Cref{app:thm:medvedev-degrees} is itself in the Medvedev degree of $S$. As there exist effectively closed sets $S \subseteq \{0,1\}^{\N}$ without any computable infinite word, this also recovers~\cite{Myers_1974:nonrecursive-tilings-of-the-plane-2}. Finally, the Medvedev equivalence between $\mathcal{X}$ and $S$ provides another proof of the undecidability of the Domino problem~\cite{Berger_1964:phd:undecidability-domino-problem}: since the emptiness of effectively closed sets is undecidable (by reducing the halting problem), so is the emptiness of shift spaces of finite type.

Going back to \Cref{app:thm:medvedev-degrees}, the Medvedev equivalence between $\X$ and $S$ is not computable because many configurations in $\X$ will represent the same sequence $z \in S$. Nevertheless, the configurations in $\X$ representing $z$ can be identified by their $S$-adic structure (\textit{c.f.}~\Cref{sof:par:comment-structure}): thus, there exists a computable left-invertible map from $\X$ onto $S \times \{0,1\}^\N$. In particular, any oracle that would allow to compute an infinite word of $S$ is computably homeomorphic to a configuration of $\X$: one says that the set of \emph{Turing degrees} of $\X$ is the union of \emph{cones} above those of $S$ (more precisions in \cite[Lemma~51 and Theorem~52]{Zinoviadis_2015:hierarchy-expansiveness-2d-subshifts-finite-type}).

\smallskip
While~\Cref{app:thm:medvedev-degrees} directly follows from~\Cref{sof:thm:bounded-d-adic-sofic} and the comments in \Cref{sof:par:comment-constructivism,sof:par:comment-structure} (using any two (even constant) distinct substitutions), we provide here a more self-contained proof that illustrates how our construction applies to effective Toeplitz shifts.

\subsubsection*{Toeplitz shifts}
In order to prove~\Cref{app:thm:medvedev-degrees}, we will design a shift $\X$ as a union of Toeplitz shifts, each lying in the Medvedev degree of an element of $S$. We first define this class of shifts.

A (rectangular) \emph{Toeplitz substitution} is an expanding substitution $\theta\colon\alphabet\mto\alphabet[B]^{\rect r}$ for some $\rect r\in\Srect_0$ (constant shape) that admits a \emph{coincidence}: there exists a position $\vec{i} \in \rect{r}$ at which all input letters $a \in \alphabet$ are mapped to the same symbol (in other words,~$\card{\sett{\theta(a)_i}{a\in\alphabet}}=1$). A (rectangular) \emph{Toeplitz shift} is an $S$-adic space whose directive sequence only involves (rectangular) Toeplitz substitutions. This definition is known to be equivalent to the more classical one, involving the existence of a (nonuniform) period for every pattern (see for instance \cite{Gjerde-Johansen_2000:bratteli-vershik-models-for-cantor-minimal-systems-applications-toeplitz-flows}). The previous definition, together with \Cref{sof:thm:square-d-adic-sofic}, implies that any effective Toeplitz shift is sofic.

\medskip
One of the simplest examples of two-dimensional Toeplitz substitutions is the following.
On the alphabet $\alphabet=\{0,1\}$, we define two substitutions $\theta_0$ and $\theta_1$ as follows: for $a,b \in \alphabet$,
\[ \theta_a(b)=\left[\begin{array}{cc}0&b\\a&1\end{array}\right]. \]
These two substitutions are mutually recognizable over polychromatic configurations:
\begin{lemma}\label{app:lem:toeplitz-recognizability}
  From any configuration $y \in \theta_0(\alphabet^{\Z^2}\setminus\{0^{\Z^2},1^{\Z^2}\}) \cup \theta_1(\alphabet^{\Z^2}\setminus\{0^{\Z^2},1^{\Z^2}\})$, one can computably recover a unique symbol $a \in \alphabet$, a unique configuration $x \in \alphabet^{\Z^2}\setminus\{0^{\Z^2},1^{\Z^2}\}$ and a unique vector $\vec j\in\{0,1\}^2$, such that the grid $\grid g=(\ibox2^2+\vec i+\vec j)_{\vec i\in\Z^2}$ of $2\times2$-squares satisfies
  \[ x \substep[\theta_a][\grid{g}] y.\]
\end{lemma}
\begin{proof}
  Any configuration $y$ in the image $\theta_0(\alphabet^{\Z^2})\cup\theta_1(\alphabet^{\Z^2})$ either has every second column full of~$0$s, in which case it lies in the image of $\theta_0$, or every second row full of~$1$s, in which case it lies in the image of $\theta_1$. This can be checked simply by looking any $2\times2$-pattern.
  Now, if $y$ is neither the image of $\smash{0^{\Z^2}}$ nor $\smash{1^{\Z^2}}$, then some row or column involves both a $0$ and a $1$, which allows to determine the unique grid $\grid{g}$ of the substitution. Once the grid is set, the injectivity of each substitution allows to recover the source configuration $x$.
\end{proof}

We now prove~\Cref{app:thm:medvedev-degrees}:
\begin{proof}[Proof of \Cref{app:thm:medvedev-degrees}]
  Let $S \subseteq \{0,1\}^{\N}$ be an effectively closed set. Denoting $\theta_0$ and $\theta_1$ the Toeplitz substitutions from above, the set of directive sequences $\Sndill = \{ (\theta_{z_\ell})_{\ell \in \N} : z \in S \}$ is also effectively closed. By \Cref{sof:thm:bounded-d-adic-sofic}, the resulting limit space $\limspace{X}[\Sndill]$ is thus a sofic shift. We now prove that $\limspace{X}[\Sndill]$ and $S$ are Medvedev-equivalent.

  Consider first an oracle $x \in \limspace{X}[\Sndill]$: by applying \Cref{app:lem:toeplitz-recognizability} iteratively, we recover a sequence of substitutions defining $x$ and thus a sequence $z \in S$. Conversely, consider an oracle $z \in S$: we build a configuration $x \in \limspace{X}[\Sndill]$ by always using the same grid $(\ibox2^2+\vec i+(-1,1))_{\vec i\in\Z^2}$ and progressively filling (in the macro-tiles of level $\ell$) the coincidences given by each $\theta_{z_{\ell}}$ (which are, by definition, positions that do not depend of the symbols $z_{\ell'}$ of higher levels $\ell' > \ell$). This grid is centered in a way that this procedure will eventually fill the whole plane $\Z^2$.
\end{proof}

\subsection{Prescribed frequencies and growth-type invariants}
For a symbol $a \in \alphabet$ and a configuration $x \in \alphabet^{\Z^d}$, the \emph{frequency} of $a$ in $x$ is defined, when it exists, as the following limit:
\[ \freq_a(x) = \lim_{n\to\infty} \frac{\card{\Big\{ \vec{i} \in \interval{-n}{n}^d : x_{\vec{i}} = a \Big\}}}{\Big|\interval{-n}{n}^d\Big|} \in [0,1]. \]
On the alphabet $\alphabet = \{0,1\}$, consider now the two substitutions defined by
\[ \theta_a(b)= \left[ \begin{array}{cc} a & b \\ b & a \end{array} \right] \]
for $a,b \in \alphabet$. Then the substitutions $\theta_0$ and $\theta_1$ result in well-defined frequencies of $1$s:
\begin{lemma}\label{app:lem:frequencies}
  For $z \in \{0,1\}^{\N}$, let $\Sndill = \{ (\theta_{z_\ell})_{\ell \in \N} \}$. Then every configuration $x \in \limspace{X}[\Sndill]$ admits a frequency of $1$s that is equal to $\sum_{\ell\in\N} z_{\ell} \cdot 2^{-(\ell+1)}$.
\end{lemma}
\begin{proof}
  Consider some $\ell\in\N$. In the image of any pattern $u \in \{0,1\}^{\ibox\ell^2}$, the number of symbols $1$ is exactly:
  \[\card{\{ \vec{i} \in \ibox{2\ell}^2 : \theta_a(u)_{\vec{i}} = 1 \}} = 2a \cdot \ell^2 + 2 \cdot \card{\{ \vec{i} \in \ibox{\ell}^2 : u_{\vec{i}} = 1 \}};\]
  and the ratio of symbols $1$ thus is:
  \[\frac{\card{\{ \vec{i} \in \ibox{2\ell}^2 : \theta_a(u)_{\vec{i}} = 1 \}}}{(2\ell)^2} = \frac{a}{2} + \frac{1}{2} \cdot \frac{\card{\{ \vec{i} \in \ibox{\ell}^2 : u_{\vec{i}} = 1 \}}}{\ell^2}.\]
  By induction, we obtain that for every symbol $b \in \{0,1\}$ and every sequence $z \in \{0,1\}^\ell$, we have:
  \[ \frac{\card{\{ \vec{i} \in \ibox{2^\ell}^2 : (\theta_{z_0} \theta_{z_1} \cdots\mspace{2mu}\theta_{z_{\ell-1}} (b))_{\vec{i}} = 1 \}}}{2^{2\ell}}
    = \sum_{i<\ell} z_i \cdot 2^{-(i+1)} + b \cdot 2^{-(\ell+1)}.\]
  We conclude this proof by considering arbitrary patterns in configurations $x \in \limspace{X}[\Sndill]$, and compute the frequency of symbols $1$ by a classical unique-ergodicity argument in the amenable setting. More precisely, let $(\grid{g}_{\ell})_{\ell \in \N}$ be a sequence of grids defining the $S$-adic structure of $x$ along the substitutions $(\theta_{z_\ell})_{\ell \in \N}$. Furthermore, let us fix any $\varepsilon > 0$, any level $\ell \geq 1 - \log \varepsilon$ and any square $\rect{r} \subseteq \Z^2$ of size $k \times k$ for $k \in \N$ large enough so that $4k \cdot 2^{\ell} < \frac{\varepsilon}{2} \mspace{2mu} k^2$. We now consider how the grid~$\grid{g}_{\ell-1}$ (which consists of squares of size $2^{\ell} \times 2^{\ell}$) and the square~$\rect{r}$ intersect.

  Geometrically, the square $\rect{r}$ contains roughly $(k \cdot 2^{-\ell}) \times (k \cdot 2^{-\ell})$ entire squares from the grid~$\grid{g}_{\ell-1}$, and the corresponding patterns in $x$ are images $\theta_{z_0}\theta_{z_1}\cdots\theta_{z_{\ell-1}}(b)$ of some bits $b \in \{0,1\}$. By the previous paragraphs, their ratios of symbols $1$ are exactly $\sum_{i<\ell} z_i \cdot 2^{-(i+1)} + b \cdot 2^{-(\ell+1)}$, \textit{i.e.}~belong to the interval $\sum_{i \in \N} z_i \cdot 2^{-(i+1)} + [-\frac{\varepsilon}{2},\frac{\varepsilon}{2}]$.

  Furthermore, the boundary of the square $\rect{r}$ intersects at most $4k \cdot 2^{-\ell}$ squares from the grid~$\grid{g}_{\ell-1}$. Since the corresponding area in $\rect{r}$ covers at most $4k \cdot 2^{\ell}$ positions, and that $4k \cdot 2^{\ell} < \frac{\varepsilon}{2} \mspace{2mu} k^2$, we deduce that these positions alter the ratio of symbols $1$ in $\restr{x}{\rect{r}}$ by at most $\frac{\varepsilon}{2}$.

  Overall, we obtain that for every $\varepsilon > 0$, there exists a window size $k \in \N$ that is large enough to witness that the ratio of symbols $1$ in $x$ locally looks like $\sum_{i\in\N} z_i \cdot 2^{-(i+1)}$, up to $\varepsilon$. This concludes the proof.
\end{proof}

\noindent With this lemma, we can realize many sets of frequencies in multidimensional sofic shifts (up to some computational restrictions). A set of real numbers $S \subseteq \R$ is said to be \emph{effectively closed} if its complement is the union of a computably enumerable family of rational balls.

\begin{theorem}[label={app:thm:frequencies}]
  For every effectively closed $S\subset[0,1]$, there exists a sofic shift $X \subseteq \{0,1\}^{\Z^2}$ in which all configurations have well-defined frequencies of $1$s, and these frequencies exactly cover $S$.
\end{theorem}

\begin{proof}[Proof of \Cref{app:thm:frequencies}]
  Let~$\hat{S} \subseteq \{0,1\}^{\N}$ be the set of binary expansions of $S \subseteq [0,1]$. Considering the substitutions $\Sndill_\ell=\{\theta_0,\theta_1\}$ and the set of directive sequences $\Sndill= \{ (\theta_{z_{\ell}})_{\ell \in \N} : z \in \hat{S} \}$, by \Cref{app:lem:frequencies}, every configuration in the limit space $\limspace{X}[\Sndill]$ has a well-defined frequency of symbols $1$, and the set of these frequencies in $\limspace{X}[\Sndill]$ is exactly $S$. Furthermore, since $S \subseteq [0,1]$ is effectively closed, so are the associated binary expansions $\hat{S} \in \{0,1\}^{\N}$: by \Cref{app:thm:effective-s-adic}, the shift $\limspace{X}[\Sndill]$ is thus sofic.
\end{proof}

\subsubsection*{Entropies}

From our considerations on frequencies of symbols, we can easily realize a large class of topological entropies: if a real number $\gamma \in \R$ is said to be \emph{$\Pi^{0}_{1}$-computable} when the set of rational numbers $\{ r \in \Q : r \leq \gamma \}$ is co-computably enumerable, then,
\begin{theorem}
  For every right-computable real number $\gamma \geq 0$, there exists an SFT $\mathcal X \subseteq \{0,1\}^{\Z^2}$ with entropy $\gamma$.
\end{theorem}

\noindent In fact, it is a well-known characterization of~\cite{Hochman-Meyerovitch_2010:characterization-entropies-multidimensional-sfts} that the set of entropies of $\Z^d$ SFT is exactly the set of positive $\Pi^{0}_{1}$-computable real numbers when $d\ge2$.

\begin{proof}
  Let us fix one such real number $\gamma$.
  Since entropy is additive under Cartesian product, we assume without loss of generality that $\gamma \leq 1$.
  Furthermore, since $\gamma$ is $\Pi^{0}_{1}$-computable, the set $[0,\gamma] \subseteq [0,1]$ is effectively closed:
  by~\Cref{app:thm:frequencies}, there exists a sofic shift $X' \subseteq \{0,1\}^{\Z^2}$ in which all configurations admit a well-defined frequency of symbols $1$, and for which these frequencies exactly cover $[0,\gamma]$.

  As a regular Toeplitz shift, $X'$ has entropy $0$.
  By considering its image under the non-deterministic substitution $\pi \colon 0 \mapsto 0$, $1 \mapsto 1$, $1 \mapsto 2$, we introduce within the set of valid patterns as many choices as there were occurrences of symbols $1$.
  This turns frequencies into entropy, thus leading to the sofic shift $\pi(X')$ having entropy $\gamma$.
  Thanks to \Cref{sof:par:comment-entropy}, we can further consider a finite-type cover for $\pi(X')$ with equal entropy, thus concluding the proof.
\end{proof}

\subsubsection*{Entropy dimensions}

Other conjugacy invariants related to the pattern complexity of shifts $X \subseteq \alphabet^{\Z^d}$ include the upper and lower \emph{entropy dimensions}, which are respectively defined as:
\[ \bar{D}_h(X) = \limsup_{n \to +\infty} \mspace{2mu} \frac{\log \log N_X(n)}{\log n}
  \qquad \text{and} \qquad
  \ubar{D}_h(X) = \liminf_{n \to +\infty} \mspace{2mu} \frac{\log \log N_X(n)}{\log n}
\]
Informally, $\bar{D}_h(X) = \gamma \in [0,d]$ if the pattern complexity of $X$ roughly amounts to $2^{n^{\gamma}}$. Let us now consider the following Toeplitz substitutions $\theta_0$ and $\theta_1$ over the alphabet $\alphabet = \{0,1\}$:
\[ \theta_0(0) = \theta_1(0) = \left[ \begin{array}{cc} 0 & 0 \\ 0 & 0 \end{array} \right]
  \qquad \text{and} \qquad
  \theta_a(1) = \left[ \begin{array}{cc} 1 & a \\ a & a \end{array} \right]
  \text{ for }a\in\alphabet.
\]
Then $\theta_0$ and $\theta_1$ allow to realize arbitrary entropy dimensions:
\begin{lemma}[label={app:lem:entropy-dimensions}]
  For $z \in \{0,1\}^{\N}$, let $\Sndill = \big\{ (\theta_{z'_{\ell}})_{\ell \in \N} : \forall \ell \in \N, z'_{\ell} \leq z_{\ell} \big\}$ and fix the non-deterministic substitution $\pi \colon 0 \mapsto 0$, $1 \mapsto 1$, $1 \mapsto 2$. Then $\pi(\limspace{X}[\Sndill])$ satisfies
  \[ \bar{D}_h(X) = \limsup_{n \to +\infty} \frac{2}{n} \sum_{i=0}^{n-1} z_i
    \qquad \text{and} \qquad
  \ubar{D}_h(X) = \liminf_{n \to +\infty} \frac{2}{n} \sum_{i=0}^{n-1} z_i.\]
\end{lemma}
\begin{proof}
  We first consider the pattern complexity of $\limspace{X}[\Sndill]$. Since any valid pattern $\smash{w \in \alphabet^{\ibox{n}^2}}$ appears in a configuration with a valid $S$-adic structure, it is covered by at most $2 \times 2$ macro-tiles of level~$\ell$ (for $\ell \in \N$ such that $2^{\ell} \leq n < 2^{\ell+1}$). By considering the position of $w$ in these macro-tiles ($(2^{\ell})^2$~possibilities), the symbols of these macro-tiles ($2^4$~possibilities) and choices of substitutions for levels $\ell' \leq \ell$ (less than $2^{\ell}$~possibilities), we deduce that $\limspace{X}[\Sndill]$ has polynomial pattern complexity.

  \smallskip
  We now turn to the pattern complexity of $\pi(\limspace{X}[\Sndill])$. By definition of $\pi$, we must take into account the pattern complexity of $\limspace{X}[\Sndill]$ and the numbers of symbols $1$ in each pattern. For every sequence $z' \in \{0,1\}^{\ell}$, we have:
  \[ \card{\{\vec{i} \in \ibox{2^{\ell}}^2 : (\theta_{z'_0} \theta_{z'_1} \cdots \mspace{2mu} \theta_{z'_{\ell-1}}(1))_{\vec{i}}=1\}} = \prod_{i=0}^{\ell-1} 4^{z'_i}. \]
  More generally, for any valid pattern $w \in \alphabet^{\ibox{n}^2}$ appearing in a configuration $x \in \limspace{X}[\Sndill]$: since $x$ admits an $S$-adic structure, there exists at most $2 \times 2$ macro-tiles of level~$\ell$ (for $\ell$ such that $2^{\ell} \leq n < 2^{\ell+1}$) covering the pattern $w$. By the previous paragraph, these macro-tiles together contain at most $4 \cdot \prod_{i=0}^{\ell-1} 4^{z'_i}$ symbols $1$ for some $z' \in \Sndill$. Since $z_i' \leq z_i$ by definition of $\Sndill$, and that the non-deterministic substitution $\pi$ generates two choices per symbols $1$ in $\limspace{X}$, we deduce that the pattern complexity of $\pi(\limspace{X}[\Sndill])$ satisfies:
  \[ \log N_{\pi(\limspace{X}[\Sndill])}(n) \leq 4 \cdot \mspace{-9mu} \prod_{i=0}^{\lfloor \log n \rfloor} \mspace{-6mu} 4^{z_i} + \bigO(\log n). \]
  while taking $z_i' = z_i$ and considering a whole macro-tile of level $\lfloor \log n \rfloor$ gives:
  \[ \prod_{i=0}^{\lfloor \log n \rfloor} \mspace{-6mu} 4^{z_i} \leq \log N_{\pi(\limspace{X}[\Sndill])}(n). \]
  This concludes the proof, since the pattern complexity then satisfies:
  \[ \frac{2}{\log n} \mspace{-4mu} \sum_{i=0}^{\lfloor \log n \rfloor} \mspace{-4mu} z_i
    \leq \frac{\log \log N_{\pi(\limspace{X}[\Sndill])}(n)}{\log n}
    \leq \frac{2}{\log n} \mspace{-4mu} \sum_{i=0}^{\lfloor \log n \rfloor} \mspace{-4mu} z_i + \bigO\Big(\frac{\log \log n}{\log n}\Big). \qedhere\]
\end{proof}

Using some elementary facts about the computability of real numbers, we then obtain a computational characterization for entropy dimensions of sofic shifts $\bar{D}_h(X)$ and $\ubar{D}_h(X)$ (\textit{c.f.}~\cite{Meyerovitch_2011:growth-type-invariants-multidimensional-sfts}). The class of $\Pi^{0}_{1}$-computable real number actually generalizes into an infinite \emph{arithmetical hierarchy}: a real number $\gamma \in \R$ is said to be $\Sigma^{0}_{n}$-computable (resp.~$\Pi^{0}_{n}$-computable) if it can be defined as an alternation of $n$ suprema and infima (resp.~$n$ infima and suprema) from a uniformly computable sequence of rationals. For more details, see~\cite{Zhen-Weihrauch_2001:arithmetical-hierarchy-of-real-numbers}.

\begin{theorem}[label={app:thm:entropy-dimensions}]
  For every $\Pi^{0}_{3}$-computable (resp.~$\Sigma^{0}_{2}$-computable) real number $\gamma \in [0,2]$, there exists a sofic shift $X\subset\{0,1\}^{\Z^2}$ with upper (resp.~lower) entropy dimension $\gamma$.
\end{theorem}

\begin{proof}
  This result follows from \Cref{app:lem:entropy-dimensions} and an alternative characterization of computable real numbers: by \cite[Lemma~4.1]{Meyerovitch_2011:growth-type-invariants-multidimensional-sfts}, a real number $\gamma \in [0,1]$ is $\Pi^{0}_{3}$-computable (resp.~$\Sigma^{0}_{2}$-computable) if there exists a computable function $g \colon \N \times \N \mapsto \{0,1\}$ such that:
  \[ \gamma = \limsup_{n \to +\infty} \frac{1}{n} \sum_{i=0}^{n-1} \inf_{j} g(i,j)
    \qquad \Big(\text{resp.} \mspace{8mu}
    \gamma = \liminf_{n \to +\infty} \frac{1}{n} \sum_{i=0}^{n-1} \inf_{j} g(i,j)
    \Big)\]
  For such function $g$, the set $\{ (\theta_{z'_\ell})_{\ell \in \N} \in \{0,1\}^{\N}  : \forall \ell,\, z'_{\ell} \leq \inf_{j} g(\ell,j) \}$ is then effectively closed, so that the resulting shift is sofic by~\Cref{sof:thm:bounded-d-adic-sofic}.
\end{proof}

\noindent
All along the proof, the $2$ can be replaced by $d$, because similar Toeplitz substitutions can be found with the wanded densities of $1$s, so that $d$-dimensional sofic shifts realize any $\gamma\in[0,d]$ with the same assumptions.

Moreover, a careful analysis of the pattern complexity of the finite-type cover $\mathcal{X}$ built in~\Cref{fix:sec} (see the comments in \Cref{sof:par:comment-entropy}) would allow to also apply the realization of entropy dimensions from~\Cref{app:thm:entropy-dimensions} to multidimensional shifts of finite type, thus fully reobtaining~\cite{Meyerovitch_2011:growth-type-invariants-multidimensional-sfts}.

\subsection{Intrinsic universality}

Since the $S$-adic structure of the fixed-point construction can be recognized locally (see the comments in \Cref{sof:par:comment-structure}), we also obtain the following rephrasing of \cite[Theorem~5.2]{Lafitte-Weiss_2010:tilings-simulation-and-universality} (see also \cite[Theorem~48]{Zinoviadis_2015:hierarchy-expansiveness-2d-subshifts-finite-type} for a deterministic version):
\begin{theorem}[label={app:thm:intrinsic-universality}]
  For every $d\ge2$ and every computably enumerable family of non-empty sofic shifts $(Y_n)_{n \in \N}$ with $\smash{Y_n \subseteq \alphabet[B]_n^{\mspace{2mu}\Z^d}}$, there exists a shift of finite type $\smash{X \subseteq \alphabet^{\Z^d}}$ that simulates every $Y_n$.
\end{theorem}
In this statement, we canonically represent sofic shifts by two alphabets $\alphabet[B]',\alphabet[B]\subset\{0,1\}^\ast$, a projection $\alphabet[B]' \to \alphabet[B]$ and a finite list of forbidden patterns over $\alphabet[B]'$.

\begin{remark}
  Notice that the set of all nonempty sofic shifts is itself not computably enumerable: as a matter of fact, there does not exist a sofic shift that simulates all its nonempty counterparts~\cite{Hochman_2009:universality-in-multidimensional-symbolic-dynamics,Ballier_2013:universality-symbolic-dynamics-constrained-by-medvedev-degrees}.
\end{remark}

The main ingredient for \Cref{app:thm:intrinsic-universality} is that, when enumerating a set without caring for speed, the algorithmic complexity is irrelevant.
\begin{lemma}[label={app:lem:ram-enumeration-turtle}]
  Let $U \subseteq \{0,1\}^{\ast}$ be any computably enumerable set of words.
  Then for any (total) computable size function $| \cdot | \colon \{0,1\}^{\ast} \to \N$, there exists a RAM program $e \in \{0,1\}^{\ast}$ such that $\funram{e} \colon \N \to \{0,1\}^{\ast}$ enumerates $U$ in a \emph{resource-efficient} manner:
  \begin{enumerate}[label=(\roman*)]
  \item For each $n \in \N$, $\funram{e}(n)$ halts in time $\bigO(n)$ with word length $\bigO(\log n)$;
  \item For any $n \in \N$, $\funram{e}(n)$ has size $|\funram{e}(n)| = \bigO(n)$.
  \end{enumerate}
\end{lemma}
\begin{proof}
  Since $U$ is computably enumeratble, there exists a program $\code{U} \in \{0,1\}^{\ast}$ such that $\funram{\code U}$ enumerates the set $U \subseteq \{0,1\}^{\ast}$.
  We build a surjective computable function $\psi \colon \N \to \N$ such that each $\funram{\code{U}}(\psi(n))$ can be computed with low complexity (relatively to $n$).
  We first set $\psi(0) = 0$; and then, for $n \in \N$, define $\psi(n+1) = \psi(n)+1$ provided that $ \funram{\code{U}}(\psi(n)+1)$ halts in time at most $\sqrt{n}$ with word length at most $\log n$, that its size can be computed with the same requirements and satisfies $|\funram{\code{U}}(\psi(n)+1)| \leq n$, and that $\psi(n)+1 \leq \sqrt{n}$;
  otherwise (if some check returns false or if the computation is too long), we set $\psi(n+1) = \psi(n)$.
  By construction (and since inequalities between short integers can be checked by log-RAM machines in constant time), $\psi(n)$ can be computed in time $\bigO(n)$, and we define the RAM program $e \in \{0,1\}$ to compute and return, on a given input $n \in \N$, the word $\funram{\code{U}}(\psi(n))$.

  The program $e$ naturally satisfies the required computational constraints and the inclusion ${\{ \funram{e}(n) : n \in \N \} \subseteq U}$.
  Conversely, $\psi$ is surjective: indeed, for any fixed $\psi(n)$, there will exist some minimal $N > n$ such that the computations of $\funram{\code{U}}(\psi(n)+1)$ will satisfy the required complexity constraints, so that $\psi(N) = \psi(n)+1$.
\end{proof}

\begin{proof}[Proof of \Cref{app:thm:intrinsic-universality}]
  Let $(Y_n)_{n \in \N}$ be a computably enumerable family of non-empty sofic shifts $\smash{Y_n \subseteq \alphabet[B]_n^{\mspace{2mu}\Z^d}}$:
  There exists a program $e \in \{0,1\}^{\ast}$ such that, for every $n \in \N$, $\funram{e}(n)$ outputs a presentation $(\alphabet[B]'_n,\alphabet[B]_n,\pi_n,\mathcal{F}_n)$ of the sofic shift $Y_n$, where $\mathcal{F}_n$ is a finite family of forbidden patterns over the alphabet $\alphabet[B]'_n$ that defines a finite-type cover $X_n$ of $Y_n$ by the projection $\pi_n \colon \alphabet[B]'_n \to \alphabet[B]_n$.

  By \Cref{app:lem:ram-enumeration-turtle}, we can assume that $e$ is a RAM program such that~$\funram{e}(n)$ halts in time~$\bigO(n)$ with word length~$\bigO(n)$, that symbols in the alphabets $\alphabet[B]'_n$ and $\alphabet[B]_n$ have bit length $\bigO(\log n)$, and that all patterns $f \in \mathcal{F}_n$ have domain $\dom(f) \subseteq \{-n,\ldots,n\}^d$.

  In what follows, we build a limit space in which the $n$\textsuperscript{th} level of the $N$-adic hierarchy embeds the sofic shifts $(Y_n)_{n \in \N}$ via their covers of finite type. More precisely, our limit space starts with the shift $Y_0$. For $\ell \in \N$, we set $N_{\ell} = 2$ so that $L_{\ell} = 2^{\ell}$, $V_{\ell} = \interval{-\ell-1}{\ell+1}^d$, and $\smash{\alphabet_{\ell} = \alphabet[B]'_{\psi(\ell)}}$. We define the parallel ndill maps $\smash{\ndill_{\ell} \colon \alphabet_{\ell+1}^{V_{\ell}} \to \alphabet_{\ell}^{\ibox{2}^d}}$ with the following algorithm: on input $(\rect{r},\vec{i},u)$ such that $\rect{r} = \ibox{2}^d$, $\vec{i} \in \rect{r}$ and $\smash{u \in \alphabet_{\ell+1}^{V_{\ell}}}$,
  \begin{itemize}
  \item Compute the presentation $\funram{e}(\ell)$ of the sofic shift $Y_{\ell}$;
  \item Check that the input pattern $u \in \alphabet_{\ell+1}^{V_{\ell}}$ does not contain any forbidden pattern from the presentation $\funram{e}(\psi(\ell+1))$. If any forbidden pattern is found, reject the computation; otherwise, continue.
  \item Finally, non-deterministically return any symbol $a \in \alphabet_{\ell}$ for the output tile $v_{\vec{i}}$.
  \end{itemize}
  Then for every $\ell \in \N$, all the configurations of the finite-type cover $X_{\ell}$ of $Y_{\ell}$ can appear at the $\ell$\textsuperscript{th}~level of the $N$-adic hierarchy; and conversely, every configuration at the $\ell$\textsuperscript{th} level belongs to~$X_{\ell}$. Furthermore, each ndill map $\tau_{\ell}$ is computable in time $\smash{\mathrm{poly}(\ell) = \bigO(\polylog L_{\ell})}$ and word length $\bigO(\log L_{\ell})$, so that the resulting limit space $\limspace{(Y_0)}[\mspace{2mu}\Sndill]$ is sofic by \Cref{sof:thm:square-d-adic-sofic}. Finally, using the recognizability of the $N$-adic structure (\textit{c.f.}~\Cref{sof:par:comment-structure}) and composing with the projections~$\pi_{\ell}$, we conclude that the associated cover of finite type simulates all the sofic shifts $(Y_n)_{n \in \N}$.
\end{proof}

\subsection{Effective subdynamics and periodic lifts}
\label{app:sec:periodic-lifts}

We recover a major result from the literature~\cite[etc\dots]{Hochman_2009:dynamics-recursive-properties-of-multidimensional-symbolic-systems,Durand-Romashchenko-Shen_2010:effective-closed-subshifts-1D-implemented-2D,Aubrun-Sablik_2014:multid-effective-s-adic-subshifts-are-sofic}: effective $\Z^{d}$ shifts are exactly the subdynamics of $\Z^{d+1}$ sofic shifts (this can be seen as a dynamical analog to Higman's theorem on effectively presented groups).
More precisely, for a $d$-dimensional shift space $\smash{X \subseteq \alphabet^{\Z^d}}$, its \emph{periodic lift} $\smash{\lift{X} \subseteq \alphabet^{\Z^{d+1}}}$ is the $(d+1)$-dimensional shift defined as
\[ \lift{X} = \{y \in \alphabet^{\Z^{d+1}} : \exists x \in X, \forall i \in \Z,\, \restr{y}{\{ i \} \times \Z^d} = x \}. \]

\begin{theorem}
  A $\Z^d$ shift $X \subseteq \alphabet^{\Z^d}$ is effective if and only if $\lift{X}$ is a sofic $\Z^{d+1}$ shift.
\end{theorem}
\noindent It is notable that the constructions from~\cite{Aubrun-Sablik_2014:multid-effective-s-adic-subshifts-are-sofic,Durand-Romashchenko-Shen_2010:effective-closed-subshifts-1D-implemented-2D} had to implement a complex mechanism (like small robots visiting neighboring macro-tiles) in order to check all patterns, even those that overlap macro-tiles at every level.
Our framework allows this check to be directly performed within the local constraints of the ndill map, which advertises the power of $N$-adicity, in particular in terms of conciseness.

\begin{figure}[ht]
  \definecolor{HRainbow0}{HTML}{ff2600}
  \definecolor{HRainbow1}{HTML}{fc6212}
  \definecolor{HRainbow2}{HTML}{f79f23}
  \definecolor{HRainbow3}{HTML}{fac713}
  \definecolor{HRainbow4}{HTML}{feef01}
  \definecolor{HRainbow5}{HTML}{7fd200}
  \definecolor{HRainbow6}{HTML}{00b500}
  \definecolor{HRainbow7}{HTML}{00c2af}
  \definecolor{HRainbow8}{HTML}{0170c4}
  \definecolor{HRainbow9}{HTML}{3a48b9}
  \definecolor{HRainbow10}{HTML}{7321b0}
  \definecolor{HRainbow11}{HTML}{d3288e}
  \begin{tikzpicture}[scale=.085,decoration={random steps,segment length=0.4em,amplitude=0.15em}]
    \pgfmathsetmacro{\N}{16}
    \begin{scope}[shift={(0,0)}]
      \pgfmathsetseed{1024}
      \draw[TilingGrid,decorate,path picture={
        \foreach \i in {-1,...,4} {
          \foreach \j in {-1,...,4} {
            \draw[very thin,pattern={Lines[angle=45,distance={3pt/sqrt(2)}]},pattern color=black!20] ($\N*(\i,\j) + (0,5)$) rectangle ++(\N,1);
            \pgfmathtruncatemacro{\st}{\j+6-4}
            \pgfmathtruncatemacro{\en}{\st+8}
            \foreach \k in {\st,...,\en} {
              \pgfmathtruncatemacro{\c}{Mod(\N*\i+\k+6,12)}
              \fill[HRainbow\c] ($\N*(\i,\j) + (\k,5)$) rectangle ++(1,1);
            }
            \draw[TilingGrid,opacity=.3] ($\N*(\i,\j) + (0,5)$) grid ++(\N,1);
            \draw[very thick,TilingGrid] ($\N*(\i,\j)$) rectangle ++(\N,\N);
            \draw[black,thin,{Classical TikZ Rightarrow[length=1.5pt]}-{Classical TikZ Rightarrow[length=1.5pt]}] ($\N*(\i,\j) + (0,7)$) -- ++($(\j+6.5,0)$);
            \node[black,font=\tiny] at ($\N*(\i,\j) + .5*(\j+6,0) + (.3,8.3)$) {$\Delta$};
          }
        }
      }] (-4,-4) rectangle ++(72,72);
    \end{scope}
    \draw[->] (71,32) -- (91,32) node[midway,above,yshift=1pt] {$\tau_{0} \circ \ldots \circ \tau_{\ell-1}$};
    \begin{scope}[shift={(98,0)}]
      \pgfmathsetseed{1024}
      \draw[TilingGrid,decorate,path picture={
        \foreach \k in {-8,...,71} {
          \pgfmathtruncatemacro{\c}{Mod(\k+6,12)}
          \fill[HRainbow\c] (\k,-8) rectangle ++(1,80);
        }
        \draw[TilingGrid,opacity=.4] (-8,-8) grid ++(80,80);
      }] (-4,-4) rectangle ++(72,72);
      \end{scope}
    \end{tikzpicture}
    \caption{Periodic lifts: a ``preimage'' configuration $x^{(\ell)} \in
      \alphabet_{\ell}^{\mspace{2mu}\Z^d}$ substituted to a vertically periodic
      configuration $x = \tau_0 \circ \ldots \circ \tau_{\ell-1}(x^{(\ell)}) \in
      \alphabet^{\Z^d}$.
      \label{app:fig:liftsigma}
    }
    \subcaption{In the configuration $x^{(\ell)}$, each tile encodes a column of width $2 \cdot \log L_{\ell}+1$ and an offset $\Delta \in \{0,\dots,L_{\ell+1}-1\}$. By construction, two vertically adjacent tiles must encode adjacent offsets (modulo $L_{\ell+1}$), thus ensuring that, in a column of tiles in $x^{(\ell)}$, all distinct slices of width $2 \cdot \log L_{\ell}+1$ (from the final configuration~$x$) are correctly represented.}
\end{figure}

\begin{proof}
  The shift $\lift{X}$ being sofic directly implies that one can computably enumerate the forbidden patterns of $\lift{X}$, and thus of $X$. In the rest of this proof, we consider the converse direction and focus on the case $d=1$ (cases $d \geq 2$ are similar, although more heavy on notations).

  Applying \Cref{sof:thm:square-d-adic-sofic} is actually not straightforward: $\basis{1}$-periodic patterns of size $n \times n$ contain $\bigO(n)$~bits of information, which is too large for our sublinear bounds. To circumvent this issue, we use the redundancy  of the $\basis{1}$-periodicity to distribute the checking of these $\bigO(n)$ bits in small chunks of size $\bigO(\log n)$ (see \Cref{app:fig:liftsigma}).

  Denote by $\mathcal{F}$ a computably enumerable family of forbidden patterns for $X$. We start with the shift $X_0 = \alphabet^{\Z^{2}}$ of all $\basis{1}$-periodic configurations, and define the sequence $(N_{\ell})_{\ell \in \N}$ as follows:
  \[ \begin{cases}
    N_{0} = 2 \\
    N_{\ell} = L_{\ell}, \quad \text{where $\smash{\displaystyle L_{\ell} = \prod_{i=0}^{\ell-1} N_{i}}$}.
  \end{cases} \]

\smallskip
\noindent A simple computation shows that $L_{\ell} = 2^{2^{\ell-1}}$ and thus $N_{\ell} = 2^{2^{\ell-2}}$ for $\ell \geq 2$. We also define the alphabet $\alphabet_{\ell}$ as the set of records on the following two fields:
\begin{itemize}
\item An \field{offset} $\Delta \in \{0,\dots,L_{\ell}-1\}$;
\item A \field{subpattern} $p \in \alphabet^{\Delta + \{-\log L_{\ell},\dots, \log L_{\ell}\}}$.
\end{itemize}
And we finally define a local map for a parallel ndill map
\[ \tau_{\ell} \colon \alphabet_{\ell+1}^{\{-1,0,1\}^{2}} \mto \alphabet_{\ell}^{N_{\ell+1} \times N_{\ell+1}}\]
by the following algorithm: on input $(\rect{r},\vec{i},u)$ such that $\rect{r} = \ibox{N_{\ell+1}}^{2}$, $\vec{i} \in \rect{r}$ and $u \in \alphabet_{\ell+1}^{\{-1,0,1\}^{2}}$,
\begin{itemize}
\item First, we check the geometry of the input tiles $u_{\vec{j}}$ ($\vec{j} \in \{-1,0,1\}^d$): if $\Delta \in \{0,\dots,L_{\ell+1}-1\}$ is the \field{offset} of $u_{\vec{0}}$, check that the \field{offsets} of the symbols $u_{\vec{0} \pm \basis{1}}$ are $\Delta \pm 1 \bmod L_{\ell+1}$; and that other symbols $u_{\vec{0} \pm \basis{2}}$ have \field{offset} $\Delta$. If any mistake is found, reject the computation; otherwise, continue.
\item Then, we check the \field{subpattern} of the input tile $u_{\vec{0}}$: after enumerating patterns from the forbidden family $\mathcal{F}$ for some $\bigO(\log L_{\ell})$ steps of computations, we check that none of the enumerated patterns appear in the \field{subpattern} of $u_{\vec{0}}$. Otherwise, we reject the computation;
\item Finally, we define the output tile $v_{\vec{i}} \in \alphabet_{\ell}$ of position $\vec{i} = (i_{1}, i_{2})$ as follows:
  \begin{itemize}
  \item We set the \field{offset} of $v_{\vec{i}}$ to $\Delta' = i_1 \in \{0,\dots,L_{\ell}-1\}$;
  \item We non-deterministically guess the \field{subpattern} $p' \in \alphabet^{\Delta' + \{-\log L_{\ell},\dots,\log L_{\ell}\}}$ of $v_{\vec{i}}$, but ensure that:
    \begin{itemize}
    \item Let $\bar{p}'$ be the translation of the pattern $p'$ of distance $L_{\ell} \cdot i_2$;
    \item For every $\vec{j} \in \{-1,0,1\}^{2}$, if $p^{(\vec{j})}$ is the \field{subpattern} of $u_{\vec{j}}$, similarly denote $\bar{p}^{(\vec{j})}$ the translation of $p^{(\vec{j})}$ of distance $L_{\ell+1} \cdot j_2$ (where $\vec{j} = (j_1,j_2)$);
    \item If the patterns $\bar{p}^{(\vec{j})}$ and $\bar{p}'$ intersect, they coincide on their common~subdomain.
    \end{itemize}
  \end{itemize}
\end{itemize}

Checking whether a pattern of size $\bigO(\log L_{\ell})$ appears inside a pattern of domain $\llbracket \log L_{\ell} \rrbracket$ can be computed in time $\bigO(\log^2 L_{\ell})$, so that each ndill map $\tau_{\ell}$ is computable in time $\bigO(\log^3 L_{\ell})$.

Furthermore, identifying $\alphabet_{0}$ with $\alphabet$, the $N$-adic configurations in the periodic shift $X_0$ exactly form the lift~$\lift{X}$: indeed, since image patterns of $\tau_{\ell}$ are of size $\ibox{L_{\ell}}^{2}$, all pixels in the domain $\dom(\bar{p}^{(\vec{j})}) \cap \ibox{L_{\ell}}$ of any input \field{subpattern} $\bar{p}^{(\vec{j})}$ will be copied to at least one image tile of~$\tau_{\ell}(u)$. Inductively, we deduce that, if a symbol $a \in \alphabet_{\ell}$ of \field{offset} $\Delta$ and \field{subpattern} $p$ is substituted towards some $\basis{1}$-periodic pattern $v \in \alphabet^{\llbracket L_{\ell} \rrbracket}$ by $\tau_0 \circ \dots \circ \tau_{\ell-1}$, then $\restr{v}{\{0\} \times (\Delta + \{-\log L_{\ell},\dots,\log L_{\ell}\})}$ coincides with the \field{subpattern} $p$. By \Cref{sof:thm:square-d-adic-sofic}, this concludes the proof.
\end{proof}

\subsection{Shifts defined by information quantification}

The informal $\bigO(n^{d-1})$ bound on the information flow between $\ibox{n}^d$-patterns and their exterior has motivated numerous developments in the sofic literature. For example,~\cite{Kass-Madden_2013:sufficient-condition-non-soficness-multidimensional} and~\cite{Ormes-Pavlov_2016:extender-sets-multidimensional} prove necessary conditions on soficity phrased in the formalism of \emph{extender sets}; while the necessary conditions from~\cite{Destombes-Romashchenko_2022:ressource-bounded-Kolmogorov-complexity-soficness-multidimensional-shifts} are based upon resource-bounded Kolmogorov complexity, and those of~\cite{Guillon-Jeandel_2015:infinite-communication-complexity} on non-deterministic communication complexity.

On the other hand, sufficient conditions for the soficity of multidimensional shifts have been relatively scarce. As illustrated by~\cite[Corollary~1]{Destombes-Romashchenko_2022:ressource-bounded-Kolmogorov-complexity-soficness-multidimensional-shifts}, such conditions would need to take into account not only the amount of ``useful information'' contained in the patterns of these shifts, but also the time complexity required to compute the associated efficient descriptions.

\enlargethispage{.5\baselineskip}
\subsubsection{Low-density shifts}
In this direction,~\cite{Destombes_2021:phd:algorithmic-complexity-soficness-subshifts-multidimensional} considers a family of shifts whose patterns have obvious and computationally efficient descriptions: shifts of low density. In particular, we generalize~\cite[Theorem~5]{Destombes_2021:phd:algorithmic-complexity-soficness-subshifts-multidimensional} from $\Z^2$ to arbitrary $\Z^d$ for $d \geq 2$:
\begin{theorem}[name=Low density shifts, label={app:thm:low-density}]
  For $\alpha < d-1$ and $\alphabet = \{\square,\square[black]\}$ the binary alphabet, let $\smash{X \subseteq \alphabet^{\mathbb{Z}^{d}}}$ be an effective shift space such that for any configuration $x \in X$ and any pattern $w \subpattern x$ of domain $\dom(w) = \ibox{n}^d$, $w$ contains fewer than $\bigO(n^{\alpha})$ black cells.
  Then $X$ is sofic.
\end{theorem}

\begin{proof}
  We apply \Cref{sof:thm:square-d-adic-sofic}. Since $X$ is an effective shift, let us fix a computably enumerable family of forbidden patterns $\mathcal{F}$ representing $X$. For $\ell \in \N$, we define $N_{\ell} = 2$, $L_{\ell} = 2^{\ell}$, and the alphabet $\alphabet_{\ell}$ as the subsets of positions $S \subseteq
  \ibox{L_{\ell}}^d$ containing at most $\bigO(L_{\ell}^{\alpha})$ elements.
  Intuitively, a tile $a \in \alphabet_{\ell}$ is meant to represent the set of $\square[black]$-colored positions in the pattern $\smash{\tau_{0} \circ \dots \circ \tau_{\ell-1}(a) \in \alphabet^{\ibox{L_{\ell}}^d}}$ with the parallel ndill map
  \[ \smash{\ndill_{\ell} \colon \alphabet_{\ell+1}^{\{-1,0,1\}^d} \mto \alphabet_{\ell}^{\ibox{2}^d}} \]
  defined by the following algorithm: on input $(\rect{r},\vec{i},u)$ such that $\rect{r} = \ibox{2}^d$, $\vec{i} \in \rect{r}$ and $\smash{u \in \alphabet_{\ell+1}^{\{-1,0,1\}^d}}$,
  \begin{itemize}
  \item From the pattern $u \in \alphabet_{\ell+1}^{\{-1,0,1\}^d}$, define the set of merged positions $S \in (\N^d)^{\ast}$ as
    \[ S = \bigcup_{\vec{j} \in \{-1,0,1\}^d} \big\{ \vec{p}_{i} + L_{\ell+1} \cdot \vec{j} : \vec{p}_{i} \in u_{\vec{j}} \big\}.\]
  \item Let $\mathcal{F}_{\ell}$ be the list of forbidden patterns computed by the first $\ell$ steps of the family $\mathcal{F}$ representing the shift $X$. In particular, each pattern $w_f \in \mathcal{F}_{\ell}$ satisfies  $\dom(w_f) \subseteq \ibox{\ell}^d$.
  \item Then for every pattern $w_f \in \mathcal{F}_{\ell}$:
    \begin{itemize}
    \item And for every position $\vec{p} \in S$:
      \begin{itemize}
      \item Check if the black cells of $w_f$ coincide with a set of positions in $S$ around $\vec{p}$;
      \item If it does, reject. Otherwise, continue.
      \end{itemize}
    \end{itemize}
  \item Finally, we define the output tile at position $\vec{i}$: from the positions stored in the input tile $u_{\vec{0}} \in \alphabet_{\ell}$, we compute the (shifted) subset of positions $S_{\vec{i}}' \subseteq \N^d$ in the quadrant of index $\vec{i}$:
    \[ S_{\vec{i}}' = \big\{ \vec{p} - L_{\ell} \cdot \vec{i} : \vec{p} \in S \cap (L_{\ell} \cdot \vec{i} + \ibox{L_{\ell}}^d) \big\}.\]
    Then $S_{\vec{i}}'$ is actually an element of $\alphabet_{\ell}$, and we return $S_{\vec{i}}'$.
  \end{itemize}

  Checking the appearance of a single forbidden pattern of domain $\ibox{\ell}^d$ around a position $\vec{p} \in S$ can be computed in time $\bigO(\ell^d \cdot \log (3^d \cdot L_{\ell}^{\alpha})) = \mathrm{poly}(\ell)$ if the set $S$ is implemented using an efficient data structure, such as a balanced binary tree. Thus, we deduce that a local function for $\ndill_{\ell}$ can be computed in time $\bigO(L_{\ell}^{\alpha} \cdot \mathrm{poly}(\ell))$ for some $\alpha < 1$. Identifying in $\alphabet_{0}$ the empty set~$\emptyset$ with the symbol $\square$ and the singleton $\{\vec{0}\}$ with the symbol $\square[black]$, we obtain the desired result.
\end{proof}

\begin{figure}[t]
  \centerline{
    \begin{tikzpicture}[scale=.114,decoration={random steps,segment length=0.4em,amplitude=0.15em}]
      \begin{scope}
        \pgfmathsetseed{4}
        \draw[TilingGrid,decorate,path picture={
          \foreach \i/\j in {-4/34,-2/-3,-2/15,-2/23,0/-2,0/28,0/29,4/27,5/30,6/13,6/16,6/24,8/24,9/17,13/4,14/19,15/-3,16/10,18/19,20/3,20/15,20/22,21/30,24/12,24/16,25/27,26/9,27/4,27/8,27/27,30/22,30/35,34/20,35/20} {
            \fill[black,opacity=.1] (\i,\j) rectangle ++(1,1);
          }
          \node[black,font=\ttfamily\tiny,align=center,text width=3.2cm] at (-16,16) {[{(0,27)}, {(2,9)}, {(3,13)}, {(6,17)}, {(7,2)}, {(8,6)}, {(8,16)}, {(11,3)}, {(30,15)}, {(30,23)}]}; 
          \node[black,font=\ttfamily\tiny,align=center,text width=3.2cm] at (16,16) {[{(0,28)}, {(0,29)}, {(4,27)}, {(5,30)}, {(6,13)}, {(6,16)}, {(6,24)}, {(8,24)}, {(9,17)}, {(13,4)}, {(14,19)}, {(16,10)}, {(18,19)}, {(20,3)}, {(20,15)}, {(20,22)}, {(21,30)}, {(24,12)}, {(24,16)}, {(25,27)}, {(26,9)}, {(27,4)}, {(27,8)}, {(27,27)}, {(30,22)}]}; 
          \node[black,font=\ttfamily\tiny,align=center,text width=3.2cm] at (48,16) {[{(2,20)}, {(3,20)}, {(5,7)}, {(5,11)}, {(5,27)}, {(12,27)}, {(16,16)}, {(19,14)}, {(22,20)}, {(22,24)}]}; 
          \node[black,font=\ttfamily\tiny,align=center,text width=3.2cm] at (16,-16) {[{(0,30)}, {(15,11)}, {(15,29)}, {(17,0)}]}; 
          \node[black,font=\ttfamily\tiny,align=center,text width=3.2cm] at (16,48) {[{(26,28)}, {(30,3)}]}; 
          \draw[TilingGrid,opacity=.2] (-8,-8) grid (40,40);
          \draw[very thick,TilingGrid,step=32] (-8,-8) grid (40,40);
        }] (-4,-4) rectangle (36,36);
      \end{scope}
      \draw[->] (38,16) -- (42,16) node[midway,above,yshift=1pt] {$\tau_{\ell-1}$};
      \begin{scope}[shift={(48,0)}]
        \pgfmathsetseed{4}
        \draw[TilingGrid,decorate,path picture={
          \foreach \i/\j in {-4/34,-2/-3,-2/15,-2/23,0/-2,0/28,0/29,4/27,5/30,6/13,6/16,6/24,8/24,9/17,13/4,14/19,15/-3,16/10,18/19,20/3,20/15,20/22,21/30,24/12,24/16,25/27,26/9,27/4,27/8,27/27,30/22,30/35,34/20,35/20} {
            \fill[black,opacity=.1] (\i,\j) rectangle ++(1,1);
          }
          \node[black,font=\ttfamily\tiny,align=center,text width=1.7cm] at ($16*(-1,0) + (8,8)$) {[{(14,15)}]};
          \node[black,font=\ttfamily\tiny,align=center,text width=1.7cm] at ($16*(-1,1) + (8,8)$) {[{(14,7)}]};
          \node[black,font=\ttfamily\tiny,align=center,text width=1.7cm] at ($16*(0,-1) + (8,8)$) {[{(0,14)}, {(15,13)}]};
          \node[black,font=\ttfamily\tiny,align=center,text width=1.7cm] at ($16*(0,0) + (8,8)$) {[{(6,13)}, {(13,4)}]};
          \node[black,font=\ttfamily\tiny,align=center,text width=1.7cm] at ($16*(0,1) + (8,8)$) {[{(0,12)}, {(0,13)}, {(4,11)}, {(5,14)}, {(6,0)}, {(6,8)}, {(8,8)}, {(9,1)}, {(14,3)}]};
          \node[black,font=\ttfamily\tiny,align=center,text width=1.7cm] at ($16*(1,0) + (8,8)$) {[{(0,10)}, {(4,3)}, {(4,15)}, {(8,12)}, {(10,9)}, {(11,4)}, {(11,8)}]};
          \node[black,font=\ttfamily\tiny,align=center,text width=1.7cm] at ($16*(1,1) + (8,8)$) {[{(2,3)}, {(4,6)}, {(5,14)}, {(8,0)}, {(9,11)}, {(11,11)}, {(14,6)}]};
          \node[black,font=\ttfamily\tiny,align=center,text width=1.7cm] at ($16*(1,2) + (8,8)$) {[{(14,3)}]};
          \node[black,font=\ttfamily\tiny,align=center,text width=1.7cm] at ($16*(2,-1) + (8,8)$) {[{(5,12)}]};
          \node[black,font=\ttfamily\tiny,align=center,text width=1.7cm] at ($16*(2,0) + (8,8)$) {[{(5,7)}, {(5,11)}]};
          \node[black,font=\ttfamily\tiny,align=center,text width=1.7cm] at ($16*(2,1) + (8,8)$) {[{(2,4)}, {(3,4)}, {(5,11)}, {(12,11)}]};
          \draw[TilingGrid,opacity=.2] (-8,-8) grid (40,40);
          \draw[very thick,TilingGrid,step=16] (-8,-8) grid (40,40);
        }] (-4,-4) rectangle (36,36);
      \end{scope}
      \draw[->] (87,16) -- (101,16) node[midway,above,yshift=1pt] {$\tau_{0} \circ \ldots \circ \tau_{\ell-2}$};
      \begin{scope}[shift={(108,0)}]
        \pgfmathsetseed{4}
        \draw[TilingGrid,decorate,path picture={
          \foreach \i/\j in {-4/34,-2/-3,-2/15,-2/23,0/-2,0/28,0/29,4/27,5/30,6/13,6/16,6/24,8/24,9/17,13/4,14/19,15/-3,16/10,18/19,20/3,20/15,20/22,21/30,24/12,24/16,25/27,26/9,27/4,27/8,27/27,30/22,30/35,34/20,35/20} {
            \fill[black] (\i,\j) rectangle ++(1,1);
          }
          \draw[TilingGrid,opacity=.5] (-8,-8) grid (40,40);
        }] (-4,-4) rectangle (36,36);
      \end{scope}
    \end{tikzpicture}}
  \caption{Density shifts: a preimage configuration $x^{(\ell)} \in \alphabet_{\ell}^{\mspace{2mu}\Z^d}$ substituted to some $\tau_{\ell-1}(x^{(\ell)}) \in \alphabet_{\ell-1}^{\mspace{2mu}\Z^d}$, then to some low density ${x = \tau_{0} \circ \ldots \circ \tau_{\ell-1}(x^{(\ell)}) \in \alphabet^{\Z^d}}\!$.}
  \subcaption{In the configuration $x^{(\ell)}$, each tile encodes a list of positions in $\ibox{L_{\ell}}^d$, thus simulating a virtual pattern of domain $\ibox{L_{\ell}}^d$. When substituting towards images of size $\ibox{2}^d$, the positions of the tiles are distributed to their $2^d$ quadrants and shifted accordingly. For clarity, this figure also illustrates (using transparent layers) the sets of positions inside each tile by drawing its corresponding grid and $\square[black!10]$ symbols.}
\label{fig:app:lowdensity}
\end{figure}

\medskip
The conclusion of~\cite{Destombes_2021:phd:algorithmic-complexity-soficness-subshifts-multidimensional} already mentions that it should be possible to generalize \Cref{app:thm:low-density} into more abstract computability-based conditions for the soficity of multidimensional shift spaces. In his thesis~\cite{Callard_2025:phd:soficity-multidimensional-subshifts}, the first author designed such a criterion by introducing a framework of \emph{inductive representations} to recursively describe patterns as efficiently-compressed strings. In this framework, a \emph{representation} maps square patterns to binary words of smaller sizes; and an \emph{induction map} merges arbitrary representations of adjacent square patterns into a representation of their union. In~\cite[Theorem~10.11]{Callard_2025:phd:soficity-multidimensional-subshifts}, shift spaces that admit small representations and efficiently-computable induction maps are proved to be sofic.

While these results had the desired information-theoretic flavor, they also resulted in rather convoluted proofs when applied to concrete examples (\textit{e.g.}~periodic lifts). A first motivation for the current article was the rephrasing of these results in a more user-friendly setting. Following the parallel between the \emph{self-simulation} of the fixed-point construction and the well-established notion of substitutions, we ended up generalizing the $(2 \times 2)$-sized induction maps into arbitrarily-sized rectangular ndill maps: in particular,~\cite[Theorem~10.11]{Callard_2025:phd:soficity-multidimensional-subshifts} is a corollary of \Cref{sof:thm:square-d-adic-sofic}.

\subsubsection{Gap shifts}\label{app:sec:gap-shifts}

For $f \colon \N \to \N$ an increasing function, the \emph{gap shift} $X_{G,f} \subseteq \{\square,\square[black]\}^{\Z^d}$ was introduced in~\cite{Törma_2026:countable-SFT-covers-sparse-shifts} as the set of configurations~$x$ in which, if two $\square[black]$~cells appear at distinct position $\vec{i}, \vec{j} \in \Z^d$, then the number of $\square[black]$ cells in $x$ is at most $f(\lVert \vec{i} - \vec{j} \rVert_1)$. It was also conjectured that:
\begin{theorem}[Gap shifts]
  In dimension $d \geq 2$, the gap shift $X_{G,f}$ is sofic if and only if the gap function $f \colon \N \to \N$ is upper semi-computable.
\end{theorem}
\noindent While I.~Törmä did not prove the conjecture to avoid the hassle of writing yet another entire fixed-point construction, he suggested in personal communications that it could be a simple application of \Cref{sof:thm:square-d-adic-sofic}:

\begin{proof}
  It is straightforward to produce an algorithm semi-computing $f \colon \N \to \N$ from above if $X_{G,f}$ is sofic. For the converse direction, we apply \Cref{sof:thm:square-d-adic-sofic} in the case of a non-parallel ndill map (see \Cref{sof:par:comment-parallelism}). Assuming that $f$ is upper semi-computable, there exists a program $\code{f} \in \{0,1\}^{\ast}$ such that for $n,t \in \N$, $\funram{\code{f}}(n,t)$ outputs an integer such that $\inf_{t \to +\infty} \funram{\code{f}}(n,t) = f(n)$.

  \smallskip
  For $\ell \in \N$, we define $N_{\ell} = 2$, $L_{\ell} = 2^{\ell}$, and the alphabet $\alphabet_{\ell}$ to encode the $\square[black]$ cells of an image pattern after $\ell$ steps of substitutions:
  \begin{itemize}
  \item A \field{bound} $B \in \{0,\dots,L_{\ell+1}-1\} \cup \{\infty\}$ represents a uniform lower bound on the distance between two neighboring $\square[black]$ cells;
  \item If $B = \infty$, a \field{position} $\vec{p} \in \ibox{L_{\ell}}^d$ represents (if non-blank) the coordinates of the $\square[black]$~cell of the tile; while a blank \field{position} represents a tile with no  $\square[black]$~cell at all;
  \item If $B < L_{\ell+1}$, a \field{number} $n \in \{0,\dots,L_{\ell}^d\}$ represents the number of $\square[black]$~cells in the tile.
  \end{itemize}
   We then define a computable ndill map
  \[ \smash{\ndill_{\ell} \colon \alphabet_{\ell+1}^{\{-2,\dots,2\}^d} \mto \alphabet_{\ell}^{\ibox{2}^d}} \]
  as the following algorithm: on input $\smash{u \in \alphabet_{\ell+1}^{\{-2,\dots,2\}^d}}$,
  \begin{itemize}
  \item All input tiles $u_{\vec{j}}$ ($\vec{j} \in \{-2,\dots,2\}^d$) must share the same \field{bound} ${B \in \{0,\dots,L_{\ell+2}-1\} \cup \{\infty\}}$;
  \item If $L_{\ell+1} \leq B$, the input tile $u_{\vec{0}}$ has a (possibly blank) \field{position} field. In which case,
    \begin{itemize}
    \item If $L_{\ell+1} \leq B < L_{\ell+2}$, we check for every distinct $\vec{j},\vec{j}' \in \{-2,\dots,2\}^d$ that if the input tiles $u_{\vec{j}}$ and $u_{\vec{j}'}$ both have non-blank \field{positions} $\vec{p}_{\vec{j}}$ and $\vec{p}_{\vec{j}'}$ in $\ibox{L_{\ell+1}}$, then the distance $\lVert (L_{\ell+1} \cdot \vec{j} + \vec{p}_{\vec{j}}) - (L_{\ell+1} \cdot \vec{j}' + \vec{p}_{\vec{0}}) \rVert_1$ must be larger than $B$;
    \end{itemize}
    and we define the output pattern $v \in \alphabet_{\ell}^{\ibox{2}^d}$ of the tile $u_{\vec{0}}$ as follows:
    \begin{itemize}
    \item All tiles $v_{\vec{i}}$ ($\vec{i} \in \ibox{2}^d$) have the same \field{bound} $B' = \infty$;
    \item If the \field{position} of the input tile $u_{\vec{0}}$ is blank, then all tiles $v_{\vec{i}}$ ($\vec{i} \in \ibox{2}^d$) have a blank \field{position} field.
    \item If the \field{position} of the input tile $u_{\vec{0}}$ is some non-blank $\vec{p} \in \ibox{L_{\ell+1}}^d$, then let $\vec{i} \in \ibox{2}^d$ be the quadrant index such that $\vec{p} = L_{\ell} \cdot \vec{i} + \vec{p'}$ for some $\vec{p'} \in \ibox{L_{\ell}}^d$. Then $v_{\vec{i}}$ has a non-blank \field{position} $\vec{p}'$, and all other tiles $v_{\vec{i}'}$ ($\vec{i}' \neq \vec{i}$) have blank \field{position}.
    \end{itemize}
  \item If $L_{\ell+1} > B$, all input tiles $u_{\vec{j}}$ ($\vec{j} \in \{-2,\dots,2\}^d$) must have a non-blank \field{number} ${n_{\vec{j}} \in \{0,\dots,L_{\ell+1}^d-1\}}$. In which case,
    \begin{itemize}
    \item Denoting
      \[ n = \sum_{\vec{j} \in \{-2,\dots,2\}^d} n_{\vec{j}} \]
      the total amount of $\square[black]$ cells stored in the tiles of the input pattern $u$, sequentially compute $\funram{\code{f}}(M,t)$ for increasing values of $t \in \N$ and let $N \in \N$ be the minimal value enumerated in at most $\bigO(\ell)$ steps of computation. We check that $n$ satisfies $n \leq N$: otherwise, we reject the computation;
    \end{itemize}
    And we then define the output pattern $v \in \alphabet_{\ell}^{\ibox{2}^d}$ of the tile $u_{\vec{0}}$ as follows:
    \begin{itemize}
    \item Non-deterministically define some \field{numbers} $n'_{\vec{i}} \in \{0,\dots,L_{\ell}^d\}$ such that $\smash{\sum_{\vec{i} \in \ibox{2}^d} n'_{\vec{i}} = n_{\vec{0}}}$;
    \item All tiles $v_{\vec{i}}$ ($i \in \ibox{2}^d$) have the same \field{bound} $B' = B$;
    \item If $B < L_{\ell}$, the tile $v_{\vec{i}}$ ($i \in \ibox{2}^d$) has \field{number} $n'_{\vec{i}}$;
    \item If $B \geq L_{\ell}$, check that each $n'_{\vec{i}}$ is either $0$ or $1$ (otherwise, reject). If $n'_{\vec{i}} = 0$, then $v_{\vec{i}}$ has blank \field{position}. Otherwise, $v_{\vec{i}}$ has any non-deterministic $\vec{p}'_{\vec{i}} \in \ibox{L_{\ell}}^d$.
    \end{itemize}
  \end{itemize}
  The idea is that a unique level $\ell$ (the one such that $L_{\ell+1} \leq B < L_{\ell+1}$) is in charge of checking the distance condition against the embedded \field{bound} $B$. The ndill maps $\tau_{\ell}$ are then computable in time~$\bigO(L_{\ell})$, and have neighborhood $\{-2,\dots,2\}^d$ so that, if $L_{\ell+1} \leq B < L_{\ell+2}$, the \field{positions} in a neighborhood can represent a whole ball of radius $B$. Identifying symbols from $\alphabet_0$ of non-blank \field{position} with $\square[black]$, and those of blank \field{position} with $\square[white]$, we claim that the associated limit space is actually $X_{G,f}$ (since $f$ is assumed to be increasing, and that \field{bound} fields can never grow larger than the distance between two $\square[black]$ cells). Applying the non-parallel version of \Cref{sof:thm:square-d-adic-sofic}, we conclude that $X_{G,f}$ is sofic.
\end{proof}

\subsubsection{Seas of squares}
\label{app:sec:seas-of-squares}

\newcommand{\Xsquare}[1]{\ensuremath{X_{\begin{tikzpicture}[x=1.4ex,y=1.4ex,baseline=.15ex]
        \fill[black] (0,0) rectangle (1,1);
        \draw[thin,TilingGrid,step=0.3333,opacity=0.9] (0,0) grid (1,1);
        \draw[black!30] (0,0) rectangle (1,1);
      \end{tikzpicture},#1}}}

We now recover another example from the literature about soficity: the ``seas of squares'' from~\cite{Westrick_2017:seas-of-squares}. For $S \subseteq \N$ a set of integers, the $S$-square shift $\Xsquare{S} \subseteq \smash{\{\square, \square[black]\}^{\mathbb{Z}^{2}}}$ is the set of configurations where independent connected components of black cells (as a grid subgraph) form (axis-aligned) squares of sizes~$S$: that is, $\Xsquare{S}$ is the set of all configuration $x \in \smash{\{\square, \square[black]\}^{\mathbb{Z}^{2}}}$ such that, if $D \subseteq \mathbb{Z}^{2}$ is a finite and a maximal connected component of $\square[black]$-cells in $x$, then $D$ is a rectangle in~$\Srect$ such that $\rnorm{D} = \ibox{n}^2$ for some $n \in S$.

\begin{theorem}[note={\cite{Westrick_2017:seas-of-squares}}]
  For $S \subseteq \mathbb{N}$ a computably co-enumerable set, the shift $\Xsquare{S}$ is sofic.
\end{theorem}

\begin{figure}[t]
  \centerline{
    \begin{tikzpicture}[scale=.114,decoration={random steps,segment length=0.4em,amplitude=0.15em}]
      \begin{scope}
        \pgfmathsetseed{6}
        \draw[TilingGrid,decorate,path picture={
          \fill[black!80] (-2,26) rectangle ++(10,10);
          \fill[black!80] (10,6) rectangle ++(18,18);
          \fill[black!80] (12,-23) rectangle ++(25,25);
          \draw[TilingGrid,opacity=.2] (-8,-8) grid (40,40);
          \node[black,font=\ttfamily\scriptsize] at (20,27.5) {[1,2,4,5,8,10,18,25]};
          \draw[very thick,TilingGrid,step=32] (-8,-8) grid (40,40);
        }] (-4,-4) rectangle (36,36);
      \end{scope}
      \draw[->] (38,16) -- (42,16) node[midway,above,yshift=1pt] {$\tau_{\ell-1}$};
      \begin{scope}[shift={(48,0)}]
        \pgfmathsetseed{6}
        \draw[TilingGrid,decorate,path picture={
          \fill[black!80] (22,30) rectangle ++(8,8);
          \fill[black!80] (31,5) rectangle ++(8,8);
          \fill[black!80] (-2,26) rectangle ++(10,10);
          \fill[black!80] (10,6) rectangle ++(18,18);
          \fill[black!80] (12,-23) rectangle ++(25,25);
          \draw[TilingGrid,opacity=.2] (-8,-8) grid (40,40);
          \node[black,font=\ttfamily\scriptsize] at (24,27) {[2,5,8]};
          \node[black,font=\ttfamily\scriptsize] at (5,22) {[2,5,10]};
          \node[black,font=\ttfamily\scriptsize] at (5,7) {[4,5]};
          \node[black,font=\ttfamily\scriptsize] at (24,4) {[1,8]};
          \draw[very thick,TilingGrid,step=16] (-8,-8) grid (40,40);
        }] (-4,-4) rectangle (36,36);
      \end{scope}
      \draw[->] (87,16) -- (101,16) node[midway,above,yshift=1pt] {$\tau_{0} \circ \ldots \circ \tau_{\ell-2}$};
      \begin{scope}[shift={(108,0)}]
        \pgfmathsetseed{6}
        \draw[TilingGrid,decorate,path picture={
          \fill[black!85] (21,3) rectangle ++(1,1);
          \fill[black!85] (27,4) rectangle ++(1,1);
          \fill[black!85] (15,27) rectangle ++(2,2);
          \fill[black!85] (2,2) rectangle ++(4,4);
          \fill[black!85] (-1,13) rectangle ++(5,5);
          \fill[black!80] (29,22) rectangle ++(5,5);
          \fill[black!85] (22,30) rectangle ++(8,8);
          \fill[black!85] (31,5) rectangle ++(8,8);
          \fill[black!85] (-2,26) rectangle ++(10,10);
          \fill[black!85] (10,6) rectangle ++(18,18);
          \fill[black!85] (12,-23) rectangle ++(25,25);
          \draw[TilingGrid,opacity=.5] (-8,-8) grid (40,40);
        }] (-4,-4) rectangle (36,36);
      \end{scope}
    \end{tikzpicture}}
  \caption{Seas of squares: a preimage configuration $x^{(\ell)} \in \alphabet_{\ell}^{\mspace{2mu}\Z^d}$ substituted to some $\tau_{\ell-1}(x) \in \alphabet_{\ell-1}^{\mspace{2mu}\Z^d}$, then to some sea of squares $x = \tau_{0} \circ \ldots \circ \tau_{\ell-1}(x^{(\ell)}) \in \protect\Xsquare{\N}$.}
  \subcaption{In the configuration $x^{(\ell)}$, each tile encodes a \emph{size list} of square sizes and a \emph{rectangle list} storing the positions of some ``large'' rectangles (\textit{i.e.}~of size $\geq L_{\ell-1}$, or touching a corner), thus simulating a virtual pattern of domain $\ibox{L_{\ell}}^2$. When substituting towards images of size $\ibox{2}^d$, the \emph{size list} of the tiles are distributed to their $2^d$ quadrants (entries $n \geq L_{\ell-1}$ are forgotten), and the new rectangles of sizes $L_{\ell-2} \leq n < L_{\ell-1}$ are introduced non-deterministically.}
  \label{app:fig:seas-of-squares-proof}
\end{figure}

\begin{proof}
  \newcommand{\sizelist}{\ensuremath{S_{\leftrightarrow}}}
    We again apply the non-parallel version of \Cref{sof:thm:square-d-adic-sofic} (see \Cref{sof:par:comment-parallelism}). Assuming that $S \subseteq \N$ is computably co-enumerable, there exists an algorithm enumerating $\N \setminus S$, \textit{i.e.}~a $\log$-RAM program $\code{S}$ such that $\funram{\code{S}}(n)$ halts if and only if $n \notin S$.

  Let us start with the full $\N$-square shift $\Xsquare{\N}$. For $\ell \in \N$, we define $N_{\ell} = 2$, $L_{\ell} = 2^{\ell}$, and the alphabet $\alphabet_{\ell}$ as the set of records on the following two fields:
  \begin{itemize}
  \item A \field{size list} $\sizelist \subseteq \{0,\dots,L_{\ell}-1\}$, which intuitively represents the various sizes of squares intersecting the tile;
  \item A \field{rectangle list} $R \subseteq \Srect$ of non-overlapping rectangles $\rect{r} \subseteq \ibox{L_{\ell}}^2$ such that:
    \begin{enumerate}[label=(\roman*)]
    \item Either $\rect{r}$ covers a corner of the cube~$\ibox{L_{\ell}}^2$;
    \item Or $\rect{r}$ intersects a facet $\facet{k}{\pm}(\ibox{L_{\ell}}^2)$ and the edge of $\rect{r}$ that is parallel to this facet is longer than~$L_{\ell-1}$;
    \item Or $\rect{r}$ is strictly contained in the interior of~$\ibox{L_{\ell}}^2$, and $\rect{r}$ is a square of size~$n \times n$ for some $n \geq L_{\ell-1}$.
    \end{enumerate}
  \end{itemize}
  In particular, a valid \field{rectangle list} $R$ cannot contain more than 9 elements (one in the center, one per edge, and one per corner). We then define a computable ndill map
  \[ \smash{\ndill_{\ell} \colon \alphabet_{\ell+1}^{\{-1,0,1\}^d} \mto \alphabet_{\ell}^{\ibox{2}^d}} \]
  in the following algorithm: on input $\smash{u \in \alphabet_{\ell+1}^{\{-1,0,1\}^d}}$:
  \enlargethispage{-\baselineskip}
  \begin{itemize}
  \item Let $\sizelist \subseteq \{0,\dots,L_{\ell}-1\}$ and $R \subseteq \Srect$ be the \field{size list} and the \field{rectangle list} of the central tile~$u_{\vec{0}}$; and denote $D = \bigcup_{\vec{j} \in \{-1,0,1\}^2} L_{\ell} \cdot \vec{j} + \ibox{L_{\ell}}^2$ the merged domains represented in $u$.
  \item We then check the geometry of the input pattern $u$: for every rectangle $\rect{r} \in R$,
    \begin{itemize}
    \item If $\rect{r}$ is a square of size $n \times n$ strictly contained in the interior of the cube $\ibox{L_{\ell}}^2$, check that $n$ belongs to the \field{size list} $\sizelist$;
    \item If $\rect{r}$ intersects a facet $\facet{k}{\pm}(\ibox{L_{\ell}}^2)$ but no corner, check that there exists a matching rectangle $\rect{r}'$ in the \field{rectangle list} of the tile $u_{\vec{0} \pm \basis{k}}$ such that $\rect{r} \cup (\pm L_{\ell} \cdot \basis{k} + \rect{r}')$ forms a valid square of size $n \times n$ for some $n \in \sizelist$;
    \item If $\rect{r}$ covers a corner of~$\ibox{L_{\ell}}^2$, consider the \field{rectangle list} of the adjacent tiles~$u_{\vec{j}}$ ($\vec{j} \in \{-1,0,1\}^2$) to check if it can be merged with (at most) three other rectangles to form a larger $\rect{r}' \subseteq D$. If $\rect{r}$' is not a rectangle, reject. If $\rect{r}'$ intersects a facet of $D$, check that $\rect{r}'$ could be extended into a square; otherwise, reject. If $\rect{r}'$ entirely fits in the interior of $D$, check that $\rect{r}'$ is a square of size $n \times n$ for some $n \in \N$; and if $n \leq L_{\ell}-1$, check that $n \in \sizelist$. Otherwise, reject.
    \end{itemize}
    Furthermore, check that all sizes of $\sizelist \cap \{L_{\ell-1},\dots,L_{\ell}-1\}$ are accounted for in this way; otherwise, reject.
  \item We then check $\sizelist$ against an enumeration of $S$: for every $n \in \{0,\dots,\ell\}$,
    \begin{itemize}
    \item Check if $\funram{\code{S}}(n)$ halts in less than $\ell$ steps. If it does, then $n \notin S$, and we thus check that $n$ does not appear in the \field{size list} $\sizelist$ (otherwise, we reject);
    \end{itemize}
  \item Finally, we define the output pattern $v \in \alphabet_{\ell}^{\ibox{2}^d}$ of the tile $u_{\vec{0}}$ as follows:
    \begin{itemize}
    \item Non-deterministically define a new \field{rectangle list} $R' \subseteq \Srect$ of non-overlapping squares $\rect{r}'$ of size $n \times n$ for $L_{\ell-2} \leq n < L_{\ell-1}$ and such that:
      \begin{enumerate}[label=(\roman*)]
      \item Each square $\rect{r}' \in R'$ intersects $\ibox{L_{\ell}}^2$;
      \item The squares $\rect{r}' \in R'$ cover exactly the sizes from $\sizelist \cap \{L_{\ell-2},\dots,L_{\ell-1}-1\}$;
      \item The squares $\rect{r}' \in R'$ and the rectangles $(L_{\ell} \cdot \vec{j} + \rect{r})$ from the \field{rectangle lists} of the tiles $u_{\vec{j}}$ ($\vec{j} \in \{-1,0,1\}^2$) are disjoint;
      \end{enumerate}
    \item For every $\vec{i} \in \ibox{2}^d$, set the \field{rectangle list} of $v_{\vec{i}}$ to $R'_{\vec{i}} = \{ (\rect{r}' - L_{\ell-1} \cdot \vec{i}) \cap \ibox{L_{\ell-1}}^2 : \rect{r}' \in R' \cup R \}$; and the \field{size list} of $v_{\vec{i}}$ to the union of $\sizelist \cap \{0,\dots,L_{\ell-2}-1\}$ with the sizes of squares in $\{L_{\ell-2},\dots,L_{\ell-1}-1\}$ from $R'$ and $R$ that intersect $L_{\ell-1} \cdot \vec{i} + \ibox{L_{\ell-1}}^2$.
    \end{itemize}
  \end{itemize}

  Identifying symbols in $\alphabet_0$ of empty \field{rectangle list} with $\square$, and those of \field{rectangle list}~$\{\ibox{1}^d\}$ with~$\square[black]$, we claim that -- starting from the $\N$-square shift $\Xsquare{\N}$ -- the resulting limit space is the $S$-square shift~$\Xsquare{S}$. Yet, \Cref{sof:thm:square-d-adic-sofic} cannot be applied as is: indeed, the \field{size list} in the alphabet $\alphabet_{\ell}$ could be of length $L_{\ell}$, and thus fail to satisfy the required sublinear information bounds. As already noted in~\cite{Westrick_2017:seas-of-squares}, we can actually restrict the cardinality of a \field{size list} in $\alphabet_{\ell}$ to be at most $\smash{\bigO(L_{\ell}^{2/3})}$, since the maximal number $N$ of different square sizes that can fit an $n \times n$ window must satisfy $\smash{\sum_{k=1}^{N} k^2 \leq n^2}$, and thus be at most $\smash{\bigO(n^{2/3})}$.

  After this modification, all the ndill maps $\tau_{\ell}$ are computable in time $\smash{\bigO(L_{\ell}^{2/3} \cdot \mathrm{poly}(\ell))}$, and we conclude that $\Xsquare{S}$ is sofic by the non-parallel version of \Cref{sof:thm:square-d-adic-sofic}.
\end{proof}

%
\section{Programs and higher-dimensional computation models}\label{calc:sec}

This section provides the computational tools required to prove \Cref{sof:lem:fixpoint} in \Cref{fix:sec}, such as classical results from computability theory about program transformations; and presents the computation model of \emph{processor arrays}, which appears especially suited for computation embeddings into multidimensional tilings.

\subsection{Classical results from computability}
By considering RAM programs $e \in \{0,1\}^{\ast}$ as text/code, it is actually possible for RAM programs to read the code of other RAM programs, or even modify such codes. In what follows, we consider some classical computability results of this flavor and their influence on time and word length. To do so, for any RAM program $e \in \{0,1\}^{\ast}$ and input array $\ram{I} \in \{0,1\}^{\ast\ast}$, we denote by
\begin{align*}
  T_{e}(\ram{I}) &: \text{the maximal length of the halting and accepting runs of $e$ on $\ram{I}$}; \\[-2pt]
  W_{e}(\ram{I}) &: \text{the maximal word length of the halting and accepting runs of $e$ on $\ram{I}$.}
\end{align*}
Note that these quantities might be infinite, if, using the non-determinism of the model, there exist arbitrarily long accepting runs. We first claim that there exists a (word-RAM preserving) universal interpreter:
\begin{proposition}
  There exists a RAM program $e_{\mathcal{U}}$ such that, for every RAM program $e \in \{0,1\}^{\ast}$ and every input $\ram{I} \in \{0,1\}^{\ast\ast}$ of bit length $n$, we have
  \[ \funram{e_{\mathcal{U}}}(e \cdot \ram{I}) = \funram{e}(\ram{I}).\]
  Furthermore, time complexity and word length respectively satisfy $T_{e_{\mathcal{U}}}(e \cdot \ram{I}) = \bigO(T_{e}(\ram{I}))$ and $W_{e_{\mathcal{U}}}(e \cdot \ram{I}) = W_{e}(\ram{I}) + \bigO(\log |e|)$.
\end{proposition}

\noindent The construction of such an interpreter is a textbook result from computability that is proved by simulating the variables of $e$ in the memory of $e_{\mathcal{U}}$. More precisely, the new machine $e_{\mathcal{U}}$ will contain at least a variable to simulate the program counter $\ram{PC}$ of $e$ and three variables to perform arithmetic operations. Each instruction from $e$ is decoded in time $\bigO(1)$ from the input of $e_{\mathcal{U}}$ using the simulated $\ram{PC}$, and is then simulated in time $\bigO(1)$ (because of the additional $\ram{LOAD}$ and $\ram{STORE}$ instructions required to load the associated variables of $e$ from the memory of $e_{\mathcal{U}}$).

\pagebreak
\enlargethispage{\baselineskip}
Another classical result is the $S^m_n$ theorem~\cite{Kleene_1938:on-notation-for-ordinal-numbers}, which allows to computably ``hardcode'' the value of some input words $\ram{I}[0],\dots$ directly in the code of a program:
\begin{proposition}[name={$S^m$ theorem},label={calc:prop:smn-theorem}]
  For all $m \in \N$, there exists a computable total function ${S^m \colon \{0,1\}^{\ast} \times (\{0,1\}^{\ast})^m \to \{0,1\}^{\ast}}$ such that, for every program $e \in \{0,1\}^{\ast}$, every input ${\ram{I} \in \{0,1\}^{\ast\ast}}$ and every words $(w_1,\dots,w_m) \in (\{0,1\}^{\ast})^m$, we have
  \[ \funram{S^m(e,w_1,\dots,w_m)}(\ram{I}) = \funram{e}(w_1 \cdot \ldots \cdot w_m \cdot \ram{I}). \]
  Furthermore, time complexity and word length satisfy $T_{S^m(e,w_1\dots,w_m)}(\ram{I}) = \bigO(T_{e}(w_1 \cdot \ldots \cdot w_m \cdot \ram{I}))$ and $W_{S^m(e,w_1\dots,w_m)}(\ram{I}) = W_{e}(w_1 \cdot \ldots \cdot w_m \cdot \ram{I}) + \bigO(\log(|e| + \sum_{i=1}^m |w_i|))$.
\end{proposition}
\noindent Indeed, we can computably replace all $\ram{READ}$ instructions $\ram{var}_i \leftarrow \ram{I}[\ram{var}_j][\ram{var}_k]$ in the input program $e \in \{0,1\}^{\ast}$ by a $\ram{GOTO}$ to the end of the program $e$, where we add instructions to either read $\ram{I}[\ram{var}_j-m]$ or the hardcoded value of the corresponding $w_{\ram{var}_j}$. This can be implemented in time $\bigO(1)$ using arithmetic directly on the program counter. Since these modifications introduce new operations inside the program $e$, all the original $\ram{GOTO}$ instructions should be updated accordingly. The word length increases in order to index all these new instructions.

\medskip
We now consider Kleene's fixed-point theorems. Somewhat hidden in~\cite{Kleene_1938:on-notation-for-ordinal-numbers}, with modern phrasing from~\cite[Chapter~11]{Rogers_1967:theory-recursive-functions-effective-computability}, the following theorem proves that all total computable functions admit programs whose behavior they leave unchanged:
\begin{proposition}[name={Fixed point theorem},label={calc:prop:fixed-point-theorem}]
  For every computable total function $F \colon \{0,1\}^{\ast} \to \{0,1\}^{\ast}$, there exists a program $e \in \{0,1\}^{\ast}$ such that $\funram{e} = \funram{F(e)}$. Furthermore, time complexity and word length respectively satisfy\footnote{All the following additive $\bigO(1)$ terms actually depend on the computable function $F$} $T_{e}(\ram{I}) = \bigO(T_{F(e)}(\ram{I})) + \bigO(1)$ and $W_{e}(\ram{I}) = W_{F(e)}(\ram{I}) + \bigO(1)$.
\end{proposition}

\begin{proof}
  We follow the classical proof and analyze its time and word length. Let us fix $b \in \{0,1\}^{\ast}$ the following program: on a given input array $\ram{I} \in \{0,1\}^{\ast\ast}$:
  \begin{enumerate}
  \item Denoting $e' = \ram{I}[0]$ and $\ram{I}' = \ram{I}[1\mspace{-2mu}\texttt{:}]$, compute $\funram{e'}(e') \in \{0,1\}^{\ast\ast}$; if this computation halts, denote $\ram{O} \in \{0,1\}^{\ast\ast}$ the output result;
  \item If $\ram{O}$ is actually of length $\len(\ram{O}) = 1$ (\textit{i.e.}~${\ram{O} \in \{0,1\}^{\ast}}$), compute $F(\ram{O})$ (otherwise, reject);
  \item Compute and return $\funram{F(\ram{O})}(\ram{I}') \in \{0,1\}^{\ast\ast}$.
  \end{enumerate}
  Applying the $S^m$ theorem, there exists some total and computable $s \colon \{0,1\}^{\ast} \times \{0,1\}^{\ast} \to \{0,1\}^{\ast}$ satisfying $\funram{s(b,e')}(\ram{I}) = \funram{b}(e' \!\cdot \ram{I})$ for all $e' \in \{0,1\}^{\ast}$ and $\ram{I} \in \{0,1\}^{\ast\ast}$. Since $s$ is computable, there exists a program $a \in \{0,1\}^{\ast}$ such that $\funram{a}(e') = s(b,e')$; and since $s$ is total, so is $\funram{a} \colon \{0,1\}^{\ast} \to \{0,1\}^{\ast}$. Thus, the program $e = \funram{a}(a) \in \{0,1\}^{\ast}$ is well-defined; and for any input $\ram{I} \in \{0,1\}^{\ast\ast}$, we have:
  \[ \funram{e}(\ram{I}) = \funram{s(b,a)}(\ram{I}) = \funram{b}(a \cdot \ram{I}) = \funram{F(\funram{a}(a))}(\ram{I}) = \funram{F(e)}(\ram{I}),\]
  since $\funram{a}(a)$ halts and $F$ is total.

  \smallskip
  We now consider the time and word length of $e$. To do so, we fix some program $\code{F}$ for $F$:
  \begin{itemize}
  \item By the $S^m$ theorem, the time complexity $T_{e}(\ram{I}) = T_{s(b,a)}(\ram{I})$ is $\bigO(T_{b}(a \cdot \ram{I}))$. By definition, $T_{b}(a \cdot \ram{I}) = T_{e_{\mathcal{U}}}(a \cdot a) + T_{\code{F}}(e) + T_{e_{\mathcal{U}}}(F(e) \cdot \ram{I})$. Since $a \in \{0,1\}^{\ast}$ and $e \in \{0,1\}^{\ast}$ only depend on our choice of $\code{F}$, $T_{e_{\mathcal{U}}}(a \cdot a)$ and $T_{\code{F}}(e)$ are both constants for varying $\ram{I}$, and by our choice of universal machine $T_{e_{\mathcal{U}}}(F(e) \cdot \ram{I}) = \bigO(T_{F(e)}(\ram{I}))$. Thus:
    \[ T_{e}(\ram{I}) = \bigO(T_{F(e)}(\ram{I})) + \bigO(1). \]
  \item With the same reasoning, we obtain that $W_{e}(\ram{I}) = W_{s(b,a)}(\ram{I}) = W_{b}(a \cdot \ram{I}) + \bigO(\log |a|)$. Since $W_{b}(a \cdot \ram{I} = W_{e_{\mathcal{U}}}(a \cdot a) + W_{\code{F}}(e) + W_{e_{\mathcal{U}}}(F(e) \cdot \ram{I})$, and that $W_{e_{\mathcal{U}}}(a \cdot a)$, $W_{\code{F}}(e)$, $\log |a|$ and $\log |F(e)|$ are constants for varying $\ram{I}$, we deduce that:
    \[ W_{e}(\ram{I}) = W_{F(e)}(\ram{I}) + \bigO(1). \qedhere\]
  \end{itemize}
  \end{proof}

  In conjunction with the \crtlnameref{calc:prop:smn-theorem}, it is classically used to obtain programs that can access their own code (second fixed-point theorem): for every computable function ${f \pcolon \{0,1\}^{\ast\ast} \mto \{0,1\}^{\ast\ast}}$, there exists a program $e \in \{0,1\}^{\ast}$ such that, for every $\ram{I} \in \{0,1\}^{\ast\ast}$, we have:
  \[ \funram{e}(\ram{I}) = f(e \cdot \ram{I}). \]
  Nevertheless, the fixed-point tiling construction in this article will be based upon the \crtlnameref{calc:prop:fixed-point-theorem} (\Cref{calc:prop:fixed-point-theorem}).

\subsection{Grid-connected processor arrays}\label{calc:sec:processors}
When embedding universal computations within tilings of domain $\ibox{n}^2$, one usually draws the space-time diagram of a Turing machine: spatially, the tape occupies $n$ cells on each line; and temporally, $n$ lines allow for $n$ steps of computations. However, we have rarely seen any consideration being given to higher-dimensional embeddings. In this direction, we consider the model of \emph{processor arrays}: a highly parallel model of computation that appears especially suitable for tilings of $d$-dimensional rectangles $\ibox{n_1,\dots,n_d}$.

\subsubsection{Processor arrays}

Processor arrays are a distributed model of computation in which the nodes of a graph, sometimes called \emph{processors}, can perform some elementary operations and most importantly communicate with their neighbors to globally compute a function~\cite[Chapter~5]{Akl_1985:parallel-sorting-algorithms}.

\begin{figure}[ht]
  \scalebox{0.6}{
    \begin{tikzpicture}[scale=0.9]
      \foreach[count=\m from 0] \i in {0,2.8,5.6,8.4} {
        \foreach[count=\n from 0] \j in {0,1.8,3.6,5.4} {
          \draw[fill=white] (\i,\j) rectangle ++(2.2,1.2) node[midway] {$P_{(\m,\n)}$};
          \ifnum\m<3{
              \draw ($(\i,\j) + (2.2,0.6)$) -- ++(0.6,0);
            }\else{
              \draw ($(\i,\j) + (2.2,0.6)$) -- ++(0.3,0);
            }\fi
          \ifnum\n<3{
              \draw ($(\i,\j) + (1.1,1.2)$) -- ++(0,0.6);
            }\else{
              \draw ($(\i,\j) + (1.1,1.2)$) -- ++(0,0.3);
            }\fi
        }
      }
      \foreach[count=\m from 0] \i in {0,2.8,5.6,8.4} {
        \draw (\i,8.4) rectangle ++(2.2,1.2) node[midway] {$P_{(\m,n_2\text{-}1)}$};
        \ifnum\m<3{
            \draw ($(\i,8.4) + (2.2,0.6)$) -- ++(0.6,0);
          }\else{
            \draw ($(\i,8.4) + (2.2,0.6)$) -- ++(0.3,0);
          }\fi
        \draw ($(\i,8.4) + (1.1,0)$) -- ++(0,-0.3);
        \draw ($(\i,7.6) + (1.1,0)$) node[font=\LARGE] {$\vdots$};
      }
      \foreach[count=\n from 0] \j in {0,1.8,3.6,5.4} {
        \draw (12.6,\j) rectangle ++(2.2,1.2) node[midway] {$P_{(n_1\text{-}1,\n)}$};
        \draw ($(12.6,\j) + (0,0.6)$) -- ++(-0.3,0);
        \draw ($(11.6,\j) + (0,0.6)$) node[font=\LARGE] {$\ldots$};
        \ifnum\n<3{
            \draw ($(12.6,\j) + (1.1,1.2)$) -- ++(0,0.6);
          }\else{
            \draw ($(12.6,\j) + (1.1,1.2)$) -- ++(0,0.3);
          }\fi
      }
      \draw (12.6,8.4) rectangle ++(2.2,1.1) node[midway] {$P_{(n_1\text{-}1,n_2\text{-}1)}$};
      \draw (13.7,8.4) -- ++(0,-0.3);
      \draw (13.7,7.6) node[font=\LARGE] {$\vdots$};
      \draw (12.6,9) -- ++(-0.3,0);
      \draw (11.6,9) node[font=\LARGE] {$\ldots$};
    \end{tikzpicture}}

  \caption{A processor array on the rectangle $\ibox{n_1,n_2}$.}
\end{figure}

We call \emph{processor} a register machine $((\ram{var}_i)_{i \in V},\ram{PC})$ holding finitely many variables ${\ram{var}_i \in \{0,1\}^{\ast}}$ and a \emph{program counter} $\ram{PC}$. In addition to performing (non-deterministic) arithmetic operations on its variables and control flow operations on its program counter, we assume that processors can be \emph{connected} to their (at most) $2d$ neighbors and that an instruction $\ram{COMM} : \ram{var}_i \leftarrow P_{\ram{var}_j}(\ram{var}_k)$ allows a processor to read the variable of index $\ram{var}_k$ in its neighbor of index $\ram{var}_j$.

A \emph{processor array} is then a pair $(\rect{r},e)$ consisting of a rectangle $\rect{r} = \ibox{n_1,\dots,n_d}$ and a global program $e \in \{0,1\}^{\ast}$, which specifies the finite set of variables $(\ram{var}_i)_{i \in V}$ and a sequence of instructions. In order to perform computations, we place processors running the program $e$ at each node in the graph of vertices $\rect{r}$, and connect facet-adjacent processors with wires. \emph{Finally, the word length of the variables in a processor is assumed to be bounded by $\smash{\bigO(\log(\sum_{k=1}^d n_k))}$}.

\medskip
For a program $e \in \{0,1\}^{\ast}$ defining variables $V$, and for a variable $\ram{var} \in V$, we can define the projection $\pi_{\ram{var}} \colon (\{0,1\}^{\ast})^{V} \times \N \to \{0,1\}^{\ast}$ projecting the \emph{state} of a processor (\textit{i.e.}~the value of its variables and of its program counter $\ram{PC}$) to the value of its variable $\ram{var}$. This projection extends to the states of whole processor arrays $\pi_{\ram{var}} \colon ((\{0,1\}^{\ast})^{V} \times \N)^{\rect{r}} \to (\{0,1\}^{\ast})^{\rect{r}}$.

\medskip
In order to compute with a processor array $(\rect{r},e)$, we will assume that the set of variables $V$ defined by $e$ includes at least three variables $\ram{pos} = (\rect{r},\vec{i})$ (the position of the processor $\vec{i} \in \rect{r}$), $\ram{input} \in \{0,1\}^{\ast}$ and $\ram{output} \in \{0,1\}^{\ast}$. A computation step of the global array then consists in synchronously executing one instruction on each of its processors. By considering the mapping between the initial $\ram{input}$ and final $\ram{output}$ variables (\textit{i.e.}~resp.~$\pi_{\ram{input}}$ and $\pi_{\ram{output}}$) after all processors have halted in an accepting state, a processor array defines a non-deterministic (partial) function $(\{0,1\}^{\ast})^{\rect{r}} \mto (\{0,1\}^{\ast})^{\rect{r}}$.

\subsubsection{The sorting problem}

In a processor array of domain $\rect{r}$, assume that each processor is given an integer as $\ram{input}$. How fast can the resulting array get sorted in boustrophedon order?

\begin{proposition}[name={\cite{Thompson-Kung_1977:sorting-on-mesh-connected-parallel-computer}},label={calc:prop:pa-sorting}]
  There exists a deterministic program $e \in \{0,1\}^{\ast}$ such that, for any rectangle $\rect{r} = \ibox{n_1,\dots,n_d}$, the processor array $(\rect{r},e)$ solves the sorting problem in time $\bigO(\sum_{k=1}^d n_k)$.
\end{proposition}

\subsubsection{Simulation of $\log$-RAM programs} \Cref{calc:prop:pa-sorting} provides a deterministic sorting algorithm that, combined with non-determinism, allows to efficiently simulate $\log$-RAM programs.

\medskip
For an input array $\ram{I} \in \{0,1\}^{\ast\ast}$ of bit length $s$, we define $\mathrm{flat}_{\ram{I}} \colon \dom(\ram{I}) \to \interval{0}{s-1}$ the indexing function that increases along the lexicographic order, \textit{i.e.}~$\mathrm{flat}_{\ram{I}}(i,j) = \sum_{i' < i} \len(\ram{I}[i']) + j$. For $\rect{r} \in \Srect_0$, we define the folding $\mathrm{fold}_{\rect{r}} \colon \{0,1\}^{\ast\ast} \to (\N^2 \times \{0,1\})^{\rect{r}}$ which maps inputs $\ram{I} \in \{0,1\}^{\ast\ast}$ to the boustrophedon folding of their flattenings (\textit{i.e.}~the patterns $w \in (\N^2 \times \{0,1\})^{\rect{r}}$ such that, if $\vec{i} \in \rect{r}$ is the position of index $\mathrm{flat}_{\ram{I}}(i,j)$ in the boustrophedon ordering of $\rect{r}$, then $w_{\vec{i}} = (i,j,\ram{I}[i][j])$).

Recall that the \emph{volume} of a rectangle $\rect{r} = \ibox{n_1,\dots,n_d}$ is, by definition, $\rvol{\rect{r}} = \prod_{k=1}^d n_k$. In what follows, we prove that there exists a program $e_{\mathcal{U}}$ such that processor arrays $(\rect{r},e_{\mathcal{U}})$ can simulate $\smash{\prod_{k=1}^d n_k}$ steps of $\log$-RAM computations in time $\smash{\bigO(\sum_{k=1}^d n_k)}$:

\begin{lemma}[name={Accelerated RAM simulations},label={calc:lem:pa-ram-simulation}]
  There exists a program $e_{\mathcal{U}}$ such that, for any RAM program $e \in \{0,1\}^{\ast}$ of logarithmic word length $w(s) = \bigO(\log s)$, for every input array $\ram{I} \in \{0,1\}^{\ast\ast}$, for every rectangle $\rect{r} = \ibox{n_1,\dots,n_d}$ and for every $t \in \N$ such that $t \leq \rvol{\rect{r}} = \prod_{k=1}^d n_k$, if the processor array $(\rect{r},e_{\mathcal{U}})$ is initialized as follows:
  \begin{itemize}
  \item The $\ram{program}$ variable in each processor contains the program $e \in \{0,1\}^{\ast}$;
  \item The $\ram{word}$ variable in each processor contains an upper bound on $w(\rvol{\rect{r}})$;
  \item The $\ram{timeout}$ variable in each processor contains the time bound $t \in \N$;
  \item The $\ram{input}$ variables of the processors (\textit{i.e.}~the projection $\pi_{\ram{input}}$) form\footnote{\label{ftn:pa-fold-encoding}Actually, $\mathrm{fold}_{\rect{r}}$ returns a rectangular pattern over the alphabet $\N^2 \times \{0,1\}^{\ast}$ instead of $\{0,1\}^{\ast}$; thus, this equality holds only up to decoding binary strings of $\{0,1\}^{\ast}$ into tuples of $\N^2 \times \{0,1\}^{\ast}$.} the pattern $\mathrm{fold}_{\rect{r}}(\ram{I})$;
  \end{itemize}
  then the outputs $w \in (\{0,1\}^{\ast})^{\rect{r}}$ yielded by accepting computations are exactly\footref{ftn:pa-fold-encoding} the patterns $\mathrm{fold}_{\rect{r}}(\ram{O})$ for $\ram{O} \in \funram{e}(\ram{I})\halt[t]$, \textit{i.e.}~for $\ram{O}$ computed by $e$ in time $t$.

  \noindent Furthermore, the processor array $(\rect{r},e_{\mathcal{U}})$ halts in time $\bigO(\sum_{k=1}^d n_k)$; and the word length in each processor is bounded by $w(s) + \bigO(\log \rvol{\rect{r}}) = \bigO(\log \rvol{\rect{r}})$ during accepting computations.
\end{lemma}

\begin{proof}[Proof]
  For $e \in \{0,1\}^{\ast}$ a RAM program, assume that we are given a \emph{trace} of the program~$e$ as a list of instructions executed in chronological order (\textit{i.e.}~for each operation, we record the addresses and values that were read and written). By sorting the trace in lexicographic order (first by memory addresses, and then by time), the successive states of any given memory cell now appear consecutively: in other words, the validity of the global computation can be checked locally.

  Thus, if $e \in \{0,1\}^{\ast}$ is a $\log$-RAM program, the word length of the processor array is large enough for processors to each non-deterministically ``guess'' a computation step in $\bigO(1)$ steps. After checking in time $\bigO(1)$ that adjacent processors contain plausibly consecutive instructions (by considering the associated \emph{program counters} $\ram{PC}$), and sorting these records (by memory addresses and time) using \Cref{calc:prop:pa-sorting}, adjacent processors check the chronological consistency of the values stored in each simulated memory cell.
\end{proof}

\begin{details}
  Let us consider more formally the \emph{trace} of a RAM computation. A computation step of a RAM programs can be summarized as:
  \begin{enumerate}
  \item Reading a few variables/memory cells: for example, the instruction $\ram{ADD} : \ram{var}_1 \leftarrow \ram{var}_2 + \ram{var}_3$ reads the values of variables $\ram{var}_2$ and $\ram{var}_3$; the instruction $\ram{LOAD} : \ram{var}_1 \leftarrow \ram{M}[\ram{var}_2]$ reads the value of the variable $\ram{var}_2$ and the memory cell $\ram{M}[\ram{var}_2]$; etc\dots
  \item (Optional) Performing an operation on these read values (\textit{e.g.}~an addition, bit shift\dots);
  \item Writing a variable/memory cell: the instruction $\ram{ADD} : \ram{var}_1 \leftarrow \ram{var}_2 + \ram{var}_3$ updates the variable $\ram{var}_1$; the instruction $\ram{STORE} : \ram{M}[\ram{var}_1] \leftarrow \ram{var}_2$ changes the value of the memory cell $\ram{M}[\ram{var}_1]$;
  \item Updating the program counter $\ram{PC}$: either increasing it by $1$ for most instructions, or setting it to an arbitrary value by a $\ram{GOTO}$ instruction.
  \end{enumerate}

  \medskip
  We now design of a program $e_{\mathcal{U}} \in \{0,1\}^{\ast}$ for processor arrays. Among its variables, it defines a $\ram{program}$ variable that is assumed to contain a whole RAM program $e \in \{0,1\}^{\ast}$, a $\ram{word}$ variable holding an upper bound of the allowed word length, a $\ram{timeout}$ variable and a $\ram{simPC}$ variable indexing an instruction of~$e$. The program operates in two phases:

  \pagebreak
  \paragraph*{Part 1: Simulation} Reading its $\ram{pos} = (\rect{r},\vec{i})$ variable and checking with its neighbors that all processors share a common value $\ram{timeout} \in \N$, each processor computes in time $\bigO(1)$ its index $n \in \interval{0}{\rvol{r}-1}$ in the boustrophedon ordering of $\rect{r}$. If $n \leq \ram{timeout}$, it then non-deterministically ``guesses'' the $n$\textsuperscript{th} step of the RAM computation of $e$:
  \begin{enumerate}[topsep=3pt,itemsep=3pt]
  \item First, if $n=0$, the processor sets its $\ram{simPC}$ variable to $0$; otherwise, it uses non-determinism to fill a value $\ram{simPC}$ that indexes an instruction in $e$;
  \item Then, consider the instruction $\ram{INSTR}$ in $e$ indexed by $\ram{simPC}$:
    \begin{itemize}[topsep=2pt,itemsep=3pt]
    \item The instruction $\ram{INSTR}$ will result in at most three memory readings and one memory writing: using non-determinism, each processor fills three words $u, v, w \in \{0,1\}^{\ast}$ of word length $\ram{word}$ to act as ``guesses'' for these readings; and stores the associated records in (at most) four of its variables (say, $\ram{trace1}$, $\ram{trace2}$, $\ram{trace3}$ and $\ram{trace4}$).

      For example, an instruction $\ram{INSTR} = \ram{LOAD} : \ram{var}_i \leftarrow \ram{M}[\ram{var}_j]$ will generate two reading records and one writing record:
      \[
        \begin{array}{l}
          \ram{trace1} \leftarrow (\ram{type} = \ram{V}, \ram{READ}, \mathtt{time} = n, \mathtt{address} = \ram{var}_j, \mathtt{value} = u) \\
          \ram{trace2} \leftarrow (\ram{type} = \ram{M}, \ram{READ}, \mathtt{time} = n, \mathtt{address} = u, \mathtt{value} = v) \\
          \ram{trace3} \leftarrow (\ram{type} = \ram{V}, \ram{WRITE}, \mathtt{time} = n, \mathtt{address} = \ram{var}_i, \mathtt{value} = v);
        \end{array}
      \]
      an instruction $\ram{INSTR} = \ram{ADD} : \ram{var}_i \leftarrow \ram{var}_j + \ram{var}_k$ will also generate three records:
      \[
        \begin{array}{l}
          \ram{trace1} \leftarrow (\ram{type} = \ram{V}, \ram{READ}, \mathtt{time} = n, \mathtt{address} = \ram{var}_j, \mathtt{value} = u) \\
          \ram{trace2} \leftarrow (\ram{type} = \ram{V}, \ram{READ}, \mathtt{time} = n, \mathtt{address} = \ram{var}_k, \mathtt{value} = v) \\
          \ram{trace3} \leftarrow (\ram{type} = \ram{V}, \ram{WRITE}, \mathtt{time} = n, \mathtt{address} = \ram{var}_i, \mathtt{value} = x);
        \end{array}
      \]
      where $x = u+v$ is the result of the addition $u+v$, as computed by the processor directly; and an instruction $\ram{INSTR} = \ram{WRITE} : \ram{O}[\ram{var}_i][\ram{var}_j] \leftarrow \ram{var}_k$ will generate four records:
        \[
          \begin{array}{l}
            \ram{trace1} \leftarrow (\ram{V}, \ram{READ}, \mathtt{time} = n, \mathtt{address} = \ram{var}_i, \mathtt{value} = u) \\
            \ram{trace2} \leftarrow (\ram{V}, \ram{READ}, \mathtt{time} = n, \mathtt{address} = \ram{var}_j, \mathtt{value} = v) \\
            \ram{trace3} \leftarrow (\ram{V}, \ram{READ}, \mathtt{time} = n, \mathtt{address} = \ram{var}_k, \mathtt{value} = w) \\
            \ram{trace4} \leftarrow (\ram{O}, \ram{WRITE}, \mathtt{time} = n, \mathtt{address} = (u,v), \mathtt{value} = w);
          \end{array}
        \]
      \item To conclude the simulation of $\ram{INSTR}$, the processor stores in a variable $\ram{simPC'}$ the new value of the instruction pointer, as either $\ram{simPC}+1$ or as written in $\ram{res}$.
      \end{itemize}
    \item Finally, the processor program $e_{\mathcal{U}}$ checks the \emph{chronological continuity} of the simulated computation step: reading the $\ram{simPC}$ variable of its successor (in boustrophedon order) inside the processor array, each processor checks that it equals its own variable $\ram{simPC}'$.
  \end{enumerate}

  If we think about this simulation informally, the processors have simulated a plausible sequence of RAM instructions for the program $e$ of length smaller than $\ram{timeout}$; but memory readings/writings may have been guessed incorrectly. We solve this issue in the second ``validation'' phase:

  \paragraph*{Part 2: Validation}
  \begin{enumerate}[topsep=3pt,itemsep=3pt]
  \item Verifying the chronological consistency of variable/memory interactions:
    \begin{itemize}
    \item Using the sorting procedure from \Cref{calc:prop:pa-sorting}~\cite{Thompson-Kung_1977:sorting-on-mesh-connected-parallel-computer}, the processor program $e_{\mathcal{U}}$ sorts the $\ram{input}$ and $\ram{trace}$ variables mixed together in lexicographic order (first by $\mathtt{type}$, then by $\mathtt{address}$, then by $\mathtt{time}$, and finally $\ram{READ}$ before $\ram{WRITE}$) along the boustrophedon indexing of $\rect{r}$.
    \item In the resulting array, all records of the given same address are consecutive in the boustrophedon ordering of $\rect{r}$, and these form a segment of memory interactions sorted in chronological order. In parallel time $\bigO(1)$, processors can communicate with their neighbors (in boustrophedon order) to check that these records are chronologically consistent (successive $\ram{READ}$ records, or two consecutive $\ram{WRITE}$ and $\ram{READ}$ records, should contain the same $\mathtt{value}$, etc\dots).
    \end{itemize}
  Thus, at this point, the RAM computation simulated by the processor array is both chronologically consistent and correctly initialized.
\item Correctly displaying the output:
  \begin{itemize}
  \item Using the sorting procedure described in \Cref{calc:prop:pa-sorting}~\cite{Thompson-Kung_1977:sorting-on-mesh-connected-parallel-computer}, the processor program $e_{\mathcal{U}}$ sorts the $\ram{trace}$ variables recording the writing of an output bit by addresses along the boustrophedon indexing of the domain $\rect{r}$;
  \item Finally, the processors having a $\ram{store}$ record $(\ram{type} = \ram{O},\mathtt{address} = (i,j), \mathtt{value} = b)$ copy their content to their variable $\ram{output} = (i,j,b)$.
  \end{itemize}
\end{enumerate}

If any of these local checks has failed, the processor array rejects the computation. Otherwise, the computation halts in an accepting state.
\end{details}

%
\section{Rectangles and wirings in Wang tilings}\label{pos:sec}

The fixed-point construction (\Cref{fix:sec}) will be based on computational embeddings inside the rectangles of $d$-dimensional grids $\grid{g} = (\rect{r}_{\vec{i}})_{\vec{i} \in \Z^d}$. This section will thus focus on drawing finite rectangles, rectangular grids, and transmitting information inside such rectangles.

\subsection{Drawing rectangles with Wang tiles}\label{pos:sec:rect}

In this section, we define an infinite tileset $\Swang[\Srect]$ to draw finite tilings of rectangles $\rect{r} \in \Srect_0$. To do so, consider $\alphabet[C]_{\Srect} \subseteq (\N^d \times \N^d) \cup \{\blank\}$ the set of \emph{positioning colors} defined as follows
\[ \alphabet[C]_{\Srect} = \{ ((n_1,\dots,n_d), \vec{i}) : \vec{i} \in \ibox{n_1,\dots,n_d} \} \cup \{\blank\},\]
where $\blank$ is a blank symbol. Identifying the tuple $(n_1,\dots,n_d)$ with the rectangle $\rect{r} = \ibox{n_1,\dots,n_d}$, we will always refer to a positioning color as a pair $(\rect{r},\vec{i})$ for $\vec{i} \in \rect{r}$.

\medskip
We now define $\Swang[\Srect]$ as the set of tiles $\wang{t} \in \alphabet[C]_{\Srect}^{2d+1}$ such that there exists $\rect{r} \in \Srect_0$, $\vec{i} \in \rect{r}$ and
\begin{itemize}
\item $\dec(\wang{t}) = (\rect{r},\vec{i})$;
\item If $\vec{i} \notin \facet{k}{-}(\rect{r})$, then $\facet{k}{-}(\wang{t}) = (\rect{r},\vec{i})$; otherwise, $\facet{k}{-}(\wang{t}) = \blank$ is left blank;
\item If $\vec{i} \notin \facet{k}{+}(\rect{r})$, then $\facet{k}{+}(\wang{t}) = (\rect{r}, \vec{i} + \basis{k})$; otherwise, $\facet{k}{+}(\wang{t}) = \blank$ is left blank;
\end{itemize}

\noindent We claim that the decorations of valid patterns on $\Swang[\Srect]$ exactly draw the rectangles of $\Srect_0$:
\begin{proposition}
  For any rectangle $\rect{r} \in \Srect_0$, there exists a valid pattern $w \in (\Swang[\Srect])^{\rect{r}}$ such that, for every position $\vec{i} \in \rect{r}$, we have $\dec(w_{\vec{i}}) = (\rect{r},\vec{i})$. Conversely: for any rectangle $\rect{r} \in \Srect_0$, for any valid pattern $w \in (\Swang[\Srect])^{\rect{r}}$ such that $\dec(w_{\vec{0}}) = (\rect{r},\vec{0})$ and for any position $\vec{i} \in \rect{r}$, we have $\dec(w_{\vec{i}}) = (\rect{r},\vec{i})$.
\end{proposition}

In fact, after adding a blank tile $\blank^{2d+1}$ to this tileset, the valid configurations over $\Swang[\Srect] \cup \{ \blank^{2d+1} \}$ are exactly the drawings of independent rectangles on $\Z^d$ over a blank background.

\subsection{Drawing grids with Wang tiles}\label{pos:sec:grids}

If one aims at drawing rectangular grids instead, we define a similar set of tiles $\Swang[\grid{G}]$ as the set of tiles $\wang{t} \in \alphabet[C]_{\Srect}^{2d+1}$ such that: there exists $\rect{r} \in \Srect_0$ and $\vec{i} \in \rect{r}$ for which
\begin{itemize}
\item $\dec(\wang{t}) = (\rect{r},\vec{i})$;
\item $\facet{k}{-}(\wang{t}) = (\rect{r},\vec{i})$;
\item If $\vec{i} \notin \facet{k}{+}(\rect{r})$, then $\facet{k}{+}(\wang{t}) = (\rect{r},\vec{i} + \basis{k})$; otherwise, $\facet{k}{+}(\wang{t}) = (\rect{r}', \vec{i} + \basis{k} \bmod \rect{r})$ for some rectangle $\rect{r}' \in \Srect_0$ such that $\rnorm{\facet{k}{+}(\rect{r})}$ and $\rnorm{\facet{k}{-}(\rect{r}')}$ are equal.
\end{itemize}

\noindent We claim that valid configurations on $\Swang[\grid{G}]$ draw non-regular rectangular grids:
\begin{proposition}
  For any grid $\grid{g} = (\rect{r}_{\vec{p}})_{\vec{p} \in \Z^d}$, there exists a unique valid tiling $x \in (\Swang[\grid{G}])^{\Z^d}\mspace{-8mu}$ such that, for any position $\vec{i} \in \Z^d$ appearing in the rectangle $\rect{r}_{\vec{p}}$ from the grid $\grid{g}$, we have $\dec(x_{\vec{i}}) = (\rnorm{\rect{r}_{\vec{p}}}, \vec{i} \bmod \grid{g})$.

  \noindent Conversely: for any valid tiling $x \in (\Swang[\grid{G}])^{\Z^d}\mspace{-8mu}$, there exists a unique grid $\grid{g} = (\rect{r}_{\vec{p}})_{\vec{p} \in \Z^d}$ such that, for any position $\vec{i} \in \Z^d$ appearing in the rectangle $\rect{r}_{\vec{p}}$ from the grid $\grid{g}$, we have $\dec(x_{\vec{i}}) = (\rnorm{\rect{r}_{\vec{p}}}, \vec{i} \bmod \grid{g})$.
\end{proposition}

\subsection{Drawing wires with Wang tiles}\label{pos:sec:wiring}

Finally, we will often need to route information inside rectangles using continuous wires. For any set of colors $\alphabet[C]$, we define the \emph{wiring colors} $\alphabet[C]_{\wire} = (\alphabet[C] \times \N)^{\ast} \cup (\alphabet[C] \times \{\mathtt{start},\mathtt{end},\mathtt{cross}\})^{\ast}$; and for a fixed constant $\K[][cross] \in \N$, we define the \emph{wiring tileset} $\Swang[\wire](\alphabet[C])$ as the set of tiles $\wang{t} \in (\alphabet[C]_{\wire})^{2d+1}$ such that there exists an index $I$ such that $|I| \leq \K[][cross]$ and colors $(c_i)_{i \in I} \in \alphabet[C]^{I}$ such that:
\begin{enumerate}
\item $\dec(\wang{t})$ is a list of pairs $((c_i,x_i))_{i \in I}$ (for $x_i \in \{\mathtt{start},\mathtt{end},\mathtt{cross}\}$) encoding the list of wires traversing the tile $\wang{t}$; intuitively, $x_i$ represents whether the wire starts in $\wang{t}$, ends in $\wang{t}$, or crosses $\wang{t}$;
\item The color along each facet $\facet{k}{\pm}(\wang{t})$ is a list of pairs in $\alphabet[C] \times \N$, representing respectively the color and the length of the wires;
\item There is a bijection between the wires $(c_i,\mathtt{cross})$ in $\dec(\wang{t})$ (\textit{i.e.}~the wires crossing $\wang{t}$) and the pairs $(c_i,n_i)$ and $(c_i,n_i+1)$ (for some $x_i \in \N$) appearing across the facets $\facet{k}{\pm}(\wang{t})$;
\item There is a bijection between the $(c_i,\mathtt{start})$ in $\dec(\wang{t})$ (resp.~$(c_i,\mathtt{end})$ in $\dec(\wang{t})$) and the wiring tuples $(c_i,0)$ (resp.~wiring colors $(c_i,n_i)$ without a matching $(c_i,n_i+1)$ for $n_i > 0$) appearing across the facets $\facet{k}{\pm}(\wang{t})$.
\end{enumerate}

\begin{proposition}[label={pos:prop:wiring-bijection}]
  For any valid pattern $w \in \Swang[\wire](\alphabet[C])^{\rect{r}}$ with blank borders\footnote{\textit{i.e.}~for every position $\vec{i} \in \rect{r}$ such that $\vec{i} + \basis{k} \notin \rect{r}$ (resp.~$\vec{i} - \basis{k} \notin \rect{r}$), $\facet{k}{+}(w_{\vec{i}})$ (resp.~$\facet{k}{-}(w_{\vec{i}})$) is the empty list.}, there exists a bijection between the tuples $(c,\mathtt{start})$ and $(c,\mathtt{end})$ (for $c \in \alphabet[C]$) appearing across the decorations $\dec(w_{\vec{i}})$ for $\vec{i}$ ranging over $\rect{r}$.

  \noindent Furthermore, for each $(c,\mathtt{start})$ appearing in some decoration $\dec(w_{\vec{s}})$ for $\vec{s} \in \rect{r}$, there exists a path of contiguous tiles starting at position $\vec{s}$ and ending at some position $\vec{e}$ such that $\dec(w_{\vec{e}})$ contains the wiring tuple $(c,\mathtt{end})$; and if $\vec{i}$ is the $k$\textsuperscript{th} (intermediate) position in this path, then $(c,\mathtt{cross})$ appears in $\dec(w_{\vec{i}})$, and there exists a pair of facets $(\facet{}{},\facet{}{}')$ such that $(c,k) \in \facet{}{}(w_{\vec{i}})$ and $(c,k+1) \in \facet{}{}'(w_{\vec{i}})$.
\end{proposition}

We call \emph{wire} any such path of contiguous tiles in $w$, and \emph{wiring} the bijection between the start and the end of each wire drawn by $w$. We made the wiring tiles $\wang{t}$ encode their position in the wires crossing them because of convenience (it helps to identify them when looking at patterns) and practical purposes: it eliminates cycling wires.

\bigskip
We now consider the condition upon which a wiring between a given matching of starting and ending positions can exist inside a rectangle $\rect{r}$. Inside a facet $\facet{k}{\pm}(\rect{r})$ of a rectangle $\rect{r} \in \Srect_0$, we say that a given set of positions $P \subseteq \facet{k}{\pm}(\rect{r})$ is \emph{greedy} if, for every $\vec{i} \in \facet{k}{\pm}(\rect{r})$ such that $\vec{i} \in P$, all positions $\vec{j} \in \facet{k}{\pm}(\rect{r})$ of boustrophedon index $k' < k$ (for $k$ the boustrophedon index of $\vec{i}$ in $\facet{k}{\pm}(\rect{r})$) also satisfy $\vec{j} \in P$. In other words, $P$ fills the facet $\facet{k}{\pm}(\rect{r})$ in boustrophedon order greedily.

For any $\rect{r} \in \Srect_0$, denote by $\rborder{\rect{r}} = \biguplus_{k=1}^{d} \biguplus_{s \in \{+,-\}} \facet{k}{s}(\rect{r})$ the \emph{border}\footnote{In this case, it is useful to consider this border as a disjoint union of facets: so that, if a position $\vec{i}$ appears at the intersection of two facets $\facet{k}{\pm}(\rect{r})$ and $\facet{k'}{\pm}(\rect{r})$, then it appears twice in $\rborder{\rect{r}}$.} of $\rect{r}$.

\begin{lemma}[name={Routing lemma},label={pos:lem:routing-lemma}]
  Let $\alphabet[C]$ be a set of colors, $\rect{r} \in \Srect_0$ a rectangle, $S \subseteq \rborder{\rect{r}}$ and $R \subseteq \rborder{\rect{r}}$ of size $|S| = |R| \leq \rspace(\rect{r})$, and $(s,r) \in \alphabet[C]^{S} \times \alphabet[C]^{R}$ be \emph{source} and \emph{receiver} patterns satisfying the following conditions:
  \begin{itemize}
  \item For every color $c \in \alphabet[C]$, we have $|s|_c = |r|_c$ (as many $c$ symbols in $s$ and $r$);
  \item For every facet $\facet{k}{\pm}(\rect{r})$, if $\facet{k}{\pm}(\rect{r})$ is not a smallest facet of $\rect{r}$, then both $S \cap \facet{k}{\pm}(\rect{r})$ and $R \cap \facet{k}{\pm}(\rect{r})$ are greedy in~$\facet{k}{\pm}(\rect{r})$.
  \end{itemize}
  Then there exists a valid pattern $w \in (\Swang[\wire])^{\rect{r}}$ with blank borders such that:
  \begin{itemize}
  \item For every $\vec{i} \in \rect{r}$, a tuple $(c,\mathtt{start})$ appears in $\dec(w_{\vec{i}})$ if and only if $\vec{i} \in S$ and $s_{\vec{i}} = c$;
  \item For every $\vec{i} \in \rect{r}$, a tuple $(c,\mathtt{end})$ appears in $\dec(w_{\vec{i}})$ if and only if $\vec{i} \in R$ and $r_{\vec{i}} = c$;
  \item If $\rect{r} = \ibox{n_1,\dots,n_d}$, then wires in $w$ have length $\bigO(\sum_{k=1}^d n_k)$.
  \end{itemize}
\end{lemma}

\smallskip
In this paper, we now fix the value of the constant $\K[][cross]$ (bounding the number of wires that can appear in any given tile of $\Swang[\wire]$) in order to allow the tiles of $\Swang[\wire]$ to implement the geometric construction given in the following proof.

\begin{proof}
  Fix $\rect{r} = \ibox{n_1,\dots,n_d}$. Without loss of generality, we consider a source and a receiver pattern $(s,r)$ of domains $(S,R) \subseteq \facet{S}{} \times \facet{R}{}$ for $\facet{R}{} = \facet{h}{\pm}(\rect{r})$ a smallest facet and $\facet{S}{} = \facet{k}{\pm}(\rect{r})$ an arbitrary facet of $\rect{r}$: indeed, starts and ends of wires are symmetric, and the \crtlnameref{pos:lem:routing-lemma} can be obtained by managing each pair of facets independently.

  Since $\facet{R}{}$ is a smallest facet of $\rect{r}$, the proof of \Cref{calc:prop:pa-sorting}~\cite{Thompson-Kung_1977:sorting-on-mesh-connected-parallel-computer} shows that any permutation of $\facet{R}{}$ can be realized with wires of length $\bigO(\rtime(\rect{r}))$. Thus, instead of wiring the source pattern $s$ with the exact receiver pattern $r$, it is enough to realize a wiring of $s$ with some permutation $r'$ of $r$ in the same facet $\facet{R}{}$. Since the case $\facet{S}{} = \facet{R}{}$ is already solved, we assume that $\facet{R}{} \neq \facet{S}{}$.

  \smallskip
  First, notice that the \crtlnameref{pos:lem:routing-lemma} is immediate in dimension $d=2$. In the rest of this proof, we focus on the dimension $d=3$, since the proof will generalize from dimension $d=3$ to arbitrary $d > 3$ by keeping all other coordinates intact. By symmetry, we further assume that $\facet{R}{} = \facet{1}{-}(\rect{r})$ and that $\facet{S}{} = \facet{3}{-}(\rect{r})$, so that we can denote $\facet{S}{} = \ibox{n_1,n_2,1}$, $\facet{R}{} = \ibox{1,n_2,n_3}$ and $n_1 = \max \{ n_k \}$.

  For $1 \leq k \leq d$ and $n \geq 0$, let $I^{k}_{n} = \left\{ \vec{i} = (i_1,i_2,i_3) \in \facet{S}{} : \vec{i}_k = n \right\}$ be the $n$\textsuperscript{th} slice in $\facet{S}{}$ along the direction $\basis{m}$. Denoting $1 \leq m \leq d$ the most significant direction in the boustrophedon order of $\facet{S}{}$, then the positions in $I^{m}_{n}$ are of lower boustrophedon index than those of $I^{m}_{n+1}$. Since $S$ is greedy in $\facet{S}{}$, the points of $S$ fill the segments $I^{m}_{n}$ increasingly with $n$: introducing the integer $k_m = n_m \cdot \frac{n_3}{n_1}$, any segment $I^m_n$ of index $n > k_m$ will not intersect $S$ (as such segments $I^m_n$ contain elements of boustrophedon index larger than the cardinality of $|S|$).

  \medskip
  If $m=1$, then $k_m = n_3$ and simple diagonal lines sending each position $x = (x_1,x_2,0) \in S$ towards the position $(0,x_2,x_1) \in \facet{R}{}$ draw a wiring between $S$ and a subset $R' \subseteq \facet{R}{}$:
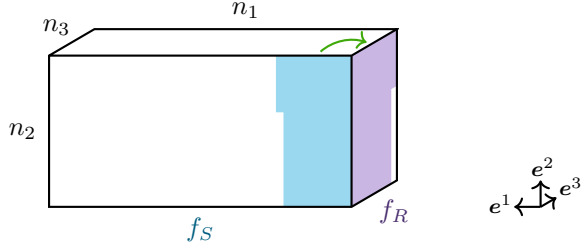
  \begin{center}
    \begin{tikzpicture}[scale=0.5,x={(-1,0)},y={(0,1)},z={(-0.6,0.35)}]
      \fill[RainbowE!40] (1.8,0,0) -- (1.8,2.5,0) -- (2,2.5,0) -- (2,4,0) -- (0,4,0) -- (0,0,0) -- cycle;
      \node[RainbowE!60!black,shift={(0,-0.3)}] at (4,0,0) {$\facet{S}{}$};
      \fill[RainbowF!40] (0,0,1.75) -- (0,2.5,1.75) -- (0,2.5,2) -- (0,4,2) -- (0,4,0) -- (0,0,0) --cycle;
      \node[RainbowF!60!black,shift={(-0.5,-0.1)}] at (0,0,0.2) {$\facet{R}{}$};
      \draw[->,RainbowD,thick] (1,4,0.3) to[bend left,looseness=0.9] (0.2,4,1);
      \begin{scope}
        \clip (8.024,-0.03) rectangle (-1.23,4.73);
        \draw[thick] (0,4,0) -- ++(8,0,0) -- ++(0,0,2) -- ++(-8,0,0) -- cycle;
        \draw[thick] (0,0,0) -- ++(0,4,0);
        \draw[thick] (8,4,0) -- ++(0,-4,0) -- ++(-8,0,0) -- ++(0,0,2) -- ++(0,4,0);
      \end{scope}
      \node[shift={(0,0.25)}] at (4,4,2) {$n_1$};
      \node[shift={(0.35,0)}] at (8,2,0) {$n_2$};
      \node[shift={(0.2,0.2)}] at (8,4,1) {$n_3$};
      \begin{scope}[shift={(-5,0)}]
        \draw[thick, ->] (0,0,0) to (0.7,0,0) node[shift={(0.18,0.04)},font=\footnotesize] {$\basis{1}$};
        \draw[thick, ->] (0,0,0) to (0,0.7,0) node[shift={(-0.03,0.18)},font=\footnotesize] {$\basis{2}$};
        \draw[thick, ->] (0,0,0) to (0,0,0.7) node[shift={(-0.19,0.15)},font=\footnotesize] {$\basis{3}$};
      \end{scope}
    \end{tikzpicture}
  \end{center}

  Otherwise, if $m=2$, then $k_m = \frac{n_2 n_3}{n_1}$. We wire $S$ and a subset $R' \subseteq \facet{R}{}$ in three steps:
  \begin{itemize}
  \item We split the set $S$ into blocks of length $n_3$ and turn them around using straight lines: each position $x = (kn_3 + i, i_2, 0)$ is sent to $(kn_3,i_2,i)$ ($i < n_3$ and $k \leq \frac{n_1}{n_3}$);

    
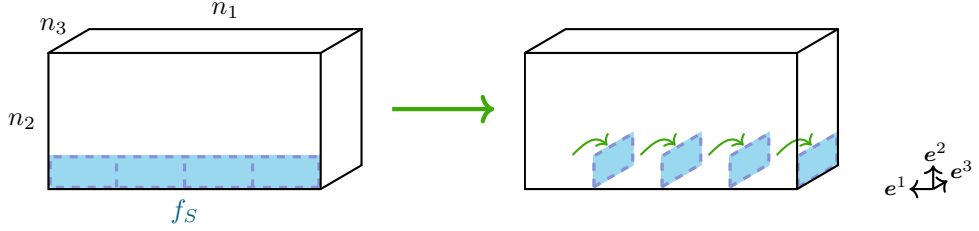
\begin{center}
      \begin{tikzpicture}[scale=0.45,x={(-1,0)},y={(0,1)},z={(-0.6,0.35)}]
        \begin{scope}[shift={(14,0)}]
          \fill[RainbowE!40] (0,0,0) -- (8,0,0) -- (8,1,0) -- (0,1,0) -- cycle;
          \begin{scope}
            \foreach \i in {2,4,6} {
              \draw[RainbowE!30!RainbowF!70,very thick, dashed] (\i,0.05,0) -- ++(0,0.9,0);
            }
            \draw[RainbowE!30!RainbowF!70,very thick, dashed] (0.05,0.05,0) -- ++(7.9,0,0) -- ++(0,0.9,0) -- ++(-7.9,0,0) -- cycle;
          \end{scope}
          \node[RainbowE!60!black,shift={(0,-0.3)}] at (4,0,0) {$\facet{S}{}$};
          \begin{scope}
            \clip (8.024,-0.03) rectangle (-1.23,4.73);
            \draw[thick] (0,4,0) -- ++(8,0,0) -- ++(0,0,2) -- ++(-8,0,0) -- cycle;
            \draw[thick] (0,0,0) -- ++(0,4,0);
            \draw[thick] (8,4,0) -- ++(0,-4,0) -- ++(-8,0,0) -- ++(0,0,2) -- ++(0,4,0);
          \end{scope}
          \node[shift={(0,0.25)}] at (4,4,2) {$n_1$};
          \node[shift={(0.35,0)}] at (8,2,0) {$n_2$};
          \node[shift={(0.2,0.2)}] at (8,4,1) {$n_3$};
        \end{scope}
        \draw[ultra thick,RainbowD,->] (12.5,2,1) -- (9.5,2,1);
        \begin{scope}
          \foreach \i in {6,4,2,0} {
            \draw[RainbowD,thick,->] ($(\i,1,0) + (0.6,0,0)$) to[bend left,in=120,looseness=1.2] (\i.05,1,0.8);
            \fill[RainbowE!40] (\i,0,0) -- ++(0,1,0) -- ++(0,0,2) -- ++(0,-1,0) -- cycle;
            \draw[RainbowE!30!RainbowF!70,very thick, dashed] (\i,0.05,0.05) -- ++(0,0.9,0) -- ++(0,0,1.9) -- ++(0,-0.9,0) --cycle;
          }
          \begin{scope}
            \clip (8.024,-0.03) rectangle (-1.23,4.73);
            \draw[thick] (0,4,0) -- ++(8,0,0) -- ++(0,0,2) -- ++(-8,0,0) -- cycle;
            \draw[thick] (0,0,0) -- ++(0,4,0);
            \draw[thick] (8,4,0) -- ++(0,-4,0) -- ++(-8,0,0) -- ++(0,0,2) -- ++(0,4,0);
          \end{scope}
        \end{scope}
        \begin{scope}[shift={(-4,0)}]
          \draw[thick, ->] (0,0,0) to (0.7,0,0) node[shift={(0.18,0.04)},font=\footnotesize] {$\basis{1}$};
          \draw[thick, ->] (0,0,0) to (0,0.7,0) node[shift={(-0.03,0.18)},font=\footnotesize] {$\basis{2}$};
          \draw[thick, ->] (0,0,0) to (0,0,0.7) node[shift={(-0.19,0.15)},font=\footnotesize] {$\basis{3}$};
        \end{scope}
      \end{tikzpicture}
    \end{center}

  \item We then align these blocks by translating them using ``zigzags'' of amplitude $k_m$ that send a position $x' = (kn_3,i_2,i)$ to $x'' = (kn_3,i_2 + k \frac{n_2 n_3}{n_1},i))$.\footnote{For rounding reasons, we might allow a small overlap between any two adjacent blocks.} As $n_2 \leq n_1$, we have that $k_m = \frac{n_2 n_3}{n_1} \leq n_3$, so that such zigzags can be drawn with at most $\bigO(1)$ superimposed wires at any position $\vec{i} \in \rect{r}$;

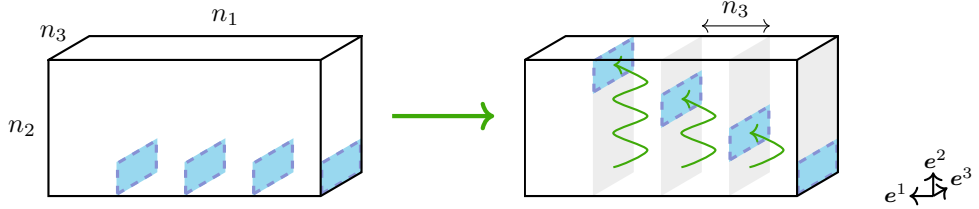
    \begin{center}
      \begin{tikzpicture}[scale=0.45,x={(-1,0)},y={(0,1)},z={(-0.6,0.35)}]
        \begin{scope}[shift={(14,0)}]
          \foreach \i in {0,2,4,6} {
            \fill[RainbowE!40] (\i,0,0) -- ++(0,1,0) -- ++(0,0,2) -- ++(0,-1,0) -- cycle;
            \draw[RainbowE!30!RainbowF!70,very thick, dashed] (\i,0.05,0.05) -- ++(0,0.9,0) -- ++(0,0,1.9) -- ++(0,-0.9,0) --cycle;
          }
          \begin{scope}
            \clip (8.024,-0.03) rectangle (-1.23,4.73);
            \draw[thick] (0,4,0) -- ++(8,0,0) -- ++(0,0,2) -- ++(-8,0,0) -- cycle;
            \draw[thick] (0,0,0) -- ++(0,4,0);
            \draw[thick] (8,4,0) -- ++(0,-4,0) -- ++(-8,0,0) -- ++(0,0,2) -- ++(0,4,0);
          \end{scope}
          \node[shift={(0,0.25)}] at (4,4,2) {$n_1$};
          \node[shift={(0.35,0)}] at (8,2,0) {$n_2$};
          \node[shift={(0.2,0.2)}] at (8,4,1) {$n_3$};
        \end{scope}
        \draw[ultra thick,RainbowD,->] (12.5,2,1) -- (9.5,2,1);
        \begin{scope}
          \foreach \i in {6,4,2,0} {
            \fill[opacity=0.7,gray!20] (\i,0,0) -- ++(0,4,0) -- ++(0,0,2) -- ++(0,-4,0) --cycle;
            \tikzmath{int \j, \jo; \j=\i/2; \jo=max(\j-2,0);}
            \fill[RainbowE!40] (\i,\j,0) -- ++(0,1,0) -- ++(0,0,2) -- ++(0,-1,0) -- cycle;
            \draw[RainbowE!30!RainbowF!70,very thick, dashed] (\i,\j.05,0.05) -- ++(0,0.9,0) -- ++(0,0,1.9) -- ++(0,-0.9,0) --cycle;
            \ifnum\i>2\relax{
                \draw[RainbowD,thick,->] (\i,0.5,1) to [out=0,in=-90,looseness=0.3] ++(-1,0.5,0) foreach \x in {0,...,\jo} {to[out=90,in=-90,looseness=0.3] ++(1,0.5,0) to[out=90,in=-90,looseness=0.3] ++(-1,0.5,0)} to[out=90,in=0,looseness=0.2] ++(1,0.5,0);
              }\fi
            \ifnum\i=2\relax{
                \draw[RainbowD,thick,->] (\i,0.5,1) to [out=0,in=-90,looseness=0.3] ++(-1,0.5,0)  to[out=90,in=0,looseness=0.3] ++(1,0.5,0);
              }\fi
          }
          \begin{scope}
            \clip (8.024,-0.03) rectangle (-1.23,4.73);
            \draw[thick] (0,4,0) -- ++(8,0,0) -- ++(0,0,2) -- ++(-8,0,0) -- cycle;
            \draw[thick] (0,0,0) -- ++(0,4,0);
            \draw[thick] (8,4,0) -- ++(0,-4,0) -- ++(-8,0,0) -- ++(0,0,2) -- ++(0,4,0);
          \end{scope}
          \draw[<->] (4,4.3,2) -- ++(-2,0,0) node[midway,above] {$n_3$};
        \end{scope}
        \begin{scope}[shift={(-4,0)}]
          \draw[thick, ->] (0,0,0) to (0.7,0,0) node[shift={(0.18,0.04)},font=\footnotesize] {$\basis{1}$};
          \draw[thick, ->] (0,0,0) to (0,0.7,0) node[shift={(-0.03,0.18)},font=\footnotesize] {$\basis{2}$};
          \draw[thick, ->] (0,0,0) to (0,0,0.7) node[shift={(-0.19,0.15)},font=\footnotesize] {$\basis{3}$};
        \end{scope}
      \end{tikzpicture}
    \end{center}

  \item Finally, we send each of these positions in a straight line towards $\facet{R}{}$: the position $x'' = (kn_3,i_2 + k\frac{n_2n_3}{n_1}, i)$ is sent to $(0,i_2 + k\frac{n_2n_3}{n_1},i)$.

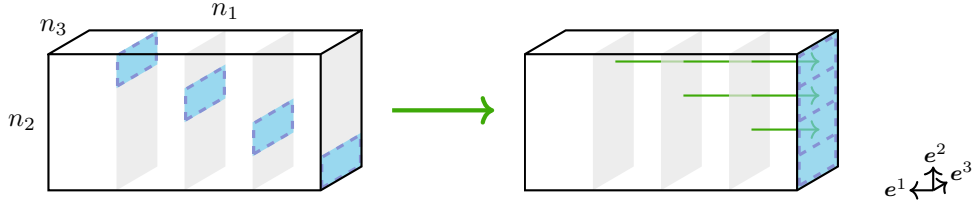
    \begin{center}
      \begin{tikzpicture}[baseline=0.5ex,scale=0.45,x={(-1,0)},y={(0,1)},z={(-0.6,0.35)}]
        \begin{scope}[shift={(14,0)}]
          \foreach \i in {6,4,2,0} {
            \fill[opacity=0.7,gray!20] (\i,0,0) -- ++(0,4,0) -- ++(0,0,2) -- ++(0,-4,0) --cycle;
            \tikzmath{int \j; \j=\i/2;}
            \fill[RainbowE!40] (\i,\j,0) -- ++(0,1,0) -- ++(0,0,2) -- ++(0,-1,0) -- cycle;
            \draw[RainbowE!30!RainbowF!70,very thick, dashed] (\i,\j.05,0.05) -- ++(0,0.9,0) -- ++(0,0,1.9) -- ++(0,-0.9,0) --cycle;
          }
          \begin{scope}
            \clip (8.024,-0.03) rectangle (-1.23,4.73);
            \draw[thick] (0,4,0) -- ++(8,0,0) -- ++(0,0,2) -- ++(-8,0,0) -- cycle;
            \draw[thick] (0,0,0) -- ++(0,4,0);
            \draw[thick] (8,4,0) -- ++(0,-4,0) -- ++(-8,0,0) -- ++(0,0,2) -- ++(0,4,0);
          \end{scope}
          \node[shift={(0,0.25)}] at (4,4,2) {$n_1$};
          \node[shift={(0.35,0)}] at (8,2,0) {$n_2$};
          \node[shift={(0.2,0.2)}] at (8,4,1) {$n_3$};
        \end{scope}
        \draw[ultra thick,RainbowD,->] (12.5,2,1) -- (9.5,2,1);
        \begin{scope}
          \foreach[count=\n from 0] \i in {6,4,2} {
            \tikzmath{\j=\i/2+0.4;}
            \fill[opacity=0.7,gray!20] (\i,0,0) -- ++(0,4,0) -- ++(0,0,2) -- ++(0,-4,0) --cycle;
            \foreach \k in {0,...,\n} {
              \ifnum\i>2\relax{
                  \draw[RainbowD,thick] ($(\i,\j,1.1) + (0,\k,0)$) -- ++(-2,0,0);
                }\else{
                  \draw[RainbowD,->,thick] ($(\i,\j,1.1) + (0,\k,0)$) -- ++(-2,0,0);
                }\fi
            }
          }
          \fill[RainbowE!55,opacity=0.7] (0,0,0) -- ++(0,4,0) -- ++(0,0,2) -- ++(0,-4,0) -- cycle;
          \draw[RainbowE!30!RainbowF!70,very thick,dashed] (0,0.1,0.1) -- ++(0,3.8,0) -- ++(0,0,1.8) -- ++(0,-3.8,0) --cycle;
          \foreach \i in {1,2,3} {
            \draw[RainbowE!30!RainbowF!70,very thick,dashed] (0,\i,0.1) -- ++(0,0,1.8);
          }
          \begin{scope}
            \clip (8.024,-0.03) rectangle (-1.23,4.73);
            \draw[thick] (0,4,0) -- ++(8,0,0) -- ++(0,0,2) -- ++(-8,0,0) -- cycle;
            \draw[thick] (0,0,0) -- ++(0,4,0);
            \draw[thick] (8,4,0) -- ++(0,-4,0) -- ++(-8,0,0) -- ++(0,0,2) -- ++(0,4,0);
          \end{scope}
        \end{scope}
        \begin{scope}[shift={(-4,0)}]
          \draw[thick, ->] (0,0,0) to (0.7,0,0) node[shift={(0.18,0.04)},font=\footnotesize] {$\basis{1}$};
          \draw[thick, ->] (0,0,0) to (0,0.7,0) node[shift={(-0.03,0.18)},font=\footnotesize] {$\basis{2}$};
          \draw[thick, ->] (0,0,0) to (0,0,0.7) node[shift={(-0.19,0.15)},font=\footnotesize] {$\basis{3}$};
        \end{scope}
      \end{tikzpicture}\qedhere
    \end{center}
  \end{itemize}
\end{proof}

\begin{remark}[label={pos:rem:deterministic-routing-lemma}]
  Notice that the wiring in this proof is actually constructive and ``positional'', in the sense that the direction of the wires appearing at a position $\vec{i} \in \rect{r}$ only depends on said position. Since the tileset $\Swang[\Srect]$ from \Cref{pos:sec:rect} labels each tile with its position in $\rect{r}$, we could construct a tileset  $\Swang[\route] \subseteq \Swang[\Srect] \times \Swang[\wire](A \times \alphabet[C])$ (for some finite set $A$ of cardinality $\bigO(d)$) such that $\Swang[\route]$ satisfies the \crtlnameref{pos:lem:routing-lemma} while also ensuring the \emph{uniqueness} of the wiring pattern $w$.
\end{remark}

%
\section{Proof of the fixed-point lemma}
\label{fix:sec}

\setcounter{subsection}{-1}
\subsection{Overview of the fixed-point construction}
At its core, the fixed-point construction is based on the notion of \emph{block simulation} from~\Cref{fix:def:tiling-simulation}.

\subsubsection{The classical construction}

Usually defined on $\Z^2$, the fixed-point construction creates a hierarchy $(\Swang[\ell][S])_{\ell \in \N}$ of tilesets such that each $\Swang[\ell][S]$ simulates $\Swang[\ell+1][S]$ through a deterministic simulation $\sigma_{\ell} \colon \Swang[\ell+1][S]^{\rspace_{\ell} \times \rspace_{\ell}} \to \Swang[\ell][S]$ with zoom $\rspace_{\ell} \times \rspace_{\ell}$ (for some size $\rspace_{\ell} \in \N$).

\smallskip

To realize this hierarchy of simulations, the tiles of $\Swang[\ell][S]$ will arrange themselves into blocks of size $\rspace_{\ell} \times \rspace_{\ell}$ called \emph{macro-tiles}, thus forming a regular grid $\grid{g}_{\ell}$ such that all rectangles $\rect{r} \in \grid{g}_{\ell}$ satisfy $\rnorm{\rect{r}} = \ibox{\mu_{\ell}}^2$. These macro-tiles then emulate the behavior of Wang tiles (\textit{c.f.}~\Cref{fix:fig:macro-tile-configuration}):
\begin{enumerate}
\item If a tile of $\Swang[\ell][S]$ appears on the border of a macro-tile, it carries an additional bit;
\item The concatenations of those bits on each side of a macro-tile form its four \emph{macro-colors};
\end{enumerate}
\begin{figure}[ht]
  \begin{tikzpicture}[scale=0.25,decoration={random steps,segment length=0.4em,amplitude=0.15em}]
    \pgfmathdeclarerandomlist{Rainbow}{{RainbowB!50!RainbowC}{RainbowE}{RainbowF}{RainbowG!50!RainbowA}}
    \begin{scope}
      \pgfmathsetseed{1024}
      \draw [TilingGrid,decorate,path picture={
        \foreach \i in {0,5,...,30} {
          \foreach \j in {0,5,...,20} {
            \pgfmathrandomitem{\Rainbow}{Rainbow}
            \fill[color=\Rainbow] ($(\i,\j) + (-0.2,0.6)$) rectangle ++(0.4,3.8);
            \draw[line width=0.7,draw=\Rainbow] ($(\i,\j) + (0,2.5)$) -- ++(0.7,0) -- ++(0,-2) -- ++(1.2,0) -- ++(0,1);
            \draw[line width=0.7,draw=\Rainbow] ($(\i,\j) + (0,2.5)$) -- ++(-0.9,0) -- ++(0,-1.8) -- ++(-1.8,0) -- ++(0,1);
            \pgfmathrandomitem{\Rainbow}{Rainbow}
            \fill[color=\Rainbow] ($(\i,\j) + (0.6,-0.2)$) rectangle ++(3.8,0.4);
            \draw[line width=0.7,draw=\Rainbow] ($(\i,\j) + (2.5,0)$) -- ++(0,0.5) -- ++(-0.4,0) -- ++(0,1);
            \draw[line width=0.7,draw=\Rainbow] ($(\i,\j) + (2.5,0)$) -- ++(0,-0.9) -- ++(-1.6,0) -- ++(0,-3.4) -- ++(0.8,0) -- ++(0,1);
          }
        }
        \foreach \i in {0,5,...,30} {
          \foreach \j in {0,5,...,20} {
            \fill[rounded corners=0.5pt,opacity=0.5,black] ($(\i,\j) + (1.24,1)$) rectangle ++(2.52,0.4);
            \fill[gray!20!white] ($(\i,\j) + (1.2,1.2)$) rectangle ++(2.6,2.6);
          }
        }
        \draw[very thin,opacity=0.5,step=0.2] (0,0) grid (30,20);
        \foreach \i in {0,5,...,30} {
          \foreach \j in {0,5,...,20} {
            \fill[gray!20!white] ($(\i,\j) + (2.5,2.5)$) node {\color{black} \scriptsize $\sigma_{\ell}(\cdot)$};
          }
        }
        \draw[step=5] (0,0) grid (30,20);
      }] (1.5,1.5) rectangle (28.5,18.5);
    \end{scope}
    \begin{scope}[shift={(29,0)}]
      \pgfmathsetseed{1024}
      \draw [TilingGrid,decorate,path picture={
        \foreach \i in {0,5,...,30} {
          \foreach \j in {0,5,...,20} {
            \pgfmathrandomitem{\Rainbow}{Rainbow}
            \draw[fill=\Rainbow] (\i,\j) -- ++(-1,1) -- ++ (0,3) -- ++(1,1) -- ++(1,-1) -- ++(0,-3) -- cycle;
            \pgfmathrandomitem{\Rainbow}{Rainbow}
            \draw[fill=\Rainbow] (\i,\j) -- ++(1,1) -- ++ (3,0) -- ++(1,-1) -- ++(-1,-1) -- ++(-3,0) -- cycle;
          }
        }
        \draw[step=5] (0,0) grid (30,20);
      }] (1.5,1.5) rectangle (28.5,18.5);
    \end{scope}
  \end{tikzpicture}
  \caption{Schematics of a macro-tile, and a partial configuration of macro-tiles.}
  \label{fix:fig:macro-tile-configuration}
\end{figure}
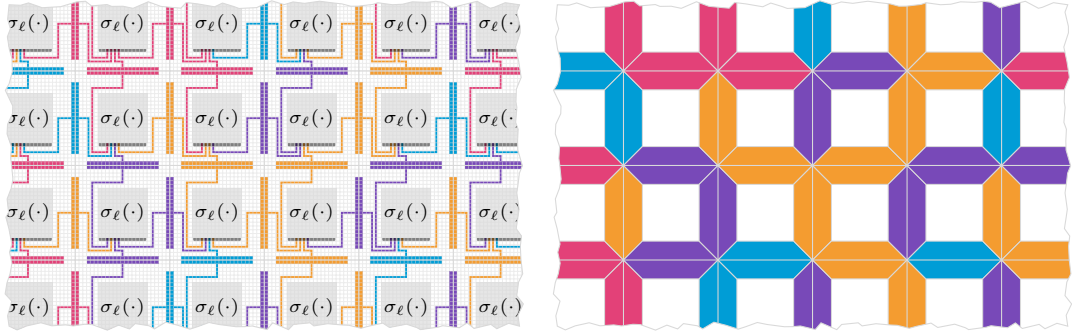

\noindent When placing such macro-tiles next to one another along the grid $\grid{g}_{\ell}$, adjacent macro-tiles share a common macro-color along their common facet; and thus, tilings of macro-tiles will behave like tilings of classical Wang tiles. In order to restrict the macro-tiles (whose macro-colors currently cover all
binary strings of length $\rspace_{\ell}$) to the tileset $\Swang[\ell+1][S]$
(which may only contain a subset of $4$-tuples of macro-colors):
\begin{enumerate}[start=3]
\item We embed arbitrary computations by drawing, inside each macro-tile, the space-time diagram of a Turing machine computing the simulation $\sigma_{\ell}$ (and, thus, the tileset $S_{\ell+1}$);
\end{enumerate}

\noindent In particular, the zoom $\rspace_{\ell}$ of the simulation $\sigma_{\ell}$ between $\Swang[\ell][S]$ and $\Swang[\ell+1][S]$ \emph{must be large enough} to provide enough time for those computations to terminate.

\medskip
Notice that this method requires to already know the programming of the tileset $\Swang[\ell+1][S]$ \emph{before} defining $\Swang[\ell][S]$.
In order to define an infinite sequence $(\Swang[\ell][S])_{\ell \in \N}$ of tilesets such that $\Swang[\ell][S]$ simulates $\Swang[\ell+1][S]$, we circumvent this dependency issue by applying the \crtlnameref{calc:prop:fixed-point-theorem}~(\Cref{calc:prop:fixed-point-theorem}), which defines a unique program upon which all these tilesets can be based.

\subsubsection{Generalizing to non-uniform $d$-dimensional rectangles}

\Cref{sof:lem:fixpoint} is stated for grids $(\grid{g}_{\ell})_{\ell \in \N}$ whose rectangles $\rect{r} \in \grid{g}_{\ell}$ might draw various shapes and sizes. We thus generalize macro-colors and macro-tiles from $2$-dimensional squares to $d$-dimensional rectangles, and:
\begin{itemize}
\item The growth conditions on the grids $(\grid{g}_{\ell})_{\ell \in \N}$ use uniform bounds $\rspace_{\ell}$, $\rtime_{\ell}$ on the smallest facets $\rspace_{\ell} \leq \rspace(\grid{g}_{\ell}) = \inf_{\rect{r} \in \grid{g}_{\ell}} \rspace(\rect{r})$ and highest edge length $\rtime_{\ell} \geq \rtime(\grid{g}_{\ell}) = \sup_{\rect{r} \in \grid{g}_{\ell}} \rtime(\rect{r})$;
\item We embed computations by drawing the space-time diagrams of a \emph{processor array}, whose parallelism enables the $d$-dimensional geometry to operate in time $\bigO(n)$ on objects of size $\bigO(n^{d-1})$ (as opposed to Turing machines, in which space is always bounded by time);
\item The transmission of information (\textit{e.g.}~macro colors) inside a macro-tile of rectangular domain now relies on the \crtlnameref{pos:lem:routing-lemma}~(\Cref{pos:lem:routing-lemma}).
\end{itemize}

\subsubsection{Computing the substitution \texorpdfstring{$\tau$}{τ}} In order to compute the substitution $\tau$ given by \Cref{sof:lem:fixpoint}, we are left with a major difficulty: while the fixed-point construction classically requires a fast growth on the sizes $(\rspace(\grid{g}_{\ell}))_{\ell \in \N}$ to provide enough time for the computations of the simulations $(\sigma_{\ell})_{\ell \in \N}$ to terminate, \Cref{sof:lem:fixpoint} allows for substitution steps $\smash{x^{(\ell+1)} \substep[\tau][\grid{g}_{\ell}] x^{(\ell)}}$ whose grid $\grid{g}_{\ell}$ contains arbitrarily small rectangles (\textit{e.g.}~$\smash{\rnorm{\rect{r}} = \ibox{2}^d}$).

\emph{Our main idea for this proof consists in compressing the computation of several substitution steps for $\tau$ at the same level of simulation}. More precisely, we make a level $\ell$ of simulation compute several substitution steps
\[ x^{(\ell+1)} \substep[\tau][\grid{g}_{\ell}] x^{(\ell)}
  , \quad
  x^{(\ell+2)} \substep[\tau][\grid{g}_{\ell+1}] x^{(\ell+1)}
  ,\quad
  x^{(\ell+3)} \substep[\tau][\grid{g}_{\ell+2}] x^{(\ell+2)}
  , \quad \text{etc\dots}
\]
until the product $\rspace(\grid{g}_{\ell}) \cdot \rspace(\grid{g}_{\ell+1}) \cdot \rspace(\grid{g}_{\ell+2}) \cdot \ldots$ becomes large enough to compute a simulation step. For a given sequence $(\rspace(\grid{g}_{\ell}))_{\ell \in \N}$, we thus define a function $\kappa \colon \N \to \N$ that partitions $\N$ into segments $\interval{\kappa(n)}{\kappa(n+1)-1}$, which correspond to substitutions steps whose computations are actually implemented at the same level of simulation. This idea of \emph{sub-hierarchies} satisfying the growth constraints of the simulations already appears in~\cite{Zinoviadis_2016:phd:expansiveness-2d-sfts}.

\subsubsection{Interleaving computations of the substitution \texorpdfstring{$\tau$}{τ}}

Such segments $\interval{\kappa(n)}{\kappa(n+1)-1}$ might unfortunately define rectangles so large that the tileset~$\Swang[\kappa(n)][S]$ can no longer simulate $\Swang[\kappa(n+1)][S]$ (for example, a macro-tile might become too big for its size to fit (in binary) inside the computation space of the tiles of $\Swang[\kappa(n)][S]$). To solve this issue, we actually interleave the computations of simulations and substitutions steps. More precisely, we introduce a second function $\iota \colon \N \to \N$ such that:
\begin{itemize}
\item We build a sequence $(\Swang[\iota(n)][S])_{n \in \N}$ of tilesets such that each $\Swang[\iota(n)][S]$ simulates $\Swang[\iota(n+1)][S]$ using the fixed-point construction;
\item For all $n \in \N$, we have $\iota(n) < \kappa(n)$ so that $\Swang[\iota(n)][S]$ can be given the responsability to embed the computation of all the substitution steps $\smash{x^{(\ell+1)} \substep[\tau][\grid{g}_{\ell}] x^{(\ell)}}$ for $\ell \in \interval{\kappa(n)}{\kappa(n+1)-1}$;
\item And for all $n \in \N$, we have $\iota(n+1) \in \interval{\kappa(n)}{\kappa(n+1)-1}$ is small enough so that $\Swang[\iota(n)][S]$ can simulate $\Swang[\iota(n+1)][S]$.
\end{itemize}

\subsection{Deciding the hierarchy of substitutions and simulations}

The whole construction actually relies on the following lemma, which proves that such interleavings of simulations $\iota(n+1)$ inside substitution $\interval{\kappa(n)}{\kappa(n+1)-1}$ segments exist.

\begin{figure}[ht]
  \begin{tikzpicture}[scale=0.4]
    \begin{scope}
      \draw[gray!50,->,>=stealth] (-0.5,0) -- ++(20.3,0);
      \foreach \i in {0,...,19} {
        \draw[gray!50] (\i,0.2) -- ++(0,-0.4);
      }
      \draw[very thick,RainbowE,font=\scriptsize] (8,0.35) -- ++(0,-0.7) node[below] {$\kappa(n)$};
      \draw[very thick,RainbowE,font=\scriptsize] (17,0.35) -- ++(0,-0.7) node[below] {$\kappa(n{+}1)$};
      \draw[very thick,RainbowB,font=\scriptsize] (1,0.35) -- ++(0,-0.7) node[below] {$\iota(n)$};
      \draw[very thick,RainbowB,font=\scriptsize] (12,0.35) -- ++(0,-0.7) node[below] {$\iota(n{+}1)$};
    \end{scope}
  \end{tikzpicture}
  \caption{The functions $\iota \colon \N \to \N$ and $\kappa \colon \N \to \N$ from the ``\protect\crtlnameref{fix:lem:staircase-lemma}''.\protect\footnotemark}
  \subcaption{In the proof, the levels $(\iota(n))_{n \in \N}$ will refer to successive levels of fixed-point simulation. Each level $\iota(n)$ will also be responsible for computing the substitutions $\smash{x^{(\ell+1)} \substep x^{(\ell)}}$ for the block $\ell \in \interval{\kappa(n)}{\kappa(n{+}1){-}1}$.}
\end{figure}
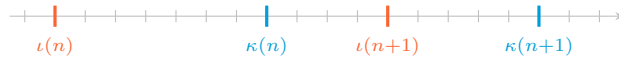

\noindent Before we proceed with the formal statement \Cref{fix:lem:staircase-lemma}, let
us motivate the numerical quantities involved. If the tiles of level $\iota(n)$ organize themselves into macro-tiles of level $\ell \geq \kappa(n)$, then:
\begin{itemize}[itemsep=3pt]
\item A macro-tile of level $\ell$ has space larger than $\prod_{i=\iota(n)}^{\ell-1} \rspace_{i}$: more precisely, each facet in such a macro-tile contains at least that many tiles of level $\iota(n)$;
\item We make macro-tiles of level $\ell$ embed $\bigO(\pxspace{\ell}[\gamma])$ steps of $\log$-RAM computations for some $\gamma < 1$ by embedding the processor array simulations from \Cref{calc:lem:pa-ram-simulation}: Equation~\eqref{eq:stair:1} ensures that each facet in such a macro-tile is large enough to embed this many processors;
\item In a macro-tile of level $\ell$, a tile of level $\iota(n)$ is actually contained in macro-tiles of intermediate levels $\iota(n) < i \leq \ell$. As these rectangles have maximal edge length $\rtime_{i}$, encoding relative positions in binary (Equation~\eqref{eq:stair:2}) requires bit length $\bigO(\sum_{i=\iota(n)}^{\ell-1} \log \rtime_{i})$;
\item Since variables in a processor array have logarithmic bit length, and that $\bigO(\pxspace{\ell}^{\gamma})$ processors are embedded in the macro-tiles of level $\ell$, the tiles of level $\iota(n)$ require $\bigO(\log \pxspace{i})$ bits to encode a processor state for each intermediate level $\iota(n) < i \leq \ell$ (Equation~\eqref{eq:stair:3}).
\end{itemize}

\footnotetext{The \protect\crtlnameref{fix:lem:staircase-lemma} was also named the «~écluse~» lemma, from the French word denoting the devices that can raise/lower boats in a waterway to transition between sections of different heights. To their dismay, the authors later learned that such devices are called ``locks'' in English, which loses all the intended meaning. We thus settled for the terminology ``staircase'', although this staircase is somewhat of a trip hazard.}

\begin{lemma}[name={Staircase lemma},label={fix:lem:staircase-lemma}]
  Let $\K \in \N$ be an integer, $\delta, \gamma \in \R_+$ be two real numbers such that $2\delta < \gamma < 1$, and let $(\rspace_{\ell}, \rtime_{\ell})_{\ell \in \N}$ be sequences of integers such that, denoting $\pxspace{\ell} = \prod_{i < \ell} \rspace_{i}$, we have $\rtime_{\ell}^{\mspace{3mu}d-1} \geq \rspace_{\ell}$, $2 \leq \rtime_{\ell} \leq 2^{\K \cdot \pxspace{\ell}[\delta]}$ and $\smash{\rspace_{\ell}^{\K \cdot \log \pxspace{\ell+1}} \geq \rtime_{\ell}}$. \\[2pt]
  For $\K[][s] \in \R_+$, define $\iota \colon \N \to \N$ and $\kappa \colon \N \to \N$ as $\iota(0) = \kappa(0) = 0$, and: \\[-\linespacing]
  \begin{align*}
    \kappa(n+1) & = \min \big\{ \ell > \kappa(n) : \prod_{i=\kappa(n)}^{\ell-1} \rspace_{i} > \K[][s]\cdot \pxspace{\ell}[\gamma] \big\}; \\[-6pt]
    \iota(n+1) & = \max \{ \ell < \kappa(n+1) : \prod_{i=\ell}^{\kappa(n+1)-1} \rspace_{i} > \K[][s]\cdot \pxspace{\kappa(n+1)}[\gamma] \big\}.
  \end{align*}
  Then if $\K[][s] \in \R_+$ is large enough, $\iota$ and $\kappa$ are well-defined, and satisfy for all $n \geq 1$:\\[-.5\linespacing]
  \begin{align}
    \kappa(n-1) \leq \iota(n) < \kappa(n)
  \end{align}
  \begin{align}
    \prod_{\ell=\iota(n)}^{\kappa(n)-1} \mspace{-9mu} \rspace_{\ell} & \geq \K[][s] \cdot \pxspace{\kappa(n)}[\gamma] \hspace{0.7cm} \label{eq:stair:1}\\
    \sum_{\ell=0}^{\kappa(n+1)-2} \mspace{-12mu} \log \rtime_{\ell} & \leq \polylog \K[][s] \cdot \pxspace{\iota(n)}[2\delta] \label{eq:stair:2} \\
    \sum_{\ell=1}^{\kappa(n+1)-1} \mspace{-12mu} \log \pxspace{\ell} & \leq \polylog \K[][s] \cdot \pxspace{\iota(n)}[2\delta] \label{eq:stair:3}
  \end{align}
\end{lemma}
\noindent Since $\gamma < 1$, (2) also implies that for every $\ell \in \interval{\kappa(n)}{\kappa(n+1)-1}$, we have $\mspace{-12mu} \displaystyle \prod_{k=\iota(n)}^{\ell-1} \mspace{-9mu} \rspace_{k} \geq \K[][s] \cdot \pxspace{\ell}[\gamma]$.

\begin{proof}
  First, the functions $\kappa$ and $\iota$ are well-defined: since $\gamma < 1$ and that the product $\smash{\prod_{\ell=\kappa(n)}^{+\infty} \rspace_{\ell}}$ diverges, there must exist some $\ell$ such that $\prod_{i=\kappa(n)}^{\ell-1} \rspace_{i} > (\K[][s] \cdot \pxspace{\kappa(n)}[\gamma]) \cdot \prod_{i=\kappa(n)}^{\ell-1} \rspace_{i}^{\gamma}$. In particular, $\kappa(n+1)$ is well-defined, and in turn so is $\iota(n+1)$. Furthermore, we have $\iota(n+1) < \kappa(n+1)$ by definition of $\iota(n+1)$; and the product $\smash{\prod_{i=\kappa(n)}^{\kappa(n+1)-1} \rspace_{i}}$ is large enough to ensure that $\iota(n+1) \geq \kappa(n)$.

  \medskip
  We now prove that \eqref{eq:stair:1}, \eqref{eq:stair:2} and \eqref{eq:stair:3} hold for all $\K[][s] \in \R_+$ large enough. We claim that~\eqref{eq:stair:1} holds by definition of $\iota(n)$, whereas~\eqref{eq:stair:2} and~\eqref{eq:stair:3} follow from the following computations:
  \reqnomode
  \begin{itemize}[topsep=0.4\linespacing]
  \item First, we have by minimality of $\kappa(n+1)$ that
    \[ \prod_{\ell=\kappa(n)}^{\kappa(n+1)-2} \mspace{-6mu} \rspace_{\ell}
      \leq \K[][s] \cdot \pxspace{\kappa(n+1)-1}[\gamma]
      = \K[][s] \cdot \pxspace{\kappa(n)}[\gamma] \cdot \mspace{-10mu} \prod_{\ell=\kappa(n)}^{\kappa(n+1)-2} \mspace{-10mu} \rspace_{\ell}^{\gamma};
    \]
    so that, by dividing, taking a $\log$ and assuming that $\K[][s]$ is large enough, we deduce
    \begin{equation}
      \tag{$\ast$}
      \log \mspace{-6mu} \prod_{\ell = \kappa(n)}^{\kappa(n+1)-2} \mspace{-6mu} \rspace_{\ell}
      \leq \frac{\gamma}{1-\gamma} \cdot \log \pxspace{\kappa(n)} + \log \K[][s];
    \end{equation}
  \item Furthermore, by $(\ast)$, we obtain that for $\K[][s]$ large enough:
    \[ \log \pxspace{\kappa(n+1)-1}
      = \log \pxspace{\kappa(n)} + \log \mspace{-8mu}\prod_{\ell=\kappa(n)}^{\kappa(n+1)-2} \mspace{-8mu} \rspace_{\ell}
      \leq \log \pxspace{\kappa(n)} \cdot 2 \log \K[][s]; \]
    \pagebreak
  \item And by maximality of $\iota(n)$, and since $\rspace_{\ell} \leq \rtime_{\ell}^{\mspace{3mu}d-1}$ and $\rtime_{\ell} \leq 2^{\K[\rtime] \cdot \pxspace{\ell}[\delta]}$ we have
    \[
      \pxspace{\kappa(n)}
      = \pxspace{\iota(n)} \cdot \rspace_{\iota(n)} \cdot \mspace{-6mu} \prod_{\ell=\iota(n)+1}^{\kappa(n)-1} \mspace{-6mu} \rspace_{\ell}
      \leq \pxspace{\iota(n)} \cdot 2^{\K[\rtime] \cdot (d-1) \cdot \pxspace{\iota(n)}[\delta]} \cdot \K[][s] \cdot \pxspace{\kappa(n)}[\gamma];
    \]
    so that, by dividing, taking a $\log$ and assuming that $\K[][s]$ is large enough, we deduce that\\[-0.3\linespacing]
    \[
      \log \pxspace{\kappa(n)}
      \leq (\log \K[][s]) \cdot \pxspace{\iota(n)}[\delta] + \log \K[][s].
    \]
  \end{itemize}
  \leqnomode

  \bigskip
  \noindent Combining these considerations with the hypothesis $\smash{\rtime_{\ell} \leq \rspace_{\ell}^{\K \cdot \log \pxspace{\ell+1}}}$, we deduce that:
  \begin{align*}
    \sum_{\ell=0}^{\kappa(n+1)-2} \log \rtime_{\ell}
    & \leq \sum_{\ell=0}^{\kappa(n+1)-2} \K \cdot \log \pxspace{\ell+1} \cdot \log \rspace_{\ell} \\[-1\linespacing]
    & \leq \K[\rspace] \cdot \log \pxspace{\kappa(n+1)-1} \cdot \log \mspace{-8mu} \prod_{\ell=0}^{\kappa(n+1)-2} \mspace{-10mu} \rspace_{\ell} \\[-0.8\linespacing]
    & \leq \K[\rspace] \cdot (\log \pxspace{\kappa(n+1)-1})^2 \\
    & \leq \K[\rspace] \cdot (2 \log \K[][s])^4 \cdot \pxspace{\iota(n)}[2\delta]
  \end{align*}
  \textit{i.e.}~that~\eqref{eq:stair:2} holds if $\K[][s] \in \R_+$ is large enough to satisfy $\log \K[][s] \geq 2^4 \cdot \K[\rspace]$. Furthermore, since all $\rspace_{\ell}$ satisfy $\rspace_{\ell} \geq 2$, we deduce and obtain by $(\ast)$ that for $\K[][s]$ large enough:
    \[ \kappa(n+1) - 1
      = \log 2^{\kappa(n+1)-1}
      \leq \log \mspace{-8mu} \prod_{\ell=0}^{\kappa(n+1)-2} \mspace{-8mu} \rspace_{\ell}
      \leq \log \pxspace{\kappa(n+1)-1}
      \leq \log \pxspace{\kappa(n)} \cdot 2 \log \K[][s];
    \]
    from which we deduce $(4)$ if $\K[][s]$ is large enough again:
  \begin{align*}
    \hspace*{-0.08cm} \sum_{\ell=1}^{\kappa(n+1)-1} \log \pxspace{\ell}
    & \leq (\kappa(n+1)-1) \cdot \log \pxspace{\kappa(n+1)-1} \\[-1\linespacing]
    & \leq \big(2 \log \K[][s] \cdot \log \pxspace{\kappa(n)} \big)^2 \\
    & \leq (2 \log \K[][s])^4 \cdot \log \pxspace{\iota(n)}[2\delta] \qedhere
  \end{align*}
\end{proof}

\subsection{\boldmath The algorithm \texorpdfstring{$F(e)$}{F(e)}}\label{fix:sec:algorithm}

We now begin the proof of \Cref{sof:lem:fixpoint}.
Any RAM program $e \in \{0,1\}^{\ast}$ will be computably transformed into a new $\log$-RAM program $F(e) \in \{0,1\}^{\ast}$ which defines a non-deterministic function of the form:
\begin{gather*}
  \funram{F(e)}
  \big((\K[][s],\K[][w]),
  \iota,
  (\rspace_{\ell},\rtime_{\ell})_{0 \leq \ell < \kappa-1},
  (\rect{r}_{\ell},\vec{i}_{\ell})_{\iota \leq \ell < \kappa-1},
  \mspace{1mu} c_1^{\scriptscriptstyle -},\dots,c_d^{\scriptscriptstyle -},
  \mspace{1mu} c_1^{\scriptscriptstyle +},\dots, c_d^{\scriptscriptstyle +},
  \mspace{1mu} c_{\decsymbol}
  \big) \\
  \mapsmapsto
  \big(\nghb_{\kappa-1},
  \rect{r}_{\kappa-1},
  \vec{i}_{\kappa-1},
  u
  \big)
\end{gather*}
where\footnote{Notice that most of these variables are not written as binary words in $\{0,1\}^{\ast}$, but as data types of the form $\N$, $(\{0,1\}^{\ast})^{2d+1}$, etc\dots\ As mentioned in \Cref{calc:rem:data-types}, it is usual in computability theory to consider such elements both as data structures and binary words through implicit encodings.}

\begin{itemize}[topsep=0.4\linespacing,leftmargin=3\parindent,labelindent=1.5\parindent,labelwidth=0pt,itemindent=!,itemsep=0.2\linespacing]
\item $(\K[][s],\K[][w]) \in \N^2$ is a pair of constants;\\[2pt]
  \footnotesize Informally, they hardcode the constant from the \crtlnameref{fix:lem:staircase-lemma} and the $\bigO(\ldots)$ factor on the logarithmic word length of the program $e \in \{0,1\}^{\ast}$.

  \normalsize
\item $\iota \in \N$ is an integer;\\[2pt]
  \footnotesize Informally, $\iota$ refers to the current level of the simulation.

  \normalsize
\item $(\rspace_{\ell},\rtime_{\ell})_{0 \leq \ell < \kappa-1} \in (\N^2)^{\ast}$ is a family of integer bounds;\\[2pt]
  \footnotesize Informally, each pair $(\rspace_{\ell},\rtime_{\ell}) \in \N^2$ consists of a lower bound on the space and an upper bound on the time of the grids~$\grid{g}_{\ell}$ ($\ell < \kappa-1$) defined by the previous steps of substitutions and simulations.

  \normalsize
\item $(\mspace{1mu} c_1^{\scriptscriptstyle -},\dots,c_d^{\scriptscriptstyle -},
  \mspace{1mu} c_1^{\scriptscriptstyle +},\dots, c_d^{\scriptscriptstyle +},
  \mspace{1mu} c_{\decsymbol}) \in (\{0,1\}^{\ast})^{2d+1}$ is a $(2d+1)$-tuple of binary strings;\\[2pt]
  \footnotesize Informally, $(\mspace{1mu} c_1^{\scriptscriptstyle -},\dots,c_d^{\scriptscriptstyle -},
  \mspace{1mu} c_1^{\scriptscriptstyle +},\dots, c_d^{\scriptscriptstyle +},
  \mspace{1mu} c_{\decsymbol})$ represents a $d$-dimensional Wang tile $\wang{t}$ of the infinite tileset $\Swang[\ast]$.

  \normalsize
\item $(\rect{r}_{\ell},\vec{i}_{\ell})_{\iota \leq \ell < \kappa-1}$ is a (possibly empty) family of rectangles $\rect{r}_{\ell} \in \Srect_0$ and positions $\vec{i}_{\ell} \in \rect{r}_{\ell}$;\\[2pt]
  \footnotesize Informally, $\vec{i}_{\ell}$ represents the position of the Wang tile $\wang{t}$ inside its rectangle $\rect{r}_{\ell} = \rnorm{\rect{r}}$ for $\rect{r} \in \grid{g}_{\ell}$, where the grids $\grid{g}_{\ell}$ ($\iota \leq \ell < \kappa-1$) were defined by the previous steps of substitutions and simulations.
\end{itemize}

and
\begin{itemize}[topsep=0.4\linespacing,leftmargin=3\parindent,labelindent=1.5\parindent,labelwidth=0pt,itemindent=!,itemsep=0.2\linespacing]
\item $\nghb_{\kappa-1} \subseteq \{\pm \basis{k} : 1 \leq k \leq d\}$ is a set of (signed) directions.\\[2pt]
  \footnotesize Informally, $\nghb_{\kappa-1}$ indicates the directions of sibling macro-tiles at level $\kappa-1$, that is, determines which macro-tiles of level $\kappa-1$ belong to the same parent macro-tile of level $\kappa$.

  \normalsize
\item $(\rect{r}_{\kappa-1},\vec{i}_{\kappa-1})$ is a rectangle $\rect{r}_{\kappa-1} \in \Srect_0$ and a position $\vec{i}_{\kappa-1} \in \rect{r}_{\kappa}$;\\[2pt]
  \footnotesize Informally, $\vec{i}_{\kappa-1}$ refers to the position of the Wang tile $\wang{t}$ inside its rectangle $\rect{r}_{\kappa-1} = \rnorm{\rect{r}}$ for $\rect{r} \in \grid{g}_{\kappa-1}$, where $\grid{g}_{\kappa-1}$ is the first newly defined grid at the current level of simulation.

  \normalsize
\item $u \in \{0,1\}^{\ast}$ is either blank, or encodes a tuple $(i,j,b) \in (\N^2 \times \{0,1\})$ for $(i,j)$ a pair of integers and $b$ a single bit.\\[2pt]
  \footnotesize Intuitively, a non-blank tuple $(i,j,b)$ corresponds to a bit $b \in \{0,1\}$ of index $(i,j) \in \interval{0}{2d} \times \N$ issued by some RAM computations and needed by the tile $\wang{t}$ from the previous level of simulation.
\end{itemize}

\paragraph*{Proof strategy}
Given any RAM program $e \in \{0,1\}^{\ast}$, and for fixed parameters $(\K[][s],\K[][w])$, $\iota$~and $(\rspace_{\ell},\rtime_{\ell})_{0 \leq \ell < \kappa-1}$, \emph{the non-deterministic function $\funram{F(e)}$ will define a set of accepted Wang tiles ${(\mspace{1mu} c_1^{\scriptscriptstyle -},\dots,c_d^{\scriptscriptstyle -},
  \mspace{1mu} c_1^{\scriptscriptstyle +},\dots, c_d^{\scriptscriptstyle +},
  \mspace{1mu} c_{\decsymbol}) \in \Swang[\ast]}$.} Shortening notations, we denote by $\Swang[\iota][S] \subseteq \Swang[\ast]$ the set
\[ \Swang[\iota][S] =
  \mspace{-20mu} \bigcup_{(\rect{r}_{\ell},\vec{i}_{\ell})_{\iota \leq \ell < \kappa-1}} \mspace{-20mu}
  \big\{ (\mspace{1mu} c_1^{\scriptscriptstyle -},\dots,c_{\decsymbol}) \in \Swang[\ast]
  \; : \;
  \funram{F(e)}((\K[][s],\K[][w]),\iota,(\rspace_{\ell},\rtime_{\ell})_{0 \leq \ell < \kappa-1},(\rect{r}_{\ell},\vec{i}_{\ell})_{\iota \leq \ell < \kappa-1},\mspace{1mu} c_1^{\scriptscriptstyle -},\dots, c_{\decsymbol})\halt
  \big\}\]
of Wang tiles $\wang{t}$ for which there exists a family of rectangles and positions $(\rect{r}_{\ell},\vec{i}_{\ell})_{\iota \leq \ell < \kappa-1}$ such that $\funram{F(e)}(\dots,\wang{t})$ admits an accepting computation; and by $\Scolor[\iota] \subseteq \{0,1\}^{\ast}$ the associated colors; although one should keep in mind that both sets $\Swang[\iota][S]$ and $\Scolor[\iota]$ strongly depend on the fixed parameters $(\K[][s],\K[][w]) \in \N^2$ and $(\rspace_{\ell},\rtime_{\ell}) \in (\N^2)^{\ast}$.

\medskip
For grids $(\grid{g}_{\ell})_{0 \leq \ell < \kappa-1}$, the tiles in $\Swang[\iota][S]$ will be designed to form tilings that informally:
\begin{itemize}[itemsep=1pt]
\item Define some new grids $(\grid{g}_{\ell})_{\kappa-1 \leq \ell < \kappa'-1}$;
\item Simulate a tiling of $\Swang[\iota'][S]$ for some level $\iota'$ in the interval $\kappa \leq \iota' < \kappa'$;
\item Emulate at each new level $\kappa \leq \ell < \kappa'$ a valid $\Swang[\ast]$-tiling $x^{(\ell)}$ such that\footnote{The configuration $x^{(\kappa')}$ and the grid $\grid{g}_{\kappa'-1}$ appear to be missing in this explanation: while we recommend to ignore this technicality at this level of informal explanations, it is actually due to the numeric bounds provided by the~\crtlnameref{fix:lem:staircase-lemma}, which do not allow to store the grid $\grid{g}_{\kappa'-1}$ in the tiles of $\Swang[\iota][S]$. They will thus be defined by the next level of simulation~$\iota'$, and will require some complex synchronization in the simulation between $\Swang[\iota][S]$ and $\Swang[\iota'][S]$.} $\smash{x^{(\ell+1)} \substep[\tau][\grid{g}_{\ell}] x^{(\ell)}}$;
\end{itemize}
To design the tileset $\Swang[\iota][S] \subseteq \Swang[\ast]$, we will progressively enumerate the restrictions that make a tuple of colors $(c_1^{\scriptscriptstyle -},\dots,c_d^{\scriptscriptstyle -}, c_1^{\scriptscriptstyle +},\dots, c_d^{\scriptscriptstyle +}, c_{\decsymbol})$ accepted by the RAM algorithm $F(e) \in \{0,1\}^{\ast}$. By implicitly encoding complex data types into binary strings, we will consider that valid colors of $\Scolor[\iota] \subseteq \{0,1\}^{\ast}$ actually represent a \emph{record}\footnote{Also known as~``compound data type'' or \texttt{struct}, a record allows to define a collection of several variables (called \emph{fields}), possibly of different data types, grouped together into a single value. The fields in a record are typically identified by some explicit names.}, \textit{i.e.}~a collection of several data fields: for both the authors' and the reader's convenience, this proof will identify these fields by explicit names.

\emph{The proof below thus consists in listing the fields appearing in the colors $\Scolor[\iota]$, and the constraints a $(2d+1)$-tuple of them must satisfy to form a valid tile $(c_1^{\scriptscriptstyle -},\dots,c_d^{\scriptscriptstyle -}, c_1^{\scriptscriptstyle +},\dots, c_d^{\scriptscriptstyle +}, c_{\decsymbol}) \in \Swang[\iota][S]$.} The algorithm $F(e) \in \{0,1\}^{\ast}$ will thus be defined somewhat implicitly, as it mostly amounts to checking whether these constraints are satisfied.

\paragraph*{Housekeeping}\label{fix:sec:housekeeping}

Before we begin the proof, let us introduce the objects given by \Cref{sof:lem:fixpoint}: let $\alphabet \subseteq \{0,1\}^{\ast}$ be a finite alphabet and $\smash{X \subseteq \alphabet^{\Z^d}}$ a sofic shift space on $\alphabet$. Fix $\code{\tau} \in \{0,1\}^{\ast}$ a $\log$-RAM program defining a parallel expanding substitution, $\alpha < 1$ and $\smash{\delta < \frac{1}{2}}$ two rationals satisfying the hypotheses of \Cref{sof:lem:fixpoint}. Let also $\K \in \N$ be such that:
\begin{enumerate}[itemsep=4pt,topsep=4pt]
\item \begin{enumerate}[label={(\roman*)},itemsep=2pt]
  \item $2 \leq \rtime(\grid{g}_{\ell}) \leq 2^{\K \cdot \pxspace{\ell}[\delta]}$;
  \item $\rspace_{\ell}^{\K \cdot \log \pxspace{\ell+1}} \geq \rtime_{\ell}$;
  \end{enumerate}
\item \begin{enumerate}[label={(\roman*)},itemsep=2pt]
  \item The bit length in $x^{(\ell+1)}$ satisfies $\smash{\norm{x^{(\ell + 1)}} \leq \K[t] \cdot \pxspace{\ell}[\alpha]}$;
  \item The substitution $\smash{x^{(\ell+1)} \substep[\tau][\grid{g}_{\ell}] x^{(\ell)}}$ is computed in time at most $\smash{\K[t] \cdot \pxspace{\ell}[\alpha]}$;
  \item The program $\code{\tau} \in \{0,1\}^{\ast}$ is a $\log$-RAM program satisfying $W_{\code{\tau}}(\ram{I}) \leq \smash{\K[w] \cdot \log |\ram{I}|}$.
  \end{enumerate}
\end{enumerate}
For the remainder of the proof, we fix a rational $\gamma < 1$ such that $\alpha < \gamma$, $2\delta < \gamma$ and $\frac{1}{2} < \gamma$.

\subsection{Hierarchical positioning}\label{fix:sec:positioning}

Let us fix $\iota \in \N$ an integer corresponding to the current level of tiles; $(\K[][s],\K[][w])$ two integers, and a sequence of pairs $\smash{(\rspace_{\ell},\rtime_{\ell})_{0 \leq \ell < \kappa-1} \in (\N^2)^{\ast}}$ as given as input to~$\funram{F(e)}$. Notice that the latter determines a value $\kappa \in \N$, \textit{via} its length.

The tileset $\Swang[\iota][S]$ resulting from these parameters will embed some ``next'' steps of substitutions
\[ x^{\kappa'} \mspace{6mu} \substep[\tau][\grid{g}_{\kappa'-1}] \mspace{6mu} x^{\kappa'-1} \mspace{6mu} \substep[\tau][\grid{g}_{\kappa'-2}] \mspace{10mu} \ldots \mspace{10mu} \substep[\tau][\grid{g}_{\kappa+1}] \mspace{6mu} x^{(\kappa+1)} \mspace{6mu} \substep[\tau][\grid{g}_{\kappa}] x^{(\kappa)} \]
for some grids $\grid{g}_{\kappa-1}, \dots, \grid{g}_{\kappa'-1}$; and decide on a level $\iota' \geq \kappa$ for the next level of simulation. To proceed, the tileset $\Swang[\iota][S]$ will start by fixing the grids $\grid{g}_{\kappa-1}, \dots, \grid{g}_{\kappa'-2}$ and the uniform bounds on time $\rtime_{\ell} \geq \rtime(\grid{g}_{\ell})$ and space $\rspace_{\ell} \leq \rspace(\grid{g}_{\ell})$ for these new grids.

\paragraph*{Uniform bounds}

The tilings over $\Swang[\iota][S]$ start by choosing an integer $\kappa' > \kappa$ and draw a uniform family of bounds $(\rspace_{\ell},\rtime_{\ell})_{\kappa -1 \leq \ell < \kappa'-1}$.
To do so, we add a \field{bound field} to the colors of~$\Scolor[\iota]$, which contains a family of bounds
\[ (\rspace_{\ell},\rtime_{\ell})_{\kappa-1 \leq \ell < \kappa'-1} \in (\N^2)^{\ast} \]
such that\\[-1.5\linespacing]
\begin{equation}\label{eq:bound-fields1}\tag{b1}
  \prod_{i=\iota}^{\kappa-1} \rspace_{i} \geq \K[][s] \cdot \pxspace{\kappa}[\gamma]
\end{equation}
and for every $\ell \in \interval{\kappa-1}{\kappa'-2}$,
\begin{equation}\label{eq:bound-fields2}\tag{b2}
  \rtime_{\ell} \leq 2^{\K \cdot \pxspace{\ell}[\delta]}\!
  \quad \text{;} \qquad
  \rspace_{\ell}^{\K \cdot \log \pxspace{\ell+1}} \geq \rtime_{\ell}
  \qquad \text{and} \qquad
  \prod_{i=\kappa}^{\ell-1} \rspace_{i} < \K[][s] \cdot \pxspace{\ell}[\gamma]
\end{equation}
for $\K$ the constant fixed in the~\crtlnameref{fix:sec:housekeeping} paragraph, $\K[][s]$ the constant given as input to $\funram{F(e)}$ and $\pxspace{\ell}$ the integer defined as $\smash{\pxspace{\ell} = \prod_{i=0}^{\ell-1} \rspace_{i}}$. In particular, \eqref{eq:bound-fields1} and \eqref{eq:bound-fields2} ensure that, if the prefix $(\rspace_{\ell},\rtime_{\ell})_{0 \leq \ell < \kappa-1}$ satisfies the \crtlnameref{fix:lem:staircase-lemma}, then so does the larger segment $(\rspace_{\ell},\rtime_{\ell})_{0 \leq \ell < \kappa'-1}$; and that parameters $\kappa$ and $\iota$ are consistent with the values given by said lemma\footnote{Notice that the space bound $\rspace_{\kappa}$ was not given as a parameter for $\Swang[\iota][S]$, and was first introduced in the \field{bound fields} of $\Swang[\iota][S]$. Thus, such consistency checks for the  \crtlnameref{fix:lem:staircase-lemma} must be performed in $\Swang[\iota][S]$, and could not have been done at the previous level of fixed-point simulation.}.

\begin{details}
  For any tile in $\wang{t} \in \Swang[\iota][S]$, the same data ${(\rspace_{\ell},\rtime_{\ell})_{\kappa-1 \leq \ell < \kappa'-1}}$ must appear in the \field{bound fields} of all its facets $\facet{k}{\pm}(\wang{t})$, thus ensuring the uniformity of these bounds across entire configurations drawn by the tiles of $\Swang[\iota][S]$. The same data is also copied to the decoration $\dec(\wang{t})$.
\end{details}

\paragraph*{Grids and positions}

The tilings over $\Swang[\iota][S]$ must then draw the grids $\grid{g}_{\kappa-1}$, \dots, $\grid{g}_{\kappa'-2}$. To do so, we add a \field{position field} to the colors of $\Scolor[\iota]$ which contains a family of the form
\[ (\rect{r}_{\ell},\vec{i}_{\ell})_{\iota \leq \ell < \kappa'-1} \in (\Srect_0,\N^d)^{\ast} \]
for $\rect{r}_{\ell} \in \Srect_0$ rectangles such that $\rspace(\rect{r}_{\ell}) \geq \rspace_{\ell}$ and $\rtime(\rect{r}_{\ell}) \leq \rtime_{\ell}$; and $\vec{i}_{\ell} \in \rect{r}_{\ell}$ positions in $\rect{r}_{\ell}$. The prefix ${(\rect{r}_{\ell},\vec{i}_{\ell})_{\iota \leq \ell < \kappa-1}}$ must equal the eponymous input parameters of $\funram{F(e)}$.  In an ``odometer-like'' fashion, the \field{position fields} in an $\Swang[\iota][S]$-tiling will uniquely determine some grids ${\grid{g}_{\iota},\dots,\grid{g}_{\kappa'-2}}$.

\begin{definition}[Odometer]
  For $(\rect{r}_{\ell})_{\iota \leq \ell < \kappa'-1}$ a family of rectangles from $\Srect_0$, the associated \emph{partial odometer} defines, for every $1 \leq k \leq d$, an addition
  \[ (\vec{i}_{\ell})_{\iota \leq \ell < \kappa'-1} + \basis{k} = (\vec{j}_{\ell})_{\iota \leq \ell < \kappa'-1}
    \qquad \text{for } \mspace{6mu} \vec{j}_{\ell} \in \rect{r}_{\ell} \cup \{\odnone\}, \]
  where \\[-1.5\linespacing]
  \[
    \vec{j}_{\ell} =
    \begin{cases}
      \vec{i}_{\ell} & \text{if there exists } \ell' < \ell \text{ such that } \vec{i}_{\ell'} \notin \facet{k}{+}(\rect{r}_{\ell'}) \\
      \vec{i}_{\ell} + \basis{k} & \text{if for all } \ell' < \ell \text{ we have } \vec{i}_{\ell'} \in \facet{k}{+}(\rect{r}_{\ell'}) \text{, and } \vec{i}_{\ell} + \basis{k} \in \rect{r}_{\ell} \\
      \odnone &  \text{if for all } \ell' < \ell \text{ we have } \vec{i}_{\ell'} \in \facet{k}{+}(\rect{r}_{\ell'}) \text{, but } \vec{i}_{\ell} + \basis{k} \notin \rect{r}_{\ell}.
    \end{cases}
  \]
\end{definition}
Similarly, one can define the subtraction $(\vec{i}_{\ell})_{\iota \leq \ell < \kappa'-1} - \basis{k}$ using facets $\facet{k}{-}(\rect{r}_{\ell})$. These odometers intuitively generalize additions with carries in a $d$-dimensional setting, and having a result $\vec{j}_{\ell}' = \odnone$ implies that an overflow occurred in the rectangle $\rect{r}_{\ell}$.\pagebreak

\begin{details}
  For any tile in $\wang{t} \in \Swang[\iota][S]$, let $(\rspace_{\ell},\rtime_{\ell})_{\iota \leq \ell < \kappa'-1}$ be given by the \field{bound field} of any color of $\wang{t}$. And let $(\rect{r}_{\ell},\vec{i}_{\ell})_{\iota \leq \ell < \kappa'-1}$ be the \field{position field} of the decoration $\dec(\wang{t})$. Then for every direction $1 \leq k \leq d$:
  \begin{itemize}
  \item The \field{position field} of the facet $\facet{k}{-}(\wang{t})$ is the same as the decoration's, \textit{i.e.}~is also $(\rect{r}_{\ell},\vec{i}_{\ell})_{\iota \leq \ell < \kappa'-1}$;
  \item The \field{position field} of the facet $\facet{k}{+}(\wang{t})$ contains $(\rect{r}_{\ell}', \vec{i}_{\ell}')_{\iota \leq \ell < \kappa'-1}$, where, if we denote $(\vec{j}_{\ell})_{\iota \leq \ell < \kappa'-1}$ the result of the partial odometer operation $(\vec{i}_{\ell})_{\iota \leq \ell < \kappa'-1} + \basis{k}$, then:
    \begin{itemize}
    \item If $\vec{j}_{\ell} \in \rect{r}_{\ell}$, then $\rect{r}_{\ell}' = \rect{r}_{\ell}$ and $\vec{i}_{\ell}' = \vec{j}_{\ell}$;
    \item If $\vec{j}_{\ell} = \odnone$, then $\rect{r}'_{\ell}$ is any rectangle such that $\rnorm{\facet{k}{+}(\rect{r}_{\ell})}$ and $\rnorm{\facet{k}{-}(\rect{r}_{\ell}')}$ are equal, and $\vec{i}_{\ell}' = \vec{i}_{\ell} + \basis{k} \bmod \rect{r}_{\ell}$.\qedhere
    \end{itemize}
  \end{itemize}
\end{details}

\begin{claim}[label={pclm:positioning}]
  In any valid tiling of the tileset $\Swang[\iota][S]$, the odometer-like structure of the \field{position fields} of the decorations uniquely determines a finite sequence of nested grids
  \[ \left(\bigcomp_{i=\iota}^{\ell-1} \mspace{4mu} \grid{g}_{i} \mspace{3mu} \right)_{\iota+1 \leq \ell < \kappa'}\]
\end{claim}

\noindent More precisely, the \field{position fields} of such a tiling determine two sequences of grids $(\grid{g}_{\ell})_{\iota \leq \ell < \kappa'-1}$ and $(\pxgrid{g}_{\ell})_{\iota+1 \leq \ell < \kappa'}$, where $\pxgrid{g}_{\ell}$ refers to $\pxgrid{g}_{\ell} = \bigcomp_{i=\iota}^{\ell-1} \grid{g}_{\ell}$. Each grid $\grid{g}_{\ell-1}$ defines rectangles $\rect{r} \subseteq \Z^d$ whose positions $\vec{i} \in \rect{r}$ are understood as positions of level $\ell-1$; while the grids $\pxgrid{g}_{\ell}$ define rectangles $\pxrect{r} \subseteq \Z^d$ whose positions $\vec{i} \in \pxrect{r}$ are understood as position of level $\iota$, \textit{i.e.}~as individual tiles of $\Swang[\iota][S]$.

Notice that, for $\wang{t} \in \Swang[\iota][S]$ a tile in such a valid tiling, the \field{position field} of $\dec(\wang{t})$ contains the tile's position in the corresponding rectangle $\rnorm{\rect{r}_{\ell-1}}$ for $\rect{r}_{\ell-1} \in \grid{g}_{\ell-1}$; but does \emph{not} allow to deduce the exact position of the tile in (nor the size of) the rectangles $\rnorm{\pxrect{r}_{\ell}}$ for the corresponding $\pxrect{r}_{\ell} \in \pxgrid{g}_{\ell}$.

\begin{figure}
  \begin{tikzpicture}[scale=0.2]
    \draw[thin,LevelI,opacity=0.2,dash pattern=on 1pt off .8pt] (0,0) grid (60,30); 
    \begin{scope}[densely dashed,opacity=0.2]
      \foreach \pos in {3,6,8,10,12,16,20,22,26,31,33,36,38,40,42,47,50,55} { 
        \draw[line width=1pt,LevelINext] (\pos,0) -- ++(0,30);
      }
      \foreach \pos in {6,12,20,26,33,38,50} { 
        \draw[line width=1.2pt,dashed, LevelK] (\pos,0) -- ++(0,30);
      }
      \draw[line width=1.4pt,loosely dash dot, LevelKLast] (26,0) -- ++(0,30); 
      \foreach \pos in {2,5,9,12,14,16,20,23,26} {
        \draw[line width=1pt,LevelINext] (0,\pos) -- ++(60,0); 
      }
      \foreach \pos in {5,12,23} {
        \draw[line width=1.2pt,dashed, LevelK] (0,\pos) -- ++(60,0);
      }
      \draw[line width=1.4pt,loosely dash dot, LevelKLast] (0,12) -- ++(60,0);
      \draw[densely dashed,LevelINext] (0,0) rectangle (60,30);
      \draw[thick,dashed,LevelK] (0,0) rectangle (60,30);
      \draw[line width=1.2pt,loosely dash dot,LevelKLast] (0,0) rectangle (60,30);
    \end{scope}
    \draw[LevelI,very thick] (44,17) rectangle ++(1,1);
    \draw[LevelI,thick,opacity=0.6] (42,16) grid ++(5,4);
    \begin{scope}[LevelINext,very thick]
      \foreach \pos in {12,14,16,20,23} {
        \draw (38,\pos) -- ++(12,0);
      }
      \foreach \pos in {38,40,42,47,50} {
        \draw (\pos,12) -- ++(0,11);
      }
    \end{scope}
    \begin{scope}[LevelK,very thick]
      \foreach \pos in {12,23,30} {
        \draw (26,\pos) -- ++(34,0);
      }
      \foreach \pos in {26,33,38,50,60} {
        \draw (\pos,12) -- ++(0,18);
      }
    \end{scope}
    \begin{scope}[LevelKLast,very thick]
      \foreach \pos in {0,12,30} {
        \draw (0,\pos) -- ++(60,0);
      }
      \foreach \pos in {0,26,60} {
        \draw (\pos,0) -- ++(0,30);
      }
      \draw[LevelKNext,very thick] (0,0) rectangle ++(60,30);
    \end{scope}
  \end{tikzpicture}
  \caption{The information in the \field{position field} of $\dec(\wang{t})$, for $\wang{t}$ a tile in $\Swang[\iota][S]$.}
  \subcaption{We draw five levels of grids: individual tiles of $\Swang[\iota][S]$, $\grid{g}_{\iota},\dots,\grid{g}_{\iota+3}$ (with respective colors
    \begin{tikzpicture}[baseline=-0.6ex,scale=0.3]\protect\draw[opacity=.6,densely dashed,LevelI,thin] (0,0) -- ++(1,0);\end{tikzpicture},
    \begin{tikzpicture}[baseline=-0.6ex,scale=0.3]\protect\draw[opacity=.6,thick,densely dashed,LevelINext] (0,0) -- ++(1,0);\end{tikzpicture},
    \begin{tikzpicture}[baseline=-0.6ex,scale=0.3]\protect\draw[opacity=.6,thick,densely dashed,LevelK] (0,0) -- ++(1,0);\end{tikzpicture},
    \begin{tikzpicture}[baseline=-0.6ex,scale=0.3]\protect\draw[opacity=.6,thick,densely dashed,LevelKLast] (0,0) -- ++(1,0);\end{tikzpicture} and
    \begin{tikzpicture}[baseline=-0.6ex,scale=0.3]\protect\draw[opacity=.6,thick,densely dashed,LevelKNext] (0,0) -- ++(1,0);\end{tikzpicture}).
    Given the tile $\wang{t}$ (drawn \begin{tikzpicture}[baseline=0ex,scale=0.2]\protect\draw[thick,LevelI] (0,0) rectangle ++(1,1);\end{tikzpicture}), the \field{position field} of $\dec(\wang{t})$ encodes the rectangles $\rect{r}_{\iota},\dots,\rect{r}_{\iota+3}$ in which $\wang{t}$ appears (whose cells are respectively drawn
    \begin{tikzpicture}[baseline=-0.6ex,scale=0.3]\protect\draw[LevelI,opacity=0.6,ultra thick] (0,0) -- ++(1,0);\end{tikzpicture},
    \begin{tikzpicture}[baseline=-0.6ex,scale=0.3]\protect\draw[LevelINext,ultra thick] (0,0) -- ++(1,0);\end{tikzpicture},
    \begin{tikzpicture}[baseline=-0.6ex,scale=0.3]\protect\draw[LevelK,ultra thick] (0,0) -- ++(1,0);\end{tikzpicture} and
    \begin{tikzpicture}[baseline=-0.6ex,scale=0.3]\protect\draw[LevelKLast,ultra thick] (0,0) -- ++(1,0);\end{tikzpicture}).}
\end{figure}

\begin{details}
 While a single tile $\wang{t}$ cannot determine its exact position in the rectangle $\rnorm{\pxrect{r}_{\ell}}$, the information in its \field{position fields} allows to determine if it appears on an external facet of $\rnorm{\pxrect{r}_{\ell}}$. Indeed, in any valid $\Swang[\iota][S]$-tiling in which $\wang{t}$ appears, denote $(\pxgrid{g}_{\ell})_{\iota +1 \leq \ell < \kappa'}$ the resulting grids and consider a level $\iota+1 \leq \ell < \kappa'$. By definition of $\pxgrid{g}_{\ell}$, there exists some $\pxrect{r} \in \pxgrid{g}_{\ell}$ such that $\wang{t}$ appears at position $\vec{i} \in \rnorm{\pxrect{r}}$. The tile $\wang{t}$ then appears on the facet $\facet{k}{\pm}(\rnorm{\pxrect{r}})$ if and only if, denoting $(\rect{r}_{i},\vec{i}_{i})_{\iota \leq i < \kappa'-1}$ the \field{position field} of $\dec(\wang{t})$, the operation $(\vec{i}_{i})_{\iota \leq i < \ell} \pm \basis{k}$ in the partial odometer of rectangles $(\rect{r}_{i})_{\iota \leq i < \ell}$ returns a full $\odnone$-sequence.
\end{details}

\begin{numerics}
  Let $\wang{t} \in \Swang[\iota][S]$ be a tile, and denote $(\rspace_{\ell},\rtime_{\ell})_{0 \leq \ell < \kappa'-1}$ the sequence obtained by the \field{bound fields} and the parameters of $\funram{F(e)}$ fixed for $\Swang[\iota][S]$. Then the~\crtlnameref{fix:lem:staircase-lemma}~(\Cref{fix:lem:staircase-lemma}) ensures that the \field{bound fields} and the \field{position fields} of $\wang{t}$ have bit length bounded by:
  \[ \sum_{\ell=0}^{\kappa'-2} \log \rspace_{\ell} + \log \rtime_{\ell} \leq \polylog \K[][s] \cdot \pxspace{\iota}[2\delta].\]

  \pagebreak
  We chose to make the \field{position fields} encode positions in this ``odometer-like'' fashion (\textit{i.e.}~relatively to the rectangles $\rnorm{\rect{r}_{\ell-1}}$ from the grid $\grid{g}_{\ell-1}$) instead of the more traditional ``pixel-like'' way (\textit{i.e.}~relatively to the rectangles $\rnorm{\pxrect{r}_{\ell}}$ from the product grid $\pxgrid{g}_{\ell}$) because it allows a better bound $\delta < \frac{1}{2}$ instead of $\delta < \frac{1}{3}$: by the \crtlnameref{fix:lem:staircase-lemma} the latter position scheme would require bit length:
  \[ \sum_{i=\iota}^{\kappa'-2} \log \left(\prod_{j=\iota}^{i-1} \rspace_{j} \cdot\rtime_{j}\right) \leq d \cdot \sum_{i=\iota}^{\kappa'-2} \sum_{j=\iota}^{i-1} \log \rtime_{j} = \bigO(\pxspace{\iota}^{3\delta}). \qedhere \]
\end{numerics}

\paragraph*{Neighbors} Unfortunately, the \field{position fields} in $\Swang[\iota][S]$ do not determine the grid $\grid{g}_{\kappa'-1}$ associated with the last substitution step: indeed, the allowed word length\footnote{Colors of $\Scolor[\iota]$ will not be allowed to grow larger than $\bigO(\pxspace{\iota}[\gamma])$ bits, as they will to be processed by the computations embedded inside the (not yet defined) macro-tiles of level $\iota$.} for the colors of $\Scolor[\iota]$ might be too small to accommodate the grid sizes for $\grid{g}_{\kappa'-1}$ allowed by \Cref{sof:lem:fixpoint}. Thus, instead of making each tile of $\Swang[\iota][S]$ store a rectangle $\rect{r}_{\kappa'-1} \in \Srect_0$ and its position $\vec{i}_{\kappa'-1} \in \rect{r}_{\kappa'-1}$, we implicitly represent the grid $\grid{g}_{\kappa'-1}$ by neighboring sets, \textit{i.e.}~sets indicating at each position $\vec{i}$ which neighboring positions $\vec{i} \pm \basis{k}$ belong to the same rectangle of $\grid{g}_{\kappa'-1}$. Neighboring sets sets will then be synchronized across the next step of fixed-point simulation.

To proceed, we add to the colors of $\Scolor[\iota]$ a \field{neighbor field}
\[ (\nghb_{\ell})_{\iota \leq \ell < \kappa'} \]
where each $\nghb_{\ell}$ is a \emph{neighboring set} $\nghb_{\ell} \subseteq \{\pm \basis{k} : 1 \leq k \leq d\}$. In a valid $\Swang[\iota][S]$-tiling, any tile $\wang{t}$ belongs to a unique rectangle $\rect{r}_{\vec{i}}$ in the grid $\grid{g}_{\ell-1}$: in the \field{neighbor field} of $\dec(\wang{t})$,  a direction $\pm \basis{k}$ is in $\nghb_{\ell}$ if the rectangles $\pxrect{r}_{\vec{i}}$ and $\pxrect{r}_{\vec{i} \pm \basis{k}}$ in the grid $\grid{g}_{\ell-1}$ belong to the same parent rectangle in the grid $\grid{g}_{\ell}$.

\begin{details}
  For $\wang{t} \in \Swang[\iota][S]$. First, the \field{neighbor field} of the decoration $\dec(\wang{t})$ and the negative facets $\facet{k}{-}(\wang{t})$ must first be equal. Furthermore, let $(\rect{r}_{\ell},\vec{i}_{\ell})_{\iota \leq \ell < \kappa'-1}$ denote the \field{position field} of $\dec(\wang{t})$; and for a fixed direction $\basis{k}$ for $1 \leq k \leq d$, denote $\smash{\vec{j}_{k} = (\vec{i}_{\ell'})_{\iota \leq \ell' < \ell} + \basis{k}}$ the result of the partial odometer operations defined by the rectangles $(\rect{r}_{\ell})_{\iota \leq \ell' < \ell}$. We have:
  \begin{itemize}
  \item If $\vec{j}_k = (\odnone,\dots,\odnone)$, then the direction $\basis{k}$ is in the entry $\nghb_{\ell}$ of the \field{neighbor field} of $\dec(\wang{t})$ if and only if the opposite direction $-\basis{k}$ is in the entry $\nghb_{\ell}$ of the \field{neighbor field} of~$\facet{k}{+}(\wang{t})$;
  \item If $\vec{j}_k = (\odnone,\dots,\odnone)$ and $\ell < \kappa' - 1$, then $\nghb_{\ell}$ is easily deduced from the \field{position field}: the direction $\basis{k}$ is in the entry $\nghb_{\ell}$ of the \field{neighbor field} of $\dec(\wang{t})$ if and only if $\vec{i}_{\ell+1} + \basis{k} \in \rect{r}_{\ell+1}$;
  \item If $\vec{j}_k = (\odnone,\dots,\odnone)$ and $\ell = \kappa' - 1$, we do not impose anything for now on the entry $\nghb_{\kappa'-1}$ of the \field{neighbor field} of $\dec(\wang{t})$.
  \end{itemize}

  Otherwise (\textit{i.e.}~$\vec{j}_k \neq (\odnone,\dots,\odnone)$), the \field{neighbor fields} $\nghb_{\ell}$  of level $\ell$ of $\facet{k}{-}(\wang{t})$ and $\facet{k}{+}(\wang{t})$ must be equal.
\end{details}

These conditions ensure in particular that the entries $\nghb_{\ell}$ in the \field{neighbor fields} of adjacent tiles will be constant inside a rectangle $\rect{r}_{\ell-1}$. To ensure that the entries $\nghb_{\kappa'-1}$ implicitly define some grids $\grid{g}_{\kappa'-1}$ in valid $\Swang[\iota][S]$-tilings, they will be synchronized with the \field{neighbor field} of the (yet to be defined) next level of fixed-point simulation~${\iota' \geq \kappa}$: see \Cref{fix:sec:macro-tile:communication-next-level}.

\medskip
\subsection{Macro-tiles}\label{fix:sec:macro-tiles}

As any valid tiling $x \in \smash{\Swang[\iota][S]^{\Z^d}}$ defines a sequence of nested grids $(\pxgrid{g}_{\ell})_{\iota+1 \leq \ell < \kappa'}$ (\Cref{pclm:positioning}), we define the \emph{macro-tiles of level $\ell$} (for $\iota \leq \ell < \kappa'$) as the rectangular patterns $\mwang{t}[\ell] = \rnorm{\restr{x}{\pxrect{r}_{\ell}}}$ for $\pxrect{r}_{\ell} \in \pxgrid{g}_{\ell}$.\footnote{Where the grid $\pxgrid{g}_{\iota} = (\rect{r}_{\vec{i}})_{\vec{i} \in \Z^d}$ of level $\iota$ is the unit cell grid of rectangles $\rect{r}_{\vec{i}} = \{ \vec{i} \}$.} In other words, a macro-tile of level $\ell$ is a valid pattern $\mwang{t}[\ell] \in \smash{\Swang[\iota][S]^{\mspace{2mu}\pxrect{r}_{\ell}}}$ (for some $\pxrect{r}_{\ell} \in \Srect_0$) such that the tiles $\wang{t}$ appearing in $\mwang{t}$ form a single complete rectangle $\rect{r}_{\ell-1}$ (in their decorations' \field{position field}).

\medskip
Let us make a few comments on the geometry of macro-tiles (see \Cref{fix:fig:macro-tile-hierarchy}):
\begin{itemize}
\item Following from the definition of grid products and the induced nested structure of the grids $(\pxgrid{g}_{\ell})_{\iota +1 \leq \ell < \kappa'}$, any macro-tile of level $\ell$ is a union of macro-tiles of level $\ell-1$;
\item Let $\mwang{t}[\ell]$ be a macro-tile of level $\ell$: it uniquely determines the macro-tiles of level $\ell-1$ whose union constitutes $\mwang{t}[\ell]$. Inductively, it thus determines a unique structure of nested macro-tiles of level $\ell'$ for every $\ell' < \ell$.
\end{itemize}

\begin{figure}[ht]
  \centerline{\hspace*{2cm}\begin{tikzpicture}[scale=0.17]
    \begin{scope}[shift={(3.765,-6.7)}]
      \clip (0,0) circle (2.4cm);
      \foreach \x in {-2.830,-2.690,-2.550,-2.430,-2.230,-2.160,-2.080,-1.950,-1.870,-1.760,-1.700,-1.550,-1.410,-1.340,-1.190,-1.090,-1.030,-0.940,-0.810,-0.660,-0.600,-0.480,-0.320,-0.140,-0.010,0.120,0.280,0.370,0.450,0.510,0.680,0.750,0.850,1.010,1.090,1.230,1.330,1.490,1.590,1.650,1.820,2.020,2.130,2.250,2.390,2.450,2.560,2.620,2.710,2.800} {
        \draw[ultra thin,LevelK!35] (\x,-2.4) -- (\x,2.4);
      }
      \foreach \y in {-2.645,-2.545,-2.475,-2.405,-2.225,-2.025,-1.865,-1.755,-1.685,-1.545,-1.375,-1.215,-1.135,-1.015,-0.895,-0.715,-0.615,-0.515,-0.425,-0.255,-0.115,0.015,0.125,0.285,0.425,0.585,0.725,0.895,1.015,1.175,1.265,1.465,1.605,1.775,1.835,1.985,2.125,2.275,2.415,2.475,2.565,2.675} {
        \draw[ultra thin,LevelK!35] (-2.4,\y) -- (2.4,\y);
      }
    \end{scope}
    \foreach \x in {-22.035,-19.385,-17.295,-16.535,-14.865,-13.535,-12.855,-11.585,-9.745,-7.855,-7.015,-5.375,-3.735,-1.065,0.455,1.685,2.675,4.855,6.325,8.655,9.905,11.655,13.295,15.665,16.535,18.515,20.925,21.785} {
      \draw[LevelINext!50,line width=0.6pt] ($(\x,-22.125) - (0,7)$) -- ($(\x,22.125) + (0,7)$);
    }
    \foreach \y in {-26.765,-24.735,-22.925,-22.125,-20.245,-18.255,-15.545,-14.635,-13.335,-11.955,-10.295,-8.455,-7.595,-5.805,-4.715,-3.705,-1.075,0.405,3.035,5.715,6.415,7.425,9.795,11.145,12.155,13.945,15.885,17.875,18.785,19.665,20.915,22.125,23.675,25.605,27.425} {
      \draw[LevelINext!50,line width=0.6pt] ($(-16.535,\y) - (7,0)$) -- ($(16.535,\y) + (7,0)$);
    }
    \foreach \x in {-16.535,16.535} {
      \draw[LevelKLast,very thick] ($(\x,-22.125) - (0,7)$) -- ($(\x,22.125) + (0,7)$);
    }
    \foreach \y in {-22.125,22.125} {
      \draw[LevelKLast,very thick] ($(-16.535,\y) - (7,0)$) -- ($(16.535,\y) + (7,0)$);
    }
    \draw[thick] (3.765,-6.7) circle (2.4cm);
    \draw[dashed,bend left] (4.6,-3.7) to (30,13);
    \begin{scope}[shift={(45,13)},scale=0.8]
      \begin{scope}
        \clip (0,0) circle (23cm);
        \fill[white] (-25,-25) rectangle (25,25);
        \begin{scope}[shift={(5.95,1.85)}]
          \clip (0,0) circle (1.7cm);
          \draw[line width=0.1pt,LevelI!35,step=0.1] (-1.75,-1.75) grid (1.75,1.75); 
        \end{scope}
        \foreach \x in {-28.300,-26.900,-25.500,-24.300,-22.300,-21.600,-20.800,-19.500,-18.700,-17.600,-17.000,-15.500,-14.100,-13.400,-11.900,-10.900,-10.300,-9.400,-8.100,-6.600,-6.000,-4.800,-3.200,-1.400,-0.100,1.200,2.800,3.700,4.500,5.100,6.800,7.500,8.500,10.100,10.900,12.300,13.300,14.900,15.900,16.500,18.200,20.200,21.300,22.500,23.900,24.500,25.600,26.200,27.100,28.000} {
          \draw[LevelK!50] ($(\x,-20.25) - (0,7)$) -- ($(\x,20.25) + (0,7)$);
        }
        \foreach \y in {-26.450,-25.450,-24.750,-24.050,-22.250,-20.250,-18.650,-17.550,-16.850,-15.450,-13.750,-12.150,-11.350,-10.150,-8.950,-7.150,-6.150,-5.150,-4.250,-2.550,-1.150,0.150,1.250,2.850,4.250,5.850,7.250,8.950,10.150,11.750,12.650,14.650,16.050,17.750,18.350,19.850,21.250,22.750,24.150,24.750,25.650,26.750} {
          \draw[LevelK!50] ($(-22.3,\y) - (7,0)$) -- ($(22.3,\y) + (7,0)$);
        }
        \foreach \x in {-20.800,-10.900,10.900} {
          \draw[LevelINext,very thick] ($(\x,-20.25) - (0,7)$) -- ($(\x,20.25) + (0,7)$);
        }
        \foreach \y in {-17.550,-8.950,8.950,19.850} {
          \draw[LevelINext,very thick] ($(-22.3,\y) - (7,0)$) -- ($(22.3,\y) + (7,0)$);
        }
      \end{scope}
      \draw[thick] (5.95,1.85) circle (1.7cm);
      \draw[dashed,bend left] (7.8,0.8) to (20,-20);
      \draw[ultra thick] (0,0) circle (23cm);
      \begin{scope}[shift={(25,-32)}]
        \begin{scope}
          \clip (0,0) circle (23cm);
          \fill[white] (-25,-25) rectangle (25,25);
          \draw[step=0.666666,LevelI!50,very thin] (-25,-25) grid (25,25);
          \foreach \x in {-30.000,-19.333,-11.333,11.333,20.667} {
            \draw[LevelK, very thick] ($(\x,-20.25) - (0,7)$) -- ($(\x,20.25) + (0,7)$);
          }
          \foreach \y in {-27.333,-8.667,8.667,23.333} {
            \draw[LevelK,very thick] ($(-22.3,\y) - (7,0)$) -- ($(22.3,\y) + (7,0)$);
          }
        \end{scope}
        \draw[ultra thick] (0,0) circle (23cm);
      \end{scope}
    \end{scope}
  \end{tikzpicture}}
  \caption{The nested structure of a macro-tile of level $\kappa'-1$.}
  \subcaption{We represent a macro-tile of level $\kappa'-1$
    (as \begin{tikzpicture}[baseline=0.2ex,scale=0.2]
      \draw[LevelKLast,thick] (0,0) rectangle ++(1,1);
    \end{tikzpicture})
    and its sub-macro-tiles of level $\iota'$
    (as \begin{tikzpicture}[baseline=0.2ex,scale=0.2]
      \draw[LevelINext,thick] (0,0) rectangle ++(1,1);
    \end{tikzpicture}). Zooming on the level $\iota'$, we draw the sub-macro-tiles of level $\kappa$
    (as \begin{tikzpicture}[baseline=0.2ex,scale=0.2]
      \draw[LevelK,thick] (0,0) rectangle ++(1,1);
    \end{tikzpicture}); and zooming again, we draw the tiles of $\Swang[\iota][S]$
    (as \begin{tikzpicture}[baseline=0.2ex,scale=0.2]
      \draw[LevelI,thick] (0,0) rectangle ++(1,1);
    \end{tikzpicture})
    composing the level $\kappa$. This figure implicitly assumes that $\iota' > \kappa$.}
  \label{fix:fig:macro-tile-hierarchy}
\end{figure}
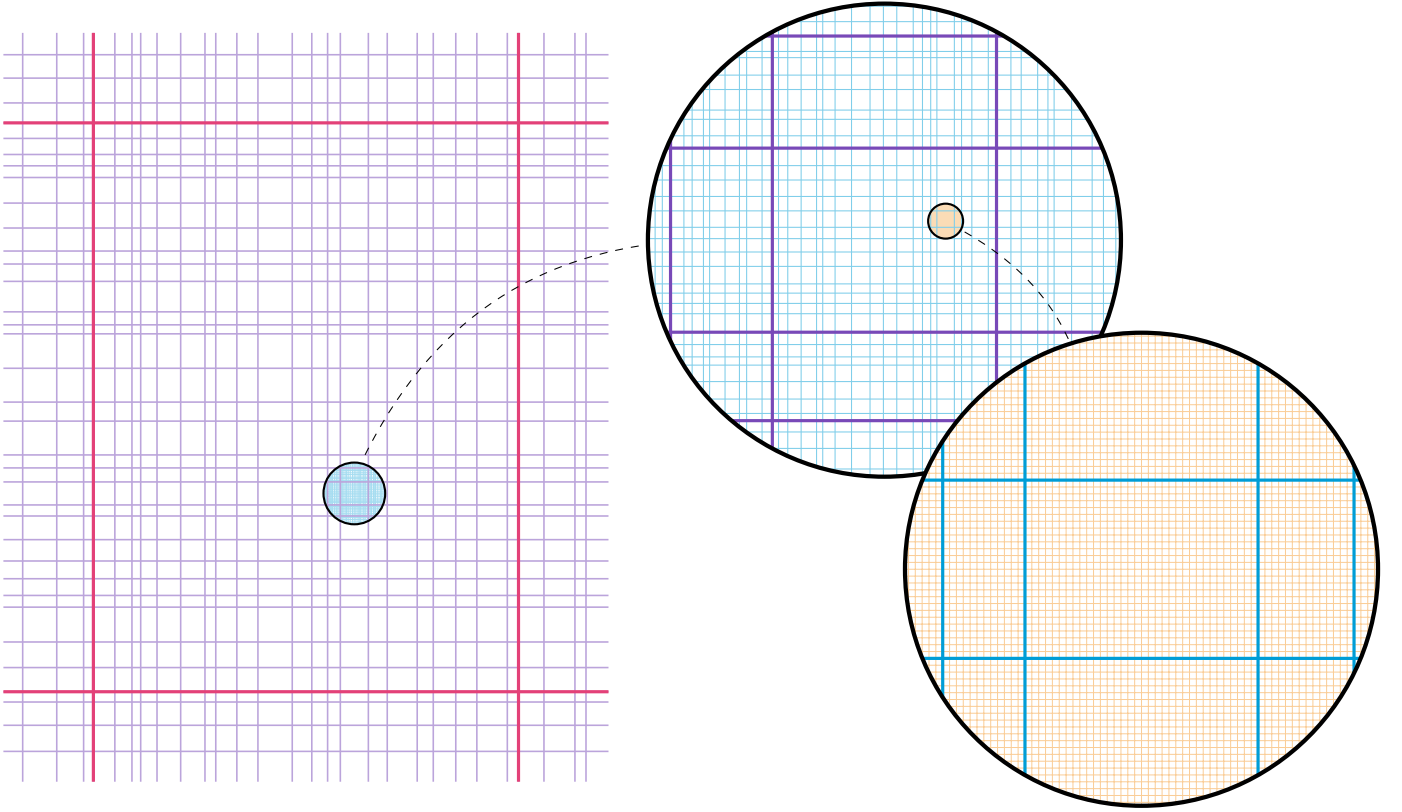

\pagebreak
\subsection{\texorpdfstring{$\bm{\sigma}$}{σ}-macro-tiles (of level \texorpdfstring{$\bm{\iota'}$}{ι'})}\label{fix:sec:sigma-macro-tiles}

Since any tiling of $\Swang[\iota][S]$ defines grids $\pxgrid{g}_{\ell}$ for the next levels of indices $\kappa \leq \ell < \kappa'$, we will use one of these levels to implement the next level of fixed-point simulation. We denote the latter~$\iota'$ by analogy with the notations of the \crtlnameref{fix:lem:staircase-lemma}.

\paragraph*{Level of simulation} We add a \field{p-level field} (for \emph{parent level}) to the colors of $\Scolor[\iota]$, encoding some
\[ \iota' \in \interval{\kappa}{\kappa'-1} \]
that is constant across any valid tiling of $\Scolor[\iota]$. Its role will be to choose the next level of simulation, which could \textit{a priori} be any level in the interval $\interval{\kappa}{\kappa'-1}$.\footnote{The correctness of this choice for $\iota'$ will be checked at the next level of simulation, as a valid \field{bound field} (in $\Swang[\iota'][S]$) must satisfy conditions \eqref{eq:bound-fields1}.}

\begin{details}
For any tile $\wang{t} \in \Swang[\iota][S]$, let $\kappa$ and $\kappa'$ be the two extremal integers indexing the \field{bound field} of its decoration $\dec(\wang{t})$. Denoting ${\iota' \in \interval{\kappa}{\kappa'-1}}$ the \field{p-level field} of $\dec(\wang{t})$, then the \field{p-level fields} of the facets~$\facet{k}{\pm}(\wang{t})$ must all be equal to $\iota'$ (thus ensuring the uniformity of $\iota'$ across entire $\Swang[\iota][S]$-tilings).

Furthermore, in the tiles of $\Swang[\iota][S]$ whose decoration's \field{bound field} determines values $\kappa$ and $\kappa'$, any value $\iota' \in \interval{\kappa}{\kappa'-1}$ may appear in the \field{p-level fields} of the decorations.
\end{details}

\paragraph*{Macro-tiles} The rest of this section focuses on implementing this next level of fixed-point simulation. Inside the macro-tiles of level $\iota'$, we will draw patterns whose role will be to emulate the behavior of Wang tiles, thus properly defining the tiling of the next simulation. For the rest of the proof, we call \emph{$\sigma$-macro-tiles} the macro-tiles whose level $\iota'$ matches their \field{p-level fields}.

\pagebreak
\subsubsection{Computation zone}\label{fix:sec:computation-zone}
The construction of the $\sigma$-macro-tiles begins by designing a way for them to embed arbitrary computations. To do so, we will define a \emph{computation zone} inside each $\sigma$-macro-tile, which will embed the space-time diagram of the processor array from \Cref{calc:lem:pa-ram-simulation}.

\paragraph*{Delimiting the computation zone} For a rectangle $\pxrect{r} = \ibox{n_1,\dots,n_d} \in \Srect_0$, we define the \emph{computation zone} of~$\pxrect{r}$ as the cut $\mathrm{Cut}_{N}(\pxrect{r}) = \ibox{\min(n_1,N),\dots,\min(n_d,N)}$ of bound $N = \pxspace{\iota'}/\pxspace{\iota} = \smash{\prod_{\ell=\iota}^{\iota'-1} \rspace_{\ell}}$ (as read from a tile's \field{bound fields}). This choice of bound follows from:

\begin{claim}
  For any rectangle $\pxrect{r} \in \Srect_0$ and for any $N \in \N$, if $\rspace(\pxrect{r}) \geq N$, then $\rspace(\mathrm{Cut}_{N}(\pxrect{r})) \geq N$.
\end{claim}

\begin{figure}[ht]
  \centerline{\begin{tikzpicture}[scale=0.2]
    \draw[LevelI!25,step=0.4] (0,0) grid (34,20);
    \node[anchor=east,RainbowG!50!RainbowA] (I) at (-2,10) {Input facet};
    \draw[RainbowG!50!RainbowA,bend left,->] (I) to (-0.2,16);
    \fill[opacity=0.4,RainbowG!50!RainbowA] (0,0) rectangle ++(0.4,20);
    \fill[opacity=0.4,RainbowG!50!RainbowA] (24.4,0) rectangle ++(0.4,20);
    \node[anchor=west,RainbowG!50!RainbowA] (O) at (36,10) {Output facet};
    \draw[RainbowG!50!RainbowA,bend left,->] (O) to (25,4);
    \draw[LevelINext,pattern={north east lines},pattern color=LevelINext,opacity=0.5] (0,0) rectangle ++(24.8,20);
    \draw[very thick,LevelINext] (0,0) rectangle (34,20);
  \end{tikzpicture}}
  \caption{The computation zone
    (in \begin{tikzpicture}[baseline=0.2ex,scale=0.25]
      \protect\draw[LevelINext,pattern={north east lines},pattern color=LevelINext] (0,0) rectangle ++(1,1);
    \end{tikzpicture})
    of a macro-tile $\mwang{t}$ of level $\iota'$
    (in \begin{tikzpicture}[baseline=0.2ex,scale=0.25]
      \protect\draw[LevelINext,thick] (0,0) rectangle ++(1,1);
    \end{tikzpicture}).}
  \subcaption{We draw the tiles of level $\iota$ composing $\mwang{t}$ as
    \begin{tikzpicture}[baseline=0.3ex,scale=0.2]
      \draw[LevelI,thick] (0,0) rectangle ++(1,1);
    \end{tikzpicture}.
    The computation zone is a rectangle $\compzone = \mathrm{Cut}_{N}(\pxrect{r})$ that forms a subset of $\pxrect{r} = \dom(\mwang{t})$. The area of its smallest facet $\sfacet(\compzone)$ is larger than $K \cdot \pxspace{\iota'}[\gamma]$, but $\compzone$ is small enough for the positions $\vec{i} \in \compzone$ to fit on $\bigO(\log \pxspace{\iota'})$ bits. The computation zone always contains the position $\vec{0} \in \dom(\wang{t})$, and may or may not span the whole domain $\pxrect{r}$.}
  \label{fix:fig:computation-zone}
\end{figure}

\noindent To mark the positions in the computation zone of a $\sigma$-macro-tile $\mwang{t}$ of domain $\pxrect{r} = \dom(\mwang{t})$, we introduce a \field{compzone field} to the colors of $\Scolor[\iota]$ based on the positioning tileset $\Swang[\Srect]$ from \Cref{pos:sec:rect}. It will contain records of the form
\[ (\compzone, \vec{i}) \in \Srect_0 \times \N^d \]
for $\compzone = \mathrm{Cut}_{\pxspace{\iota'}/\pxspace{\iota}}(\pxrect{r})$ the computation zone of $\pxrect{r}$ and $\vec{i} \in \compzone$ a position in $\compzone$. Intuitively, the \field{compzone fields} encode the positions of the $\Swang[\iota][S]$-tiles (\textit{i.e.}~of level $\iota$) appearing inside the computation zone of $\mwang{t}$ and its geometry.

\begin{details}
  For any tile $\wang{t} \in \Swang[\iota][S]$, denote $(\rect{r}_{\ell},\vec{i}_{\ell})_{\iota \leq \ell < \kappa'-1}$ and $(\rspace_{\ell},\rtime_{\ell})_{\iota \leq \ell < \kappa'-1}$ the content of the \field{position} and \field{bound fields} of $\dec(\wang{t})$. The \field{compzone fields} of $\wang{t}$ can either contain some pair $(\compzone,\vec{i}) \in \Srect_0 \times \N^d$, or be left entirely blank.

  If the \field{compzone field} of $\dec(\wang{t})$ is blank, then so are the \field{compzone fields} of all facets $\facet{k}{\pm}(\wang{t})$. Otherwise, the \field{compzone field} of $\dec(\wang{t})$ is a pair $(\compzone,\vec{i}) \in \Srect_0 \times \Z^d$ and:
  \begin{itemize}
  \item Denoting $\compzone = \ibox{n_1,\dots,n_d}$, each $n_k$ satisfies $n_k \leq \prod_{\ell=\iota}^{\iota'-1} \rspace_{\ell}$; and $\vec{i}$ satisfies $\vec{i} \in \compzone$;
  \item Additionally, if the \field{position field} of $\dec(\wang{t})$ satisfies $\vec{i}_{\ell} = \vec{0}$ for all $\iota \leq \ell < \iota'$ (\textit{i.e.}~the tile $\wang{t}$ appears at position $\vec{0}$ in $\sigma$-macro tiles), then in its \field{compzone field} the position $\vec{i} \in \compzone$ must actually be $\vec{i} = \vec{0}$;
  \end{itemize}

  Furthermore, for $(\compzone,\vec{i})$ the value of the \field{compzone field} of $\dec(\wang{t})$ and $(\rect{r}_{\ell},\vec{i}_{\ell})_{\iota \leq \ell < \iota'}$ read from the \field{position field}, the facets $\facet{k}{\pm}(\wang{t})$ satisfy the following:
  \begin{itemize}
  \item If the partial odometer given by the rectangles $\smash{(\rect{r}_{\ell})_{\iota \leq \ell < \iota'}}$ defines an operation $\smash{(\vec{i}_{\ell})_{\iota \leq \ell < \iota'} \pm \basis{k} = (\odnone,\dots,\odnone)}$ (\textit{i.e.}~the tile $\wang{t}$ appears on the facet $\facet{k}{\pm}$ of $\sigma$-macro-tiles), then the \field{compzone field} of $\facet{k}{\pm}(\wang{t})$ is left blank; in which case, $\vec{i}$ must satisfy $\vec{i} \in \facet{k}{\pm}(\compzone)$;
  \item Otherwise, the \field{compzone field} of the facet $\facet{k}{-}(\wang{t})$ is the same as the decoration's, \textit{i.e.}~$(\compzone,\vec{i})$;
  \item And the \field{compzone field} of the facet $\facet{k}{+}(\wang{t})$ contains $(\compzone,\vec{i}+\basis{k})$ if $\vec{i} \notin \facet{k}{+}(\compzone)$ is not on the border of~$\compzone$; and is left blank otherwise.\qedhere
  \end{itemize}
\end{details}

\enlargethispage{\baselineskip}
\begin{claim}
  Let $\mwang{t}$ be a $\sigma$-macro-tile of domain $\pxrect{r}$. Then for any position $\vec{i} \in \pxrect{r}$, the \field{compzone field} of $\dec(\mwang{t}_{\vec{i}})$ is non-blank if and only if $\vec{i}$ is a position in the computation zone $\compzone = \mathrm{Cut}_{\pxspace{\iota'}/\pxspace{\iota}}(\pxrect{r})$; in which case, said \field{compzone field} contains exactly the pair~$(\compzone,\vec{i})$.
\end{claim}

\begin{numerics}
  In a valid $\Swang[\iota][S]$-tiling, the volume of a $\sigma$-macro-tile of domain $\rnorm{\pxrect{r}} \in \Srect_0$ may be larger than $\smash{\prod_{\ell=\iota}^{\iota'-1} \rspace_{\ell} \cdot \rtime_{\ell}}$. The computation zone restricts the size of the computations to a cut of $\rnorm{\pxrect{r}}$ whose volume is bounded by $\smash{\prod_{\ell=\iota+1}^{\iota'} (\rspace_{\ell})^d}$, which limits the bit length of the \field{compzone fields} to $\polylog \K[][s] \cdot \bigO(\pxspace{\iota}[\delta])$.
\end{numerics}

\pagebreak
\paragraph*{Embedding computations} We now use this computation zone to embed the space-time diagram of the RAM-simulating processor array from \Cref{calc:lem:pa-ram-simulation}. Recall that it defines a program $e_{\mathcal{U}} \in \{0,1\}^{\ast}$ which, when used on a grid of processors of domain $\ibox{n_1,\dots,n_{d-1}}$, simulates $\smash{\prod_{k=1}^{d-1} n_k}$ steps of $\log$-RAM computations in time $\smash{\bigO(\sum_{k=1}^{d-1} n_k)}$.

\medskip
To proceed, let $\compzone = \ibox{n_1,\dots,n_d}$ be the computation zone of a $\sigma$-macro-tile. Denoting $\sfacet(\compzone)$ a smallest facet $\facet{k}{-}(\compzone)$ of minimal index $1 \leq k \leq d$, and denoting by $h \in \interval{1}{d}$ the index of this facet, then $\rtime(\compzone) = n_h$. For notational convenience, we will assume that $h=d$; but the general case follows by permuting indices $1 \leq k \leq d$ in the proof below.

We use the facet $\sfacet(\compzone) = \ibox{n_1,\dots,n_{d-1}} \times \{0\}$ to draw the processor array $(\ibox{n_1,\dots,n_{d-1}}, e_{\mathcal{U}})$ (the ``space'' of the computation); and use the remaining direction $\basis{d}$ to draw successive states of the array (the ``time'' of the computation\footnote{Hence naming \emph{time} and \emph{space} the longest edge length $\rtime(\rect{r})$ and smallest facet area $\rspace(\rect{r})$ of a rectangle $\rect{r}$.}). To do so, we introduce a \field{processor field} to the colors of $\Scolor[\iota]$. On a tile's facets $\facet{d}{\pm}(t)$, it will encode a tuple of the form
\[ (\ram{PC}, (\ram{var}_{i})_{i \in V}) \in \N \times (\{0,1\}^{\ast})^{I}\]
where $\ram{PC}$ is a program counter indexing an instruction of $e_{\mathcal{U}}$, and $I$ is a finite fixed set indexing the variables defined in $e_{\mathcal{U}}$; and on a tile's facets $\facet{k}{\pm}(t)$ (for $k \neq d$), it will encode two finite lists of adjacent processor communications of the form:
\[ ((\ram{COMM}^{\ram{in}}_{\,i})_{i \in I_1}, (\ram{COMM}^{\ram{out}}_{\,i})_{i \in I_2}) \in (\{0,1\}^{\ast})^{\ast} \times (\{0,1\}^{\ast})^{\ast}.\]

In the direction $\basis{d}$, colors $(\ram{PC}, (\ram{var}_{i})_{i \in V})$ will denote consecutive states of a processor during a step of computation; and along the other directions, colors $((\ram{COMM}^{\ram{in}}_{\,i})_{i \in I_1}, (\ram{COMM}^{\ram{out}}_{\,i})_{i \in I_2})$ will draw the communications between adjacent processors. Globally, the \field{processor fields} in the computation zone $\compzone$ will thus draw a space-time diagram of the processor array $\ibox{n_1,\dots,n_{d-1}}$ for $n_d-2$ steps of computation.

\begin{remark}[Linear acceleration]
  The processor program $e_{\mathcal{U}}$ from \Cref{calc:lem:pa-ram-simulation} terminates in time $\bigO(n_1 + \dots + n_{d-1})$ on processor arrays of size $\rect{r} = \ibox{n_1,\dots,n_{d-1}}$. Thus, there exists a constant $\K[][sim] \in \N$ such that $e_{\mathcal{U}}$ terminates in time $\K[][sim] \cdot \max_{1 \leq k \leq d-1} n_k$.

  Unfortunately, we only embed space-time diagrams of the processor array $(\ibox{n_1,\dots,n_{d-1}},e_{\mathcal{U}})$ over $n_d = \rtime(\compzone) = \max_{1 \leq k \leq d} n_k$ steps of computations. To allow enough time for the program $e_{\mathcal{U}}$ to terminate all its valid runs, we actually make each tile in the computation zone embed several (here, $2\K[][sim]$) computation steps of a processor (for a global running time $2\K[][sim] \cdot (\rtime(\compzone)-2)$).
\end{remark}

\begin{details}
  For any tile $\wang{t} \in \Swang[\iota][S]$, the \field{processor fields} of $\wang{t}$ are all blank if the \field{compzone field} of $\dec(\wang{t})$ is blank. Otherwise, let $(\compzone,\vec{i})$ refer to the \field{compzone field} of $\dec(\wang{t})$.

  \medskip
  If $\vec{i}$ is a position in an external facet $\facet{k}{\pm}(\compzone)$, then the \field{processor field} of the corresponding $\facet{k}{\pm}(\wang{t})$ must be blank. Apart from this constraint, we now consider several cases:

  \vspace*{-0.7em}
  \paragraph*{Initializing the computations} Assume that $\vec{i}$ is part of the smallest facet $\sfacet(\compzone) = \facet{d}{-}(\compzone)$:
  \begin{itemize}
  \item The \field{processor field} of the facet
    $\facet{d}{+}(\wang{t})$ is some $(\ram{PC}, (\ram{var}_{i})_{i \in V}) \in \N \times (\{0,1\}^{\ast})^{I}$ for $\ram{PC} = 0$ (the entry point of the program $e_{\mathcal{U}}$). Furthermore, the values of the variables $(\ram{var}_{i})_{i \in V}$ must satisfy the following conditions:
    \begin{itemize}
    \item The $\ram{pos}$ variable must be set to $\ram{pos} =
      (\rect{r},\vec{j})$ where $\rect{r} = \ibox{n_1,\dots,n_{d-1}}$ and
      $\vec{j} = (i_1,\dots,i_{d-1})$ (where $\ibox{n_1,\dots,n_d} = \compzone$ and $(i_1,\dots,i_{d-1},0) = \vec{i}$);
    \item The $\ram{word}$ variable must be set to $\ram{word} = \K[][w] \cdot \log \rspace(\compzone)$, where $\K[][w]$ is the constant appearing in the parameters $(\K[][s],\K[][w])$ of the global function $\funram{F(e)}$;
    \item The $\ram{timeout}$ variable must be set to $\ram{timeout} = \rspace(\compzone)$;
    \item The $\ram{program}$ variable must be set to $\ram{program} = e$, where $e \in \{0,1\}^{\ast}$ is the eponymous argument given to the algorithm $F$ (recall that, given a program $e \in \{0,1\}^{\ast}$, this proof builds a program $F(e)$ that will recognize the simulating tilesets $\Swang[\iota][S]$);
    \item The $\ram{input}$ variables can be set to any $\ram{input} = (i,j,b)$, for a bit $b \in \{0,1\}$ and an index\footnote{As the function $\funram{F(e)}$ will operate on at most $2d+5$ arguments, we restrict $i$ to range in $\interval{0}{2d+4}$.} $(i,j) \in \interval{0}{2d+4} \times \interval{0}{\rvol{\compzone}}$. Restrictions on $\ram{input}$ variables will be added later;
    \item Any other variable $\ram{var}_{i}$ of the program $e_{\mathcal{U}}$ is assumed to be the empty word $\ram{var}_{i} = \varepsilon$.
    \end{itemize}
  \item The \field{processor field} of any other facet $\facet{k}{\pm}(\wang{t})$ ($k \neq d$) is blank;
  \end{itemize}

  \pagebreak
  \paragraph*{Computations} Assume that $\vec{i}$ is not part of $\facet{d}{-}(\compzone)$ nor its opposite facet $\facet{d}{+}(\compzone)$. Then:
  \begin{itemize}
  \item The \field{processor fields} of the facets $\facet{d}{-}(\wang{t})$ and $\facet{d}{+}(\wang{t})$ must respectively be some $(\ram{PC}, (\ram{var}_{i})_{i \in V})$ and $(\ram{PC}', (\ram{var}_{i}')_{i \in V})$ in $\N \times (\{0,1\}^{\ast})^{I}$;
  \item The \field{processor fields} of the other facets $\facet{k}{\pm}(\wang{t})$ ($k \neq d$) must contain two lists of length $2\K[][sim]$ whose elements $\ram{COMM}$ are either blank, or encode a tuple $(i,w)$ for $i \in V$ the index of a variable and $w \in \{0,1\}^{\ast}$ a value.\\
    On facets $\facet{k}{-}(\wang{t})$, the first list (resp.~second) is referred to as the $\ram{in}$-list (resp.~$\ram{out}$-list), and its elements are denoted $\ram{COMM}^{\ram{in}}_{\,n}$ (resp.~$\ram{COMM}^{\ram{out}}_{\,n}$) for $0 \leq n < 2\K[][sim]$; conversely, on facets $\facet{k}{+}(\wang{t})$, the first list and the second lists are respectively referred to as the $\ram{out}$-list and the $\ram{in}$-list;\footnote{This is a form of twisted pair communication!}
  \end{itemize}
  and there exists some computation steps $(\ram{PC}^{(n)}, (\ram{var}_{i}^{(n)})_{i \in V})_{0 \leq n \leq 2\K[][sim]}$ such that $(\ram{PC}^{(0)}, (\ram{var}_{i}^{(0)})_{i \in V}) = (\ram{PC}, (\ram{var}_{i})_{i \in V})$ and $(\ram{PC}^{(2\K[][sim])}, (\ram{var}_{i}^{(2\K[][sim])})_{i \in V}) = (\ram{PC}', (\ram{var}_{i}')_{i \in V})$ that satisfies, for all $0 \leq n < 2\K[][sim]$:
  \begin{itemize}
  \item (Incoming communications) If $\ram{PC}^{(n)}$ points in the program $e_{\mathcal{U}}$ to a $\ram{COMM}$ instruction reading the variable of index $i \in V$ of the processor in direction $\pm e_k$, then the element $\ram{COMM}^{\ram{in}}_{\,n}$ from the $\ram{in}$-list in the \field{processor field} of the corresponding\footnote{Notice that, in the general case where $h$ can be any direction in the interval $\interval{1}{d}$, this requires a re-indexing of the facets between the $d$-dimensional Wang tiles and the $(d-1)$-dimensional processor array.} $\facet{k}{\pm}(\wang{t})$ must satisfy $\ram{COMM}^{\ram{in}}_{\,n} =  (i,w)$ for some $w \in \{0,1\}^{\ast}$; otherwise, $\ram{COMM}^{\ram{in}}_{\,n}$ must be blank;
  \item (Outgoing communications) On any facet $\facet{k}{\pm}(\wang{t})$ for $k \neq d$, if the entry $\ram{COMM}^{\ram{out}}_{\,n}$ in the $\ram{out}$-list of the \field{processor field} is non-blank and contains some $(i,w)$ for $i \in V$ and $w \in \{0,1\}^{\ast}$, then $w$ must satisfy $\smash{w = \ram{var}_{i}^{(n)}}$.
  \item (Computation step) $\ram{PC}^{(n+1)}$ and $(\ram{var}_{i}^{(n+1)})_{i \in V}$ must be consistent with a valid step of computation running the instruction of index $\smash{\ram{PC}^{(n)}}$ in $e_{\mathcal{U}}$ on the values of variables $(\smash{\ram{var}_{i}^{(n)}})_{i \in V}$ (and, in the case of a $\ram{COMM}$ instruction, the value $w$ from the corresponding $\ram{COMM}^{\ram{in}}_{\,n}$ entry).
  \end{itemize}
  Finally, if the processor halts at some step $0 \leq n < 2\K[][sim]$, then $\ram{PC}^{(n+1)}$ and $(\ram{var}_{i}^{(n+1)})_{i \in V}$ are just copies of $\smash{\ram{PC}^{(n)}}$ and $(\smash{\ram{var}_{i}^{(n)}})_{i \in V}$.

  \paragraph*{End of the computation} Assume that $\vec{i}$ is part of the opposite facet $\facet{d}{+}(\compzone)$. Then:
  \begin{itemize}
  \item The \field{processor field} of the facet $\facet{d}{-}(\wang{t})$ must be some $(\ram{PC}, (\ram{var}_{i})_{i \in V}) \in \N \times (\{0,1\}^{\ast})^{I}$ such that $\ram{PC}$ denotes an accepting state of completed computation in $e_{\mathcal{U}}$;
  \item The \field{processor field} of any other facet $\facet{k}{\pm}(\wang{t})$ is blank.
  \end{itemize}

  Furthermore, all valid possibilities $(\ram{PC}, (\ram{var}_{i})_{i \in V})$ and $(\ram{PC}', (\ram{var}_{i}')_{i \in V})$ in $\N \times (\{0,1\}^{\ast})^{I}$ for the \field{processor fields} of the facets $\facet{d}{-}(\wang{t})$ and $\facet{d}{+}(\wang{t})$ appear when ranging across the tiles $\Swang[\iota][S]$ with fixed \field{compzone field} $(\compzone,\vec{i})$ (with $\vec{i} \notin \facet{k}{\pm}(\compzone)$) in their decorations, thus ensuring that all valid computations can actually be drawn in the computation zones of a $\sigma$-macro-tile.
\end{details}

In a $\sigma$-macro-tile $\mwang{t}$ of computation zone $\compzone = \ibox{n_1,\dots,n_d}$, we respectively call \emph{input facet} and \emph{output facet} the facets $\sfacet(\compzone)$ and its opposite, as their \field{processor fields} draw the first and the last step of the embedded computation (\textit{c.f.}~\Cref{fix:fig:computation-zone}). If we still assume that ${\sfacet(\compzone) = \facet{d}{-}(\compzone)}$ for notational convenience, the \emph{input facet} of $\compzone$ is $\ibox{n_1,\dots,n_{d-1}} \times \{0\}$; and the \emph{output facet} is $\ibox{n_1,\dots,n_{d-1}} \times \{n_d-1\}$. We summarize the previous paragraphs in the following claim:

\begin{claim}[label={pclm:processor-array-macro-tile}]
  Let $\mwang{t}$ be a $\sigma$-macro-tile of computation zone $\compzone = \ibox{n_1,\dots,n_d}$ and smallest facet ${\sfacet(\compzone) = \facet{d}{-}(\compzone)}$.\footnote{Other cases $\sfacet(\compzone) = \facet{h}{-}(\compzone)$ for $h \neq d$ follow from a permutation of the vectors $\basis{k}$.} Then for all $0 \leq n < n_d-1$, the \field{processor fields} of the facets $\facet{d}{+}(\wang{t})$ for tiles $\wang{t}$ at positions $\ibox{n_1,\dots,n_{d-1}} \times \{n\}$ draw a state of the processor array $(\ibox{n_1,\dots,n_{d-1}},e_{\mathcal{U}})$ after $2\K[][sim] \cdot n$ steps of computations.
\end{claim}

\enlargethispage{\baselineskip}
\paragraph*{Folding the arguments of the computation} The previous section embeds arbitrary computations of the $\log$-RAM program $e \in \{0,1\}^{\ast}$ (given to $F$ to define $F(e)$) in the form of space-time diagrams of the simulating processor array from \Cref{calc:lem:pa-ram-simulation}.

Since said simulation from \Cref{calc:lem:pa-ram-simulation} requires its input and output bits $\ram{I}[i][j]$ and $\ram{O}[i][j]$ to be folded (in lexicographic order) along the boustrophedon indexing of the processors, we introduce an \field{arg field} to the colors of $\Scolor[\iota][S]$, which contains a tuple of the form
\[ (i,j) \in \N^2. \]
It will force the input and output bits (appearing in the \field{processor fields}) of adjacent tiles from the input and output facets to follow their respective boustrophedon indexing.

\begin{details}
  For any tile $\wang{t} \in \Swang[\iota][S]$, the \field{arg fields} of $\wang{t}$ are all blank if the \field{compzone field} of $\dec(\wang{t})$ is blank. Otherwise, let $(\compzone,\vec{i})$ refer to the \field{compzone field} of $\dec(\wang{t})$, and denote $\facet{h}{-}(\compzone) = \sfacet(\compzone)$ the input facet of $\compzone$.

  \medskip
  If $\vec{i} \notin \facet{h}{-}(\compzone) \cup \facet{h}{+}(\compzone)$ does not appear on the input/output facet of $\compzone$, then the \field{arg fields} of $\wang{t}$ are blank.

  \medskip
  Otherwise, assume that $\vec{i} \in \facet{h}{-}(\compzone)$ belongs to the input facet of $\compzone$:
  \begin{itemize}
  \item If $\vec{i} = \vec{0}$, we set the \field{arg field} of $\dec(\wang{t})$ to be $(i,0)$ for some $i \in \N$, or blank.
  \item Otherwise, there exists some facet $\facet{p}{\pm}(\wang{t})$ pointing towards the predecessor of $\vec{i}$ along the boustrophedon indexing of the input facet $\sfacet(\compzone)$. Then the \field{arg field} of $\facet{p}{\pm}(\wang{t})$ and the \field{arg field} of the decoration $\dec(\wang{t})$ should be equal.
  \item In any case, the $\ram{input}$ variable (from the \field{processor field} of $\facet{h}{+}(\wang{t})$) contains a non-blank $(i,j,b)$ for some $(i,j) \in \N^2$ and bit $b \in \{0,1\}$ if and only if the \field{arg field} of $\dec(\wang{t})$ is not blank and contains the same pair $(i,j)$.
  \end{itemize}
  Furthermore, if $\vec{i}$ is not the last position in the boustrophedon indexing of $\sfacet(\compzone)$, there exists some facet $\facet{s}{\pm}(\wang{t})$ pointing towards the successor of $\vec{i}$ in $\sfacet(\compzone)$. Then:
  \begin{itemize}
  \item If the \field{arg field} of $\dec(\wang{t})$ is some $(i,j) \in \N^2$, then the \field{arg field} towards the successor $\facet{s}{\pm}(\wang{t})$ can either be $(i,j+1)$, or $(i+1,0)$, or blank;
  \item If the \field{arg field} of $\dec(\wang{t})$ is blank, then the \field{arg field} of $\facet{s}{\pm}(\wang{t})$ is also blank.
  \end{itemize}
  Finally, all other unmentioned facets $\facet{k}{\pm}(\wang{t})$ must have a blank \field{arg field}.

  \medskip
  We proceed similarly with tiles $\vec{i} \in \facet{h}{+}(\compzone)$ in the output facet of $\compzone$, but synchronize with the $\ram{output}$ variable from the \field{processor field} of the facet $\facet{h}{-}(\wang{t})$.
\end{details}

We thus claim that $\ram{input}$ variables on the input facet are folded as in \Cref{calc:lem:pa-ram-simulation}:
\begin{claim}
  In any $\sigma$-macro-tile $\mwang{t}$, let $\rect{r} = \facet{h}{-}(\compzone)$ be the input facet of the computation zone~$\compzone$. Then the $\ram{input}$ variables of the \field{processor fields} of $\facet{h}{+}(\mwang{t}_{\vec{i}})$ for $\vec{i} \in \rect{r}$ form the folded pattern $\mathrm{fold}_{\rect{r}}(\ram{I}) \in (\N^2 \times \{0,1\})^{\rect{r}}$ of some input array $\ram{I} \in \{0,1\}^{\ast\ast}$.\\
  A similar claim holds for the $\ram{output}$ variables on the output facet of $\mwang{t}$.
\end{claim}

\begin{numerics}
  Let $\mwang{t}$ be a $\sigma$-macro-tile of computation zone $\compzone$, and let $\ram{I} \in \{0,1\}^{\ast\ast}$ be the folded input array appearing on the input facet of $\mwang{t}$. Then for any tile $\wang{t}$ appearing in $\mwang{t}$, the bit length of:
  \begin{itemize}
  \item The \field{processor fields} of $\wang{t}$ are bounded by $\bigO(\K[][w] \cdot \log \rvol{\compzone}) = \polylog \K[][s] \cdot \K[][w] \cdot \bigO(\pxspace{\iota}[\delta])$;
  \item The \field{arg fields} of $\wang{t}$ are bounded by $\log \rvol{\compzone} = \polylog \K[][s] \cdot \bigO(\pxspace{\iota}[\delta])$;
  \end{itemize}
  where $\log \rvol{\compzone} = \bigO(\log \pxspace{\iota'}) = \polylog \K[][s] \cdot \bigO(\pxspace{\iota}[\delta])$ follows from the proof of the \crtlnameref{fix:lem:staircase-lemma}.
\end{numerics}

\enlargethispage{2.4\baselineskip}
\subsubsection{Macro-colors}

In order to make $\sigma$-macro-tiles emulate the behavior of Wang tiles, we now implement in the tiles of $\Swang[\iota][S]$ a way for $\sigma$-macro-tiles to emulate colored facets.

Consider a $\sigma$-macro-tile $\mwang{t}$ of domain $\pxrect{r}$. To emulate colored facets with $\mwang{t}$, we make the individual tiles appearing on the border of $\mwang{t}$ bear a single bit facing towards the exterior of~$\mwang{t}$: along a facet $\facet{k}{\pm}(\pxrect{r})$, the concatenation of those bits (folded along a specific path) will form a binary string $c \in \{0,1\}^{\ast}$ called a \emph{macro-color}. Since these macro-colors are outward-facing, adjacent $\sigma$-macro-tiles will share the same macro-color along their adjacent facets; thus ensuring that tilings of $\sigma$-macro-tiles really simulate valid Wang tilings.

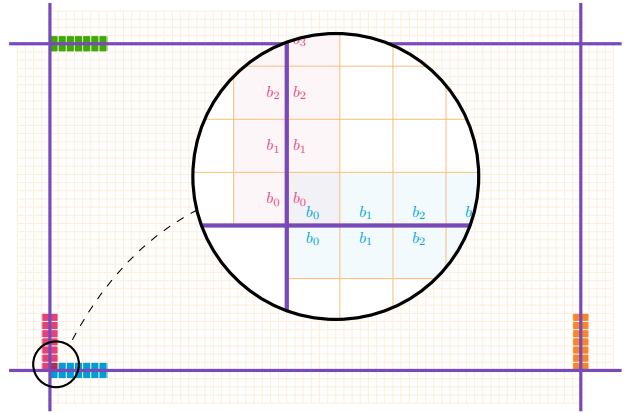
\begin{figure}[ht]
  \begin{tikzpicture}[scale=0.27]
    \dimendef\mylinewidth=0
    \pgfmathsetlength{\mylinewidth}{0.54cm}
    \fill[RainbowG!50!RainbowA] (-0.4,0) rectangle ++(0.8,2.8);
    \fill[RainbowB!70!RainbowC] (25.6,0) rectangle ++(0.8,2.8);
    \fill[RainbowE] (0,-0.4) rectangle ++(2.8,0.8);
    \fill[RainbowG!50!RainbowA!75!black] (0,0) rectangle ++(0.4,0.4); 
    \fill[RainbowD] (0,15.6) rectangle ++(2.8,0.8);
    \draw[LevelI!15,step=0.4] (0,-1.6) grid ++(26,1.6);
    \draw[LevelI!15,step=0.4] (0,16) grid ++(26,1.6);
    \draw[LevelI!15,step=0.4] (-1.6,0) grid ++(1.6,16);
    \draw[LevelI!15,step=0.4] (26,0) grid ++(1.6,16);
    \draw[LevelI!15,step=0.4] (0,0) grid (26,16);
    \foreach \x in {0,26} {
      \draw[LevelINext,very thick] (\x,-2) -- ++($(0,16) + (0,4)$);
    }
    \foreach \y in {0,16} {
      \draw[LevelINext,very thick] (-2,\y) -- ++($(26,0) + (4,0)$);
    }
    \draw[thick] (0.3,0.3) circle (1.12cm);
    \draw[dashed,bend left] (1.1,1.5) to (14,9.5);
    \begin{scope}[shift={(14,9.5)}]
      \begin{scope}
        \clip (0,0) circle (7cm);
        \fill[white] (-7,-7) rectangle (7,7);
        \begin{scope}[shift={(-2.4,-2.4)},scale=2.6]
          \fill[RainbowG!50!RainbowA,opacity=0.05] (-1,0) rectangle ++(2,5);
          \fill[RainbowE,opacity=0.05] (0,-1) rectangle ++(5,2);
          \foreach \i in {0,1,2,3,4} {
            \draw[RainbowG!50!RainbowA] (-0.25,\i.5) node {\scalebox{0.6}{$b_{\i}$}};
            \draw[RainbowG!50!RainbowA] (0.25,\i.5) node {\scalebox{0.6}{$b_{\i}$}};
          }
          \foreach \i in {0,1,2,3,4} {
            \draw[RainbowE] (\i.5,0.25) node {\scalebox{0.6}{$b_{\i}$}};
            \draw[RainbowE] (\i.5,-0.25) node {\scalebox{0.6}{$b_{\i}$}};
          }
          \draw[LevelI!60,step=1] (-2,0) grid (0,6);
          \draw[LevelI!60,step=1] (0,0) grid (6,6);
          \draw[LevelI!60,step=1] (0,-2) grid (6,0);
          \draw[LevelINext,ultra thick] (-6,0) -- (6,0);
          \draw[LevelINext,ultra thick] (0,-6) -- (0,6);
        \end{scope}
      \end{scope}
      \draw[very thick] (0,0) circle (7cm);
    \end{scope}
  \end{tikzpicture}
  \caption{Macro-colors inside a $\sigma$-macro-tile
    (drawn as\hspace{1.2ex}\begin{tikzpicture}[baseline=0.2ex,scale=0.25]
      \protect\draw[LevelINext,thick] (0,0) rectangle ++(1,1);
    \end{tikzpicture}).}
  \subcaption{Macro-colors appear as
    \begin{tikzpicture}[baseline=0.1ex,scale=0.2]
      \fill[RainbowE] (0,0) rectangle (5,1);
      \draw[LevelI!30!TilingGrid!40] (0,0) grid (5,1);
    \end{tikzpicture},
    \begin{tikzpicture}[baseline=0.1ex,scale=0.2]
      \fill[RainbowG!50!RainbowA] (0,0) rectangle (5,1);
      \draw[LevelI!30!TilingGrid!40] (0,0) grid (5,1);
    \end{tikzpicture},
    \begin{tikzpicture}[baseline=0.1ex,scale=0.2]
      \fill[RainbowD] (0,0) rectangle (5,1);
      \draw[LevelI!30!TilingGrid!40] (0,0) grid (5,1);
    \end{tikzpicture} and
    \begin{tikzpicture}[baseline=0.1ex,scale=0.2]
      \fill[RainbowB!70!RainbowB] (0,0) rectangle (5,1);
      \draw[LevelI!30!TilingGrid!40] (0,0) grid (5,1);
    \end{tikzpicture}.
    Since the \field{color fields} are facing outwards, adjacent $\sigma$-macro-tiles bear the same macro-color along their shared facet.}
\end{figure}

More precisely, we define a \field{color field} in the colors of $\Scolor[\iota]$, which consists of a single bit
\[ b \in \{0,1\}. \]
Any tile $\wang{t}$ appearing in the interior of a $\sigma$-macro-tile $\mwang{t}$ will have blank color fields, and only the tiles $\wang{t}$ appearing on the border of $\sigma$-macro-tiles will bear an outward-facing \field{color field}.

\begin{details}
  Let $\wang{t}$ be a tile of $\Swang[\iota][S]$, and let $(\rect{r}_{\ell},\vec{i}_{\ell})_{\iota \leq \ell < \iota'}$ be rectangles and positions read from the \field{position field} of its decoration $\dec(\wang{t})$. We make the \field{color field} of the decoration $\dec(\wang{t})$ always blank.

  For any direction $1 \leq k \leq d$, consider the results of the operations $\smash{\vec{j}^{\scriptscriptstyle +} = (\vec{i}_{\ell})_{\iota \leq \ell < \iota'} + \basis{k}}$ and $\smash{\vec{j}^{\scriptscriptstyle -} = (\vec{i}_{\ell})_{\iota \leq \ell < \iota'} - \basis{k}}$ in the partial odometer defined by the rectangles $(\rect{r}_{\ell})_{\iota \leq \ell < \iota'}$:
  \begin{itemize}
  \item If $\vec{j}^{\scriptscriptstyle +} = (\odnone,\dots,\odnone)$ (\textit{i.e.}~if $\wang{t}$ appears on the border of a $\sigma$-macro-tile), then the \field{color field} of $\facet{k}{+}(\wang{t})$ can either contain a bit $b \in \{0,1\}$ or be blank;
  \item If $\vec{j}^{\scriptscriptstyle -} = (\odnone,\dots,\odnone)$, the \field{color field} of $\facet{k}{-}(\wang{t})$ can similarly contain a bit $b \in \{0,1\}$ or be blank;
  \item In any other case, the \field{color field} of $\facet{k}{\pm}(\wang{t})$ is blank. \qedhere
  \end{itemize}
\end{details}

\begin{claim}
  In any $\sigma$-macro-tile $\mwang{t}$, the \emph{macro-colors} of $\mwang{t}$ define a tuple ${(c_1^{\scriptscriptstyle -},\dots,c_d^{\scriptscriptstyle -}, c_1^{\scriptscriptstyle +},\dots, c_d^{\scriptscriptstyle +})}$ of binary strings. Furthermore, in any valid $\Swang[\iota][S]$-tiling $x$ of $\Z^d$, denote $\pxgrid{g}_{\iota'} = (\pxrect{r}_{\vec{i}})_{\vec{i} \in \Z^d}$ the product grid of level~$\iota'$: then for each $\vec{i} \in \Z^d$, the macro-colors of the $\sigma$-macro-tile $\restr{x}{\pxrect{r}_{i}}$ describe a tile $\wang{t}^{(\vec{i})} \in \Swang[\ast]$ with blank decoration such that the configuration $x' \colon \vec{i} \in \Z^d \mapsto \wang{t}^{(\vec{i})}$ is a valid Wang tiling.
\end{claim}

\subsubsection{Wirings} To ensure that $\sigma$-macro-tiles actually emulate the tiles of $\Swang[\iota'][S]$ (instead of arbitrary tiles from $\Swang[\ast]$), we will use the embedded computations to restrict their tuples of possible macro-colors $(c_1^{\scriptscriptstyle -},\dots,c_d^{\scriptscriptstyle -}, c_1^{\scriptscriptstyle +},\dots, c_d^{\scriptscriptstyle +})$. To do so, we will ``route'' the macro-colors with wires from the external facets of the $\sigma$-macro-tiles to the input facet of their computation zone.

To proceed, we add a \field{wire field} to the colors of $\Swang[\iota][S]$, which is either of the form
\[ b \in \{0,1\} \]
to draw wires outside of the computation zone transmitting bits in a straight line; or, inspired by the (possibly crossing) wires from \Cref{pos:sec:wiring}, of the form
\[ (w_i)_{i \in I} \in ((\N^2 \times \{0,1\}) \times \N)^{\ast} \quad \text{or} \quad (w_i)_{i \in I} \in ((\N^2 \times \{0,1\}) \times \{\mathtt{start}, \mathtt{end}, \mathtt{cross} \})^{\ast} \]
(respectively for a tile's facets and decoration) to carry entries $(i,j,b) \in \N^2 \times \{0,1\}$ from the borders of the computation zone to its input facet, using the~\crtlnameref{pos:lem:routing-lemma}~(\Cref{pos:lem:routing-lemma}).

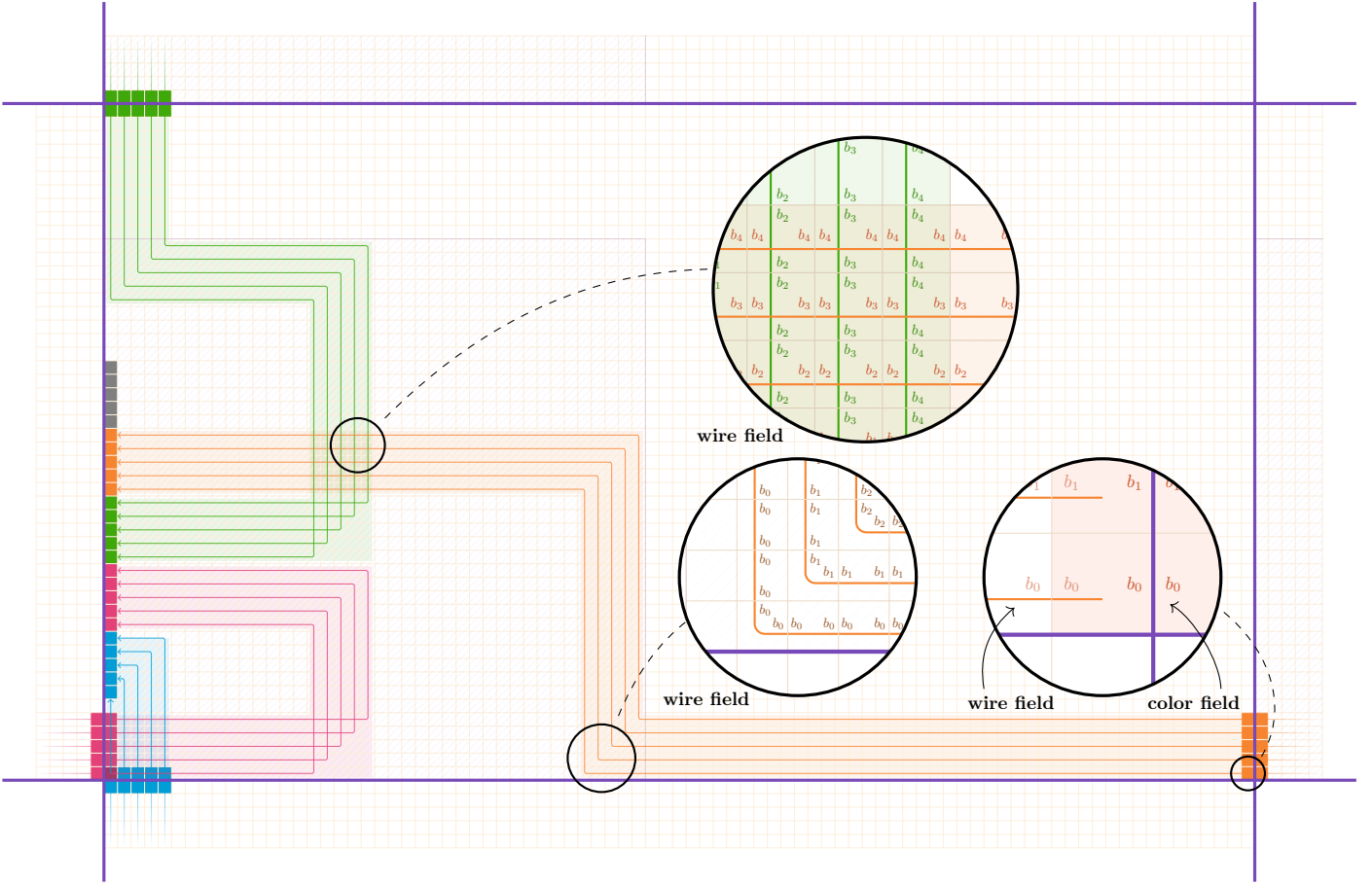
\begin{figure}[ht]
  \centerline{\hspace*{1cm}\begin{tikzpicture}[scale=0.46]
      \dimendef\mylinewidth=0
      \pgfmathsetlength{\mylinewidth}{0.85cm}
      \fill[RainbowG!50!RainbowA] (-0.4,0) rectangle ++(0.8,2);
      \fill[RainbowB!70!RainbowC] (33.6,0) rectangle ++(0.8,2);
      \fill[RainbowE] (0,-0.4) rectangle ++(2,0.8);
      \fill[RainbowG!50!RainbowA!75!black] (0,0) rectangle ++(0.4,0.4); 
      \fill[RainbowD] (0,19.6) rectangle ++(2,0.8);
      \draw[draw=RainbowE,line width=\mylinewidth,opacity=0.1] (1,0) -- (1,4.2);
      \draw[draw=RainbowG!50!RainbowA,line width=\mylinewidth,opacity=0.1] (0,1) -- (7,1) -- (7,5.4) -- (0,5.4);
      \draw[draw=RainbowD,line width=\mylinewidth,opacity=0.1] (1,20) -- (1,15) -- (7,15) -- (7,7.4) -- (0,7.4);
      \draw[draw=RainbowB!70!RainbowC,line width=\mylinewidth,opacity=0.1] (34,1) -- (15,1) -- (15,9.4) -- (0,9.4);
      \fill[RainbowE] (0,2.4) rectangle ++(0.4,2);
      \fill[RainbowG!50!RainbowA] (0,4.4) rectangle ++(0.4,2);
      \fill[RainbowD] (0,6.4) rectangle ++(0.4,2);
      \fill[RainbowB!70!RainbowC] (0,8.4) rectangle ++(0.4,2);
      \fill[gray] (0,10.4) rectangle ++(0.4,2);
      \draw[LevelI!15,step=0.4] (0,-2) grid ++(34,2);
      \draw[LevelI!15,step=0.4] (0,20) grid ++(34,2);
      \draw[LevelI!15,step=0.4] (-2,0) grid ++(2,20);
      \draw[LevelI!15,step=0.4] (34,0) grid ++(2,20);
      \draw[LevelI!15,step=0.4] (0,0) grid (34,20);
      \foreach \k in {0,...,4} {     
        \tikzmath{\x=(0.4*\k); \ox=(0.4*(4-\k)); \xn=(0.4*(\k+1)); \oxn=(0.4*(5-\k));}
        \draw[draw=RainbowG!50!RainbowA,thin,opacity=0.8,rounded corners=0.8pt,-{Classical TikZ Rightarrow[length=1.3pt]}] ($(0.2,0.2) + (0,\x)$) -- ++(\xn,0) -- ++(5.6,0) -- ++(0,\oxn) -- ++(0,2) -- ++(0,\xn) -- ++(-\xn,0) -- ++(-5.4,0);
      }
      \foreach \k in {0,...,4} {    
        \tikzmath{\x=(0.4*\k); \ox=(0.4*(4-\k)); \xn=(0.4*(\k+1)); \oxn=(0.4*(5-\k));}
        \draw[draw=RainbowE,thin,opacity=0.8,rounded corners=0.8pt,-{Classical TikZ Rightarrow[length=1.3pt]}] ($(0.2,0.2) + (\x,0)$) -- ++(0,\xn) \ifnum\k<1\relax{-- ++(0,1.8)}\else{-- ++(0,2) -- ++($(-\x,0) + (0.2,0)$)}\fi;
      }
      \foreach \k in {0,...,4} {    
        \tikzmath{\x=(0.4*\k); \ox=(0.4*(4-\k)); \xn=(0.4*(\k+1)); \oxn=(0.4*(5-\k));}
        \draw[draw=RainbowD,thin,opacity=0.8,rounded corners=0.8pt,-{Classical TikZ Rightarrow[length=1.3pt]}] ($(0.2,19.8) + (\x,0)$) -- ++(0,-3.6) -- ++(-0,-\oxn) -- ++(6,0) -- ++(0,-7.6) -- ++(-\xn,0) -- ++(-5.4,0);
      }
      \foreach \k in {0,...,4} {    
        \tikzmath{\x=(0.4*\k); \ox=(0.4*(4-\k)); \xn=(0.4*(\k+1)); \oxn=(0.4*(5-\k));}
        \draw[draw=RainbowB!70!RainbowC,thin,opacity=0.8,rounded corners=0.8pt,-{Classical TikZ Rightarrow[length=1.3pt]}] ($(33.8,0.2) + (0,\x)$) -- ++(-17.6,0) -- ++(-\oxn,0) -- ++(0,8.4) -- ++(-\xn,0) -- ++(-13.4,0);
      }
      \begin{scope}
        \path[clip,scope fading=west,fit fading=true] (0,0) rectangle ++(-2,2);
        \foreach \k in {0,...,4} {     
          \tikzmath{\x=(0.4*\k); \ox=(0.4*(4-\k)); \xn=(0.4*(\k+1)); \oxn=(0.4*(5-\k));}
          \draw[draw=RainbowG!50!RainbowA,thin,opacity=0.6,rounded corners=0.8pt] ($(-0.2,0.2) + (0,\x)$) -- ++(-2,0);
        }
      \end{scope}
      \begin{scope}
        \path[clip,scope fading=south,fit fading=true] (0,0) rectangle ++(2,-2);
        \foreach \k in {0,...,4} {    
          \tikzmath{\x=(0.4*\k); \ox=(0.4*(4-\k)); \xn=(0.4*(\k+1)); \oxn=(0.4*(5-\k));}
          \draw[draw=RainbowE,thin,opacity=0.6,rounded corners=0.8pt] ($(0.2,-0.2) + (\x,0)$) -- ++(0,-2);
        }
      \end{scope}
      \begin{scope}
        \path[clip,scope fading=north,fit fading=true] (0,20) rectangle ++(2,2);
        \foreach \k in {0,...,4} {    
          \tikzmath{\x=(0.4*\k); \ox=(0.4*(4-\k)); \oxn=(0.4*(6-\k));}
          \draw[draw=RainbowD,thin,opacity=0.6,rounded corners=0.8pt] ($(0.2,20.2) + (\x,0)$) -- ++(0,1.8);
        }
      \end{scope}
      \begin{scope}
        \path[clip,scope fading=east,fit fading=true] (34,0) rectangle ++(2,2);
        \foreach \k in {0,...,4} {    
          \tikzmath{\x=(0.4*\k); \ox=(0.4*(4-\k)); \xn=(0.4*(\k+1)); \oxn=(0.4*(5-\k));}
          \draw[draw=RainbowB!70!RainbowC,thin,opacity=0.6,rounded corners=0.8pt] ($(34.2,0.2) + (0,\x)$) -- ++(1.8,0);
        }
      \end{scope}
      \draw[LevelINext,pattern={north east lines},pattern color=LevelINext,opacity=0.2] (0,0) rectangle ++(16,16);
      \begin{scope}
        \clip (0,20) rectangle ++(34,2);
        \draw[LevelINext,pattern={north east lines},pattern color=LevelINext,opacity=0.2] (0,20) rectangle ++(16,3);
      \end{scope}
      \begin{scope}
        \clip (34,0) rectangle ++(2,20);
        \draw[LevelINext,pattern={north east lines},pattern color=LevelINext,opacity=0.2] (34,0) rectangle ++(3,16);
      \end{scope}
      \foreach \x in {0,34} {
        \draw[LevelINext,very thick] (\x,-3) -- ++($(0,20) + (0,6)$);
      }
      \foreach \y in {0,20} {
        \draw[LevelINext,very thick] (-3,\y) -- ++($(34,0) + (6,0)$);
      }
      \draw[thick] (7.5,9.9) circle (0.8cm);
      \draw[dashed,bend left] (8.3,10.7) to (22.5,14.5);
      \begin{scope}[shift={(22.5,14.5)}]
        \begin{scope}
          \clip (0,0) circle (4.5cm);
          \fill[white] (-5,-5) rectangle (5,5);
          \begin{scope}[shift={(2.5,-7.5)},scale=2]
            \fill[RainbowB!70!RainbowC,opacity=0.1] (1,0) rectangle (-5,5);
            \fill[RainbowD,opacity=0.08] (-5,-1) rectangle (0,7);
            \foreach \k in {0,1,2,3,4} {
              \tikzmath{int \ok; \ok=4-\k;}
              \draw[thick,RainbowD] (-\k.65,-1) -- (-\k.65,7);
              \foreach \y in {-1,0,...,6} {
                \draw[RainbowD!80!black] ($(-\k,\y) + (-0.47,0.15)$) node {\scalebox{0.56}{$b_{\ok}$}};
                \draw[RainbowD!80!black] ($(-\k,\y) + (-0.47,-0.15)$) node {\scalebox{0.56}{$b_{\ok}$}};
              }
            }
            \foreach \k in {0,1,2,3,4} {
              \draw[thick,RainbowB!70!RainbowC] (1,\k.35) -- (-5,\k.35);
              \foreach \x in {-1,0,...,4} {
                \draw[RainbowB!80!black] ($(-\x,\k) + (-0.15,0.55)$) node {\scalebox{0.56}{$b_{\k}$}};
                \draw[RainbowB!80!black] ($(-\x,\k) + (+0.16,0.55)$) node {\scalebox{0.56}{$b_{\k}$}};
              }
            }
            \draw[LevelI!18!TilingGrid] (-5,-1) grid (1,7);
            \draw[LevelINext,pattern={Lines[angle=45,distance={6pt/sqrt(2)}]},pattern color=LevelINext,opacity=0.1] (-5,-1) rectangle (1,6);
          \end{scope}
        \end{scope}
        \draw[very thick] (0,0) circle (4.5cm);
        \node at (-3.7,-4.3) {\scalebox{0.72}{\field{wire field}}};
      \end{scope}
      \draw[thick] (14.7,0.65) circle (1cm);
      \draw[dashed,bend left] (15.2,1.9) to (20.5,6);
      \begin{scope}[shift={(20.5,6)}]
        \begin{scope}
          \clip (0,0) circle (3.5cm);
          \fill[white] (-4,-4) rectangle (4,4);
          \begin{scope}[shift={(-1.8,-2.2)},scale=1.5,RainbowB!70!RainbowC!60!black]
            \foreach \k in {0,...,4} {
              \draw[rounded corners,thick,RainbowB!70!RainbowC] (5,\k.35) -- (\k.35,\k.35) -- (\k.35,5);
            }
            \draw[LevelI!30!TilingGrid!60] (-1,-1) grid (5,5);
            \foreach \k in {0,...,4} {
              \draw ($(\k,\k) + (0.82,0.56)$) node {\scalebox{0.52}{$b_{\k}$}};
              \draw ($(\k,\k) + (0.56,0.82)$) node {\scalebox{0.52}{$b_{\k}$}};
              \foreach \i in {1,...,3} {
                \draw ($(\k,\k) + (\i,0) + (0.18,0.56)$) node {\scalebox{0.52}{$b_{\k}$}};
                \draw ($(\k,\k) + (\i,0) + (0.82,0.56)$) node {\scalebox{0.52}{$b_{\k}$}};
                \draw ($(\k,\k) + (0,\i) + (0.56,0.18)$) node {\scalebox{0.52}{$b_{\k}$}};
                \draw ($(\k,\k) + (0,\i) + (0.56,0.82)$) node {\scalebox{0.52}{$b_{\k}$}};
              }
            }
            \draw[LevelINext,pattern={Lines[angle=45,distance={6pt/sqrt(2)}]},pattern color=LevelINext,opacity=0.15] (-1,0) rectangle (5,5);
            \draw[LevelINext,ultra thick] (-1,0) -- (5,0);
          \end{scope}
        \end{scope}
        \draw[very thick] (0,0) circle (3.5cm);
        \node at (-2.7,-3.6) {\scalebox{0.72}{\field{wire field}}};
      \end{scope}
      \draw[thick] (33.8,0.2) circle (0.5cm);
      \draw[dashed,looseness=1.1] (34.2,0.7) to[out=60,in=10] (29.5,6);
      \begin{scope}[shift={(29.5,6)}]
        \begin{scope}
          \clip (0,0) circle (3.5cm);
          \fill[white] (-4,-4) rectangle (4,4);
          \begin{scope}[shift={(1.5,-1.7)},scale=3]
            \fill[RainbowB!70!RainbowB,opacity=0.1] (-1,0) rectangle (1,5);
            \foreach \k in {0,...,4} {
              \draw[thick,RainbowB!70!RainbowC] (-0.5,\k.35) -- ++(-5,0);
              \draw[RainbowB!80!black] (-0.18,\k.5) node {\scalebox{0.7}{$b_{\k}$}};
              \draw[RainbowB!80!black] (0.2,\k.5) node {\scalebox{0.7}{$b_{\k}$}};
              \foreach \x in {1,2,3} {
                \draw[RainbowB!80!black,opacity=0.6] ($(-\x,\k) + (-0.18,0.5)$) node {\scalebox{0.7}{$b_{\k}$}};
                \draw[RainbowB!80!black,opacity=0.6] ($(-\x,\k) + (0.2,0.5)$) node {\scalebox{0.7}{$b_{\k}$}};
              }
            }
            \draw[LevelI!30!TilingGrid!60] (-4,-1) grid (0,5);
            \draw[LevelI!30!TilingGrid!60] (0,0) grid (1,5);
            \draw[LevelINext,ultra thick] (0,-1) -- (0,5);
            \draw[LevelINext,ultra thick] (-5,0) -- (1,0);
          \end{scope}
        \end{scope}
        \node at (-2.7,-3.7) {\scalebox{0.72}{\field{wire field}}};
        \draw[->,bend left] (-3.5,-3.3) to (-2.6,-0.9);
        \node at (2.7,-3.7) {\scalebox{0.72}{\field{color field}}};
        \draw[->,bend right,looseness=0.7] (3.5,-3.3) to (2,-0.8);
        \draw[very thick] (0,0) circle (3.5cm);
      \end{scope}
    \end{tikzpicture}}

  \caption{Macro-colors and wires inside a $\sigma$-macro-tile
    (drawn as\hspace{1.2ex}\begin{tikzpicture}[baseline=0.2ex,scale=0.25]
      \protect\draw[LevelINext,thick] (0,0) rectangle ++(1,1);
    \end{tikzpicture})
    and its computation zone
    (drawn with\hspace{1.2ex}\begin{tikzpicture}[baseline=0.2ex,scale=0.25]
      \protect\draw[LevelINext,pattern={north east lines},pattern color=LevelINext] (0,0) rectangle ++(1,1);
    \end{tikzpicture}).}
  \subcaption{Macro-colors appear as
    \begin{tikzpicture}[baseline=0.1ex,scale=0.2]
      \fill[RainbowE] (0,0) rectangle (5,1);
      \draw[LevelI!30!TilingGrid!40] (0,0) grid (5,1);
    \end{tikzpicture},
    \begin{tikzpicture}[baseline=0.1ex,scale=0.2]
      \fill[RainbowG!50!RainbowA] (0,0) rectangle (5,1);
      \draw[LevelI!30!TilingGrid!40] (0,0) grid (5,1);
    \end{tikzpicture},
    \begin{tikzpicture}[baseline=0.1ex,scale=0.2]
      \fill[RainbowD] (0,0) rectangle (5,1);
      \draw[LevelI!30!TilingGrid!40] (0,0) grid (5,1);
    \end{tikzpicture} and
    \begin{tikzpicture}[baseline=0.1ex,scale=0.2]
      \fill[RainbowB!70!RainbowB] (0,0) rectangle (5,1);
      \draw[LevelI!30!TilingGrid!40] (0,0) grid (5,1);
    \end{tikzpicture}. Outside of the computation zone, wires carry the individual bits of macro-colors along straight lines. Inside the computation zones, the~\crtlnameref{pos:lem:routing-lemma} wires these bits (indices $(i,j)$ are not drawn) towards the input facet (on the left).
    On the input facet, the gray argument
    \begin{tikzpicture}[baseline=0.1ex,scale=0.2]
      \fill[gray] (0,0) rectangle (5,1);
      \draw[LevelI!30!TilingGrid!40] (0,0) grid (5,1);
    \end{tikzpicture}
    corresponds to the macro-decoration (not wired) of the $\sigma$-macro-tile. For readability purposes, this figure draws a very neat wiring inside the computation zone, even though the wiring tileset may not always do so.}
\end{figure}

\begin{details} For any tile $\wang{t}$, let $(\rect{r}_{\ell},\vec{i}_{\ell})_{\iota \leq \ell < \iota'}$ be read from the \field{position field} of $\dec(\wang{t})$.

  \smallskip
  First, let $\vec{j}^{\scriptscriptstyle +} = (\vec{i}_{\ell})_{\iota \leq \ell < \iota'} + \basis{k}$ and $\vec{j}^{\scriptscriptstyle -} = (\vec{i}_{\ell})_{\iota \leq \ell < \iota'} - \basis{k}$ denote the results of the partial odometer operations defined by the rectangles $(\rect{r}_{\ell})_{\iota \leq \ell < \iota'}$. Then if $\vec{j}^{\scriptscriptstyle +} = (\odnone,\dots,\odnone)$ (resp.~$\vec{j}^{\scriptscriptstyle +} = (\odnone,\dots,\odnone)$), then the \field{wire field} of $\facet{k}{+}(\wang{t})$ (resp.~$\facet{k}{-}(\wang{t})$) must be blank. This ensures that the \field{wire fields} between adjacent $\sigma$-macro-tiles are completely isolated from one another.

  \medskip
  Let us now consider whether $\wang{t}$ appears inside a computation zone or not. On the one hand, if the \field{compzone field} of $\dec(\wang{t})$ is blank, bits travel in straight lines:
  \begin{itemize}
  \item The \field{wire field} of $\dec(\wang{t})$ must be blank; and the \field{wire field} of any $\facet{k}{\pm}(\wang{t})$ can either be blank or a single bit $b \in \{0,1\}$;
  \item The \field{wire field} of a facet $\facet{k}{\pm}(\wang{t})$ is a bit $b \in \{0,1\}$ if and only if either the \field{wire field} of its opposite facet $\facet{k}{\mp}(\wang{t})$ contains the same bit $b \in \{0,1\}$ (\textit{i.e.}~inside a $\sigma$-macro-tile), or if the \field{color field} of said $\facet{k}{\mp}(\wang{t})$ contains the same bit $b \in \{0,1\}$ (\textit{i.e.}~on the border of a $\sigma$-macro-tile).
  \end{itemize}

  Otherwise, let $(\compzone,\vec{i})$ denote the \field{compzone field} of $\dec(\wang{t})$. Then the \field{wire fields} of $\wang{t}$ will implement the \emph{routing tileset} from \Cref{pos:sec:wiring}:
  \begin{itemize}
  \item The \field{wire field} of $\dec(\wang{t})$ contains a list of wires $(w_i^{\decsymbol})_{i \in I} \in ((\N^2 \times \{0,1\}) \times \{\mathtt{start}, \mathtt{end}, \mathtt{cross} \})^{\ast}$ of length $|I| \leq \K[][cross]$, where $\K[][cross] \in \N$ bounds the crossings required in the~\crtlnameref{pos:lem:routing-lemma};
  \item If $\vec{i} \in \facet{k}{\pm}(\compzone)$, then the \field{wire field} of $\facet{k}{\pm}(\wang{t})$ is either blank or a single bit $b \in \{0,1\}$. Otherwise, $\facet{k}{\pm}(\wang{t})$ is a list of wires $(w_i)_{i \in I} \in ((\N^2 \times \{0,1\}) \times \N)^{\ast}$ of length $|I| \leq \K[][cross]$;
  \item The \field{wire fields} of $\dec(\wang{t})$ and $\facet{k}{\pm}(\wang{t})$ must form a valid tile in the routing tileset $\Swang[\wire](\N^2 \times \{0,1\})$ (preservation of wires between facets\dots); we also bound the length of any wire $n \in \N$ by $\rvol{\compzone}$;
    \pagebreak
  \item (Starts of wires) An entry $((i,j,b), \mathtt{start})$ appears in the \field{wire field} of $\dec(\wang{t})$ if and only if $\vec{i}$ belongs to a facet $\facet{k}{\pm}(\compzone)$ for some $1 \leq k \leq d$, and the \field{color field} or the \field{wire field} of $\facet{k}{\pm}(\wang{t})$ consists of a single bit $b \in \{0,1\}$. In which case, we set $i = 3 + k$ (on $\facet{k}{-}(\compzone)$) or $i = 3+d+k$ (on $\facet{k}{+}(\compzone)$) and $j$ to the index of the position $\vec{i}$ in the boustrophedon indexing of $\facet{k}{\pm}(\compzone)$;
  \item (Ends of wires) An entry $((i,j,b), \mathtt{end})$ appears in the \field{wire field} of $\dec(\wang{t})$ if and only if $\vec{i} \in \sfacet(\compzone) = \facet{h}{-}(\compzone)$ belongs to the input facet of the computation zone, and the variable $\ram{input}$ in the \field{processor field} of the facet $\facet{h}{+}(\wang{t})$ contains the same tuple $(i,j,b)$.
  \end{itemize}

  \medskip
  Finally, in order to prevent bits of macro-colors to traverse a $\sigma$-macro-tile in straight lines without ever crossing its computation zone, we enforce the following condition:
  \begin{itemize}
  \item If the \field{color field} of $\facet{k}{-}(\wang{t})$ is non-blank, then so is the \field{compzone field} of $\dec(\wang{t})$;
  \end{itemize}
  thus ensuring that macro-colors on negative facets must already appear inside the computation zone.
\end{details}

\newpage
We then claim that macro-colors occupy, on an input facet, the arguments $\ram{I}[4]$ to $\ram{I}[2d+4]$:
\begin{claim}
  Let $\mwang{t}$ be a $\sigma$-macro-tile of macro-colors ${(c_1^{\scriptscriptstyle -},\dots,c_d^{\scriptscriptstyle -}, c_1^{\scriptscriptstyle +},\dots, c_d^{\scriptscriptstyle +})}$, and let $\ram{I} \in \{0,1\}^{\ast\ast}$ be the input array that appears on the input facet of its computation zone. Then for every $1 \leq k \leq d$, we have $\ram{I}[3+k] = c_{k}^{\scriptscriptstyle -}$ and $\ram{I}[3+d+k] = c_{k}^{\scriptscriptstyle +}$.
\end{claim}

\begin{remark}
  The previous claim does not imply anything about the argument $\ram{I}[2d+4]$. In this construction, this argument plays the role of the \emph{decoration} $c_{\decsymbol}$ of the Wang tile $(c_1^{\scriptscriptstyle -},\dots,c_d^{\scriptscriptstyle -}, c_1^{\scriptscriptstyle +},\dots, c_d^{\scriptscriptstyle +},c_{\decsymbol})$. In the same way that a Wang tile's decoration does not intervene in the adjacency constraints with its neighbors, the argument $\ram{I}[2d+4]$ is processed computationally but is not wired.\\
  In the rest of this proof, when we say that the tuple $(c_1^{\scriptscriptstyle -},\dots,c_d^{\scriptscriptstyle -}, c_1^{\scriptscriptstyle +},\dots, c_d^{\scriptscriptstyle +},c_{\decsymbol})$ forms the \emph{macro-colors} of a macro-tile $\mwang{t}$, the decoration $c_{\decsymbol}\!$ will not actually be read from its \field{color fields}, but from this argument $\ram{I}[2d+4]$ in the \field{compzone fields}.
\end{remark}

\begin{numerics}
  Let $\mwang{t}$ be a $\sigma$-macro-tile of domain $\pxrect{r}$ and computation zone $\compzone$. Since the length of wires inside $\compzone$ is bounded by $\rvol{\compzone} \leq \rspace_{\iota'}^d$, and that in the carried colors indices $(i,j)$ belong to $\interval{4}{2d+4} \times \interval{0}{\rvol{\compzone}}$ (because they are mapped to the $\ram{input}$ variables of some \field{processor fields}), we deduce that the \field{wire fields} in $\mwang{t}$ have bit length $\bigO(\log \pxspace{\iota'}) = \polylog \K[][s] \cdot \bigO(\pxspace{\iota}[\delta])$.

  (This also justifies why we use straight wires outside of the computation zone: otherwise, wires could be of length $\bigO(\rtime(\pxrect{r}))$, which would result in \field{wire fields} of bit length $\polylog \K[][s] \cdot \bigO(\pxspace{\iota}[2\delta])$.)
\end{numerics}

\medskip
\subsubsection{Communication with the next level of the simulation}\label{fix:sec:macro-tile:communication-next-level}

While the arguments $\ram{I}[4]$ to $\ram{I}[2d+4]$ of the input facets are \emph{wired} from their macro-colors, we have yet to mention anything about the arguments $\ram{I}[0],\dots,\ram{I}[3]$. These will actually be used as a way to communicate information from the tiles of $\Swang[\iota][S]$ to the next level of the simulation~$\iota'$, and are thus initialized differently.

\paragraph*{The copying method} Assume that a \field{foo} field in the colors o4f $\Scolor[\iota]$ consists of a binary string that is constant across all the tiles $\wang{t} \in \Swang[\iota][S]$ composing a $\sigma$-macro-tile $\mwang{t}$. We then claim that \field{foo} can be used to initialize some argument $\ram{I}[i]$ on the input facet of the computation zone of $\mwang{t}$.

Indeed, let $\wang{t} \in \Swang[\iota][S]$ be a tile appearing in a $\sigma$-macro-tile $\mwang{t}$, and assume that $\dec(\wang{t})$ has a non-blank \field{compzone field} $(\compzone,\vec{i})$ and a \field{foo} field $u \in \{0,1\}^{\ast}$. In the definition of the tileset $\Swang[\iota][S]$, we can enforce that, if $\wang{t}$ appears on the input facet of the computation zone $\compzone$ (\textit{i.e.}~if $\vec{i} \in \sfacet(\compzone)$) and that the $\ram{input}$ variable of the corresponding \field{processor field} is some $(i',j,b)$, then if $i = i'$ the bit $b \in \{0,1\}$ must actually satisfy $b = u_{j}$. Collectively, the tiles $t \subpattern \mwang{t}$ thus ensure that their \field{foo} field equals the macro-tile's input argument $\ram{I}[i]$: we refer to this process as the \emph{copying method}.

\paragraph*{Initializing the input facet}

Using the copying method, we initialize the arguments $\ram{I}[0]$ to $\ram{I}[3]$ of the input facets of $\sigma$-macro-tiles as follows:
\begin{itemize}
\item We initialize $\ram{I}[0]$ by copying the parameters $(\K[][s],\K[][w])$ from the global function $\funram{F(e)}$;
\item We initialize $\ram{I}[1]$ by copying the \field{p-level field} $\iota' \in \N$ of the decorations $\dec(\wang{t})$;
\item We initialize $\ram{I}[2]$ by copying $(\rspace_{\ell},\rtime_{\ell})_{0 \leq \ell < \kappa'-1}$, the concatenation of $(\rspace_{\ell},\rtime_{\ell})_{0 \leq \ell < \kappa-1}$ from the function $\funram{F(e)}$ and the family $(\rspace_{\ell},\rtime_{\ell})_{\kappa-1 \leq \ell < \kappa'-1}$ from the \field{bound field} of $\dec(\wang{t})$;
\item We initialize $\ram{I}[3]$ by copying $(\rect{r}_{\ell},\vec{i}_{\ell})_{\iota' \leq \ell < \kappa'-1}$, as extracted from the family $(\rect{r}_{\ell},\vec{i}_{\ell})_{\iota \leq \ell < \kappa'-1}$ from the \field{position field} of $\dec(\wang{t})$.
\end{itemize}

\paragraph*{Reading the output facet} Instead of writing inside the input facet of the computation zone, it is also possible to use the copying method to read the values returned by the computation of the embedded processor array on the output facet. For example:
\begin{itemize}
\item We copy the argument $\ram{O}[0]$ from the output facet, which consists of a single set of directions $\nghb_{\kappa'-1} \subseteq \{\pm \basis{k} : 1 \leq k \leq d\}$, into the eponymous entry $\nghb_{\kappa'-1}$ in the \field{neighbor field} of $\dec(\wang{t})$.
\end{itemize}
We will similarly use the other arguments $\ram{O}[1]$, $\ram{O}[2]$ and $\ram{O}[3]$ later in the proof, as a way to read the results of some computations embedded in the next level of fixed-point simulation.

\newpage
\subsection{\texorpdfstring{$\bm{\tau}$}{τ}-macro-tiles (of level \texorpdfstring{$\bm{\kappa \leq \ell < \kappa'}$}{κ ≤ ℓ < κ'})}\label{fix:sec:tau-macro-tiles}

Now that the next level of simulation is implemented in the $\sigma$-macro-tiles, we turn towards the embedding of the $\tau$-substitution steps
\[ x^{(\kappa')} \mspace{6mu} \substep[\tau][\grid{g}_{\kappa'-1}] \mspace{6mu} x^{(\kappa'-1)} \mspace{6mu} \substep[\tau] \mspace{10mu} \ldots \mspace{10mu} \substep[\tau] \mspace{6mu} x^{(\kappa+1)} \mspace{6mu} \substep[\tau][\grid{g}_{\kappa}] x^{(\kappa)} \]
for $\grid{g}_{\kappa}, \dots, \grid{g}_{\kappa'-2}$ the grids determined by the \field{position fields} of a valid tiling, and $\grid{g}_{\kappa'-1}$ the grid determined by the last entry in the \field{neighbor fields}.

\medskip
\emph{In the rest of this proof, we will call $\tau$-macro-tiles of level $\ell$ the macro-tiles of level $\kappa \leq \ell < \kappa'$.\footnote{The notions of $\sigma$-macro-tiles and $\tau$-macro-tiles of level $\iota'$  actually coincide. When referring to a macro-tile of level $\iota'$ as a $\tau$-macro-tile (instead of a~$\sigma$-macro-tile), we will implicitly refer to the $\tau$-variants of its fields.}} In particular, $\tau$-macro-tiles respect the nested hierarchy already mentioned in \Cref{fix:sec:macro-tiles}: we kindly refer the reader to \Cref{fix:fig:macro-tile-hierarchy} again. This section will design the computational embeddings associated with the substitution steps $\smash{x^{(\ell+1)} \substep[\tau][\grid{g}_{\ell}] x^{(\ell)}}$ for $\kappa \leq \ell < \kappa'$ inside the $\tau$-macro-tiles of level $\ell$. In terms of configurations, a valid tiling of $\Swang[\iota][S]$ will thus draw multiple levels of substitutions superimposed on top of one another.

\medskip
Before we begin designing these $\tau$-macro-tiles, recall that the substitution $\tau$ is obtained by a $\log$-RAM program $\code{\tau} \in \{0,1\}^{\ast}$ defining:
\[
  \funram{\code{\tau}}
  \big(\rect{r},\vec{i},
  \mspace{2mu}
  a_1^{\scriptscriptstyle -},\dots,a_d^{\scriptscriptstyle -},
  a_1^{\scriptscriptstyle +},\dots,a_d^{\scriptscriptstyle +},
  a_{\decsymbol}
  \big)
  \mapsmapsto
  \big(
  b_1^{\scriptscriptstyle -},\dots,b_d^{\scriptscriptstyle -},
  b_1^{\scriptscriptstyle +},\dots,b_d^{\scriptscriptstyle +},
  b_{\decsymbol}
  \big)
\]
where $\rect{r} \in \Srect_0$ is a rectangle, $\vec{i} \in \rect{r}$ a position in $\rect{r}$, and the $a_k^{\scriptscriptstyle \pm}$'s and $a_{\decsymbol}$, and the $b_k^{\scriptscriptstyle \pm}$'s and $b_{\decsymbol}$, respectively form Wang tiles $a,b \in \Swang[\ast]$. The program $\code{\tau}$ thus implements the local function associated with the parallel substitution:
\[ \tau(a) = \bigcup_{\rect{r} \in \Srect_0} \Big\{ w \in \Swang[\ast]^{\mspace{2mu}\rect{r}} : \forall \vec{i} \in \rect{r},\;
  w_{\vec{i}} \in \funram{\code{\tau}}\big(\rect{r},\vec{i},
  a_1^{\scriptscriptstyle -},\dots,a_d^{\scriptscriptstyle -},
  a_1^{\scriptscriptstyle +},\dots,a_d^{\scriptscriptstyle +},
  a_{\decsymbol}\big) \Big\}.\]

\subsubsection{\texorpdfstring{$\tau$}{τ}-computation zones} The design of $\tau$-macro-tiles is based upon the construction of $\sigma$\nobreakdash-macro-tiles in \Cref{fix:sec:macro-tiles}.
More precisely, for every level $\ell$ in the interval $\kappa \leq \ell < \kappa'$, we introduce a \field{$\tau$-compzone field}, a \field{$\tau$-processor field} and a \field{$\tau$-arg field of level $\ell$} (in other words, there are $\kappa'-\kappa$ new \field{$\tau$-processor} (resp~\field{$\tau$-compzone}, \field{$\tau$-arg}\dots) \field{fields}: one per such level).

In a $\tau$-macro-tile of level $\ell$ and domain $\pxrect{r} \in \Srect_0$, the \field{$\tau$-compzone fields} of the form $(\compzone,\vec{i})$ (for $\compzone \in \Srect_0$ and $\vec{i} \in \compzone$) delimit the computation zone $\compzone = \mathrm{Cut}_{N}(\pxrect{r})$ with bounds $\smash{N = \prod_{i = \iota}^{\ell-1} \rspace_{i} = \pxspace{\ell}/\pxspace{\iota}}$. The \field{$\tau$-processor fields} collectively draw the space-time diagram of a processor array simulating (a parallel computation of) the substitution $\code{\tau}$; and the \field{$\tau$-arg fields} ensure that the arguments of the input and output arrays are correctly folded along the Boustrophedon indexing of the input and output facets.

\begin{details}
  The adaptation of the previously defined fields into their $\tau$-counterparts is straightforward; for clarity, we highlight a few important modifications on the \field{$\tau$-processor fields}. For any tile $\wang{t} \in \Swang[\iota][S]$, let $(\compzone,\vec{i})$ denote the \field{$\tau$-compzone field} of level $\ell$ of $\dec(\wang{t})$. Assuming that $\vec{i}$ is part of the input facet $\sfacet(\compzone) = \facet{h}{-}(\compzone)$, the variables appearing in the \field{$\tau$-processor field} of level $\ell$ of $\facet{h}{+}(\wang{t})$ satisfy:
  \begin{itemize}
  \item The $\ram{program}$ variable must be set to $\ram{program} = \code{\tau}$, where $\code{\tau} \in \{0,1\}^{\ast}$ is the code of the $\log$-RAM program defining the parallel substitution $\tau$;
  \item The $\ram{timeout}$ variable must be set to $\ram{timeout} = \smash{\K[t] \cdot \pxspace{\ell}[\alpha]}$ (\textit{c.f.}~\crtlnameref{fix:sec:housekeeping} paragraph);
  \item The $\ram{word}$ variable must be set to $\ram{word} = \K[][w] \cdot \log \rspace(\compzone)$, where $\K[][w]$ is the constant appearing in the parameters $(\K[][s],\K[][w])$ of the global function $\funram{F(e)}$;
  \item The $\ram{input}$ variable must be any value $\ram{input} = (i,j,b)$ for a bit $b \in \{0,1\}$ and an index ${(i,j) \in \interval{0}{2d+2} \times \interval{0}{\smash{\K[t]} \cdot \pxspace{\ell}[\alpha]}}$.\qedhere
  \end{itemize}
\end{details}

\begin{numerics}
  Let $\mwang{t}$ be a $\tau$-macro-tile of level $\ell$. Following similar computations from the \field{compzone} and \field{processor fields} of $\sigma$-macro-tiles, the \field{$\tau$-compzone fields of level $\ell$} and the \field{$\tau$-processor fields of level $\ell$} in $\mwang{t}$ are all of bit length $\K[][w] \cdot \bigO(\log \pxspace{\ell}) = \polylog \K[][s] \cdot \K[][w] \cdot \bigO(\pxspace{\iota}[\delta])$.

  In particular, in a valid tile $\wang{t}$, all the \field{$\tau$-compzone fields} and the \field{$\tau$-processor fields} of all levels $\kappa \leq \ell < \kappa'$ are of collective bit length $\polylog \K[][s] \cdot \K[][w] \cdot \bigO(\pxspace{\iota}[2\delta])$.
\end{numerics}

\subsubsection{\texorpdfstring{$\tau$}{τ}-colors and \texorpdfstring{$\tau$}{τ}-wirings}

Since $\tau$ substitutes Wang tilings, we create and enforce the local validity of the substituted Wang tilings by introducing, for every $\ell$ in the interval $\kappa \leq \ell < \kappa'$, a \field{$\tau$-color field} and \field{$\tau$-wire field of level $\ell$}.

In a configuration, \field{$\tau$-color fields} of adjacent $\tau$-macro-tiles of level $\ell$ define a Wang tiling $x^{(\ell)} \in (\Swang[\ast])^{\Z^d}$; and, inside each such macro-tile, \field{$\tau$-wire fields} ``route'' the corresponding $\tau$-macro-colors to the \emph{output facet} of the computation zone (which contains the corresponding pixel of the substitution step $x^{(\ell+1)} \substep[\tau][\pxgrid{g}_{\ell}] x^{(\ell)}$).

\begin{details}
  While the adaptation of the previous \field{wire fields} is straightforward, we highlight for clarity the important modifications. For any tile $\wang{t}$, let $(\compzone,\vec{i})$ denote the \field{$\tau$-compzone field} of level $\ell$ of $\dec(\wang{t})$; and denote by $\facet{h}{-}(\compzone) = \sfacet(\compzone)$ the \emph{input facet} of $\compzone$. Then the \field{$\tau$-wire field of level $\ell$} of $\dec(\wang{t})$ satisfies:
  \begin{itemize}
  \item (Start of wires) An entry $((i,j,b),\mathtt{start})$ appears in the \field{$\tau$-wire field of level $\ell$} of $\dec(\wang{t})$ if and only if $\vec{i} \in \facet{k}{\pm}(\compzone)$ for some facet $1 \leq k \leq d$ of the computation zone; and the \field{$\tau$-color field} or the \field{$\tau$-wire field of level $\ell$} of the corresponding facet $\facet{k}{\pm}(\wang{t})$ consists of a single bit $b \in \{0,1\}$. In which case, we set $i = d \pm k$ and $j$ to the index of the position $\vec{i}$ in the boustrophedon indexing of the facet $\facet{k}{\pm}(\compzone)$;
  \item (End of wires) An entry $((i,j,b), \mathtt{end})$ appears in the \field{$\tau$-wire field of level $\ell$} of $\dec(\wang{t})$ if and only if $\vec{i}$ belongs to the \emph{output facet} $\facet{h}{+}(\compzone)$ of the computation zone; and the variable $\ram{output}$ in the \field{$\tau$-processor field of level $\ell$} of the facet $\facet{h}{-}(\wang{t})$ contains the same tuple $(i,j,b)$.\qedhere
  \end{itemize}
\end{details}

\begin{numerics}
    Let $\mwang{t}$ be a $\tau$-macro-tile of level $\kappa \leq \ell < \kappa'$. Similarly to the \field{wires fields}, the \field{$\tau$-wire fields of level~$\ell$} in $\mwang{t}$ are of bit length $\bigO(\log \pxspace{\ell}) = \polylog \K[][s] \cdot \bigO(\pxspace{\iota}[\delta])$.
\end{numerics}

\subsubsection{\texorpdfstring{$\tau$}{τ}-consistency}
\label{fix:sec:tau-consistency}
In an $\Swang[\iota][S]$-valid tiling, the embedded substitution steps $\smash{x^{(\ell+1)} \substep[\tau][\grid{g}_{\ell}] x^{(\ell)}}$ must respect the structure of the grid $\grid{g}_{\ell}$ (as determined by the \field{position fields} of the tiles if $\ell < \kappa'-1$, or their \field{neighbor fields} if $\ell = \kappa'-1$). In particular, all parallel computations $\funram{\code{\tau}}(\rnorm{\rect{r}},\vec{i},a)$ belonging to the same rectangle $\rect{r}$ in $\grid{g}_{\ell}$ must be initialized with the same tile $(a_0,\dots,a_{2d}) \in \Swang[\ast]$.

For every $\kappa \leq \ell < \kappa'$, we thus add a \field{$\tau$-consistency field of level $\ell$}. Also based on wires, it ensures that the $\tau$-macro-tiles of level $\ell$ appearing in some common rectangle of the grid $\pxgrid{g}_{\ell+1}$ all substitute the same tile $(a_0,\dots,a_{2d}) \in \Swang[\ast]$ by drawing wires between their  respective \emph{input facets}:
\begin{itemize}
\item Outside the computation zone of a $\tau$-macro-tile, bits are carried in a straight line;
\item Inside the computation zone of a $\tau$-macro-tile, the \crtlnameref{pos:lem:routing-lemma}~(\Cref{pos:lem:routing-lemma}) carries bits from the border of the computation zone to its input facet;
\item Finally, wires are drawn between two adjacent $\tau$-macro-tiles of level $\ell$ if and only if they belong to the same $\tau$-macro-tile of level $\ell+1$.
\end{itemize}

\begin{details}
  Similar to \field{$\tau$-wire fields}, a \field{$\tau$-consistency field of level $\ell$} is either a single input bit
  \[ (i,j,b) \in \N^2 \times \{0,1\} \]
  to transmit bits in straight lines between the computation zones of adjacent macro-tiles; or of the form
  \[ (w_i)_{i \in I} \in ((\N^3 \times \{0,1\}) \times \N)^{\ast} \quad ; \; \text{or} \quad (w_i)_{i \in I} \in ((\N^3 \times \{0,1\}) \times \{\mathtt{start}, \mathtt{end}, \mathtt{cross} \})^{\ast} \]
  to carry entries $(d \pm k,i,j,b) \in \N^3 \times \{0,1\}$ from the borders of the computation zone to its input facet. The adaptation is straightforward, but we highlight the important modifications for clarity.

  \medskip
  Let $\wang{t} \in \Swang[\iota][S]$ be a tile whose decoration $\dec(\wang{t})$ has \field{position field} $(\rect{r}_{\ell},\vec{i}_{\ell})_{\iota \leq \ell < \kappa'-1}$, and in the the partial odometer of rectangles  $(\rect{r}_{\ell})_{\iota \leq \ell' < \ell}$ consider $\vec{j}_{k}^{\scriptscriptstyle +} = (\vec{i}_{\ell'})_{\iota \leq \ell' < \ell} + \basis{k}$ and $\vec{j}_{k}^{\scriptscriptstyle -} = (\vec{i}_{\ell})_{\iota \leq \ell' < \ell} - \basis{k}$. As opposed to the previously defined \field{wire fields}, having $\vec{j}_{k}^{\scriptscriptstyle -} = (\odnone,\dots,\odnone)$ or~$\vec{j}_{k}^{\scriptscriptstyle +} = (\odnone,\dots,\odnone)$ does not necessarily imply that the \field{$\tau$-consistency field} of $\facet{k}{-}(\wang{t})$ and $\facet{k}{+}(\wang{t})$ must be blank. More precisely, for $1 \leq k \leq d$:
  \begin{itemize}
  \item If $\vec{j}_{k}^{\scriptscriptstyle +} = (\odnone,\dots,\odnone)$: the \field{$\tau$-consistency field of level $\ell$} of $\facet{k}{+}(\wang{t})$ can contain a non-blank entry $(i,j,b)$ if and only if the direction $\basis{k}$ appears in the entry $\nghb_{\ell}$ of the \field{neighbor field} of $\dec(\wang{t})$;
  \item If $\vec{j}_{k}^{\scriptscriptstyle -} = (\odnone,\dots,\odnone)$: the \field{$\tau$-consistency field of level $\ell$} of $\facet{k}{-}(\wang{t})$ can contain a non-blank entry $(i,j,b) \in \N^2 \times \{0,1\}$ if and only if $\dec(\wang{t})$ has a non-blank \field{$\tau$-compzone field of level $\ell$};
  \end{itemize}
  Otherwise, the \field{$\tau$-consistency fields of level $\ell$} of $\wang{t}$ behave similarly to the previously defined \field{wire fields}. If the \field{$\tau$-compzone field of level $\ell$} of $\dec(\wang{t})$ is some non-blank $(\compzone,\vec{i})$, then:
  \begin{itemize}
  \item (Start of wires) An entry $((\pm k,i,j,b),\mathtt{start})$ appears in the \field{$\tau$-consistency field of level~$\ell$} of $\dec(\wang{t})$ if and only if $\vec{i} \in \smash{\facet{k}{\pm}(\compzone)}$ is on some facet $1 \leq k \leq d$ of the computation zone, and the \field{$\tau$-consistency field of level $\ell$} of the corresponding color $\smash{\facet{k}{\pm}(\wang{t})}$ contains $(i,j,b) \in \N^2 \times \{0,1\}$;
  \item (End of wires) If $\vec{i} \in \sfacet(\compzone) = \facet{h}{-}(\compzone)$ appears on the input facet of the computation zone, then the \field{$\tau$-consistency field of level $\ell$} of $\dec(\wang{t})$ contains an entry of the form ${((\pm k,i,j,b),\mathtt{end})}$ for each signed direction $\pm \basis{k}$ in the entry $\nghb_{\ell}$ of its \field{neighbor field} if and only if the variable $\ram{input}$ of the \field{$\tau$-processor field of level $\ell$} of the color $\facet{h}{+}(\wang{t})$ contains the same $(i,j,b)$ for $2 \leq i \leq 2d+2$. Otherwise, it contains no entry of the form $((\cdot,\cdot,\cdot),\mathtt{end})$.
  \end{itemize}
  By applying the \crtlnameref{pos:lem:routing-lemma} $2d$ times (one per value $\pm k$ for $1 \leq k \leq d$), this defines a valid routing between the input facet of the computation zone and each facet $\facet{k}{\pm}(\compzone)$.
\end{details}

\begin{numerics}
  Similarly to the \field{$\tau$-wire fields}, the \field{$\tau$-consistency fields of level $\ell$} (for $\kappa \leq \ell < \kappa'$) in a valid $\tau$-macro-tile~$\mwang{t}$ of level $\ell$ have bit length $\polylog \K[][s] \cdot \bigO(\pxspace{\iota}[\delta])$.
\end{numerics}

\subsubsection{\texorpdfstring{$\tau$}{τ}-pipes}
\label{fix:sec:tau-pipes}

We now initialize the input words $\ram{I}[0]$, $\ram{I}[1]$ and $\ram{I}[2],\ldots,\ram{I}[2d+2]$ on the input facet of $\tau$-macro-tiles of level $\ell$. In a valid tiling, they should respectively encode a rectangle $\rnorm{\rect{r}_{\ell}}$ from the grid $\rect{r}_{\ell} \in \grid{g}_{\ell}$, a position $\vec{i}_{\ell} \in \rnorm{\rect{r}_{\ell}}$, and a tile ${a_{\ell+1} = (a_1^{\scriptscriptstyle -},\dots,a_d^{\scriptscriptstyle -}, a_1^{\scriptscriptstyle +},\dots, a_d^{\scriptscriptstyle +}, a_{\decsymbol}) \in \Swang[\ast]}$. To ensure the consistency of this across successive substitution steps $\smash{x^{(\ell+2)} \substep x^{(\ell+1)} \substep x^{(\ell)}}$, we will synchronize the output facet of the $\tau$-macro-tiles of level $\ell+1$ with the input facets of their nested sub-$\tau$-macro-tiles of level $\ell$.

\medskip
We first state this elementary fact about the structure of nested $\tau$-macro-tiles:
\begin{claim}[label={pclm:tau-macro-tile-facet-inclusion}]
  Let $\ell' \leq \ell < \kappa'$, and consider a $\tau$-macro-tile $\mwang{t}[\ell]$ of level $\ell$ and domain $\pxrect{r}$. Then:
  \begin{itemize}
  \item There exists a unique $\tau$-macro-tile $\mwang{t}[\ell']$ of level $\ell'$ in $\mwang{t}[\ell]$ whose domain $\pxrect{r}' \subseteq \pxrect{r}$ contains the corner position $\vec{0} \in \pxrect{r}$;
  \item The computation zone $\compzone'$ of $\mwang{t}[\ell']$ has a facet $\facet{p}{-}(\compzone')$ that is included in the input facet $\sfacet(\compzone)$ of the computation zone $\compzone$ of $\mwang{t}[\ell]$.\footnote{Note that $\facet{p}{-}(\compzone')$ has no reason to actually be the input facet of the computation zone $\compzone'$.}
  \end{itemize}
\end{claim}
\noindent Informally, this claim ensures that there exists a facet in $\mwang{t}[\ell']$ through which the \crtlnameref{pos:lem:routing-lemma} will be able to wire information between $\mwang{t}[\ell']$ and $\mwang{t}[\ell]$. By analogy with the chaining of standard input and output streams in Unix systems, \emph{we will refer to these wirings across levels as ``pipes''}.

\paragraph*{Substitution shapes} In a $\tau$-macro-tile of level $\ell$, the first two arguments $\ram{I}[0]$ and $\ram{I}[1]$ of $\funram{\code{\tau}}$ should respectively be a rectangle $\rnorm{\rect{r}_{\ell}}$ from the grid $\grid{g}_{\ell}$ and a position $\vec{i}_{\ell} \in \rnorm{\rect{r}_{\ell}}$.

\medskip
If $\ell < \kappa' - 1$, these arguments $\ram{I}[0]$ and $\ram{I}[1]$ can be directly initialized on the input facet by using the \emph{copying method} on the entry $(\rect{r}_{\ell},\vec{i}_{\ell})$ in the \field{position field} of every tile $\wang{t} \in \Swang[\iota][S]$.

\begin{claim}
  For any $\tau$-macro-tile $\mwang{t}$ of level $\ell$ for $\ell < \kappa'-1$, let $\ram{I} \in \{0,1\}^{\ast\ast}$ be the input array appearing on the input facet of $\mwang{t}$, and let $(\rect{r}_{\ell},\vec{i}_{\ell})$ be the positions of index $\ell$ in the \field{position field} of the decoration $\dec(\wang{t})$ of every tile $\wang{t}$ in $\mwang{t}$. Then $\ram{I}[0] = \rect{r}_{\ell}$ and $\ram{I}[1] = \vec{i}_{\ell}$.
\end{claim}

\smallskip
If $\ell = \kappa' - 1$, the tiles from $\Swang[\iota][S]$ that appear in a valid tiling do not determine a grid $\grid{g}_{\kappa'-1}$ with their \field{position fields}: \emph{this information actually appears in the next level of simulation}. To solve this communication problem between $\Swang[\iota][S]$ and the tiles of $\Swang[\iota'][S]$ simulated in its $\sigma$-macro-tiles, we use the output of the latter's computations. More precisely, if the RAM programs $e$ (run by the $\sigma$-macro-tiles) and $F(e)$ (defining $\Swang[\iota'][S]$) have the same specification\footnote{And as this proof is a fixed-point argument on the transformation $F$, they will eventually do!}, then the computations of $e$ will output a pair $\ram{O}[1] = \rect{r}_{\kappa'-1}$ and $\ram{O}[2] = \vec{i}_{\kappa'-1}$ if we so design $F(e)$. In a $\tau$-macro-tile $\mwang{t}[\kappa'-1]$ of level $\kappa'-1$, we thus wire the arguments $\ram{I}[0]$ and $\ram{I}[1]$ from the computation zone of $\mwang{t}[\kappa'-1]$ from the arguments $\ram{O}[1]$ and $\ram{O}[2]$ of the $\sigma$-macro-tiles $\mwang{t}$ contained in~$\mwang{t}[\kappa'-1]$ (\textit{c.f.}~\Cref{fix:sec:macro-tile:communication-next-level}).

\smallskip
To proceed, we define a \field{$\tau$-pos-pipe field} containing a pair $(W_1,W_2)$, which
informally corresponds to two applications of the routing lemma transporting a color $(i,j,b) \in \N^2 \times \{0,1\}$ for $b = \ram{O}[i][j]$. Since by \Cref{pclm:tau-macro-tile-facet-inclusion}, the computation zones of some macro-tile $\mwang{t}$ and of $\mwang{t}[\kappa'-1]$ intersect along a common facet, we apply the \crtlnameref{pos:lem:routing-lemma} twice:
\begin{itemize}
\item Inside the $\sigma$-macro-tile $\mwang{t}$, we route the information from the output facet of the computation zone towards this common facet using the entries $W_1$ from \field{$\tau$-pos-pipe fields};
\item Inside the $\tau$-macro-tile $\mwang{t}[\kappa'-1]$, we reorganize this information towards the corresponding arguments of the input array using the entries $W_2$ from \field{$\tau$-pos-pipe fields}.
\end{itemize}

\begin{details}
  Let $\wang{t} \in \Swang[\iota][S]$ be a tile and let $(\rect{r}_{\ell},\vec{i}_{\ell})_{\iota \leq \ell < \kappa'-1}$ be the \field{position field} of $\dec(\wang{t})$. The \field{$\tau$-pos-pipe fields} of $\wang{t}$ consist of tuples of the form $(W_1,W_2)$, where $W_1$  and $W_2$ are lists $(w_i)_{i \in I}$ of the form:
  \[ (w_i)_{i \in I} \in ((\N^2 \times \{0,1\}) \times \N)^{\ast} \quad ; \; \text{or} \quad (w_i)_{i \in I} \in ((\N^2 \times \{0,1\}) \times \{\mathtt{start}, \mathtt{end}, \mathtt{cross} \})^{\ast} \]
  to carry entries $(i,j,b) \in \interval{0}{2d} \times \N \times \{0,1\}$ using the~\crtlnameref{pos:lem:routing-lemma}.

  \medskip
  Assume that $\dec(\wang{t})$ has \field{position field} $(\rect{r}_{\ell},\vec{i}_{\ell})_{\iota \leq \ell < \kappa'-1}$, and that its \field{compzone field} and \field{$\tau$-compzone field of level $\kappa'-1$} are non-blank and of respective values $(\compzone,\vec{i})$ and $(\compzone',\vec{i}')$. Furthermore, assume that $(\vec{i}_{\ell})_{\iota' \leq \ell < \kappa'-1} = (\vec{0},\dots,\vec{0})$ (\textit{i.e.}~any $\sigma$-macro-tile containing $\wang{t}$ covers the position $\vec{0}$ in $\compzone'$). By \Cref{pclm:tau-macro-tile-facet-inclusion}, there exists a facet $\facet{p}{-}(\compzone)$ that is included in the input facet $\sfacet(\compzone')$. Then the list of wires $W_1$ in the \field{$\tau$-pos-pipe fields} of $\wang{t}$ is used to carry entries $(i,j,b) \in \interval{0}{2d} \times \N \times \{0,1\}$ from the output facet of $\compzone$ towards the facet $\facet{p}{-}(\compzone)$:
  \begin{itemize}
  \item (Border condition) If $\vec{i} \pm \basis{k} \notin \compzone$ (in the \field{compzone field} of $\wang{t}$), then the list $W_1$ of the \field{$\tau$-pos-pipe field} of $\facet{k}{\pm}(\wang{t})$ is empty;
  \item (Start of wires) Let $(W_1,\cdot)$ be the \field{$\tau$-pos-pipe field} of $\dec(\wang{t})$. If the position $\vec{i}$ belongs to the output facet $\facet{h}{+}(\compzone)$ of $\compzone$ (if~$\facet{h}{-}(\compzone) = \sfacet(\compzone)$), then $W_1$ contains an entry $((i,j,b),\mathtt{start})$ if and only if the $\ram{output}$ variable of the \field{processor field} of $\facet{h}{-}(\wang{t})$ encodes the same tuple $(i,j,b)$;
  \item (End of wires) Let $(W_1,\cdot)$ be the \field{$\tau$-pos-pipe field} of $\dec(\wang{t})$. If the position $\vec{i}$ belongs to the common facet $\facet{p}{-}(\compzone)$, then $W_1$ may contain a single entry $((i,j,b),\mathtt{end})$ for some $i \in \{1,2\}$, $j \in \interval{0}{\rvol{\compzone}}$ and $b \in \{0,1\}$.
  \end{itemize}
  In any other case, the lists $W_1$ from the \field{$\tau$-pos-pipe fields} of $\wang{t}$ must be empty.

  \medskip
  Assume now that $(\compzone',\vec{i}')$ is a non-blank \field{$\tau$-compzone field of level $\kappa'-1$} of $\wang{t}$. Then the list of wires $W_2$ in the \field{$\tau$-pos-pipe fields} of $\wang{t}$ is used to carry entries $(i,j,b) \in \interval{0}{2d} \times \N \times \{0,1\}$ from the input facet of $\compzone'$ to their correct positions in input array:
  \begin{itemize}
  \item (Border condition) If $\vec{i}' \pm \basis{k} \notin \compzone'$ (in the \field{$\tau$-compzone field of level $\kappa'-1$} of $\wang{t}$), then the list $W_2$ of the \field{$\tau$-pos-pipe field} of $\facet{k}{\pm}(\wang{t})$ is empty;
  \item (Start of wires) Let $(W_1,W_2)$ be the \field{$\tau$-pos-pipe field} of $\dec(\wang{t})$. Then $W_2$ contains an entry $((i,j,b),\mathtt{start})$ if and only if $W_1$ contains the same entry $((i,j,b),\mathtt{end})$;
  \item (End of wires) Let $(\cdot,W_2)$ be the \field{$\tau$-pos-pipe field} of $\dec(\wang{t})$. If the position $\vec{i}'$ belongs to the input facet $\facet{h}{-}(\compzone') = \sfacet(\compzone')$, then $W_2$ contains an entry $((i,j,b),\mathtt{end})$ for some $b \in \{0,1\}$ if and only if the $\ram{input}$ variable of the \field{processor field} of $\facet{h}{+}(\wang{t})$ encodes the tuple $(i-1,j,b)$.\qedhere
  \end{itemize}
\end{details}

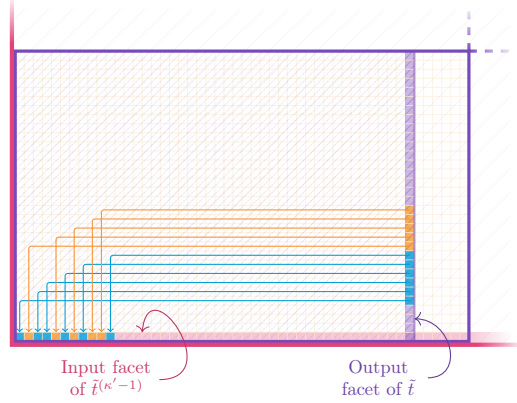
\begin{figure}
  \begin{tikzpicture}[scale=0.24]
    \colorlet{ArgA}{RainbowE!82}
    \colorlet{ArgB}{RainbowB!50!RainbowC!82}
    \draw[LevelI,opacity=0.12,step=0.5] (0,0) grid (25,16);
    \fill[opacity=0.3,pattern={north east lines},pattern color=LevelINext] (0,0.5) rectangle (21.5,16);
    \fill[opacity=0.25,pattern={Lines[angle=45,distance={9pt/sqrt(2)}]},pattern color=LevelKLast] (0,0) rectangle ++(27.3,18.3);
    \fill[opacity=0.3,color=LevelKLast] (0,0) rectangle (25,0.5);
    \fill[opacity=0.3,color=LevelKLast,path fading=east,fit fading=true] (25,-0.2) rectangle (27.3,0.5);
    \foreach[count=\k from 0] \c in {A,B,A,A,B,A,B,A,B,B,A} {
      \tikzmath{\x=\k*0.5;}
      \fill[fill=Arg\c] (\x,0) rectangle ++(0.5,0.5);
    }
    \fill[opacity=0.3,pattern={north east lines},pattern color=LevelINext] (0,-0.2) rectangle (21.5,0.5);
    \draw[LevelI!20,opacity=0.7,step=0.5] (0,0) grid (25,0.5);
    \fill[opacity=0.4,color=LevelINext] (21.5,0) rectangle (22,16);
    \fill[ArgA] (21.5,2) rectangle ++(0.5,3);
    \fill[ArgB] (21.5,5) rectangle ++(0.5,2.5);
    \fill[opacity=0.3,pattern={north east lines},pattern color=LevelINext] (21.5,-0.2) rectangle (22,16); 
    \draw[LevelI!20,opacity=0.7,step=0.5] (21.5,0.5) grid (22,16);
    \draw[opacity=0.5,LevelINext,thick] (0,-0.2) rectangle (22,16);
    \pgfmathsetmacro{\ca}{0}
    \pgfmathsetmacro{\cb}{0}
    \foreach[count=\k from 0,remember=\ca,remember=\cb] \c in {A,B,A,A,B,A,B,A,B,B,A} {
      \tikzmath{\x=\k*0.5; \ta=\ca*0.5+2; \tb=\cb*0.5+5;}
      \expandafter\ifstrequal\expandafter{\c}{A}{%
        \draw[-{Classical TikZ Rightarrow[length=1.3pt]},rounded corners=1pt,RainbowE] ($(21.5,\ta) +(0.25,0.25)$) -- ($(\x,\ta) + (0.25,0.25)$) -- ($(\x,0) + (0.25,0.52)$);%
        \pgfmathsetmacro{\ca}{\ca+1}}
      {\draw[-{Classical TikZ Rightarrow[length=1.3pt]},rounded corners=1pt,RainbowB!60!RainbowC] ($(21.5,\tb) +(0.25,0.25)$) -- ($(\x,\tb) + (0.25,0.25)$) -- ($(\x,0) + (0.25,0.52)$);%
        \pgfmathsetmacro{\cb}{\cb+1}}
    }
    \draw[very thick,LevelINext,dashed,path fading=north] (25,16) -- ++(0,2.5);
    \draw[very thick,LevelINext,dashed,path fading=east] (25,16) -- ++(2.5,0);
    \draw[very thick,LevelINext] (0,0) rectangle ++(25,16);
    \draw[ultra thick,LevelKLast,path fading=north,fit fading=true] (-0.18,16.499) -- ++(0,2.501);
    \draw[ultra thick,LevelKLast,path fading=east,fit fading=true] (25.499,-0.18) -- ++(2.501,0);
    \draw[ultra thick,LevelKLast] (-0.18,16.5) -- ++(0,-16.68) -- ++(25.68,0);
    \node[LevelKLast!90!black] (A) at (5,-2) {\scalebox{0.7}{\parbox{1.8cm}{\centering Input facet of $\mwang{t}[\kappa'-1]$}}};
    \draw[LevelKLast!80!black,->,looseness=3.5] ($(A.east) + (-0.2,-0.1)$) to[in=85,out=0] (7,0.5);
    \node[LevelINext!90!black] (B) at (20,-2) {\scalebox{0.7}{\parbox{1.8cm}{\centering Output facet of $\mwang{t}$}}};
    \draw[LevelINext!80!black,->,looseness=1.6] ($(B.east) + (-0.4,0)$) to[out=0,in=0] (22,1.25);
  \end{tikzpicture}
  \caption{First wiring induced by the \field{$\tau$-pos pipes}.}
  \subcaption{A (partial) $\tau$-macro-tile $\mwang{t}[\kappa'-1]$ of level $\kappa'-1$ (drawn as\hspace{1.2ex}\begin{tikzpicture}[baseline=0.2ex,scale=0.25]
      \protect\draw[path fading=east,fading transform={rotate=45},LevelKLast,thick] (0,1) -- (0,0) -- (1,0);
    \end{tikzpicture}),
    the $\sigma$-macro-tile $\mwang{t}$ covering the corner $\vec{0}$
    (drawn as\hspace{1.2ex}\begin{tikzpicture}[baseline=0.2ex,scale=0.25]
      \protect\draw[LevelINext,thick] (0,0) rectangle ++(1,1);
    \end{tikzpicture}),
    and their respective computation zones $\compzone'$
    (as\hspace{1.2ex}\begin{tikzpicture}[baseline=0.2ex,scale=0.25]
      \protect\draw[LevelKLast!50,pattern={north east lines},pattern color=LevelKLast!50] (0,0) rectangle ++(1,1);
    \end{tikzpicture})
    and $\compzone$
    (as\hspace{1.2ex}\begin{tikzpicture}[baseline=0.2ex,scale=0.25]
      \protect\draw[LevelINext!50,pattern={north east lines},pattern color=LevelINext!50] (0,0) rectangle ++(1,1);
    \end{tikzpicture}). The first half of the \field{$\tau$-pos pipes} applies the \crtlnameref{pos:lem:routing-lemma} inside $\compzone$ to wire the output facet of $\compzone$ towards the facet $\facet{p}{-}(\compzone) \subseteq \sfacet(\compzone')$. The second half of the \field{$\tau$-pos pipes} (not drawn here) applies the \crtlnameref{pos:lem:routing-lemma} inside $\compzone'$ to rearrange the bits on the input facet of $\mwang{t}[\kappa'-1]$.}
\end{figure}
\begin{claim}
  Let $\mwang{t}[\kappa'-1]$ be a $\tau$-macro-tile of level $\kappa'-1$, and $\mwang{t}$ its $\sigma$-macro-tile covering the corner $\vec{0} \in \dom(\mwang{t}[\kappa'-1])$. If $\rect{r} \in \Srect_0$ and $\vec{i} \in \rect{r}$ are the arguments $\ram{O}[1]$ and~$\ram{O}[2]$ from the output facet of~$\mwang{t}$, then the arguments $\ram{I}[0]$ and $\ram{I}[1]$ on the input facet of $\mwang{t}[\kappa'-1]$ satisfy $\ram{I}[0] = \rect{r}$ and $\ram{I}[1] = \vec{i}$.
\end{claim}

\begin{numerics}
  Let $\mwang{t}[\kappa'-1]$ be a $\tau$-macro-tile of level $\kappa'-1$. Similarly to the \field{$\tau$-wire fields of level $\kappa'-1$}, the \field{$\tau$-pos-pipe fields} in $\mwang{t}$ are of bit length $\bigO(\log \pxspace{\kappa'-1}) = \polylog \K[][s] \cdot \bigO(\pxspace{\iota}[\delta])$.
\end{numerics}

\paragraph*{Substitution symbols} The last $2d+1$ arguments $\ram{I}[3]$ to $\ram{I}[2d+3]$ of $\funram{\code{\tau}}$ on the input facet of every $\tau$-macro-tile of level $\ell$ should be initialized with the symbol ${a_{\ell+1} = (a_1^{\scriptscriptstyle -},\dots,a_d^{\scriptscriptstyle -}, a_1^{\scriptscriptstyle +},\dots, a_d^{\scriptscriptstyle +},a_{\decsymbol}) \in \Swang[\ast]}$ that appears on the output facet of the $\tau$-macro-tile of the next level.

\medskip
If $\ell < \kappa'-1$, we transmit this information using a new \field{$\tau$-sym-pipe field of level $\ell$} encoding a pair~$(W_1,W_2)$, which corresponds to two applications of the routing lemma transporting a color $(i,j,b) \in \N^2 \times \{0,1\}$ for $b = \ram{I}[i][j]$. Since the \field{$\tau$-consistency field of level $\ell$} ensures that all macro-tiles $\mwang{t}[\ell]$ forming $\mwang{t}[\ell+1]$ share the same input symbol, it is enough to synchronize the symbol $a_{\ell+1} \in \Swang[\ast]$ from $\mwang{t}[\ell+1]$ with the input of the $\tau$-macro-tile $\mwang{t}[\ell]$ covering the corner $\vec{0} \in \dom(\mwang{t}[\ell+1])$.

We skip the details here because the \field{$\tau$-sym-pipe fields of level $\ell$} (for $\kappa \leq \ell < \kappa'-1$) are a straightforward adaptation of the already-defined \field{$\tau$-pos-pipe fields}.

\begin{claim}
  Let $\mwang{t}[\ell+1]$ be a $\tau$-macro-tile of level $\ell+1 < \kappa'$, and $\mwang{t}[\ell]$ any of its $\tau$-macro-tile of level~$\ell$. If $(a_0,\dots,a_{2d}) \in \Swang[\ast]$ denotes the arguments ${\ram{I}[3]},\dots,{\ram{I}[2d+3]}$ on the input facet of $\mwang{t}[\ell]$, then the arguments $\ram{O}[0],\dots,\ram{O}[2d]$ on the output facet of $\mwang{t}[\ell+1]$ satisfy $\ram{O}[i] = a_i$ for every $i \in \interval{0}{2d}$.
\end{claim}

\medskip
If $\ell = \kappa'-1$, there is actually no $\tau$-macro-tile of level $\kappa'$ to synchronize with. Indeed, the substitution step $\smash{x^{(\kappa'+1)} \substep[\tau] x^{(\kappa')}}$ is embedded in the tilings of $\Swang[\iota'][S]$, \textit{i.e.}~the next level of simulation. Therefore, the synchronization of the symbol $a_{\kappa'} = (a_1^{\scriptscriptstyle -},\dots,a_d^{\scriptscriptstyle -}, a_1^{\scriptscriptstyle +},\dots, a_d^{\scriptscriptstyle +},a_{\decsymbol}) \in \Swang[\ast]$ again requires to \emph{communicate through a level of simulation}, which unfortunately turns out to be rather technical.

More precisely, the tileset $\Swang[\iota'][S]$ forms $\tau$-macro-tiles $\mwang{t}[\kappa']$ of level $\kappa'$ whose output arguments $\ram{O}[0],\dots\ram{O}[2d]$ define a symbol $a_{\kappa'} = (a_1^{\scriptscriptstyle -},\dots,a_d^{\scriptscriptstyle -}, a_1^{\scriptscriptstyle +},\dots, a_d^{\scriptscriptstyle +},a_{\decsymbol}) \in \Swang[\ast]$. Assuming again that the programs $e$ (run by $\sigma$-macro-tiles) and $F(e)$ (defining the tileset $\Swang[\iota'][S]$) have the same specification, we will use the remaining output argument $\ram{O}[3]$ on the output facet of the $\sigma$-macro-tiles to synchronize the individual bits from the symbol $a_{\kappa'}$. We divide this synchronization in two parts:
\begin{itemize}
\item Let $\mwang{t}[\kappa'-1] \in \Swang[\iota][S]^{\mspace{2mu}\Srect_0}$ be a $\tau$-macro-tile of level $\kappa'-1$ whose input arguments $\ram{I}[2],\dots,\ram{I}[2d+2]$ define a symbol $(a_0,\dots,a_{2d}) \in \Swang[\ast]$. We first design an \field{in pipe} (based on wirings again) to ensure that, if a $\sigma$-macro-tile $\mwang{t}$ in $\mwang{t}[\kappa'-1]$ outputs a non-blank argument $\ram{O}[3] = (i,j,b)$ for some $(i,j,b) \in \interval{0}{2d} \times \N \times \{0,1\}$, then $b = a_i(j)$;
\end{itemize}
However, the~\crtlnameref{pos:lem:routing-lemma} is not able to accommodate arbitrary dispositions of wires inside the rectangle $\dom(\mwang{t}[\kappa'-1])$ (wire ends must appear on the border, form greedy sets of positions\dots). Hence, for the wiring of the \field{in pipe} to be possible, we further add:
\begin{itemize}
\item Let $\mwang{t}[\kappa'] \in \Swang[\iota'][S]^{\mspace{2mu}\Srect_0}$ be a $\tau$-macro-tile of level $\kappa'$ whose output arguments $\ram{O}[0],\dots,\ram{O}[2d]$ define a symbol $a_{\kappa'} = (a_0,\dots,a_{2d}) \in \Swang[\ast]$. We then design an \field{out pipe} (based on wirings again) to distribute the associated bits towards some ``well-positioned'' tiles of $\mwang{t}[\kappa']$.
\end{itemize}

\begin{details}[Details (pixel positions)]
  In order to identify the shape (and, in particular, the smallest facet) of some rectangles $\pxrect{r}_{\ell}$ in the respective grid $\pxgrid{g}_{\ell}$, we begin by defining three \field{pixel position fields} (respectively \field{of level $\kappa-1$}, \field{$\kappa$} and \field{$\kappa'-1$}, which all contain tuples $(\pxrect{r}_{\ell},\pxvec{i}_{\ell})$ such that $\pxrect{r}_{\ell} \in \Srect_0$ and $\pxvec{i}_{\ell} \in \pxrect{r}_{\ell}$.

  For a tile $\wang{t} \in \Swang[\iota][S]$, let $(\rect{r}_{\ell},\vec{i}_{\ell})_{\iota \leq \ell < \kappa'-1}$ denote the \field{position field} of $\dec(\wang{t})$. For $\ell \in \{\kappa-1,\kappa,\kappa'-1\}$, the \field{pixel position field of level $\ell$} of $\dec(\wang{t})$ is a non-blank pair $(\pxrect{r}_{\ell},\pxvec{i}_{\ell}) \in \Srect_0 \times \N^d$ such that:
  \begin{itemize}
  \item If $\vec{i}_{j} = \vec{0}$ for every $\iota \leq j < \ell$, then $\pxvec{i}_{\ell}$ is actually $\vec{i}_{\ell}' = \vec{0}$;
  \end{itemize}
  Furthermore, if the \field{pixel position field of level $\ell$} contains $(\pxrect{r}_{\ell},\pxvec{i}_{\ell})$, the facets $\facet{k}{\pm}(\wang{t})$ satisfy:
  \begin{itemize}
  \item If the partial odometer given by the rectangles $(\rect{r}_{j})_{\iota \leq j < \ell}$ yields $(\vec{i}_{j})_{\iota \leq j < \ell} \pm \basis{k} = (\odnone,\dots,\odnone)$ (\textit{i.e.}~$\wang{t}$ appears on the border of macro-tiles of level $\ell$), then the \field{pixel position field of level~$\ell$} of $\facet{k}{\pm}(\wang{t})$ is left blank; in which case, the position $\pxvec{i}_{\ell}$ must satisfy $\pxvec{i}_{\ell} \in \facet{k}{\pm}(\pxrect{r}_{\ell})$;
  \item Otherwise, the \field{pixel position field of level $\ell$} of the facet $\facet{k}{-}(\wang{t})$ is also~$(\pxrect{r},\pxvec{i})$;
  \item And the \field{pixel position field of level $\ell$} of the facet $\facet{k}{+}(\wang{t})$ contains $(\pxrect{r}',\pxvec{i} + \basis{k})$.\qedhere
  \end{itemize}
\end{details}

\begin{claim}
  For any $\ell \in \{\kappa-1,\kappa,\kappa'-1\}$, and for any macro-tile $\mwang{t}[\ell]$ of level $\ell$ and domain $\pxrect{r}_{\ell}$, the \field{pixel position field of level $\ell$} of $\dec(\mwang{t}[\ell]_{\mspace{-4mu}\pxvec{i}})$ contains the pair $(\pxrect{r}_{\ell},\vec{i}')$ for every $\pxvec{i} \in \pxrect{r}_{\ell}$.
\end{claim}

Consider a $\tau$-macro-tile $\mwang{t}[\kappa'-1]$ of level $\kappa'-1$ and domain $\pxrect{r}_{\kappa'-1}$ and $\mwang{t}$ a $\sigma$-macro-tile appearing in $\mwang{t}[\kappa'-1]$. Depending on the position of the $\sigma$-macro-tile $\mwang{t}$ inside $\mwang{t}[\kappa'-1]$, it may provide a single bit from the symbol $a_{\kappa'} \in \Swang[\ast]$ by encoding a tuple $(i,j,b) \in \interval{0}{2d} \times \N \times \{0,1\}$ in the argument $\ram{O}[3]$ of its output facet. To implement the synchronization of the bit $b$ with the bit $\ram{I}[i+2][j]$ from the input facet of $\mwang{t}[\kappa'-1]$, we introduce two \field{$\tau$-in pipe fields}.

\pagebreak
The \field{first $\tau$-in pipe field} is of the form $(i,j,b) \in \interval{0}{2d} \times \N \times \{0,1\}$; and:
\begin{itemize}
\item If the $\sigma$-macro-tile $\mwang{t}$ in $\mwang{t}[\kappa'-1]$ touches the smallest facet $\sfacet(\pxrect{r}_{\kappa'-1})$, the \field{first $\tau$-in pipe field} of the tiles $\wang{t}$ in $\mwang{t}$ contains some $(i,j,b)$ that is is initialized with the \emph{copying method} from the output argument $\ram{O}[3]$ of the $\sigma$-macro-tile $\mwang{t}$;
\item If the $\sigma$-macro-tile $\mwang{t}$ in $\mwang{t}[\kappa'-1]$ does not touch the smallest facet of $\mwang{t}[\kappa'-1]$, the \field{first $\tau$-in pipe field} of any tile $\wang{t}$ in $\mwang{t}$ must be blank;
\end{itemize}
The \field{second $\tau$-in pipe field} is very similar to the already-defined \field{wire fields}, and allows to apply the \crtlnameref{pos:lem:routing-lemma} from tiles with non-blank \field{first $\tau$-in pipe field} $(i,j,b) \in \N^2 \times \{0,1\}$ towards the input arguments $\ram{I}[i+2][j] = b$ of the input facet of $\mwang{t}[\kappa'-1]$.

\begin{details}
  Let $\wang{t} \in \Swang[\iota][S]$ be a tile and let $(\rect{r}_{\ell},\vec{i}_{\ell})_{\iota \leq \ell < \kappa'-1}$ be the \field{position field} and $(\pxrect{r},\pxvec{i})$ be the \field{pixel position field of level $\kappa'-1$} of $\dec(\wang{t})$. The \field{first $\tau$-in pipe field} of $\dec(\wang{t})$ satisfies the following:
  \begin{itemize}
  \item The \field{first $\tau$-in pipe field} of $\dec(\wang{t})$ can either be blank, or contain a tuple $(i,j,b)$ such that $i \in \interval{0}{2d}$, $j \in \interval{0}{\K[t] \cdot \pxspace{\kappa'-1}[\alpha]}$ and $b \in \{0,1\}$;
  \item Additionally, if $(\vec{i}_{\ell})_{\iota \leq \ell < \iota'} = (\vec{0},\dots,\vec{0})$, but that the position $\pxvec{i}$ does not belong to the smallest facet $\sfacet(\pxrect{r})$, then the \field{first $\tau$-in pipe field} of $\dec(\wang{t})$ must be blank.
  \end{itemize}
  Furthermore, if the \field{first $\tau$-in pipe field} of $\dec(\wang{t})$ is blank, then so are the \field{first $\tau$-in pipe fields} of $\facet{k}{\pm}(\wang{t})$. Otherwise, assuming that $\dec(\wang{t})$ encodes some tuple $(i,j,b)$, then for any direction $1 \leq k \leq d$:
  \begin{itemize}
  \item If the partial odometer defined by the rectangles $(\rect{r}_{\ell})_{\iota \leq \ell < \iota'}$ yields $(\vec{i}_{\ell})_{\iota \leq \ell < \iota'} \pm \basis{k} \neq (\odnone,\dots,\odnone)$, then the \field{first $\tau$-in pipe field} of $\facet{k}{\pm}(\wang{t})$ also contains $(i,j,b)$; otherwise, it is blank.
  \end{itemize}
  Thus, in any macro-tile $\mwang{t}$ appearing inside a $\tau$-macro tile $\mwang{t}[\kappa'-1]$ of level $\kappa'-1$, the \field{first $\tau$-in pipe fields} of the decorations are constant; and can only be non-blank across the smallest facet of $\mwang{t}[\kappa'-1]$. We can thus apply the \emph{copying method}: if the \field{compzone field} of $\dec(\wang{t})$ contains some tuple $(\compzone,\vec{i})$, that $\facet{h}{-}(\compzone) = \sfacet(\compzone)$ denotes the input facet of the computation zone $\compzone$, and that $\vec{i} \in \facet{h}{+}(\compzone)$ belongs to its output facet then,
  \begin{itemize}
  \item Let $u \in \{0,1\}^{\ast}$ be the (bit encoding\footnote{Until now, this proof has mostly hidden the encodings of data types into binary strings: as the most recent example, we just assumed that a \field{first $\tau$-in pipe field} usually encodes a tuple $(i,j,b) \in \N^2 \times \{0,1\}$ without specifying its binary encoding. Nevertheless, these fields are concretely written as binary strings, and we actually manipulate said field as a binary string in this specific proof item.})
    of the \field{first $\tau$-in pipe field} of $\dec(\wang{t})$. If the $\ram{output}$ variable of \field{processor field} of $\facet{h}{-}(\wang{t})$ contains a tuple $(i',j',b') \in \N^2 \times \{0,1\}$ and if $i' = 3$, then $u$ must satisfy $u_{j'} = b'$.
  \item Reciprocally, if the \field{arg fields} of $\wang{t}$ show that the argument $\ram{O}[3]$ of the $\sigma$-macro-tiles containing $\wang{t}$ are empty, then the \field{first $\tau$-in pipe field} of $\dec(\wang{t})$ must be blank.
  \end{itemize}

  \medskip
  We then use the \field{second $\tau$-in pipe fields} to implement the \crtlnameref{pos:lem:routing-lemma}. They are of the form:
  \[ (w_i)_{i \in I} \in ((\N^2 \times \{0,1\}) \times \N)^{\ast} \quad ; \; \text{or} \quad (w_i)_{i \in I} \in ((\N^2 \times \{0,1\}) \times \{\mathtt{start}, \mathtt{end}, \mathtt{cross} \})^{\ast} \]
  to carry entries $(i,j,b) \in \interval{0}{2d} \times \N \times \{0,1\}$ inside $\tau$-macro-tiles of level $\kappa'-1$ from the borders of their computation zone to their input facet. Based upon similar wirings, we skip some details and only highlight these important points:
  \begin{itemize}
  \item (Border condition) Let $(\pxrect{r},\pxvec{i})$ be the \field{pixel position field of level $\kappa'-1$} of $\dec(\wang{t})$). If $\pxvec{i} \pm \basis{k} \notin \pxrect{r}$, then the \field{second $\tau$-in pipe field} of the corresponding color $\facet{k}{\pm}(\wang{t})$ must be blank;
  \item (Start of wires) Let $(\rect{r}_{\ell},\vec{i}_{\ell})_{\iota \leq \ell < \kappa'-1}$ be the \field{position field} of $\dec(\wang{t})$. An entry $((i,j,b),\mathtt{start})$ appears in the \field{second $\tau$-in pipe field} of $\dec(\wang{t})$ if and only if $(\vec{i}_{\ell})_{\iota \leq \ell < \iota'} = (\vec{0},\dots,\vec{0})$ and its \field{first $\tau$-in pipe field} encodes the same $(i,j,b)$; otherwise, no entry $(\cdot,\mathtt{start})$ appears;
  \item (End of wires) Assume that the \field{$\tau$-compzone field of level $\kappa'-1$} of $\dec(\wang{t})$ is some non-blank~$(\compzone,\vec{i})$, and that $\vec{i}$ belongs to the smallest facet $\facet{h}{-}(\compzone) = \sfacet(\compzone)$. If the $\ram{input}$ variable of the \field{$\tau$-processor field of level $\kappa'-1$} of $\facet{h}{+}(\wang{t})$ encodes a tuple $(i+2,j,b)$, then either the tuple $((i,j,b),\mathtt{end})$ is the only wire end in the \field{second $\tau$-in pipe field} of $\dec(\wang{t})$, or none appear at all. Otherwise (\textit{i.e.}~if $\vec{i} \notin \sfacet(\compzone)$), it contains no wire end $(\cdot,\mathtt{end})$.\qedhere
  \end{itemize}
\end{details}

\begin{claim}[label={pclm:in-pipe}]
  Let $\mwang{t}[\kappa'-1]$ be a $\tau$-macro-tile of level $\kappa'-1$, and let $a_{\kappa'} = (a_0,\dots,a_{2d}) \in \Swang[\ast]$ be the symbol given by the input arguments $(\ram{I}[2],\dots,\ram{I}[2d+2])$ of its computation zone. Then for any $\sigma$-macro-tile $\mwang{t}$ appearing in $\mwang{t}[\kappa'-1]$ with non-blank output argument $\ram{O}[3] = (i,j,b)$, we have:
  \begin{itemize}
  \item $\mwang{t}$ appears along the smallest facet of $\mwang{t}[\kappa'-1]$;
  \item $i$ ranges in $\interval{0}{2d}$, $j$ ranges in $\interval{0}{|a_{i}|-1}$ and $b = a_{i}(j)$.
  \end{itemize}
\end{claim}
\noindent In particular, the first item ensures that the wire starts of the \field{$\tau$-in pipes} satisfy the hypotheses of the~\crtlnameref{pos:lem:routing-lemma}, and thus ensure that a wiring is always possible.

\begin{figure}[ht]
  \begin{tikzpicture}[scale=0.07]
    \draw[LevelI!30!TilingGrid,opacity=0.12] (0,0) grid (128,80);
    \foreach[count=\k] \y/\desty in {0/5,11/30,24/26,63/11} {
      \fill[RainbowB!20!RainbowC,opacity=0.5] (0,\y) rectangle ++(1,1);
    }
    \foreach[count=\k from 0] \x in {18, 36, 47, 59, 76, 86, 100, 118} {
      \ifnum0<\k\relax{
        \draw[very thick,LevelINext,dashed,opacity=0.3] (\x,0) -- ++(0,80);
      }\else{%
        \draw[very thick,LevelINext] (\x,0) -- ++(0,80);
      }\fi
    }
    \foreach \y in {11, 24, 34, 48, 63} {
      \draw[very thick,LevelINext] (0,\y) -- (18,\y);
      \draw[very thick,LevelINext,dashed,opacity=0.35] (18,\y) -- (128,\y);
    }
    \foreach[count=\k] \y/\desty in {0/5,11/26,24/18,63/11} {
      \tikzmath{\nb=\k*5+20;}
      \draw[RainbowE,thick,-{Classical TikZ Rightarrow[length=1.8pt]},rounded corners=0.8pt] ($(0,\y) + (0.5,0.5)$) foreach \p in {1,...,5} {-- ++(0,1) -- ++(1,0)} -- ($(\nb,\y) + (0.5,5.5)$) -- ($(\nb,\desty) + (0.5,0.5)$) -- ($(109,\desty) + (-0.2,0.5)$);
    }
    \fill[pattern={Lines[line width=0.7pt,angle=45,distance={6pt/sqrt(2)}]},pattern color=LevelKLast!40,opacity=0.5] (0,0) rectangle (110,80);
    \draw[LevelKLast,opacity=0.5] (110,0) -- (110,80);
    \fill[LevelKLast,opacity=0.5] (109,0) rectangle (110,80);
    \fill[RainbowE] (109,0) rectangle ++(1,30);
    \draw[LevelI!20,opacity=0.7] (109,0.01) grid (110,79.99);
    \draw[ultra thick,LevelKLast] (-0.3,-0.3) rectangle (128.3,80.3);
    \node[LevelKLast!90!black] (A) at (-15,40) {\scalebox{0.7}{\parbox{2.2cm}{\centering Smallest facet of $\mwang{t}[\kappa'-1]$}}};
    \draw[LevelKLast!80!black,->,looseness=1] ($(A.south) + (0,0.5)$) to[in=180,out=-90] (-0.6,30);
    \node[LevelINext!90!black] (B) at (143,40) {\scalebox{0.7}{\parbox{2.2cm}{\centering Output facet of $\mwang{t}[\kappa'-1]$}}};
    \draw[LevelINext!80!black,->,looseness=0.6] ($(B.south) + (0,0.5)$) to[out=-90,in=0] (110.2,30);
  \end{tikzpicture}
  \caption{A wiring of \field{second $\tau$-in pipes}.}
  \subcaption{In a $\tau$-macro-tile $\mwang{t}[\kappa'-1]$ of level $\kappa'-1$
    (drawn as\hspace{1.2ex}\begin{tikzpicture}[baseline=0.2ex,scale=0.25]
      \protect\draw[LevelKLast,thick] (0,0) rectangle ++(1,1);
    \end{tikzpicture}),
    the corner $\vec{0}$ of some $\sigma$-macro-tiles~$\mwang{t}$
    (drawn as\hspace{1.2ex}\begin{tikzpicture}[baseline=0.2ex,scale=0.25]
      \protect\draw[LevelINext,thick] (0,0) rectangle ++(1,1);
    \end{tikzpicture})
    appearing along the smallest facet of $\mwang{t}[\kappa'-1]$ start a wire carrying a tuple $(i,j,b) \in \N^2 \times \{0,1\}$, which goes towards the bit $\ram{I}[i+2][j]$ on the input facet of $\mwang{t}[\kappa'-1]$.}
\end{figure}
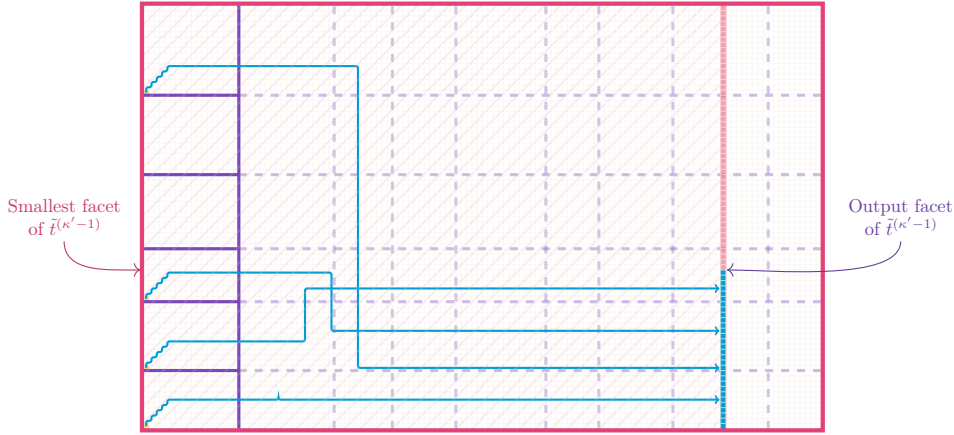

\bigskip
Let us now consider the second half of these pipes, also called \field{out pipes}: in the tilings of~$\Swang[\iota'][S]$, the $\tau$-macro-tile $\mwang{t}[\kappa']$ of level $\kappa'$ must distribute bits from their output symbol $a_{\kappa'} = (\ram{O}[0],\dots,\ram{O}[2d])$ towards some individual tiles of $\Swang[\iota'][S]$. As this proof is based upon a fixed-point argument, we must implement the same behavior in the tiles of $\Swang[\iota][S]$ (\emph{mutatis mutandis}, \textit{i.e.}~at level $\kappa$ instead of $\kappa'$).

Furthermore, \Cref{pclm:in-pipe} imposes that the individual tiles of $\Swang[\iota][S]$ receiving these bits must appear on some facet of the macro-tiles of level $\kappa-1$ \emph{that we cannot identify}\footnote{Indeed, the smallest facet of a $\tau$-macro-tile of level $\kappa-1$ \emph{in the previous level of simulation} may not coincide with the smallest facet of the corresponding macro-tile of level $\kappa-1$ \emph{in the current level of simulation}.}, we resort to a non-deterministic guess: more precisely, we introduce an \field{out-destination field} that encodes inside each macro-tile of level $\kappa-1$ the index $k_0 \in \interval{1}{d}$ of the destination facet\footnote{It would also be possible to introduce yet another input argument $\ram{I}[\dots]$ in the program $F(e)$, and to carry this information to the \field{out-destination fields} with the \emph{copying method}; but a non-deterministic guess is sufficient for our purposes.}.

\begin{details} Let $\wang{t} \in \Swang[\iota][S]$ be a tile and let $(\rect{r}_{\ell},\vec{i}_{\ell})_{\iota \leq \ell < \kappa'-1}$ be the \field{position field} of $\dec(\wang{t})$. The \field{out-destination field} of $\dec(\wang{t})$ is always some non-blank $k_0 \in \interval{1}{d}$; and for every $k \in \interval{1}{d}$:
  \begin{itemize}
  \item If the partial odometer of rectangles $(\rect{r}_{\ell})_{\iota \leq \ell < \kappa-1}$ yields $(\vec{i}_{\ell})_{\iota \leq \ell < \kappa-1} \pm \basis{k} = (\odnone,\dots,\odnone)$, then the \field{out-destination field} of $\facet{k}{\pm}(\wang{t})$ is blank;
  \item Otherwise, the \field{out-destination field} of $\facet{k}{\pm}(\wang{t})$ also contains $k_0$.
  \end{itemize}
  In any macro-tile of level $\kappa-1$, the individual tiles $\wang{t} \in \Swang[\iota][S]$ thus have the same \field{out-destination field} in their decoration; and that the associated integer $k_0$ can take any value in $k_0 \in \interval{1}{d}$.
\end{details}

\noindent Let $\mwang{t}[\kappa]$ be a $\tau$-macro-tile of level $\kappa$. We design some \field{$\tau$-out pipes} that will carry tuples $(i,j,b)$ for $b = \ram{O}[i][j]$ from the output facet of $\mwang{t}[\kappa]$ towards some well-positioned tiles of $\Swang[\iota][S]$. To proceed, we introduce three \field{$\tau$-out pipe fields}. Based on the already-defined \field{$\tau$-wire fields}, they implement three successive iterations of the \crtlnameref{pos:lem:routing-lemma}:
\begin{itemize}[itemsep=3pt]
\item The \field{first $\tau$-out pipe fields} implement a routing of the bits $(i,j,b) \in \N^2 \times \{0,1\}$ from the arguments $\ram{O}[i][j]$ on the output facet of $\mwang{t}[\kappa]$ to its smallest facet $\sfacet(\mwang{t}[\kappa])$;\footnote{Since the computation zone of $\mwang{t}[\kappa]$ is a cut of the rectangle $\dom(\mwang{t}[\kappa])$, its input facet is not necessarily included in the smallest facet of $\mwang{t}[\kappa]$ (as in \Cref{fix:fig:out-pipes}).}
\item The \field{second $\tau$-out pipe fields} distribute the bits $(i,j,b)$ by permuting them on the smallest facet $\sfacet(\mwang{t}[\kappa])$, thus ensuring that each macro-tile of level $\kappa-1$ receives at most $\smash{\prod_{\ell=\iota}^{\kappa-2} \rspace_{\ell}}$ such bits\footnote{Where $\smash{\prod_{\ell=\iota}^{\kappa-2} \rspace_{\ell}}$ is a lower bound on the size of the smallest facet of any macro-tile of level $\kappa-1$, and is thus necessary to ensure that the third part of the \field{$\tau$-out pipes} will always admit a valid wiring.};
\item The \field{third $\tau$-out pipe fields} permute these received bits, inside each macro-tile $\mwang{t}[\kappa-1]$ of level $\kappa-1$, towards the facet $\facet{k_0}{-}(\mwang{t}[\kappa-1])$ given by the \field{out-destination fields}.
\end{itemize}

\begin{figure}[ht]
  \colorlet{ArgA}{RainbowF!88}
  \colorlet{ArgB}{RainbowG!50!RainbowA}
  \vspace*{-0.2cm}
  \centerline{\begin{tikzpicture}[scale=0.32]
    \pgfmathsetmacro{\compzone}{12.5}
    \pgfmathsetmacro{\nbw}{8}
    \pgfmathsetmacro{\nbwires}{6*\nbw-1}
    \clip (-26.2,-2.8) rectangle (26.4,27.2);
    \draw[LevelI!10,step=0.25] (0,0) grid (20,25);
    \begin{scope}[transparency group]
      \fill[LevelK,opacity=0.3] (0,0) rectangle ++(0.25,\compzone);
      \draw[LevelI!20,step=0.25] (0,0) grid ++(0.25,\compzone);
      \fill[LevelK,opacity=0.3] (\compzone,\compzone) rectangle ++(-0.25,-\compzone);
      \fill[ArgA] ($(\compzone,0) + (-0.25,0)$) rectangle ++($(0.25,0) + \nbwires*(0,0.25)$);
      \draw[LevelI!20,step=0.25] (\compzone,\compzone) grid ++(-0.2501,-\compzone);
      \fill[opacity=0.3,pattern={Lines[line width=0.6pt,angle=45,distance={0.25cm/sqrt(2)}]},pattern color=LevelK] (0,0) rectangle (\compzone,\compzone);
      \draw[thick,LevelK!60] (0,0) rectangle ($(\compzone,\compzone) + (0.02,0.02)$);
    \end{scope}
    \foreach \x in {3, 7, 10, 13, 17} {
      \draw[LevelI+1,thick] (\x,0) -- (\x,25);
    }
    \foreach \y in {3, 5, 9, 12, 14, 19, 21, 24} {
      \draw[LevelI+1,thick] (0,\y) -- (20,\y);
    }
    \draw[ultra thick,LevelK] (-0.1,-0.1) rectangle (20.1,25.1);
    \begin{scope}[transparency group]
      \foreach \i in {0,...,\nbwires} {
        \tikzmath{int \nb; \oi=\nbwires-\i; \y=\i*0.25; \nx=\oi*0.25;
          \nb=div(\oi,\nbw); \nm=0.25*Mod(\oi,\nbw);
          \nnx=array({0,3,7,10,13,17},\nb); \dy=\nbw*\nb*0.25;}
        \draw[very thin,opacity=0.4,ArgA] ($(\compzone,\y) + (-0.1,0.1)$) -- ($(\nx,\y) + (0.1,0.1)$)
        \ifnum0<\i\relax{-- ($(\nx,0) + (0.1,0.1)$)}\fi
        \ifnum0<\nb\relax{-- ($(\nnx,\nnx) + (\nm,-\dy) + (0.1,0.1)$)}\fi;
        \ifnum0=\nb\relax{%
          \draw[thin,opacity=0.9,ArgA!60!ArgB,-{Classical TikZ Rightarrow[length=0.9pt]}] ($(\nx,0) + (0.1,1)$) -- ($(\nx,0) + (0.1,0.1)$);
        }\else{%
          \draw[thin,opacity=0.9,ArgA!60!ArgB,-{Classical TikZ Rightarrow[length=0.9pt]}] ($(\nnx,\nnx) + (\nm,-\dy) + (0.1,0.1)$) -- ($(\nnx,0) + (\nm,0) + (0.1,0.1)$);
        }\fi
      }
    \end{scope}
    \draw[thick] (16.5,1.5) circle (4.1cm);
    \draw[thick, dashed,looseness=1.2] (-4,20) to[out=40,in=90] (16.5,5.8);
    \begin{scope}[shift={(-14,15)}]
      \begin{scope}
        \clip (0,0) circle (12cm);
        \fill[white] (-12,-12) rectangle (12,12);
        \begin{scope}[scale=3.02,shift={(-3.45,-1.5)}]
          \fill[ArgA] (-0.5,0) rectangle ++($(0.166666,0) + \nbwires*(0,0.1666666)$);
          \draw[LevelI!20,step=0.166666666] (-3,0) grid (7,9);
          \draw[LevelK!60,thick] (-0.333333,0) -- ++(0,9);
          \begin{scope}[transparency group,opacity=0.8]
            \path[clip,scope fading=north,fit fading=true] (-3,0) rectangle (7,7);
            \foreach \mx in {-3,0,4} {
              \foreach \i in {0,...,15} {
                \tikzmath{\x=\i*0.16666666;}
                \draw[ArgA!60!ArgB,-{Classical TikZ Rightarrow[length=0.9pt]}] ($(\mx,0) + (\x,0) + (0.12444,7)$) -- ($(\mx,0) + (\x,0) + (0.12444,0.0833)$);
              }
            }
          \end{scope}
          \begin{scope}
            \fill[ArgB!80!blue,opacity=0.5] (0,0) rectangle ++(2.66666,0.166666);
            \draw[LevelI!20,step=0.166666666] (0,0) grid ++(2.66666,0.166666);
            \foreach \i in {0,...,15} {
              \tikzmath{\x=\i*0.1666666; \ny=array({6,7,0,5,0,2,2,1,1,0,0,3,4,8,0,1},\i)*0.1666666; \nx=array({8,9,0,14,3,1,12,4,10,6,11,5, 2,7,15,13},\i)*0.1666666;}
              \draw[ArgB,-{Classical TikZ Rightarrow[length=1.5pt]},rounded corners=1pt] ($(\x,0) + (0.08333,0.08333)$) -- ($(0,0.1666) + (\x,\ny) + (0.08333,0.08333)$) -- ($(0,0.1666) + (\nx,\ny) + (0.12555,0.08333)$) -- ($(\nx,0) + (0.12555,0.16666)$);
            }
          \end{scope}
          \begin{scope}[shift={(4,0)}]
            \fill[ArgB!80!blue,opacity=0.5] (0,0) rectangle ++(0.166666,2.666666);
            \draw[LevelI!20,step=0.166666666] (0,0) grid ++(0.166666,2.666666);
            \draw[ArgB,-{Classical TikZ Rightarrow[length=1.2pt]},rounded corners=1pt] (0.08333,0.08333) -- (0.083333,0.13555) [rounded corners=0.6pt] -- (0.3,0.13555) -- (0.3,0.08333) -- (0.16666,0.08333);
            \foreach \i in {1,...,15} {
              \tikzmath{\x=\i*0.1666666; \ny=array({0,3,9,15,6,8,11,2,4,10,12,5,13,14,1,7},\i)*0.1666666;}
              \draw[ArgB,-{Classical TikZ Rightarrow[length=1.5pt]},rounded corners=1pt] ($(\x,0) + (0.08333,0.08333)$) -- ($(\x,\ny) + (0.08333,0.08333)$) -- ($(0,\ny) + (0.16666,0.0833)$);
            }
          \end{scope}
          \foreach \x in {0,4,7} {
            \draw[LevelI+1,very thick] (\x,0) -- (\x,5);
          }
          \foreach \y in {3,5} {
            \draw[LevelI+1,very thick] (0,\y) -- (7,\y);
          }
          \draw[LevelK,ultra thick] (-2,-0.02) -- (7.02,-0.02) -- (7.02,9);
        \end{scope}
      \end{scope}
      \draw[ultra thick] (0,0) circle (12cm);
    \end{scope}
    \node[LevelK!90!black] (A) at (-4,2) {\scalebox{0.9}{\parbox{2.1cm}{\centering Smallest \\ facet of $\mwang{t}[\kappa]$}}};
    \draw[LevelK!80!black,->,looseness=1] ($(A.south) + (0,0.1)$) to[out=-90,in=-90] (2,-0.2);
    \node[LevelK!90!black] (B) at (24,10) {\scalebox{0.9}{\parbox{2.1cm}{\centering Output \\ facet of $\mwang{t}[\kappa]$}}};
    \draw[LevelK!80!black,->,looseness=1] ($(B.west) + (1.4,0.2)$) to[out=180,in=0] ($(\compzone,7) + (0.04,0)$);
  \end{tikzpicture}}

  \vspace*{-0.6cm}
  \caption{A wiring of \field{$\tau$-out pipe fields}.}
  \subcaption{Inside a $\tau$-macro tile $\mwang{t}[\kappa]$ of level $\kappa$
    (drawn as\hspace{1.2ex}\begin{tikzpicture}[baseline=0.2ex,scale=0.25]
      \protect\draw[LevelK,thick] (0,0) rectangle ++(1,1);
    \end{tikzpicture}),
    wires
    ~\begin{tikzpicture}[baseline=0.2ex,scale=0.25]
      \protect\draw[RainbowF!80,thick,-{Classical TikZ Rightarrow[length=2.2pt]}] (-0.1,0.5) -- (0.5,0.5) -- (1,0);
    \end{tikzpicture}\hspace{1.2ex}
    join the output facet of $\mwang{t}[\kappa]$ with the smallest facet, and then organize a permutation of the output bits. Inside the macro-tiles of level $\kappa-1$
    (drawn as\hspace{1.2ex}\begin{tikzpicture}[baseline=0.2ex,scale=0.25]
      \protect\draw[LevelI+1,thick] (0,0) rectangle ++(1,1);
    \end{tikzpicture})
    appearing on the smallest facet of $\mwang{t}[\kappa]$, wires
    ~\begin{tikzpicture}[baseline=0.2ex,scale=0.25]
      \protect\draw[RainbowG!50!RainbowA,thick,-{Classical TikZ Rightarrow[length=2.2pt]}] (-0.1,0.5) -- (0.5,0.5) -- (1,0);
    \end{tikzpicture}\hspace{1.2ex}
    carry these bits toward the facets determined by their respective \field{out-destination fields}.\\
    This figure implicitly assumes that $\iota < \kappa-1$: if $\iota = \kappa-1$, the
    ~\begin{tikzpicture}[baseline=0.2ex,scale=0.25]
      \protect\draw[RainbowG!50!RainbowA,thick,-{Classical TikZ Rightarrow[length=2.2pt]}] (-0.1,0.5) -- (0.5,0.5) -- (1,0);
    \end{tikzpicture}\hspace{1.2ex}
    wiring becomes trivial.}
  \label{fix:fig:out-pipes}
\end{figure}
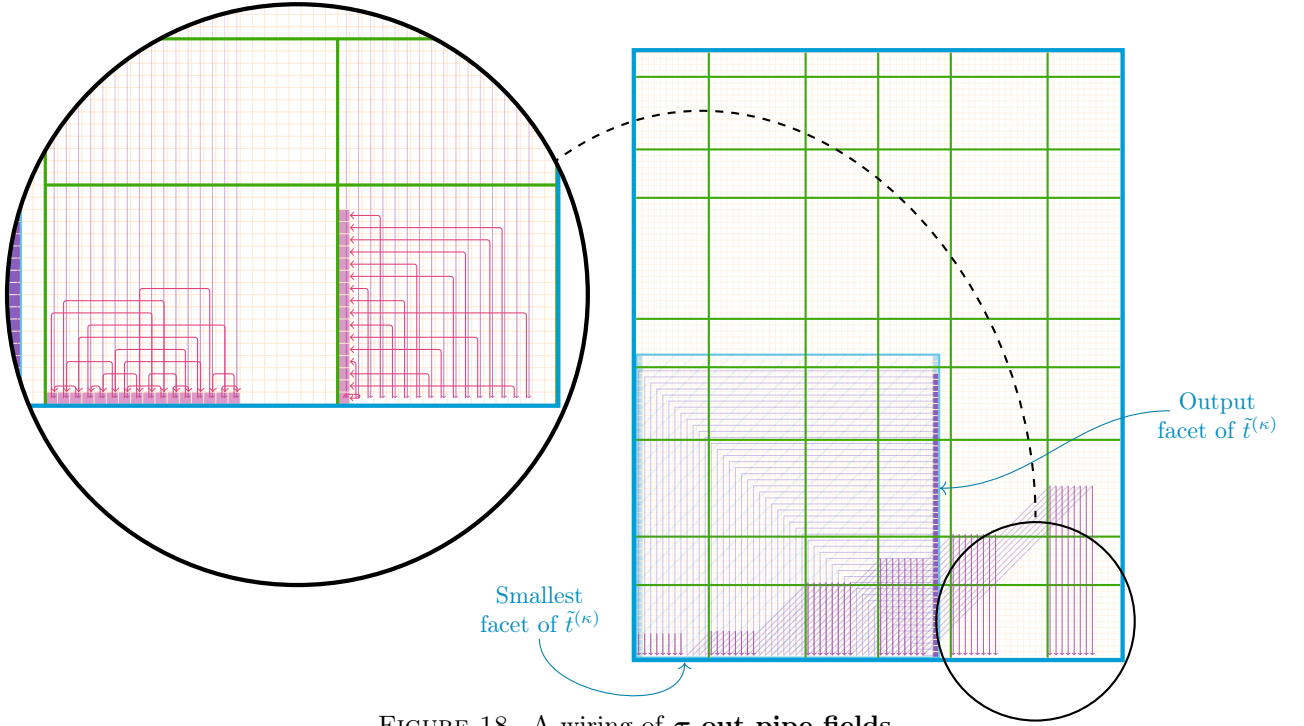

\begin{details}
  The \field{first}, \field{second} and \field{third $\tau$-out pipe fields} are of the form:
  \[ (w_i)_{i \in I} \in ((\N^2 \times \{0,1\}) \times \N)^{\ast} \quad ; \; \text{or} \quad (w_i)_{i \in I} \in ((\N^2 \times \{0,1\}) \times \{\mathtt{start}, \mathtt{end}, \mathtt{cross} \})^{\ast} \]
  to carry entries $(i,j,b) \in \interval{0}{2d} \times \N \times \{0,1\}$ using the \crtlnameref{pos:lem:routing-lemma}. Based upon similar wirings, we skip the details and only highlight the bounds, starting and ending points of these wires.

  \medskip
  Let $\wang{t} \in \Swang[\iota][S]$ be a tile. For the \field{first $\tau$-out pipe field}, let $(\compzone,\vec{i})$ and $(\pxrect{r}_{\kappa},\pxvec{i}_{\kappa})$ denote the \field{$\tau$-compzone field of level $\kappa$} and the \field{pixel position field of level $\kappa$} of $\dec(\wang{t})$:
  \begin{itemize}
  \item (Border condition) If $\vec{i} \pm \basis{k} \notin \compzone$ (\textit{i.e.}~$\wang{t}$ appears on the border of the computation zone $\compzone$), then the corresponding $\facet{k}{\pm}(\wang{t})$ has blank \field{first $\tau$-out pipe field};
  \item (Start of wires) If $\facet{h}{+}(\compzone)$ denotes the output facet of the computation zone $\compzone$, and $\vec{i}$ belongs to $\facet{h}{+}(\compzone)$, then the \field{first $\tau$-out pipe field} of $\dec(\wang{t})$ contains an entry $((i,j,b),\mathtt{start})$ if and only if the $\ram{output}$ variable of $\facet{h}{-}(\wang{t})$ encodes the same $(i,j,b)$. Otherwise (\textit{i.e.}~$\vec{i}$ is not part of the output facet), it contains no wire start $(\cdot,\mathtt{start})$;
  \item (End of wires) Denote $\facet{p}{-}(\compzone)$ the facet of $\compzone$ that is included in the smallest facet $\sfacet(\pxrect{r}_{\kappa})$ (\textit{c.f.}~\Cref{pclm:tau-macro-tile-facet-inclusion}). If $\vec{i} \in \facet{p}{-}(\compzone)$, the \field{first $\tau$-out pipe field} of $\dec(\wang{t})$ can either contain a single entry $((i,j,b),\mathtt{end})$ for any $i \in \interval{0}{2d}$, $j \in \interval{0}{\rvol{\compzone}}$ and $b \in \{0,1\}$; or no entry $(\cdot,\mathtt{end})$ at all. Otherwise (\textit{i.e.}~$\vec{i}$ is not part of the facet $\facet{p}{-}(\compzone)$), it contains no wire end $(\cdot,\mathtt{end})$.
  \end{itemize}
  For the \field{second $\tau$-out pipe field}, let $(\pxrect{r}_{\kappa},\pxvec{i}_{\kappa})$ and $(\pxrect{r}_{\kappa-1},\pxvec{i}_{\kappa-1})$ denote the \field{pixel position fields of level $\kappa$} and \field{$\kappa-1$} of $\dec(\wang{t})$:
  \begin{itemize}
  \item (Border condition) If $\vec{i} \pm \basis{k} \notin \pxrect{r}_{\kappa}$ (\textit{i.e.}~$\wang{t}$ appears on the border of a macro-tile of level $\kappa$), then the corresponding $\facet{k}{\pm}(\wang{t})$ has blank \field{second $\tau$-out pipe field};
  \item (Start of wires) The \field{second $\tau$-out pipe field} of $\dec(\wang{t})$ contains an entry $((i,j,b),\mathtt{start})$ for $i,j \in \N^2$ and $b \in \{0,1\}$ if and only if the \field{first $\tau$-out pipe field} of $\dec(\wang{t})$ contains the same entry $((i,j,b),\mathtt{end})$;
  \item (End of wires) If $\pxvec{i}_{\kappa}$ belongs to the smallest facet $\facet{h}{-}(\pxrect{r}_{\kappa}) = \sfacet(\pxrect{r}_{\kappa})$, then $\pxvec{i}_{\kappa-1}$ belongs to the facet $\facet{h}{-}(\pxrect{r}_{\kappa-1})$ and we denote $n \in \N$ its boustrophedon indexing in $\facet{h}{-}(\pxrect{r}_{\kappa-1})$). In this case, if $n \leq \prod_{\ell=\iota+1}^{\kappa-1} \rspace_{\ell}$, then the \field{second $\tau$-out pipe field} of $\dec(\wang{t})$ can either contain an entry $((i,j,b),\mathtt{end})$ for some $(i,j,b) \in \interval{0}{2d} \times \interval{0}{\rvol{\pxrect{r}_{\kappa}}} \times \{0,1\}$ or be blank. Otherwise (\textit{i.e.}~$n$~is too large or $\pxvec{i}_{\kappa}$ does not belong to $\sfacet(\pxrect{r}_{\kappa})$), it must be blank.
  \end{itemize}
  For the \field{third $\tau$-out pipe field}, let $(\pxrect{r}_{\kappa},\pxvec{i}_{\kappa})$ and $(\pxrect{r}_{\kappa-1},\pxvec{i}_{\kappa-1})$ denote the \field{pixel position fields of level $\kappa$} and \field{$\kappa-1$} of $\dec(\wang{t})$:
  \begin{itemize}
  \item (Border condition) If $\vec{i} \pm \basis{k} \notin \pxrect{r}_{\kappa-1}$ (\textit{i.e.}~$\wang{t}$ appears on the border of a macro-tile of level $\kappa-1$), then the corresponding $\facet{k}{\pm}(\wang{t})$ has blank \field{third $\tau$-out pipe field};
  \item (Start of wires) The \field{third $\tau$-out pipe field} of $\dec(\wang{t})$ contains en entry $((i,j,b),\mathtt{start})$ for $i,j \in \N^2$ and $b \in \{0,1\}$ if and only if the \field{second $\tau$-out pipe field} of $\dec(\wang{t})$ contains the same entry $((i,j,b),\mathtt{end})$;
  \item (End of wires) Let $k_0 \in \interval{1}{d}$ denote the \field{out-destination field} of $\dec(\wang{t})$. If $\pxvec{i}_{\kappa-1}$ belongs to the facet $\facet{k_0}{-}(\pxrect{r}_{\kappa-1})$, then the \field{third $\tau$-out pipe field} of $\dec(\wang{t})$ can either contain an entry $((i,j,b),\mathtt{end})$ for some $(i,j,b) \in \interval{0}{2d} \times \interval{0}{\rvol{\pxrect{r}_{\kappa}}} \times \{0,1\}$ or be blank. Otherwise (\textit{i.e.}~$\pxvec{i}_{\kappa-1}$ does not belong to $\facet{k_0}{-}(\pxrect{r}_{\kappa-1})$), it contains no wire end $(\cdot,\mathtt{end})$. \qedhere
  \end{itemize}
\end{details}

Since the three \field{$\tau$-out pipes} are chained together, we obtain by \Cref{pos:prop:wiring-bijection} that the wire starts of the \field{first $\tau$-out pipes} and the wire ends of the \field{third $\tau$-out pipes} are in bijection inside any $\tau$-macro-tile of level $\kappa$. We thus deduce:

\begin{claim}[label={pclm:out-pipe}]
  Let $\mwang{t}[\kappa]$ be a $\tau$-macro-tile of level $\kappa$, and let $a_{\kappa} = (a_0,\dots,a_{2d}) \in \Swang[\ast]$ be the symbol given by the output arguments $(\ram{O}[0],\dots,\ram{O}[2d])$ of its computation zone. Then for every $i \in \interval{0}{2d}$ and $j \in \interval{0}{|a_i|-1}$, there exists a unique Wang tile $\wang{t}$ in the $\tau$-macro-tile $\mwang{t}[\kappa]$ such that the \field{third $\tau$-out pipe field} of $\dec(\wang{t})$ contains an entry $((i,j,b'),\mathtt{end})$; in which case:
  \begin{itemize}
  \item The bit $b'$ satisfies $b' = a_{i}(j)$;
  \item Denoting $\mwang{t}[\kappa-1]$ the macro-tile of level $\kappa-1$ containing $\wang{t}$, the tile $\wang{t}$ appears on the facet $\facet{k_0}{-}$ of $\mwang{t}[\kappa-1]$, where $k_0$ is the \field{out-destination field} of $\dec(\wang{t})$.
  \end{itemize}
  Furthermore, every alternative choice for the \field{out-destination fields} of the sub-macro-tiles $\mwang{t}[\kappa-1]$ (of level $\kappa-1$ in $\mwang{t}[\kappa]$) results (up to a rewiring of the \field{third $\tau$-out pipes}) in an alternative valid $\tau$-macro-tile of level $\kappa$.
\end{claim}

\begin{numerics}
  The three \field{pixel position fields} are of bit length $\bigO(\sum_{\ell=\iota}^{\kappa'-2} \log \rtime_{\ell}) = \polylog \K[][s] \cdot \bigO(\pxspace{\iota}[2\delta])$. Similarly to the already defined wirings, the two \field{$\tau$-in pipe fields} of a tile $\wang{t} \in \Swang[\iota][S]$ are of bit length $\bigO(\sum_{\ell=\iota}^{\kappa'-2} \log \rtime_{\ell}) + \bigO(\log \pxspace{\kappa'-1}) = \polylog \K[][s] \cdot \bigO(\pxspace{\iota}[2\delta])$; and its \field{$\tau$-out pipe fields} are of bit length $\bigO(\sum_{\ell=\iota}^{\kappa-1} \log \rtime_{\ell} + \log \pxspace{\kappa}) = \polylog \K[][s] \cdot \bigO(\pxspace{\iota}[2\delta])$.
\end{numerics}

\enlargethispage{1.5\baselineskip}
\subsection{\boldmath Finalizing the algorithm \texorpdfstring{$F(e)$}{F(e)}}

\subsubsection{Exit and return}
There ends the definition of the tileset $\Swang[\iota][S]$, and we now finalize the algorithm $F(e)$ recognizing it. For fixed arguments $(\K[][s],\K[][w])$, $\iota \in \N$, $(\rspace_{\ell},\rtime_{\ell})_{0 \leq \ell < \kappa-1}$ and $(\rect{r}_{\ell},\vec{i}_{\ell})_{\iota \leq \ell < \kappa-1}$, we design the program $F(e)$ so that tiles ${\wang{t} = (c_1^{\scriptscriptstyle -},\dots,c_d^{\scriptscriptstyle -}, c_1^{\scriptscriptstyle +},\dots, c_d^{\scriptscriptstyle +},c_{\decsymbol}) \in \Swang[\ast]}$ are accepted by a valid computation of $\funram{F(e)}((\K[][s],\K[][w]), \iota, \dots, \wang{t})$ if and only if they satisfy all the conditions enumerated in the previous sections of this proof; in which case, it returns a tuple $(\nghb_{\kappa-1},\rect{r}_{\kappa-1},\vec{i}_{\kappa-1},u)$ where:
\begin{itemize}
\item $\nghb_{\kappa-1} \subseteq \{\pm \basis{k} : 1 \leq k \leq d\}$ is the eponymous entry in the \field{neighbor field} of $\dec(\wang{t})$;
\item $\rect{r}_{\kappa-1}$ and $\vec{i}_{\kappa-1}$ is the eponymous entry in the \field{position field} of $\dec(\wang{t})$;
\item $u \in \{0,1\}^{\ast}$ encodes some $(i,j,b) \in \interval{0}{2d} \times \N \times \{0,1\}$ if and only if the \field{third $\tau$-out pipe field} of $\dec(\wang{t})$ contains an entry $((i,j,b),\mathtt{end})$; and is otherwise blank.
\end{itemize}
Otherwise, we design $\funram{F(e)}$ to reject all computations on $\wang{t}$ and thus not return anything.

\smallskip
\noindent Intuitively, outputting the \field{neighbor field}, the \field{position field}'s entry of level $\kappa-1$  and the ends of the \field{$\tau$-out pipes} allows for the previous simulation level to read this data on the output facets of its $\sigma$-macro-tiles, and wire it towards its own computations (\textit{c.f.}~\Cref{fix:sec:tau-consistency,fix:sec:tau-pipes}).

\subsubsection{Time complexity and word length} Let $\ram{I} \in \{0,1\}^{\ast\ast}$ be an input array whose argument $\ram{I}[0]$ is interpreted as a pair of integers $\ram{I}[0] = (\K[][s],\K[][w])$, and whose \emph{bit size} is denoted $\smash{|\ram{I}| = \sum_{i=0}^{\len(\ram{I})-1} |\ram{I}[i]|}$.
Since all the validity conditions of $\funram{F}(e)$ can be checked by performing elementary operations (additions, multiplications, copies, comparisons\dots) on large integers and binary strings of size $|\ram{I}| + \K[][w] \cdot |\ram{I}[3]| + |e|$, which can all be implemented in the $\log$-RAM model in quasi-linear time, we deduce that $F(e)$ can be implemented to run in time
\[ \bigO \big(|\ram{I}| \cdot \polylog |\ram{I}| + |e| \big) + \bigO \big(\K[][w] \cdot |\ram{I}[3]| \cdot \polylog |\ram{I}[3]| \big) \]
and with word length
\[ \bigO \big(\log |\ram{I}| + \log \big(|e| + \K[][w] \big)\big). \]

\subsection{Fixed point theorem}

We now build a program $e \in \{0,1\}^{\ast}$ such that $e$ and $F(e)$ have the same behavior, thus ensuring that all levels of simulation run on the same program.

\subsubsection{Fixed point theorem} From any RAM program $e \in \{0,1\}^{\ast}$, the previous sections create a new program $F(e) \in \{0,1\}^{\ast}$ responsible for defining the various Wang tilesets $\Swang[\iota][S]$ (for parameters $\iota \in \N$, etc\dots). We now notice that the transformation $e \mapsto F(e)$ is computable and total: indeed, the program $e \in \{0,1\}^{\ast}$ is only ever used in $F(e)$ to determine the program implemented in the \field{processor fields} of the next level of simulation.

We thus deduce from the~\crtlnameref{calc:prop:fixed-point-theorem}~(\Cref{calc:prop:fixed-point-theorem}) that there exists a program $e \in \{0,1\}^{\ast}$ such that $\funram{e} = \funram{F(e)}$. Furthermore, the time complexity and the word length of this program $e \in \{0,1\}^{\ast}$ respectively satisfy $T_{e}(\ram{I}) = \bigO(T_{F(e)}(\ram{I})) + \bigO(1)$ and $W_{e}(\ram{I}) = W_{F(e)}(\ram{I}) + \bigO(1)$. In particular, since $F(e)$ is a $\log$-RAM program, so is $e$.

\subsubsection{Fixing constants}\label{fix:sec:fixing-constants}
We now fix, once and for all, the values of the integer constants $(\K[][s],\K[][w])$ (that are respectively the constant factors applied to the sizes of the macro-tiles (\textit{c.f.}~\Cref{fix:sec:positioning}) and the word length of their embedded computations (\textit{c.f.}~\Cref{fix:sec:computation-zone})). Notice that, for the construction to be correct, we need the $\sigma$-macro-tiles and $\tau$-macro-tiles to embed enough $\log$-RAM computation steps, and with large enough word length, to allow the computations of (respectively) $\funram{e}(\dots)$ and $\funram{\code{\tau}}(\dots)$ to run properly.

Recall that, in a valid tiling $x \in \Swang[\iota][S]^{\Z^d}$, a macro-tile of level $\kappa \leq \ell < \kappa'$ embeds, in a computation zone of domain $\compzone$, at least $\rspace(\compzone) \geq \K[][s] \cdot \pxspace{\ell}[\gamma]$ steps of RAM computations operating on variables of word length at most $2 \K[][w] \cdot \log \rvol{\compzone} \geq \K[][w] \cdot \log \pxspace{\ell}$.

\paragraph*{$\sigma$-macro-tiles}  Following the \emph{numerics} paragraphs from the previous sections, we deduce that for all fixed parameters $(\K[][s],\K[][w])$, $\iota$ and $(\rspace_{\ell},\rtime_{\ell})_{0 \leq \ell < \kappa-1}$, the macro-colors of a Wang tile $\wang{t} \in \Swang[\ast]$ accepted by $\funram{F(e)}$ have bit length $\polylog \K[][s] \cdot \K[][w] \cdot \bigO(\pxspace{\iota}[2\delta])$; and so do the remaining arguments $\ram{I}[0]$, $\ram{I}[1]$, $\ram{I}[2]$ and $\ram{I}[3]$ on their input facets.

We will thus choose $\K[][s]$ and $\K[][w]$ so that any $\sigma$-macro-tile has a computation zone large enough to embed the computations for the next level of simulation $\iota' > \iota$; \textit{i.e.}~for every $\ell \in \N$:
\[ \K[][w] \cdot \polylog \K[][s] \cdot \bigO(\pxspace{\ell}[2\delta] \cdot \polylog \pxspace{\ell}) \leq \K[][s] \cdot \pxspace{\ell}[\gamma]; \]
and similarly, that the word length of the embedded computations satisfies for every $\ell \in \N$:
\[ \bigO(\log \pxspace{\ell}[2\delta]) + \bigO(\log (\K[][s] + \K[][w])) \leq \K[][w] \cdot \log \pxspace{\ell}. \]
Since $2\delta < \gamma$, these bounds hold for all $\pxspace{\ell} \in \N$ if we set $\K[][w] = \sqrt{\K[][s]}$ and pick $\K[][s]$ large enough.

\paragraph*{$\tau$-macro-tiles} In the hypotheses of \Cref{sof:lem:fixpoint}, we assume that the substitution step $\smash{x^{(\ell+1)} \substep[\tau] x^{(\ell)}}$ is computed in $\smash{\K[t] \cdot \pxspace{\ell}[\alpha]}$ steps of RAM computations, with word length $\smash{W_{\code{\tau}}(\ram{I}) \leq \K[w] \log |\ram{I}|}$ on input arrays of bit size $|\ram{I}|$, and that $\smash{\norm{x^{(\ell+1)}} \leq \K[t] \cdot \pxspace{\ell}[\alpha]}$ (\textit{c.f.}~\crtlnameref{fix:sec:housekeeping}). We will thus choose $\K[][s]$ and $\K[][w]$ so that the computation zone of any $\tau$-macro-tile of level $\ell$ is large enough to both embed the arguments and the computations of $\smash{x^{(\ell+1)} \substep[\tau] x^{(\ell)}}$, \textit{i.e.} that for every $\ell \in \N$:
\[ \K[t] \cdot \pxspace{\ell}[\alpha] \leq \K[][s] \cdot \pxspace{\ell}[\gamma]; \]
and similarly, that the word length of the embedded computations satisfies for every $\ell \in \N$:
\[ \K[w] \cdot \bigO(\log (2d \cdot \rtime_{\ell} + \K[t] \cdot \pxspace{\ell}[\alpha])) \leq \K[][w] \cdot \log \pxspace{\ell}. \]
Since $2\delta < \gamma$ and $\log \lambda_{\ell} = \bigO(\log \pxspace{\ell})$, this indeed holds for any value of $\pxspace{\ell}$ if we set $\K[][w] = \sqrt{\K[][s]}$ and pick $\K[][s]$ large enough. In what follows, we thus fix such a pair $(\K[][s],\K[][w])$.

\subsection{End of the proof}
\label{fix:sec:final-sft}

The previous section has thus defined a function $\funram{e}(\dots)$ that recognizes a subset of valid tiles from $\Swang[\ast]$ that, in a sense, implements an infinite hierarchy of simulations. We now create a shift of finite type $\mathcal{X}$ that implements the ``level $0$'' of these simulations.

\subsubsection*{The shift $\mathcal{X}$}

Let $X \subseteq \alphabet^{\Z^d}$ be the initial sofic shift space, for a finite $\alphabet \subseteq \{0,1\}^{\ast}$. Since $X$ is sofic, there exists some shift of finite type $\smash{X' \subseteq \alphabet[B]^{\Z^d}}$ on some finite alphabet $\alphabet[B]$ and a projection $\pi' \colon \alphabet[B] \to \alphabet$ such that $\pi'(X') = X$.
\smallskip
We now define a shift space $\mathcal{X}$ whose symbols are of the form
\[ (a,b,\iota',\kappa',(\rspace_{\ell},\rtime_{\ell})_{0 \leq \ell < \kappa'}, (\rect{r}_{\ell},\vec{i}_{\ell})_{0 \leq \ell < \kappa'},(a_{\ell})_{0 \leq \ell \leq \kappa'}, \nghb_{\kappa'-1}, \wang{t}) \]
where:
\begin{itemize}
\item $a \in \alphabet$ is a letter from the original shift $X$;
\item $b \in \alphabet[B]$ is a letter from the cover of finite type $X'$ of $X$, and $a = \pi(b)$;
\item $\iota', \kappa'$ are two integers $\iota' < \kappa'$; furthermore, we assume that $\kappa' \leq N$ for some constant $N \in \N$ to be determined later;
\item $(\rspace_{\ell},\rtime_{\ell})_{0 \leq \ell < \kappa'}$ are pairs of integers satisfying the hypotheses of the \crtlnameref{fix:lem:staircase-lemma}, \textit{i.e.}~denoting $\pxspace{\ell} = \prod_{i < \ell} \rspace_{i}$:
  \begin{itemize}[label={$\triangleright$}]
  \item $2 \leq \rtime_{\ell} \leq 2^{\K \cdot \pxspace{\ell}[\delta]}$ ; $\rspace_{\ell}^{\K \cdot \log \pxspace{\ell+1}} \geq \rtime_{\ell}$ and $\prod_{i = 0}^{\kappa'-1} \rspace_{i} \geq \K[][s] \cdot \pxspace{\kappa'}[\gamma]$;
  \item The integer $\kappa'$ is the minimal integer satisfying $\prod_{i = 0}^{\kappa'-1} \rspace_{i} \geq \K[][s] \cdot \pxspace{\kappa'}[\gamma]$;
  \item The integer $\iota'$ is the maximal integer in $\interval{0}{\kappa'-1}$ satisfying $\prod_{i=\iota'}^{\kappa'-1} \rspace_{i} \geq \K[][s] \cdot \pxspace{\kappa'}[\gamma]$;
  \end{itemize}
\item $(\rect{r}_{\ell},\vec{i}_{\ell})_{0 \leq \ell < \kappa'}$ are pairs $\rect{r}_{\ell} \in \Srect_0$ and $\vec{i}_{\ell} \in \rect{r}_{\ell}$ such that $\rspace_{\ell} \leq \rspace(\rect{r}_{\ell})$ and $\rtime(\rect{r}_{\ell}) \leq \rtime_{\ell}$;
\item $(a_{\ell})_{0 \leq \ell \leq \kappa'}$ is a sequence of symbols such that:
  \begin{itemize}[label={$\triangleright$}]
  \item $\dec(a_0) = a$;
  \item For every $0 \leq \ell \leq \kappa'$, the Wang tile $a_{\ell} \in \Swang[\ast]$ is of word length $\norm{a_{\ell}} \leq \K[\code{\tau}] \cdot \pxspace{\ell-1}[\alpha]$;
  \item For every $0 \leq \ell < \kappa'$, the Wang tile $a_{\ell} \in \Swang[\ast]$ is a valid output of $\funram{\code{\tau}}(\rect{r}_{\ell},\vec{i}_{\ell},a_{\ell+1})$ computed in time $\K[\code{\tau}] \cdot \pxspace{\ell}[\alpha]$.
  \end{itemize}
\item $\nghb_{\kappa'-1}$ is a valid set of directions $\nghb_{\kappa'-1} \subseteq \{\pm \basis{k} : 1 \leq k \leq d\}$;
\item $\wang{t} \in \Swang[\ast]$ is a Wang tile such that:
  \begin{itemize}[label={$\triangleright$}]
  \item $\funram{e}((\K[][s],\K[][w]), \iota', (\rspace_{\ell},\rtime_{\ell})_{0 \leq \ell < \kappa'-1}, (\rect{r}_{\ell},\vec{i}_{\ell})_{\iota' \leq \ell < \kappa'-1}, \wang{t})$ is an accepting computation of $\funram{e}$ that returns $(\nghb_{\kappa'-1},\rect{r}_{\kappa'-1},\vec{i}_{\kappa'-1},u)$ in time $\K[][s] \cdot \pxspace{\iota'}[\gamma]$ for some $u \in \{0,1\}^{\ast}$;
  \item Let us denote $(c_0,\dots,c_{2d}) = a_{\kappa'}$. If $u$ is non-blank, then $u$ encodes a tuple $(i,j,b)$ for some $i \in \interval{0}{2d}$, $j \in \interval{0}{|c_i|-1}$ and such that $b = c_i(j)$.
  \end{itemize}
\end{itemize}

\noindent Since $N \in \N$ is fixed, the alphabet of the shift space $\mathcal{X}$ is in fact finite. We actually set
\[ N = \frac{1}{1-\gamma} \log \K[][s] \]
to ensure that there always exists $\kappa' \leq N$ such that $\pxspace{\kappa'} \geq \K[][s] \cdot \pxspace{\kappa'}[\gamma]$: indeed, for any sequence $(\rspace_{\ell},\rtime_{\ell})_{0 \leq \ell < N}$, we have $\pxspace{N} \geq 2^N > \K[][s]^{\frac{1}{1-\gamma}}$ and thus $\pxspace{N} \geq \K[][s] \cdot \pxspace{N}[\gamma]$.

\medskip
As they follow a hierarchy of nested grids, the configurations of the shift $\mathcal{X}$ are straightforward:
\begin{itemize}
\item The integers $\iota'$, $\kappa'$ and the sequence $(\rspace_{\ell},\rtime_{\ell})_{0 \leq \ell < \kappa'}$ are constant across any configuration;
\item The sequences $(\rect{r}_{\ell},\vec{i}_{\ell})_{0 \leq \ell < \kappa'}$ should draw nested grids in an odometer-like fashion (\textit{c.f.}~the \field{position field} from \Cref{fix:sec:positioning}). Thus, any valid configuration $\chi \in \mathcal{X}$ defines a unique sequence of grids $(\grid{g}_{\ell})_{0 \leq \ell < \kappa'}$;
\item In a valid configuration $\chi \in \mathcal{X}$ defining grids $(\grid{g}_{\ell})_{0 \leq \ell < \kappa'}$, letters $a_{\ell}$ should be constant across adjacent position $\vec{i}_{\ell-1}$ in the same rectangle $\rect{r}_{\ell-1}$ of $\grid{g}_{\ell-1}$; furthermore, across two adjacent rectangles $\rect{r} = \rect{r}_{\vec{j}}$ and $\rect{r}' = \rect{r}_{\vec{j} + \basis{k}}$ in a grid $\grid{g}_{\ell-1}$, we enforce adjacency conditions between the corresponding symbols $a_{\ell}$ and $a_{\ell}'$, \textit{i.e.}~$\facet{k}{+}(a_{\ell}) = \facet{k}{-}(a_{\ell}')$;
\item In a valid configuration $\chi \in \mathcal{X}$ defining the grid $\grid{g}_{\kappa'-1}$, sets $\nghb_{\kappa'-1}$ should be constant across adjacent position $\vec{i}_{\kappa'-1}$ in the same rectangle $\rect{r}_{\kappa'-1}$ of $\grid{g}_{\kappa'-1}$; furthermore, across adjacent rectangles $\rect{r} = \rect{r}_{\vec{j}}$ and $\rect{r}' = \rect{r}_{\vec{j} + \basis{k}}$ in $\grid{g}_{\kappa'-1}$, the corresponding sets $\nghb_{\kappa'-1}$ and $\nghb_{\kappa'-1}'$ should respectively contain $\basis{k}$ and $-\basis{k}$;
\item In a valid configuration $\chi \in \mathcal{X}$ defining a value $\iota' < \kappa'$ and a grid $\grid{g}_{\iota'-1}$, the tile $\wang{t}$ should be constant across adjacent position $\vec{i}_{\iota'-1}$ in the same rectangle $\rect{r}_{\iota'}$ of $\grid{g}_{\iota'-1}$; furthermore, across two adjacent rectangles $\rect{r} = \rect{r}_{\vec{j}}$ and $\rect{r}' = \rect{r}_{\vec{j} + \basis{k}}$ in the grid $\grid{g}_{\iota'-1}$, we enforce adjacency conditions between the corresponding symbols $\wang{t}$ and $\wang{t}'$, \textit{i.e.}~$\facet{k}{+}(\wang{t}) = \facet{k}{-}(\wang{t}')$;
\item Finally, using the finite family of forbidden patterns of $X'$, we ensure that the projection of a valid configuration $\chi \in \mathcal{X}$ to the alphabet $\alphabet[B]$ results in a valid configuration of $X'$;
\end{itemize}

Since there exists a uniform bound (depending on $N$) on the largest possible cells in the grids $\grid{g}_{\ell}$ for $\ell < N$, the shift space $\mathcal{X}$ is actually a shift of finite type.

\subsubsection*{End of the proof}

Let us denote $\pi_{\alphabet} \colon \mathcal{X} \to \alphabet^{\Z^{d}}$ the projection of $\mathcal{X}$ its first component:
\begin{claim}
  The shift space $\mathcal{X}$ satisfies $\pi_{\alphabet}(\mathcal{X}) = \limspace{X}[\code{\tau}]$. In particular, $\limspace{X}[\code{\tau}]$ is a sofic shift.
\end{claim}

We prove both inclusions.
\begin{proof}[{Proof ($\limspace{X}_{\!\code{\tau}} \subseteq \pi_{\alphabet}(\mathcal{X})$)}]
  Let $x \in \limspace{X}[\code{\tau}]$ be a valid configuration. By definition, there exists a sequence of grids $(\grid{g}_{\ell})_{\ell \in \N}$ such that $(\rspace(\grid{g}_{\ell}),\rtime(\grid{g}_{\ell}))_{\ell \in \N}$ satisfies the \crtlnameref{fix:lem:staircase-lemma}; and a sequence of configurations $(x^{(\ell)})_{\ell \geq 0}$ such that $\smash{x^{(\ell+1)} \substep[\tau][\grid{g}_{\ell}] x^{(\ell)}}$ is computed in time $\smash{\K[t] \cdot \pxspace{\ell}[\alpha]}$, that the bit length $\norm{x^{(\ell)}}$ satisfies $\smash{\norm{x^{(\ell)}} \leq \K[t] \cdot \pxspace{\ell-1}[\alpha]}$ and that $\dec(x^{(0)}) = x$ (where, as usual, $\smash{\pxspace{\ell} = \prod_{i=0}^{\ell-1} \rspace(\grid{g}_{\ell})}$).

  We denote $\grid{g}_{\ell} = (\rect{r}^{(\ell)}_{\vec{i}})_{\vec{i} \in \Z^d}$. Up to re-indexing of the grids $\grid{g}_{\ell}$ and shifting the configurations $\smash{(x^{(\ell)})_{\ell \in \N}}$, we assume that $\smash{\vec{0} \in \rect{r}^{(\ell)}_{\vec{0}}}$ for every $\ell \in \N$. In particular, if $\pxgrid{g}_{\ell} = (\pxrect{r}_{\vec{i}})_{\vec{i} \in \Z^d}$ refers to any product grid $\smash{\pxgrid{g}_{\ell} = \bigcomp_{i=0}^{\ell-1} \grid{g}_{i}}$, then $\vec{0} \in \pxrect{r}_{\vec{0}}$.

  We now prove that there exists a configuration $\chi \in \mathcal{X}$ such that $\pi_{\alphabet}(\chi) = x$. To do so, consider the sequence $(\iota(n),\kappa(n))_{n \geq 1}$ given by the \crtlnameref{fix:lem:staircase-lemma}. We build in parallel a sequence of configurations $(\chi^{(\iota(n))})_{n \geq 1}$ corresponding to the different levels of fixed-point simulation, \textit{i.e.}~such that tiles $\wang{t}$ appearing in $\chi^{(\iota(n))}$ are accepted by a valid computation $\funram{e}(\dots,\iota(n),\dots,\wang{t})$.

  \medskip
  Let $n \in \N$. For a position $\vec{i} \in \Z^d$, we inductively define a sequence $f_{\iota(n)}(\vec{i}) = (\rect{r}_{\ell},\vec{i}_{\ell})_{\ell \geq \iota(n)}$: initializing $\vec{i}_{\iota(n)} = \vec{i}$, there exists for each $\ell \geq \iota(n)$ a unique rectangle $\rect{r}_{\ell}$ in $\grid{g}_{\ell}$ such that $\vec{i}_{\ell} \in \rect{r}_{\ell}$; we then set $\vec{i}_{\ell+1}$ to the index of the rectangle $\rect{r}_{\ell}$ in $\grid{g}_{\ell}$. Intuitively, $(\rect{r}_{\ell},\vec{i}_{\ell})_{\ell \geq \iota(n)}$ refers to the shapes and positions in the higher-level macro-tiles that contain the position $\vec{i} \in \Z^d$ from the level $\iota(n)$.

  We begin by fixing the as many fields in $\chi^{(\iota(n))}$ as we can: for $\vec{i} \in \Z^d$, let $\wang{t}$ refer to the tile appearing at position $\vec{i}$ in $\chi^{(\iota(n))}$, and denote $f_{\iota(n)}(\vec{i}) = (\rect{r}_{\ell},\vec{i}_{\ell})_{\ell \geq \iota(n)}$; then:
  \begin{itemize}
  \item The \field{bound field} of $\dec(\wang{t})$ is set to $(\rspace(\grid{g}_{\ell}),\rtime(\grid{g}_{\ell}))_{\kappa(n)-1 \leq \ell < \kappa(n+1)-1}$;
  \item The \field{position field} of $\dec(\wang{t})$ is set to $ (\rnorm{\rect{r}_{\ell}},\vec{i}_{\ell} \bmod \grid{g}_{\ell})_{\iota(n) \leq \ell < \kappa(n+1)-1}$;
  \item The \field{neighbor field} of $\dec(\wang{t})$ is set to $(\nghb_{\ell})_{\iota(n) \leq \ell < \kappa(n+1)}$, where $\nghb_{\ell}$ is the set of directions $\pm \basis{k}$ such that both positions $\vec{i}_{\ell-1}$ and $\vec{i}_{\ell-1} \pm \basis{k}$ belong to the same rectangle $\rect{r}_{\ell}$;
  \item The \field{p-level field} of $\dec(\wang{t})$ is set to $\iota(n+1)$.
  \end{itemize}
  The \field{bound} and \field{position fields} uniquely determine the \field{compzone fields} in $\chi^{(\iota(n))}$. Using our hypotheses on the configurations $(x^{(\ell)})_{\kappa(n) \leq \ell \leq \kappa(n+1)}$ and the computation of the substitution $\smash{x^{(\ell+1)} \substep[\tau][\grid{g}_{\ell}] x^{(\ell)}}$ for $\kappa(n) \leq \ell < \kappa(n+1)$, we furthermore entirely fix all \field{$\tau$-(\dots) fields} in $\chi^{(\iota(n))}$:
  \begin{itemize}
  \item The \field{position fields} of $\chi^{(\iota(n))}$ uniquely determine the $\tau$-macro-tiles of level $\ell$ in $\chi^{(\iota(n))}$ for every $\kappa(n) \leq \ell < \kappa(n+1)$;
  \item Such a $\tau$-macro-tile $\mwang{t}$ of level $\ell$ covering the position $\vec{i} \in \Z^d$ in $\chi^{(\iota(n))}$ should compute (a pixel of) the substitution $\smash{x^{(\ell+1)}_{\vec{i}_{\ell+1}} \substep[\tau] \restr{x^{(\ell)}}{\rect{r}_{\ell}}}$, where $(\rect{r}_{\ell},\vec{i}_{\ell})_{\ell \geq \iota(n)} = f_{\iota(n)}(\vec{i})$;
  \item In $\chi^{(\iota(n))}$, we arbitrarily fix a layout of valid \field{$\tau$-wires} and \field{$\tau$-consistency fields}; similarly, we arbitrarily fix a layout of \field{$\tau$-in} and \field{$\tau$-out pipes}, that is consistent with their counterparts in $\chi^{(\iota(n+1))}$ and $\chi^{(\iota(n-1))}$. The \crtlnameref{pos:lem:routing-lemma} ensures that all these wiring-based fields can be filed in a valid way inside the tilings~$(\chi^{(\iota(n))})_{n \geq 1}$.
  \end{itemize}

  \enlargethispage{2\baselineskip}
  We have yet to fix the \field{processor}, \field{color} and \field{wire fields} in $\chi^{(\iota(n))}$. To fix these, let us first consider some level $\iota(N)$ and the associated (partially defined) configuration $\chi^{(\iota(N))}$. For any $\vec{i} \in \Z^d$, consider the tile $\wang{t}$ at position $\vec{i}$ in $\chi^{(\iota(N))}$. Assume that there exists a completion of the \field{processor}, \field{color} and \field{wire fields} of $\wang{t}$ such that
  \[ \funram{e} \big((\K[][s],\K[][w]), \iota(N), (\rspace(\grid{g}_{\ell}),\rtime(\grid{g}_{\ell}))_{0 \leq \ell < \kappa(N)-1}, (\rnorm{\rect{r}_{\ell}},\vec{i}_{\ell} \bmod \grid{g}_{\ell})_{\iota(N) \leq \ell < \kappa(N)-1}, \wang{t} \big) \]
  admits an accepting computation of length at most $\K[][s] \cdot \pxspace{\iota(N)}[\gamma]$ (\textit{c.f.}~\Cref{fix:sec:fixing-constants}). Then this computation can be used to fix the \field{processor}, \field{color} and \field{wire fields} of the corresponding $\sigma$-macro-tile of level $\iota(N-1)$:
  \begin{itemize}
  \item More precisely, let $(\pxrect{r}_{\vec{j}})_{\vec{j} \in \Z^d}$ denote the rectangles of the product grid $\smash{\bigcomp_{\ell=\iota(N-1)}^{\iota(N)-1} \grid{g}_{\ell}}$, and let $\mwang{t}$ refer to the $\sigma$-macro-tile $\restr{\chi^{\iota(N-1)}}{\pxrect{r}_{\vec{i}}}$;
  \item We fix the \field{color fields} of $\mwang{t}$ to draw the colors of the tile $\wang{t}$;
  \item We fix the \field{processor fields} in $\mwang{t}$ to implement the accepting computation of $\wang{t}$ by $\funram{e}(\dots)$ of length $K \cdot \pxspace{\iota(N)}[\gamma]$; In particular, by design, the computation zone $\compzone$ of $\mwang{t}$ is large enough to embed the whole computation;
  \item We fix any valid layout for the \field{wire fields} in $\mwang{t}$, which is possible by the \crtlnameref{pos:lem:routing-lemma}.
  \end{itemize}
  \noindent This inductively defines a completion of these fields for increasingly large rectangles of tiles in the configurations $\chi^{(\iota(n))}$ for $\iota(n) < \iota(N)$.

  We now conclude the construction of the configurations $(\chi^{(\iota(n))})_{n \geq 1}$ by a compactness argument. Instead of considering a single Wang tile $\smash{\chi^{(\iota(N))}_{\vec{i}}}$ accepted by $\funram{e}$, consider a completion of the \field{processor}, \field{color} and \field{wire fields} of $\smash{\restr{\chi^{\iota(N)}}{\{-1,0,1\}^d}}$ (it is always possible to define a locally valid pattern in $\chi^{\iota(N)}$ of such a small size). Since we assumed that $\vec{0} \in \rect{r}_{\vec{0}}^{(\ell)}$ for every $\ell \in \N$, and that the substitution $\tau$ is \emph{expanding} (\textit{i.e.}~the rectangles in the grids $\grid{g}_{\ell}$ of edge length at least~$2$), this small pattern in the configuration $x^{(\iota(N))}$ recursively defines for every $\iota(n) < \iota(N)$ a locally valid completion of the \field{processor}, \field{color} and \field{wire fields} of the configuration $\chi^{(\iota(n))}$ in a domain that includes $\{-2^{\iota(N)-\iota(n)},\dots,2^{\iota(N)-\iota(n)}\}^d$. In particular, for increasing values of $\iota(N)$, the resulting partial (and valid) tiling in $x^{(\iota(n))}$ eventually fills the entire domain $\Z^d$.

  Using compactness and a diagonal argument, we deduce the existence of configurations $(\chi^{(\iota(n))})_{n \geq 1}$ such that, for every $n \geq 1$ and $\vec{i} \in \Z^d$,
  \[ \funram{e} \big((\K[][s],\K[][w]), \iota(n), (\rspace(\grid{g}_{\ell}),\rtime(\grid{g}_{\ell}))_{0 \leq \ell < \kappa(n)-1}, (\rnorm{\rect{r}_{\ell}},\vec{i}_{\ell} \bmod \grid{g}_{\ell})_{\iota(n) \leq \ell < \kappa(n)-1}, \chi^{\iota(n)}_{\vec{i}} \big) \]
  results in an accepting computation, where $(\rect{r}_{\ell},\vec{i}_{\ell})_{\ell \geq \iota(n)} = f_{\iota(n)}(\vec{i})$. By considering a valid configuration $x' \in \mathcal{X}$ implementing the initialization of the configuration $\chi^{(\iota(1))}$ and such that $\pi_{\alphabet}(x') = x$, we conclude that $x \in \pi_{\alphabet}(\mathcal{X})$, and thus that $\limspace{X}[\code{\tau}] \subseteq \pi_{\alphabet}(\mathcal{X})$.
\end{proof}

We now prove the converse inclusion:
\begin{proof}[{Proof ($\limspace{X}_{\!\code{\tau}} \supseteq \pi_{\alphabet}(\mathcal{X})$)}]
  Let $\chi \in \mathcal{X}$. By definition, $\pi_{\alphabet}(\chi)$ is a valid configuration in the sofic shift space $X$. We now prove that it belongs to $\limspace{X}_{\!\code{\tau}}$ by exhibiting a sequence of configurations $\smash{(x^{(\ell)})_{\ell \geq 0}}$ such that $\dec(x^{(0)}) = \pi_{\alphabet}(\chi)$ and $\smash{x^{(\ell+1)} \substep[\tau] x^{(\ell)}}$ satisfying the hypotheses of \Cref{sof:lem:fixpoint}.

  We build this sequence $(x^{(\ell)})_{\ell \geq 0}$ by defining two sequences $(\iota(n),\kappa(n))_{n \geq 1}$ as in the~\crtlnameref{fix:lem:staircase-lemma}, and inductively defining the segment $(x^{(\ell)})_{\kappa(n) \leq \ell \leq \kappa(n+1)}$.

  \paragraph*{Initialization}
  Consider the configuration $\chi \in \mathcal{X}$: by definition of $\mathcal{X}$, there exists some integers $\iota' < \kappa'$, a sequence of bounds $(\rtime_{\ell},\rspace_{\ell})_{0 \leq \ell < \kappa'}$, a sequence of grids $(\grid{g}_{\ell})_{0 \leq \ell < \kappa'}$ and a sequence of configurations $(x^{(\ell)})_{0 \leq \ell \leq \kappa'}$ such that $\smash{x^{(\ell+1)} \substep[\tau][\grid{g}_{\ell}] x^{(\ell)}}$ in time $\smash{\K[t] \cdot \pxspace{\ell}[\alpha]}$ (for $\pxspace{\ell} = \prod_{i < \ell} \rspace_{i}$). Furthermore, it defines a configuration $\smash{\chi' \in \Swang[\ast]^{\Z^d}}$ such that each tile $\chi'_{\vec{i}}$ is accepted by
  \[ \funram{e} \big((\K[][s],\K[][w]), \iota', (\rspace_{\ell},\rtime_{\ell})_{\iota' \leq \ell < \kappa'-1}, (\rect{r}_{\ell},\vec{i}_{\ell})_{\iota' \leq \ell < \kappa'-1}, \chi'_{\vec{i}} \big)\]
  for some $(\rect{r}_{\ell},\vec{i}_{\ell})_{0 \leq \ell < \kappa'-1}$ given by $f_{\iota(0)}(\vec{i})$. We denote $\iota(1) = \iota'$ and $\kappa(1) = \kappa'$.

  \paragraph*{Induction}
  Assume that $\chi'$ is a valid $\Swang[\ast]$-tiling whose every tiles are accepted by a computation of
  \[ \funram{e} \big((\K[][s],\K[][w]), \iota(n), (\rspace_{\ell},\rtime_{\ell})_{0 \leq \ell < \kappa(n)-1}, (\rect{r}_{\ell},\vec{i}_{\ell})_{\iota(n) \leq \ell < \kappa(n)-1}, \chi'_{\vec{i}} \big),\]
  where $(\rect{r}_{\ell},\vec{i}_{\ell})_{\iota(n) \leq \ell < \kappa(n)-1}$ is deduced from $f_{\iota(n)}(\vec{i})$ in the grids $(\grid{g}_{\ell})_{0 \leq \ell < \kappa(n)}$. Then by construction of $\funram{e}$, we have:
  \begin{itemize}
  \item The \field{bound fields} and \field{p-level fields} in the tiling $\chi'$ uniquely determine two integers $\iota' < \kappa'$ and the next segment of bounds $(\rspace_{\ell},\rtime_{\ell})_{\kappa(n)-1 \leq \ell < \kappa'-1}$. Furthermore, they ensure that the sequence $(\iota(m),\kappa(m))_{m < \kappa(n)}$ satisfies the hypotheses of the~\crtlnameref{fix:lem:staircase-lemma};
  \item The \field{position fields} in the tiling $\chi'$ uniquely determine a extension of the grid sequence $(\grid{g}_{\ell})_{0 \leq \ell < \kappa(n)-1}$ into $(\grid{g}_{\ell})_{0 \leq \ell < \kappa'-1}$;
  \end{itemize}

  Furthermore, the $\sigma$-macro-tiles in $\chi'$ define a grid $\pxgrid{g}_{\iota'} = (\pxrect{r}_{\vec{i}})_{\vec{i} \in \Z^d}$ and a new tiling $\chi'' \in \Swang[\ast]^{\Z^d}$. More precisely, for every $\vec{i} \in \Z^d$, let $\mwang{t}$ be the $\sigma$-macro-tile $\restr{\chi'}{\pxrect{r}_{\vec{i}}}$:
  \begin{itemize}
  \item The \field{color fields} in $\mwang{t}$ define the Wang tile $\chi''_{\vec{i}} \in \Swang[\ast]$;
  \item The \field{processor fields} in $\mwang{t}$ embed an accepting computation of the form
    \[ \funram{e} \big((\K[][s],\K[][w]), \iota', (\rspace_{\ell},\rtime_{\ell})_{0 \leq \ell < \kappa'-1}, (\rect{r}_{\ell},\vec{i}_{\ell})_{\iota' \leq \ell < \kappa'-1}, \chi''_{\vec{i}} \big)\]
    where $(\rect{r}_{\ell},\vec{i}_{\ell})_{\iota' \leq \ell < \kappa'-1}$ is deduced from the \field{position fields} of the tiles in $\mwang{t}$;
  \item And this computation outputs some tuple $(\nghb_{\kappa'-1},\rect{r}_{\kappa'-1},\vec{i}_{\kappa'-1},u)$ for $\nghb_{\kappa'-1} \subseteq \{\pm \basis{k} : 1 \leq k \leq d\}$, $\rect{r}_{\kappa'-1} \in \Srect_0$, $\vec{i}_{\kappa'-1} \in \rect{r}_{\kappa'}$ and $u \in \{0,1\}^{\ast}$.
  \end{itemize}
  Furthermore, the pair $(\rect{r}_{\kappa'-1},\vec{i}_{\kappa'-1})$ output by the $\sigma$-macro-tile $\restr{\chi'}{\pxrect{r}_{\vec{i}}}$ coincides with the eponymous pair from the \field{position field} of $\chi''_{\vec{i}}$; thus, since $\chi''$ is a valid tiling whose every tile is accepted by $\funram{e}(\dots,\iota',\dots)$, these pairs collectively define a valid grid $\grid{g}_{\kappa'-1}$.

  Consider now the $\tau$-macro-tiles of $\chi'$. For every $\kappa(n) \leq \ell < \kappa'$, the $\tau$-macro-tiles of level $\ell$ define a grid $\pxgrid{g}_{\ell} = (\pxrect{r}_{\vec{i}})_{\vec{i} \in \Z^d}$; and they define an image configuration~$x^{(\ell)}$ and a source configuration~$\tilde{x}^{(\ell+1)}$ as follows. For every $\vec{i} \in \Z^d$, let $\mwang{t}$ be the corresponding $\tau$-macro-tile $\smash{\restr{\chi'}{\pxrect{r}_{\vec{i}}}}$ of level~$\ell$:
  \begin{itemize}
  \item The \field{$\tau$-color fields} in $\mwang{t}$ define, on the border of $\mwang{t}$, a valid Wang tile $x^{(\ell)}_{\vec{i}} \in \Swang[\ast]$;
  \item The \field{$\tau$-consistency fields} in $\mwang{t}$ define, on the border of $\mwang{t}$, a valid Wang tile $\tilde{x}^{(\ell+1)}_{\vec{i}} \in \Swang[\ast]$;
  \item The \field{$\tau$-processor fields} in $\mwang{t}$ embed an accepting computation of $\funram{\code{\tau}}(\rect{r}_{\ell},\vec{i}_{\ell},\tilde{x}^{(\ell+1)}_{\vec{i}}) \rightarrow x^{(\ell)}_{\vec{i}}$, where $(\rect{r}_{\ell},\vec{i}_{\ell})$ is the eponymous pair from the \field{position fields} of the tiles in $\mwang{t}$;
  \end{itemize}
  Furthermore, the \field{$\tau$-wire fields} ensure that every configuration $x^{(\ell)}$ forms a valid Wang tiling\footnote{Notice that the configuration $\tilde{x}^{(\ell+1)}$ is not a valid tiling, since the \field{$\tau$-consistency fields} of adjacent $\tau$-macro-tiles of level $\ell$ will repeat the same symbol if they belong to the same macro-tile of level~${\ell+1}$.}; and for every $\kappa(n) \leq \ell < \kappa' - 1$, the \field{$\tau$-pipes} ensure that for every $\vec{i} \in \Z^d$, every tile $\smash{\wang{t} \subpattern \restr{\tilde{x}^{(\ell+1)}}{\rect{r}_{\vec{i}}}}$ is actually equal to $\smash{x^{(\ell+1)}_{\vec{i}}}$ (for $(\rect{r}_{\vec{i}})_{\vec{i} \in \Z^d} = \grid{g}_{\ell}$). In order to conclude, we are left with comparing the source configuration $\smash{\tilde{x}^{(\kappa')}}$ from the last step of substitution in $\chi'$ and the image configuration $\smash{x^{(\kappa')}}$ of the first substitution step implemented in $\smash{\chi''}$.

  More precisely, applying the previous paragraphs on $\chi''$, the \field{$\tau$-color fields} of the $\tau$-macro-tiles of level $\kappa'$ in $\chi''$ define a valid Wang tiling $\smash{x^{(\kappa')}}$. For notations, we introduce $(\pxrect{r}'_{\vec{i}})_{\vec{i} \in \Z^d}$ the rectangles of the product grid $\smash{\bigcomp_{\ell=\iota}^{\kappa'-1} \grid{g}_{\ell}}$ and $(\pxrect{r}''_{\vec{i}})_{\vec{i} \in \Z^d}$ those of product grid $\smash{\bigcomp_{\ell=\iota'}^{\kappa'-1} \grid{g}_{\ell}}$. For $\vec{i} \in \Z^d$, we prove that every tile $\smash{t \subpattern \restr{\tilde{x}^{(\kappa')}}{\rect{r}'_{\vec{i}}}}$ (\textit{i.e.}~symbol of level $\kappa'$ given by the $\tau$-macro-tiles of level $\kappa'-1$ in $\chi'$) is equal to $\smash{\tilde{x}^{(\kappa')}_{\vec{i}}}$ (symbol of level $\kappa'$ given by the $\tau$-macro-tiles of level $\kappa'$ in $\chi''$):
  \begin{itemize}
  \item Denote $\smash{(c_0,\dots,c_{2d}) = x^{(\kappa')}_{\vec{i}}}$, and let us fix some $(i,j) \in \interval{0}{2d} \times \interval{0}{|c_i|-1}$;
  \item In the configuration $\chi''$, the \field{$\tau$-out pipes} in the $\tau$-macro-tile $\smash{\restr{\chi''}{\pxrect{r}_{\vec{i}}''}}$ define a unique position $\vec{p}_{i,j} \in \pxrect{r}_{\vec{i}}''$ such that the \field{third $\tau$-out pipe field} of $\smash{\dec(\chi''_{\vec{p}_{i,j}})}$ contains an entry $((i,j,b),\mathtt{end})$ for some $b \in \{0,1\}$; in which case, $b = c_i(j)$;
  \item Consider the corresponding $\sigma$-macro-tile in $\chi'$: denoting $(\pxrect{r}_{\vec{i}})_{\vec{i} \in \Z^d}$ the rectangles of the grid product $\smash{\bigcomp_{\ell=\iota(n)}^{\iota'-1} \grid{g}_{\ell}}$, and considering the output of the algorithm $\funram{e}(\dots)$, we deduce that $\vec{p}_{i,j} \in \Z^d$ is the index of the unique $\sigma$-macro-tile $\smash{\restr{\chi'}{\pxrect{r}_{\vec{p}_{i,j}}}}$ (inside the domain $\pxrect{r}'_{\vec{i}}$) in which the \field{first $\tau$-in pipe fields} contain en entry $((i,j,b),\mathtt{start})$; in which case, $b = c_i(j)$;
  \item Following the \field{$\tau$-in pipes}, this bit is correctly wired to the input of the parent $\tau$-macro-tile of level $\kappa'-1$ in $\chi'$.
  \end{itemize}
  Since this argument applies for all bits of indices $(i,j) \in \interval{0}{2d} \times \interval{0}{|c_i|-1}$, and that the \field{$\tau$-consistency fields} ensure, inside the pattern $\smash{\restr{\chi'}{\pxrect{r}'_{\vec{i}}}}$, that all the $\tau$-macro-tiles of level $\kappa'-1$ have the same input symbol, we conclude that for every $\vec{i} \in \Z^d$, every tile $\smash{\wang{t} \subpattern \restr{\tilde{x}^{(\kappa')}}{\rect{r}'_{\vec{i}}}}$ is equal to~$\smash{x^{(\kappa')}_{\vec{i}}}$. We conclude this induction step by setting $\iota(n+1) = \iota'$ and $\kappa(n+1) = \kappa'$.

  \bigskip
  By induction, we conclude that there exists a sequence $(x^{(\ell)})_{\ell \geq 0}$ such that $\smash{x^{(\ell+1)} \substep[\tau][\grid{g}_{\ell}] x^{(\ell)}}$, and computed in time $\smash{\K[t] \cdot \pxspace{\ell}[\gamma]}$. Since the sequence $(\rspace_{\ell},\rtime_{\ell})_{\ell \in \N}$ satisfies the growth conditions of \Cref{sof:lem:fixpoint}, so does the sequence $(\rspace_{\ell}(\grid{g}_{\ell}),\rtime_{\ell}(\grid{g}_{\ell}))_{\ell \in \N}$; and we conclude that $\dec(x^{(0)}) = x$ is a valid configuration in $\limspace{X}[\code{\tau}]$.
\end{proof}

\enlargethispage{2\baselineskip}
\subsection{Final word}

We now make a few comments on \Cref{sof:lem:fixpoint} and its proof.

\paragraph*{Pre-composition vs.\ parallel computations}

One of the main difficulties of this proof (compared to previous fixed-point constructions in the literature) comes from simulating several substitution steps $\smash{x^{(\kappa(n+1))} \substep \dots \substep x^{(\kappa(n))}}$ at a single level of fixed-point simulation $\iota(n)$. This is actually necessary because the substitution $\tau$ can output arbitrarily small patterns, \textit{e.g.}~of domain~$\ibox{2}^d$.

Instead of drawing multiple levels of $\tau$-macro-tiles in parallel, one could prefer to make single $\tau$-macro-tiles compute several compositions of $\tau$ with itself algorithmically (informally, using a standard fixed-point construction computing some power $\tau^{n}$ instead of $\tau$), thus ensuring that the resulting patterns grow fast enough to apply a step of fixed-point simulation. We claim that, while this choice is definitely possible, it results in significantly worse bounds on the possible growths of the grids $(\grid{g}_{\ell})_{\ell \in \N}$ and the computational requirements on $\tau$.

Indeed, consider a macro-tile $\mwang{t}$ of level $\ell$ and domain $\pxrect{r}_{\ell} \in \Srect_0$. If we are to compute some $n \in \N$ steps of $\tau$ inside $\mwang{t}$, the embedding of \Cref{calc:lem:pa-ram-simulation} only allows to draw $\bigO(\rspace(\pxrect{r}_{\ell})^{\alpha})$ steps of RAM computations in $\mwang{t}$. In particular, this limits the word length of the symbols in $x^{(\ell+n)}$ to $\norm{x^{(\ell+n)}} = \bigO(\pxspace{\ell}[\alpha])$, instead of the more permissive $\bigO(\pxspace{\ell+n-1}[\alpha])$ from Condition 2.(ii) in \Cref{sof:lem:fixpoint}.

\paragraph*{Algorithmic vs.\ geometric computations} One may also wonder why we actually consider a computable parallel substitution $\funram{\code{\tau}}$ instead of a more classical notion of computability on substitutions. The answer is, once again, that it allows for better bounds in \Cref{sof:lem:fixpoint}.

For fixed $n \in \N$, assume that we are designing a finite tileset $\Swang$ over the colors $\{0,1\}^n$ in which individual tiles can embed $t$ steps of RAM computations. In a finite cube $\ibox{N}^d$ (for $N \leq 2^n$), there are two possibilities to embed computations in the patterns of $\smash{\Swang^{\ibox{N}^d}}$:
\begin{itemize}
\item On the one hand, the individual tiles in $\ibox{N}^d$ can perform independent computations (with, maybe, some small $\bigO(n)$ communications between adjacent tiles). Globally, \emph{such a pattern will embed
  $t \cdot N^d$
  steps of RAM computations}, at the cost of distributing them as independent segments of length $t$;
\item On the other hand, the tiles in $\ibox{N}^d$ can collectively perform the RAM simulation of \Cref{calc:lem:pa-ram-simulation}. This embedding is based upon drawing the \emph{space-time diagram} of a processor array, thus requiring the use of a dimension to act as the ``time'' of the computation. Since this processor array is limited by its size $\bigO(N^{d-1})$ to embed RAM computations, \emph{such a pattern will only embed
  $t \cdot N^{d-1}$
  steps of RAM computations}. Nevertheless, these computations may be performed by a single (non-distributed) machine, thus allowing for as many steps of sequential computations.
\end{itemize}
Distributing the computations of $\tau$ into independent pixels thus reduces the size of the output of~$\funram{\code{\tau}}$, and in turn allows more flexibility on the computation times of $\code{\tau}$ and the growth of the grids $(\grid{g}_{\ell})_{\ell \in \N}$.

\paragraph*{Recognizable $S$-adic structure}
Let us briefly explain why the $S$-adic structures of the cover $\mathcal{X}$ and the limit space $\limspace{X}{\code{\tau}}$ (\textit{c.f.}~\Cref{sof:par:comment-structure}) is recognizable.
As seen in the proof by double inclusion in \Cref{fix:sec:final-sft}, any configuration $\chi$ in the finite-type cover $\mathcal{X}$ uniquely and inductively determines two functions $\iota,\kappa \colon \N \to \N$, and some sequences of grids $(\grid{g}_{\ell})_{\ell \in \N}$, configurations $(\chi^{(\iota(n))})_{n \geq 1}$ and simulations $(\sigma_{\iota(n)})_{n \in \N})$ such that:
\[ \chi \rsubstep[\sigma_{\iota(0)}^{-1}][\grid{g}_{\iota(1)-1} \circ \dots \circ \grid{g}_0] \chi^{(\iota(1))} \rsubstep[\sigma_{\iota(1)}^{-1}][\grid{g}_{\iota(2)-1} \circ \dots \circ \grid{g}_{\iota(1)}] \chi^{(\iota(2))}  \rsubstep[\sigma_{\iota(2)}^{-1}][\grid{g}_{\iota(3)-1} \circ \dots \circ \grid{g}_{\iota(2)}] \chi^{(\iota(3))} \dots \]
In particular, the configurations $\chi^{(\iota(n))}$ and the grid products $\grid{g}_{\iota(n)-1} \circ \dots \circ \grid{g}_{\iota(n-1)}$ form a sequence of simulated configurations by the cover configuration $\chi \in \mathcal{X}$, hence a recognizable $S$-adic structure.

Furthermore, let us denote $x = \pi_{\alphabet}(\chi) \in \limspace{X}[\code{\tau}]$. By definition of $\mathcal{X}$, every configuration $\chi^{(\iota(n))}$ draws (and thus uniquely determines) some configurations $x^{(\kappa(n))},\dots,x^{(\kappa(n+1)-1)}$ such that $x = \dec(x^{(0)})$ and, altogether, the sequence
\[ x^{(0)} \rsubstep[\tau][\grid{g}_{0}] x^{(1)} \rsubstep[\tau][\grid{g}_1] x^{(2)} \rsubstep[\tau][\grid{g}_2] x^{(3)} \dots\]
is a valid $S$-adic structure for $x^{(0)}$. Denoting by $\tilde{x}^{(\ell)}$ the ``zoomed'' configuration given by $x^{(\ell)}$ and the grids $(\grid{g}_k)_{0 \leq k < \ell}$, \textit{i.e.}~the configuration such that, for every $\vec{p} \in \Z^d$, we have $(\tilde{x}^{(\ell)})_{\vec{p}} = (x^{(\ell)})_{\vec{i}}$, where $\vec{i}$ is the unique rectangle index in the grid product $\smash{\pxgrid{g}_{\ell} = \bigcirc_{k=1}^{\ell} \grid{g}_{k}}$ such that $\vec{p} \in \pxgrid{g}_{\ell}(\vec{i})$; then the configuration $\chi$ factors onto every $\smash{\tilde{x}^{(\ell)}}$ by mapping the $\tau$-macro-tiles of level $\ell$ towards their symbol in $x^{(\ell)}$.

Since every macro-tile explicitly stores its position in the sequence of grids $(\grid{g}_{\ell})_{\ell \in \N}$, the shape and the symbol of a macro-tile can be determined locally. Furthermore, since for fixed $\alpha < 1$ and $\delta < \frac{1}{2}$ the $\tau$-macro-tiles of level $\ell$ are all of uniformly bounded radius, this mapping can be implemented as a factor map $\smash{\pi^{(\ell)} \colon X \to \Swang[\ast]^{\mspace{2mu}\Z^d}}$. We thus deduce that any configuration $\chi \in \mathcal{X}$ uniquely determines an $S$-adic structure (\textit{i.e.}~grids $(\grid{g}_{\ell})_{\ell \in \N}$ and configurations $(x^{(\ell)})_{\ell \in \N}$) for the configuration $\pi_{\alphabet}(\chi) \in \limspace{X}[\code{\tau}]$, which can be recovered by the factor maps $\pi^{(\ell)}$ that recognize the macro-tiles of level $\ell$.

\smallskip
In addition to the recognizability of the $S$-adic structures, the other inclusion in \Cref{fix:sec:final-sft} proves that any $S$-adic structure for $x$ using the substitution $\tau$ and that satisfies the hypotheses of \Cref{sof:lem:fixpoint} can be implemented as a cover configuration in $\mathcal{X}$. Thus, we conclude that there exists a correspondence between the $S$-adic structures of a configuration $x \in \limspace{X}[\code{\tau}]$ and the structure given by the factor maps $(\pi^{(\ell)})_{\ell \geq 0}$ on the preimages $\chi \in \pi_{\alphabet}^{-1}(x)$.

\paragraph*{Entropy of the construction}
Let us sketch a justification for \Cref{sof:par:comment-entropy}. In order to study the pattern complexity of the construction, we consider below its main sources of non-determinism:
\begin{itemize}
\item The embedded $\log$-RAM simulations from \Cref{calc:lem:pa-ram-simulation} create exactly one pattern per valid run of the program $\code{\tau}$. In particular, if the computations of $\code{\tau}$ are unambiguous, then the embedded computations will not increase the entropy of the construction;
\item The \crtlnameref{pos:lem:routing-lemma} may result in several possible wirings for a given set of end points, thus increasing the overall entropy.
  This can actually be avoided by considering its deterministic variant (\textit{c.f.}~\Cref{pos:rem:deterministic-routing-lemma}).
\item The \field{$\tau$-out pipes} may also result in several wirings. This issue can also be fixed by implementing a fixed layout for these pipes to follow (\textit{e.g.}~depending on the position of the $\sigma$-macro-tiles in each $\tau$-macro-tile of level $\kappa$).
\end{itemize}

We could thus define a slightly more specified cover $\mathcal{X}$ whose pattern complexity corresponds -- in terms of exponential growth -- to the number of patterns in $\limspace{X}$, counted with the multiplicity of their ambiguous $\code{\tau}$-computations and all their valid $S$-adic structures.
Note that the contribution of the latter becomes trivial when the directive sequence is assumed recognizable (and is polynomial when the substitutions are assumed injective over finite patterns, because placing a macro-tiles fixes its whole sub-hierarchy).

%
\section{Perspectives}
\label{persp:sec}

Future studies about the~\crtlnameref{sof:lem:fixpoint} will focus on better understanding the resulting finite-type cover. By going beyond the informal discussions of \Cref{sof:par:comment-constructivism} and \Cref{sof:par:comment-structure} towards expliciting the sequence of simulations (\textit{i.e.}~dill maps) that defines this cover, we could hope to prove its minimality and unique ergodicity (provided that the same properties hold for the original $N$-adic space), or that it has nearly the same pattern complexity (up to a polynomial factor?): is the projection map almost $1$-to-$1$ in some sense? is it the case at each level?

Some other properties, however, are not that easily preserved between sofic shifts and their covers.
For example, the second author and Charalampos Zinoviadis have studied directions of determinism in multidimensional shifts of finite type (see~\cite{Zinoviadis_2016:phd:expansiveness-2d-sfts}): they also defined suitable $N$-adic structures, but relied on a suitable computation model (partial reversible cellular automata) to ensure that configurations in the cover preserve the same directions of determinism.

As mentioned in \Cref{sof:par:comment-entropy}, we prove that a huge class of sofic shifts (\textit{e.g.}~$N$-adic shifts with injective ndill maps) have a cover of equal entropy. It is an open question whether every sofic shift admits such a cover. Some candidates for not having one (like the odd shift from~\cite{Cassaigne-Guillon_misc:odd-shift}) do not seem to admit an efficiently computable $N$-adic structure of low pattern complexity. Our construction can also be interpreted as a sufficient condition for multidimensional soficity in terms of bound on recursive information contained in the patterns/configurations of a shift; and while similar necessary conditions exist \cite{Guillon-Jeandel_2015:infinite-communication-complexity,Destombes-Romashchenko_2022:ressource-bounded-Kolmogorov-complexity-soficness-multidimensional-shifts}, a gap still remains when the communication complexity between two zones is exponential in their boundary.

Finally, a mysterious question would be to better understand what obstacles exist when generalizing the fixed-point construction from $\Z^d$ to other Cayley graphs, or other forms of ``substitutive'' graphs.
Note that an $S$-adic formalism also exists in all countable groups \cite{Bitar-Cabezas-Guillon_misc:substitutions-and-hierarchical-structures-in-countable-groups}, and that soficity of substitutive limit sets has been proven in a more general graph setting \cite{PavietSalomon_2024:phd:undecidability-geometric-invariants-tilings}.

\printbibliography
\end{document}
